\documentclass[11pt,a4paper]{article}
\usepackage[utf8]{inputenc}
\usepackage{lmodern} %

\ifdefined\FOCS\else
\usepackage[bookmarks,colorlinks,breaklinks]{hyperref}
\hypersetup{urlcolor=blue, colorlinks=true, citecolor=green!50!black, linkcolor=blue}
\fi
\usepackage{multirow}
\usepackage[T1]{fontenc} %
\usepackage{amsmath}
\usepackage{amsfonts}
\usepackage{amssymb}
\usepackage{amsthm}
\usepackage{thmtools}
\usepackage{pgffor}
\usepackage{xcolor}
\usepackage[lined,boxed,ruled,norelsize,algo2e,linesnumbered,noend]{algorithm2e}
\usepackage{cleveref}
\usepackage{graphicx}
\ifdefined\FOCS
\else
\usepackage{geometry}
\fi
\usepackage{enumitem}
\usepackage{url} %

\makeatletter
\newcommand\footnoteref[1]{\protected@xdef\@thefnmark{\labelcref{#1}}\@footnotemark}
\makeatother
\makeatletter
\renewcommand{\paragraph}{%
	\@startsection{paragraph}{4}%
	{\z@}{1.25ex \@plus 1ex \@minus .2ex}{-1em}%
	{\normalfont\normalsize\bfseries}%
}
\makeatother

\ifdefined\DEBUG
	\usepackage{showlabels}
	\def\jan#1{\textcolor{red}{JAN: #1}}
	\def\richard#1{\textcolor{green}{RICHARD: #1}}
	\def\danupon#1{\textcolor{orange}{DN: #1}}
	\def\thatchaphol#1{\textcolor{purple}{TS: #1}}
	\newcommand{\Zhao}[1]{{\color{red}[Zhao: #1]}}
	\newcommand{\yintat}[1]{{\color{red}[YinTat: #1]}}
	\def\wadi#1{\textcolor{blue}{DW: #1}}
	\definecolor{darkgreen}{rgb}{0.0, 0.42, 0.24}
	\newcommand{\sidford}[1]{{\color{darkgreen}[Aaron: #1]}}

\else

	\def\jan#1{}
\def\richard#1{}
\def\danupon#1{}
\def\thatchaphol#1{}
\newcommand{\Zhao}[1]{}
\def\wadi#1{}
\newcommand{\sidford}[1]{}
\newcommand{\yintat}[1]{}
	
\fi

\declaretheorem[numberwithin=section,refname={Theorem,Theorems},Refname={Theorem,Theorems}]{theorem}
\declaretheorem[numberlike=theorem]{lemma}
\declaretheorem[numberlike=theorem]{proposition}
\declaretheorem[numberlike=theorem]{corollary}
\declaretheorem[numberlike=theorem]{definition}
\declaretheorem[numberlike=theorem]{claim}

\newcommand{\R}{\mathbb{R}}
\newcommand{\Rn}{\mathbb{R}^{n}}
\newcommand{\N}{\mathbb{N}}
\newcommand{\Z}{\mathbb{Z}}
\renewcommand{\S}{\mathbb{S}}%

\renewcommand{\P}{\mathbb{P}}
\newcommand{\E}{\mathbb{E}}

\renewcommand{\tilde}{\widetilde}
\renewcommand{\hat}{\widehat}
\renewcommand{\bar}{\overline}

\newcommand{\median}{\operatorname{median}}

\DeclareMathOperator*{\argmin}{argmin}  %
\DeclareMathOperator*{\argmax}{argmax}  %
\newcommand{\poly}{\operatorname{poly}}
\newcommand{\polylog}{\operatorname{polylog}}
\newcommand{\mdiag}{\mathbf{Diag}} %

\newcommand{\im}{\operatorname{im}}
\newcommand{\grad}{\nabla}%
\newcommand{\nnz}{\operatorname{nnz}}%
\newcommand{\cnorm}{C_{\mathrm{norm}}}%

\newcommand{\Ot}{{\tilde O}}
\newcommand{\otilde}{\Ot}

\newcommand{\zerovec}{\vec{0}} %
\newcommand{\onevec}{{\vec{1}}} %
\newcommand{\mzero}{\boldsymbol{0}}%

\newcommand{\mDelta}{\Delta}%
\newcommand{\mproj}{\mathbf{P}}%
\newcommand{\mSigma}{\boldsymbol{\Sigma}}%
\newcommand{\mLambda}{\boldsymbol{\Lambda}}%
\foreach \x in {A,...,Z}{%
	\expandafter\xdef\csname m\x\endcsname{\noexpand\mathbf{\x}}
}
\newcommand{\ma}{\mathbf{A}}%
\newcommand{\mb}{\mathbf{B}}%
\newcommand{\md}{\mathbf{D}}%
\newcommand{\mi}{\mathbf{I}}%
\newcommand{\mm}{\mathbf{M}}%
\newcommand{\mn}{\mathbf{N}}%
\newcommand{\mr}{\mathbf{R}}%
\newcommand{\ms}{\mathbf{S}}%
\newcommand{\mq}{\mathbf{Q}}%
\newcommand{\mv}{\mathbf{V}}%
\newcommand{\mw}{\mathbf{W}}%
\newcommand{\mx}{\mathbf{X}}%
\newcommand{\mz}{\mathbf{Z}}%

\newcommand{\otau}{\noexpand{\overline{\tau}}}
\newcommand{\omu}{\noexpand{\overline{\mu}}}
\newcommand{\osigma}{\noexpand{\overline{\sigma}}}%
\newcommand{\os}{\noexpand{\overline{s}}}
\newcommand{\ot}{\noexpand{\overline{t}}}
\newcommand{\og}{\noexpand{\overline{g}}}
\newcommand{\ou}{\noexpand{\overline{u}}}
\newcommand{\ov}{\noexpand{\overline{v}}}
\newcommand{\ow}{\noexpand{\overline{w}}}
\newcommand{\ox}{\noexpand{\overline{x}}}
\newcommand{\oy}{\noexpand{\overline{y}}}
\newcommand{\oz}{\noexpand{\overline{z}}}

\foreach \x in {A,...,Z}{%
	\expandafter\xdef\csname c\x\endcsname{\noexpand\mathcal{\x}}
}

\newcommand{\ttau}{\widetilde{\tau}}
\newcommand{\tPi}{\widetilde{\Pi}}
\foreach \x in {a,...,z}{%
	\expandafter\xdef\csname t\x\endcsname{\noexpand\widetilde{\x}}
}

\foreach \x in {A,...,Z}{%
	\expandafter\xdef\csname om\x\endcsname{\noexpand\mathbf{\overline{\x}}}
}
\newcommand{\omw}{\omW}
\newcommand{\omx}{\omX}%
\newcommand{\oma}{\omA}%
\newcommand{\oms}{\omS}%
\newcommand{\omh}{\omH}%

\foreach \x in {A,...,Z}{%
	\expandafter\xdef\csname tm\x\endcsname{\noexpand\mathbf{\widetilde{\x}}}
}

\renewcommand{\t}{^{(t)}}

\newcommand{\tmp}{{\mathrm{(tmp)}}} %
\newcommand{\safe}{_{(\mathrm{safe})}}%
\newcommand{\new}{\mathrm{(new)}}%
\newcommand{\init}{\mathrm{(init)}}%
\newcommand{\target}{\mathrm{(end)}}%
\newcommand{\sample}{\mathrm{sample}} %

\newcommand{\LineComment}[1]{\tcc{#1}}
\newcommand{\State}[0]{} %
\newcommand{\Comment}[1]{\tcp*[h]{#1}}

\newcommand{\VecMaintainer}{\textsc{VectorMaintainer}\xspace}

\newcommand{\ProjHeavyHitter}{\textsc{HeavyHitter}\xspace}

\newcommand{\CAP}{\vec{\mathrm{cap}}}
\newcommand{\opt}{\mathrm{OPT}}
\newcommand{\OPT}{\mathrm{OPT}} %

\newcommand{\norm}[1]{\|#1\|}
\newcommand{\normFull}[1]{\left\Vert #1\right\Vert }
\newcommand{\defeq}{\stackrel{\mathrm{{\scriptscriptstyle def}}}{=}}
\newcommand{\vones}{\onevec}

\newcommand{\phicent}{\Phi}%

\newcommand{\taulog}{\tau_{\mathrm{log}}}
\newcommand{\tauLS}{\tau_{\mathrm{LS}}}

\allowdisplaybreaks %

\title{Bipartite Matching in Nearly-linear Time\\ on Moderately Dense Graphs }

\author{
Jan van den Brand\thanks{\texttt{janvdb@kth.se}. KTH Royal Institute of Technology, Sweden.}
\and
Yin Tat Lee\thanks{\texttt{yintat@uw.edu}. University of Washington and Microsoft Research Redmond, USA.}
\and
Danupon Nanongkai\thanks{\texttt{danupon@kth.se}. KTH Royal Institute of Technology, Sweden.}
\and
Richard Peng\thanks{\texttt{rpeng@cc.gatech.edu}. Georgia Institute of Technology, USA.}
\and
Thatchaphol Saranurak\thanks{\texttt{saranurak@ttic.edu}. Toyota Technological Institute at Chicago, USA}
\and
Aaron Sidford\thanks{\texttt{sidford@stanford.edu}. Stanford University, USA.} 
\and
Zhao Song\thanks{\texttt{zhaos@ias.edu}. Columbia University, Princeton University and Institute for Advanced Study, USA.}
\and 
Di Wang\thanks{\texttt{wadi@google.com}. Google Research, USA.}
}

\begin{document}
	
\begin{titlepage}
	\maketitle
	
	\pagenumbering{roman}
	We present an $\tilde O(m+n^{1.5})$-time randomized algorithm for maximum cardinality 
bipartite matching and related problems (e.g. transshipment, negative-weight shortest 
paths, and optimal transport) on $m$-edge, $n$-node graphs. For maximum cardinality  
bipartite matching on moderately dense graphs, i.e. $m = \Omega(n^{1.5})$, our 
algorithm runs in time nearly linear in the input size and constitutes the first 
improvement over the classic $O(m\sqrt{n})$-time [Dinic 1970; Hopcroft-Karp 1971; 
Karzanov 1973] and $\tilde O(n^\omega)$-time algorithms [Ibarra-Moran 1981; Mucha-Sankowski 2004] 
(where currently $\omega\approx 2.373$). On sparser graphs, i.e. when $m =  n^{9/8 + \delta}$ 
for any constant $\delta>0$, our result improves upon the recent advances of [Madry 
2013] and [Liu-Sidford 2020b, 2020a] which achieve an 
$\tilde O(m^{4/3+o(1)})$ runtime.

We obtain these results by combining and advancing recent lines of research in interior point methods (IPMs)  and dynamic graph algorithms. First, we simplify and improve the IPM of [v.d.Brand-Lee-Sidford-Song 2020], providing a general primal-dual IPM framework and new sampling-based techniques for handling infeasibility induced by  approximate linear system solvers. Second, we provide a simple sublinear-time algorithm for detecting and sampling high-energy edges in electric flows on expanders and show that when combined with recent advances in dynamic expander decompositions, this yields efficient data structures for maintaining the iterates of both [v.d.Brand~et~al.] and our new IPMs. 
Combining this general machinery yields a simpler $\tilde O(n \sqrt{m})$ time algorithm for matching based on the logarithmic barrier function, and 
our state-of-the-art $\tilde O(m+n^{1.5})$ time algorithm for matching based on the [Lee-Sidford 2014] barrier (as regularized in [v.d.Brand~et~al.]).

	\newpage
	\setcounter{tocdepth}{2}
	\tableofcontents
	
\end{titlepage}

\newpage
\pagenumbering{arabic}

\newpage

\section{Introduction}
\label{sec:intro}

The {\em maximum-cardinality bipartite matching} problem is to compute a matching of maximum size in an $m$-edge $n$-vertex bipartite graph $G=(V,E)$. This problem is one of the most fundamental and well-studied problems in combinatorial optimization, theoretical computer science, and operations research. It naturally encompasses a variety of practical assignment questions and is closely related to a wide range of prominent optimization problems, e.g. optimal transport, shortest path with negative edge-lengths, minimum mean-cycle, etc.

Beyond these many applications, this problem has long served as a barrier towards efficient optimization and a proving ground for new algorithmic techniques. Though numerous combinatorial and continuous approaches have been proposed for the problem, improving upon the classic time complexities of $O(m \sqrt{n})$~\cite{HopcroftK73,Dinic70,Karzanov73} and $O(n^\omega)$~\cite{IbarraM81,MuchaS04}\footnote{$\omega$ is the matrix multiplication exponent and currently $\omega\approx 2.373$ \cite{Gall14a,Williams12}} has proven to be notoriously difficult.
Since the early 80s, these complexities have only been improved by polylogarithmic factors (see, e.g., \cite{schrijver2003combinatorial})
until a breakthrough result of Madry \cite{m13}  showed that faster algorithms could be achieved when the graph is {\em moderately sparse}. In particular, Madry \cite{m13} showed that the problem could be solved in $\tilde{O}(m^{10/7})$ and a line of research \cite{m16,ls19,cmsv17,ls20_stoc,amv20} led to the recent $\tilde{O}(m^{4/3+o(1)})$-time algorithm \cite{ls20_focs}.\footnote{For simplicity, in the introduction we use $\tilde{O}(\cdot)$ to hide $\polylog n$ and sometimes $\polylog(W)$, where $W$ typically denotes the largest absolute value used for specifying any value in the problem. In the rest of the paper, we define $\tilde O$ (in \Cref{sec:preliminaries}) to hide $\polylog(n)$ and $\poly\log\log W$ (but not $\polylog(W)$). Further, throughout the introduction when we state runtimes for an algorithm, the algorithm may be randomized and these runtimes may only old w.h.p.}
Nevertheless, for {\em moderately dense graphs}, i.e. $m\geq n^{1.5+\delta}$ for any constant $\delta>0$, 
the $O(\min(m \sqrt{n}, n^\omega))$ runtime bound has remained unimproved for decades. 

The more general problem of {\em minimum-cost perfect bipartite $b$-matching}, where an edge can be used multiple times and the goal is to minimize the total edge costs in order to match every node $v$ for $b(v)$ times, for given non-negative integers $b(v)$, has been even more resistant to progress. 
An $\tilde{O}(m\sqrt{n})$ runtime for this problem with arbitrary polynomially bounded integer costs and $b$ was achieved only somewhat recently by \cite{ls14}. Improving this runtime by even a small polynomial factor for moderately dense graphs, is a major open problem (see \Cref{tab:intro:min-weight b-matching}).

The minimum-cost perfect bipartite $b$-matching problem can encode a host of problems ranging from transshipment, to negative-weight shortest paths, to  maximum-weight bipartite matching. Even more recently, the problem has been popularized in machine learning through its encapsulation of the optimal transport problem \cite{BlanchetJKS18,Quanrud19,LinHJ19,AltschulerWR17}. There has been progress on these problems in a variety of settings (see \Cref{sec:intro:results} and \Cref{sec:matching} and the included \Cref{tab:intro:max-C matching} through   
\Cref{table:results:perfect matching}), including recent improvements for sparse-graphs \cite{cmsv17, amv20}, and nearly linear time algorithms for computing $(1+\epsilon)$ approximate solutions for maximum-cardinality/weight matching \cite{HopcroftK73,Dinic70,Karzanov73,GabowT89,GabowT91,DuanP14} and undirected transshipment \cite{Sherman17,AndoniSZ19,Li19}. However, obtaining {\em nearly linear time} algorithms for solving these problems to high-precision for any density regime has been elusive.

\subsection{Our Results} \label{sec:intro:results}

In this paper, we show that maximum cardinality bipartite matching, and more broadly minimum-cost perfect bipartite $b$-matching, can be solved in  $\tilde O(m+n^{1.5})$ time. \Cref{tab:intro:max-C matching,tab:intro:min-weight b-matching} compare these results with previous ones. 
Compared to the state-of-the-art algorithms for maximum-cardinality matching our bound is the fastest whenever $m =  n^{9/8 + \delta}$ for any constant $\delta>0$, ignoring polylogarithmic factors. Our bound is the first (non-fast-matrix multiplication based) improvement in decades for the case of dense graphs. 
More importantly, our bound is nearly linear when $m \geq n^{1.5}$.
This constitutes the first near-optimal runtime in any density regime for the bipartite matching problem.

\begin{table}[t]
	\begin{tabular}{|l|c|p{0.23\textwidth}|c|c|}
		\hline
		\multirow{2}{*}{\textbf{Year}} & \multirow{2}{*}{\textbf{Authors}} & \multirow{2}{*}{\textbf{References}}                   & \multicolumn{2}{c|}{\textbf{Time ($\tilde O(\cdot)$)}}             \\ \cline{4-5} 
		&                                   &                                                        & \textbf{Sparse} & \textbf{Dense} \\ \hline
		1969-1973                      &  Hopcroft, Karp, Dinic, Karzanov   & \cite{HopcroftK73,Dinic70,Karzanov73} & \multicolumn{2}{c|}{$m\sqrt{n}$}               \\ \hline
		1981, 2004 & Ibarra, Moran, Mucha, Sankowski & \cite{IbarraM81,MuchaS04}	&                                  &                  $n^\omega$               \\ \hline
		2013 & Madry&\cite{m13}	&    $m^{10/7}$                              &                                 \\ \hline
		2020& Liu, Sidford & \cite{ls20_stoc} & $m^{11/8+o(1)}$ & \\ \hline
		2020 & Liu, Sidford &\cite{ls20_focs} & $m^{4/3+o(1)}$ & \\ \hline
		2020 & \bf This paper & & & $m+n^{1.5}$\\\hline
	\end{tabular}
	\caption{The summary of the results for the {\bf max-cardinality bipartite matching} problem. For a more comprehensive list, see \cite{DuanP14}.}\label{tab:intro:max-C matching}
\end{table}

\begin{table}[t]
	\begin{center}
		\begin{tabular}{ | l | l | l | l | l | l | }
			\hline
			{\bf Year} & {\bf Authors} & {\bf References} &  \bf Time ($\tilde O(\cdot)$) \\ \hline
			1972 & Edmonds and Karp  & \cite{EdmondsK72} & $mn$  \\ \hline					
			2008 & Daitch, Spielman & \cite{ds08} & $m^{3/2}$ \\ \hline
			2014 & Lee, Sidford & \cite{ls14} & $m\sqrt{n}$ \\ \hline
			2020 & \bf This paper & & $m + n^{1.5}$ \\ \hline
		\end{tabular}
	\end{center}\caption{The summary of the results for the {\bf min-cost perfect bipartite $b$-matching}
		problem (equivalently, {\bf transshipment}) for polynomially bounded integer costs and $b$.
		For a more comprehensive list, see Chapters 12 and 21 in \cite{schrijver2003combinatorial}. 
		Note that there have been further runtime improvements to this problem (not included in the table) under the assumption that $\|b\|_1 = O(m)$, see \cite{cmsv17}. Recently, a state-of-the-art runtime of $\tilde{O}(m^{4/3 + o(1)})$ was achieved by \cite{amv20} under this assumption.
	}\label{tab:intro:min-weight b-matching}
\end{table}

\begin{table}[ht]
	\begin{center}
		\begin{tabular}{ | l | l | l | l | l | l | }
			\hline
			{\bf Year} & {\bf Authors} & {\bf References} & {\bf Time ($\tilde O(\cdot)$)}  \\ \hline
			1972 & Edmonds, Karp & \cite{EdmondsK72} & $n^3$ \\ \hline
			2014 & Lee, Sidford & \cite{ls14} & $n^{2.5}$ \\ \hline
			2017-19 & Altschuler, Weed, Rigollet & \cite{AltschulerWR17} & $n^2 W^2/\epsilon^2$ \\\hline
			2018-19 & Lin, Ho, Jordan & \cite{LinHJ19} & $n^{2.5}W^{0.5}/\epsilon$\\\hline
			2018 & Quanrud / Blanchet, Jambulapati, Kent, Sidford & \cite{Quanrud19,BlanchetJKS18} & $n^2 W / \epsilon$\\\hline
			2020 & {\bf This paper} & & $n^2$\\\hline
		\end{tabular}
	\end{center}\caption{The summary of the results for the {\bf optimal transport}
		problem. 
	}\label{tab:intro:optimal transport}
\end{table}

As a consequence, by careful application of standard reductions, we show that we can solve a host of problems within the same time complexity. 
These problems are those that can be described as or reduced to the following  {\em transshipment} problem. Given $b\in \R^n$, 
$c\in \R^m$, and matrix $\mA \in \{0,1,-1\}^{m \times n}$  where each row of $A$ consists of two nonzero entries, one being $1$ and the other being $-1$, we want to find $x\in \R^m_{\ge 0}$ that achieves the following objective:
\begin{align}
\min_{\mA^\top x = b,\ x\geq 0} c^\top x. \label{eq:intro:LP}
\end{align}
Viewed as a graph problem, we are given a directed graph $G=(V, E)$, a demand function $b: V\rightarrow \R$ and a cost function $c: E \rightarrow \R$. This problem is then to compute a transshipment $f:E\rightarrow \R_{\geq 0}$ that minimizes $\sum_{uv\in E} f(uv)c(uv)$, where a transshipment is a flow $f:E\rightarrow \R_{\geq 0}$ such that for every node $v$, 
\begin{align}\sum_{uv\in E} f(uv)-\sum_{vw\in E} f(vw)=b(v).\label{eq:intro:demand}
\end{align}
The main result of this paper is the following \Cref{thm:intro:transshipment} providing our runtime for solving the transshipment problem.
\begin{theorem}\label{thm:intro:transshipment}
The transshipment problem can be solved to $\epsilon$-additive accuracy  in $\tilde O((m+n^{1.5})\log^2( {\|b\|_\infty\|c\|_\infty}/{\epsilon}))$ time. 
For the integral case, where all entries in $b$, $c$, and $x$ are integers, the problem can be solved exactly in $\tilde O((m+n^{1.5})\log^2(\|b\|_\infty\|c\|_\infty))$ time. 
\end{theorem}

Leveraging \Cref{thm:intro:transshipment} we obtain the following results. 
\begin{enumerate}[noitemsep]
	\item A maximum-cardinality bipartite matching can be computed in  $\tilde O(m+n^{1.5})$ time, where $\tilde O$ hides $\poly\log(n)$ factors. 
	\item The minimum-cost perfect bipartite $b$-matching on graph $G=(V, E)$,  with integer edge costs in $[-W, W]$ and non-negative integer $b(v)\leq W$ for all $v\in V$,  can be computed in $\tilde O((m+n^{1.5})\log^2(W))$ time.
	\item The $\tilde O((m+n^{1.5})\log^2(W))$ time complexity also holds for maximum-weight bipartite matching, negative-weight shortest paths,  uncapaciated min-cost flow, vertex-capacitated min-cost $s$-$t$ flow, minimum mean cost cycle, and deterministic Markov decision processes (here, $W$ denotes the largest absolute value used for specifying any value in the problem). 
	\item The optimal transport problem can be solved to $\epsilon$-additive accuracy in  $\tilde O(n^2 \log^2(W/\epsilon))$ time.
\end{enumerate}

We have already discussed the first two results. Below we briefly discuss some additional results. See \Cref{sec:app} for the details of all results.

\paragraph{Single-source shortest paths with negative weights and minimum weight bipartite perfect matching.} Due to Gabow and Tarjan's algorithm from 1989 \cite{GabowT89}, this problem can be solved in $O(m\sqrt{n}\log(nW))$ time where $W$ is the absolute maximum weight of an edge in the graph. For sparse graphs, this has been improved to $\tilde{O}(m^{10/7}\log W)$ \cite{cmsv17} and recently to  $\tilde{O}(m^{4/3+o(1)}\log W)$ \cite{amv20}. Here, our algorithm obtains a running time of $\tilde O((m+n^{1.5})\log^2(W))$, which again is near-linear for dense graphs and is the lowest known when $m =  n^{9/8 + \delta}$ for any constant $\delta>0$ and $W$ is polynomially bounded. \Cref{table:results:SSSP,table:results:perfect matching} compare our results with the previous ones.

\paragraph{Optimal Transport.} Algorithms with $\epsilon$-additive error received much attention from the machine learning community since the introduction of the algorithm of Altschuler, Weed, Rigollet  \cite{AltschulerWR17} (e.g. \cite{BlanchetJKS18,Quanrud19,LinHJ19}). The algorithm of \cite{AltschulerWR17} runs in time $\tilde O(n^2 W^2/\epsilon^2)$, and \cite{Quanrud19,LinHJ19} runs in time $\tilde O(n^2 W/\epsilon)$. We improve these running times to $\tilde{O}(n^2\log^2(W/\epsilon) )$.
(Note that the problem size is $\Omega(n^2)$.)
\Cref{tab:intro:optimal transport} summarizes previous results.

\subsection{Techniques}

Here we provide a brief high-level overview of our approach (see \Cref{sec:overview} for a much more detailed and formal overview which links to the main theorems of the paper).

Our results constitute a successful fusion and advancement 
of two distinct lines of research on {\em interior point methods} (IPMs) for linear programming 
\cite{Karmarkar84,Renegar88,Vaidya87,VaidyaA93,Anstreicher96,NesterovT97,ls14,cls19,lsz19,ls19,b20,blss20,JiangSWZ20}
and {\em dynamic graph} algorithms \cite{sw19,nsw17,cglnps20,BernsteinBNPSS20,ns17,w17}.
This fusion was precipitated by a breakthrough result 
of Spielman and Teng \cite{SpielmanT04} in 2004 
that Laplacian systems could be solved in nearly linear time.
As IPMs for linear programming 
essentially reduce all the problems considered in this paper 
to solving Laplacian systems in each iteration, 
one can hope for a faster algorithm via a combination 
of fast linear system solvers and interior point methods. 
Via this approach, Daitch and Spielman \cite{ds08} showed in 2008 
that minimum cost flow and other problems could be solved in $\tilde O(m^{3/2})$ time.
Additionally, along this line the results of Madry and others 
\cite{m13,m16,ls19,cmsv17,ls20_stoc,amv20}
all showed that a variety of problems could be solved faster. 
However, as discussed, none of these results lead to improved runtimes for computing maximum cardinality bipartite matching in significantly dense graphs.

More recently, the result of v.d.Brand, Lee, Sidford and Song~\cite{blss20}, 
which in turn was built on \cite{ls14,cls19,b20}, led to new possibilities. 
These methods provide a {\em robust} IPM framework 
which allows one to solve many sub-problems required by each iteration {\em approximately}, 
instead of doing so {\em exactly} as required by the previous interior point frameworks.
Combining this framework with sophisticated dynamic matrix data structures 
(e.g., \cite{Vaidya89a,Sankowski04,BrandNS19,LeeS15,cls19,lsz19,adil2019iterative,b20}) 
has led to the linear programming algorithm of v.d.Brand~et~al.~\cite{blss20}. 
Unfortunately, this algorithm runs in time $\tilde O(mn)$ for graph problems, and this running time seems inherent 
to the data structures used.  Moreover, solving sub-problems only approximately in each iteration leads to infeasible solutions, which were handled by techniques which somewhat complicated and in certain cases inefficient (as this work shows).

Correspondingly this paper makes two major advances. First we show that the data structures (from sparse recovery literature \cite{glps10,knpw11,hikp12a,p13,lnnt16,k17,ns19,nsw19}) used by v.d.Brand~et~al.~\cite{blss20} can be replaced by more efficient data structures in the case of graph problems. These data structures are based on the {\em dynamic expander decomposition} data structure developed in the series of works in \cite{ns17,w17,nsw17,sw19,cglnps20,BernsteinBNPSS20}. For an unweighted undirected graph $G$ undergoing edge insertions and deletions given as input, this data structure maintains a partition of edges in $G$ into expanders. This data structure was originally developed for the dynamic connectivity problem~\cite{nsw17,w17,ns17}, but has recently found applications elsewhere (e.g.~\cite{BernsteinBNPSS20,BernsteinGS20scc,GoranciRST20hierarchy}).
We can use this data structure to detect entries in the solution that change significantly between consecutive iterations of the robust interior point methods. It was known that this task is a key bottleneck in efficiently implementing prior IPMs methods. Our data structures solve this problem near optimally. We therefore hope that they may serve in obtaining even faster algorithms in the future.

The above data structures allow us to solve the problem needed for each iteration (in particular, some linear system) {\em approximately}. It is still left open how to use this approximate solution. The issue is that we might not get a feasible solution (we may get $x$ such that $\ma x\neq b$ when we try to solve the LP \eqref{eq:intro:LP}).
In \cite{blss20}, this was handled in a complicated way that would at best give an $\tilde O(n^2)$ time complexity for the graph problems we consider. In this paper, we simplify and further improve the method of \cite{blss20} 
by sub-sampling entries of the aforementioned approximate solution (and we show that such sampling can be computed  efficiently using the aforementioned dynamic expander decompositions). Because of the sparsity of the sampled solution, we can efficiently measure the infeasibity (i.e. compute $\ma x-b$) and then fix it in a much simpler way than \cite{blss20}.
We actually provide a general framework and analysis for these types of interior point methods that (i) when instantiated on the log barrier, with our data structures, yields a $\tilde O(n\sqrt{m})$-time algorithm (see \Cref{sec:matching:slow}) and (ii) when applied using the more advanced barriers of \cite{ls14} gives our fastest running time (see \Cref{sec:matching:fast}).

We believe that our result opens new doors for combining continuous and combinatorial techniques for graph related problems. The recent IPM advances for maximum flow and bipartite matching problems, e.g. \cite{ds08, m13, ls14, m16, cmsv17, ls20_stoc, ls20_focs, amv20} all use Laplacian system solvers~\cite{SpielmanT04} or more powerful smoothed-$\ell_p$ solver \cite{kpsw19,adil2019iterative,AdilS20} {\em statically} and ultimately spend almost linear work per iteration. In contrast, in addition to using such solvers, we leverage dynamic data-structures for maintaining expanders to implement IPM iterations possibly in sublinear time. Our ultimate runtimes are then achieved by considering the amortized cost of these data structures. 
We hope this proof of concept of intertwining continuous and combinatorial techniques opens the door to new algorithmic advances.

\subsection{Organization}

After the preliminaries and overview in \Cref{sec:preliminaries,sec:overview}, 
we present our IPMs in \Cref{sec:ipm}. 
Our main new data structure called ``\textsc{HeavyHitter}'' is in \Cref{sec:matrix_vector_product}.
In \Cref{sec:vector_maintenance} we show how this \textsc{HeavyHitter} data structure
can be used to efficiently maintain an approximation of the slack of the dual solution.
Maintaining an approximation of the primal solution is described in \Cref{sec:gradient_maintenance}.
In \Cref{sec:matching}, we put everything together to obtain our $\tilde{O}(n\sqrt{m})$-time and $\tilde{O}(m+n^{1.5})$-time algorithms for minimum cost perfect bipartite matching and other problems such as minimum cost flow and shortest paths.

Some additional tools, required by our algorithms, are in \Cref{sec:initial_point}
(constructing the initial point for the IPM),
\Cref{sec:leverage_score} (maintaining leverage scores efficiently),
and \Cref{sec:degenerate} (handling the degeneracy of incidence matrices).

\newpage

\section{Preliminaries}
\label{sec:preliminaries}

We write $[n]$ for the interval $\{1,2,...,n\}$.
For a set $I \subset [n]$ we also use $I$ as $0/1$-vector with $I_{i} = 1$ when $i \in I$ and $I_i = 0$ otherwise.
We write $e_i$ for the $i$-th standard unit vector.
We use $\tilde{O}(\cdot)$ notation to hide $(\log \log W)^{O(1)}$ and $(\log n)^{O(1)}$ factors, where $W$ typically denotes the largest absolute value used for specifying any value in the problem (e.g. demands and edge weights) and $n$ denotes the number of nodes.

When we write \emph{with high probability} (or w.h.p), we mean with probability $1-n^c$ for any constant $c > 0$.

For $x\in\R^{n}$, we use $x_{i}$ to denote
the $i$-th coordinate of vector $x$ if the symbol $x$ is simple. If the symbol is complicated, we
use $(x)_i$ or $[x]_{i}$ to denote the $i$-th coordinate of vector $x$ (e.g. $(\delta_s)_i$).

We write $\mathbf{1}_{\text{condition}}$ for the indicator variable, 
which is $1$ if the condition is true and $0$ otherwise.

\paragraph{Diagonal Matrices}
Given a vector $v \in \R^d$ for some $d$,
we write $\mdiag(v)$ for the $d\times d$ diagonal matrix with $\mdiag(v)_{i,i} = v_i$.
For vectors $x,s,\os,\ox,x_{t},s_{t},w,\ow,w_{t}, \tau, g$ we let $\mx\defeq\mdiag(x)$,
$\ms\defeq\mdiag(s)$, and define $\omx$, $\oms$, $\mx_{t}$, $\ms_{t}$,
$\mw$, $\omw$, $\mw_{t}$, $\mT$, $\mG$ analogously.

\paragraph{Matrix and Vector operations}
Given vectors $u,v \in \R^d$ for some $d$,
we perform arithmetic operations $\cdot,+,-,/,\sqrt{\cdot}$ element-wise.
For example $(u\cdot v)_i = u_i\cdot v_i$ or $(\sqrt{v})_i = \sqrt{v_i}$.
For the inner product we will write $\langle u, v \rangle$ and $u^\top v$ instead.
For a vector $v \in \R^d$ and a scalar $\alpha \in \R$ we have $(\alpha v)_i = \alpha v_i$
and we extend this notation to other arithmetic operations, e.g. $(v + \alpha)_i = v_i + \alpha$. 

For symmetric matrices $\mA,\mB\in \R^{n\times n}$ we write $\mA\preceq \mB$ to indicate that $x^\top \mA x \leq x^\top \mB x$ for all $x\in \R^n$ and define $\succ$, $\prec$, and $\succeq$ analogously. 
We let $\S_{>0}^{n\times n}\subseteq\R^{n\times n}$
denote the set of $n\times n$ symmetric positive definite matrices.
We call any matrix (not necessarily symmetric) {\em non-degenerate} if its rows are all non-zero and it has full column
rank. 

We use $a\approx_{\epsilon}b$ to denote that $\exp(-\epsilon)b\leq a\leq\exp(\epsilon)b$
entrywise and $\ma\approx_{\epsilon}\mb$ to denote that $\exp(-\epsilon)\mb\preceq\ma\preceq\exp(\epsilon)\mb$.
Note that this notation implies $a \approx_\epsilon b \approx_\delta c$ $\Rightarrow$ $a \approx_{\epsilon + \delta} c$, and $a \approx_\epsilon b$ $\Rightarrow$ $a^x \approx_{\epsilon\cdot x} b^x$ for $x \ge 0$.

For any matrix $\ma$ over reals, let $\nnz(\ma)$ denote the number of non-zero entries in $\ma$.

\paragraph{Leverage Scores and Projection Matrices}
For any non-degenerate matrix $\ma\in\R^{m\times n}$ we let $\mproj(\ma)\defeq\ma(\ma^{\top}\ma)^{-1}\ma^{\top}$
denote the orthogonal projection matrix onto $\ma$'s image. 
The definition extends to degenerate matrices via the Penrose-Pseudoinverse,
i.e. $\mproj(\ma)=\ma(\ma^{\top}\ma)^{\dagger}\ma^{\top}$.
Further, we let $\sigma(\ma) \in \R^m$ with $\sigma(\mA)_i \defeq \mproj(\ma)_{i,i}$ denote $\ma$'s \emph{leverage scores} and let $\mSigma(\ma) \defeq \mdiag(\sigma(\ma))$, and we let \emph{$\tau(\ma)\defeq\sigma(\ma)+\frac{n}{m}\vones$} denote $\ma$'s regularized leverage scores and  $\mT(\ma)\defeq\mdiag(\tau(\ma))$.  Finally, we let $\mproj^{(2)}(\ma)\defeq\mproj(\ma)\circ\mproj(\ma)$
(where $\circ$ denotes entrywise product), and $\mLambda(\ma)\defeq\mSigma(\ma)-\mproj^{(2)}(\ma)$.

\paragraph{Norms}
We write $\| \cdot \|_p$ for the $\ell_p$-norm, i.e. $\|v\|_p := (\sum_i |v_i|^p)^{1/p}$, $\|v\|_\infty = \max_i | v_i |$ and $\| v \|_0$ being the number of non-zero entries of $v$.
For a positive definite matrix $\mM$ we define $\| v \|_\mM = \sqrt{v^\top \mM v}$. 
For a vector $\tau$ we define $\| v \|_\tau := (\sum_i \tau_i v_i^2)^{1/2}$
and $\|v\|_{ \tau + \infty } := \| v \|_\infty + 40 \log(4m/n) \|v\|_\tau$, 
where $m\ge n$ are the dimensions of the constraint matrix of the linear program (we define $\|v\|_{ \tau + \infty }$ again in \Cref{def:mixed norm}). 

\paragraph{Graph Matrices}

Given a directed graph $G=(V,E)$, 
we define the (edge-vertex) incidence matrix $\mA \in \{-1,0,1\}^{E\times V}$ 
via $\mA_{e,u} = -1,$ $\mA_{e,v}=1$ for every edge $e = (u,v) \in E$.
We typically refer to the number of edges by $m$ and the number of nodes by $n$,
so the incidence matrix is an $m \times n$ matrix, which is why we also allow indices $\mA_{i,j}$ 
for $i \in [m]$, $j\in[n]$ by assuming some order to the edges and nodes.

For edge weights $w \in \R^E_{\ge0}$ we define the Laplacian matrix as $\mL = \mA^\top \mW \mA$. 
For an unweighted undirected simple graph the Laplacian matrix 
has $\mL_{u,v} = -1$ if $\{u,v\} \in E$ and $\mL_{v,v} = \deg(v)$,
which is the same as the previous definition 
when assigning arbitrary directions to each edge.

Our algorithm  must repeatedly solve Laplacian systems. 
These types of linear systems are well studied
\cite{Vaidya90,SpielmanT03,SpielmanT04,KoutisMP10,KoutisMP11,KelnerOSZ13,lee2013efficient,CohenKMPPRX14, KyngLPSS16, KyngS16}
and we use the following result for solving Laplacian systems (see e.g.~Theorem 1.2 of \cite{KyngS16}):%
\begin{lemma}\label{lem:laplacian_solver}
There is a randomized procedure that given any $n$-vertex $m$-edge graph $G$ with incidence matrix $\mA$, 
diagonal non-negative weight matrix $\mw$, and vector $b \in \R^V$ 
such that there exists an $x \in \R^V$ with $(\ma^\top \mw \ma) x = b$ 
computes $\ox \in \R^V$ with $\norm{\ox - x}_{\ma^\top \mw \ma} \leq \epsilon \norm{x}_{\ma^\top \mw \ma}$ 
in $\tilde{O}(m \log\epsilon^{-1})$ w.h.p.
\end{lemma}
Note that we can express the approximation error of \Cref{lem:laplacian_solver} 
as some spectral approximation, i.e. that there exists some $\mH \approx_{20\epsilon} \mA^\top\mw\mA$ such that $\mH \ox = b$ \cite[Section 8]{blss20}.

\paragraph{Expanders}
We call an undirected graph $G=(V,E)$ a $\phi$-expander, 
if 
$$
\phi \le \min_{\emptyset \neq S \subsetneq V} \frac{|\{\{u,v\}\in E \mid u \in S, v \in V \setminus S\} | 
}{
\min \{\sum_{v \in S} \deg(v), \sum_{v \in V \setminus S} \deg(v)\}
}.
$$
For an edge partition $\bigcup_{i=1}^t E_i = E$,
consider the set of subgraphs $G_1,...,G_t$, where $G_i$ is induced by $E_i$ (with isolated vertices removed).  
We call this edge partition and the corresponding set of subgraphs a {\em $\phi$-expander decomposition} of $G$ if each $G_i$ is a $\phi$-expander.

\newpage

\section{Overview}
\label{sec:overview}

We start the overview by explaining how our new interior point method (IPM) works in \Cref{sec:overview:newIPM}. A graph-algorithmic view of this IPM can be found in \Cref{sec:over:graph view} and the full detail can be found in \Cref{sec:ipm}.
This new IPM reduces solving linear programs to efficiently performing a number of computations approximately. To efficiently perform these computations for graph problems and implement our IPM we provide new data structures, outlined in \Cref{sec:over:data structures}.
Some of these data structures are easy to obtain via known tools, e.g. Laplacian solvers, and some constitute new contributions. In \Cref{sec:overview:vector} we outline our main data structure contributions. The details for these data structures are found in \Cref{sec:matrix_vector_product,sec:vector_maintenance,sec:gradient_maintenance}.

\subsection{Interior Point Method} 
\label{sec:overview:newIPM}

Here we provide an overview of our new efficient sampling-based primal-dual IPMs given in \Cref{sec:ipm}. Our method builds upon the recent IPM of v.d.Brand, Lee, Sidford, and Song~\cite{blss20} and a host of recent IPM advances \cite{cls19,ls19,lsz19,b20}. As with many of these recent methods, given a non-degenerate $\ma \in \R^{m \times n}$ and $b\in\R^{n}$ and $c\in\R^{m}$, these IPMs are applied to linear programs represented in the following primal $(P)$ and dual $(D)$ form: 
\begin{equation}
(P) \defeq \min_{x\in\R_{\geq0}^{m}:\ma^{\top}x=b}c^{\top}x
\enspace \text{ and } \enspace
(D) \defeq \max_{y\in\R^{n},s \in\R^m_{\geq 0}:\ma y + s = c}b^{\top}y\,.\label{eq:over:primal_dual}
\end{equation}
In the remainder of this subsection we explain, motivate, and compare our IPM for this general formulation. For more information about how this IPM is applied in the case of matching problems, see the next subsections.

\paragraph{Path following}
As is typical for primal-dual IPMs, both our IPM and the IPMs in  \cite{cls19,lsz19,b20,blss20,JiangSWZ20} maintain primal $x^{(i)} \in \R^m_{\geq 0}$
and dual slack $s^{(i)} \in \R^m_{\geq 0}$ 
and proceed for iterations $i = 0, 1, \ldots$ attempting to iteratively improve their quality.
In each iteration $i$, they attempt to compute $(x^{(i)},s^{(i)})$
so that  
\begin{align}
\label{eq:approx_centered}
x^{(i)}s^{(i)}\approx \mu^{(i)}\tau(x^{(i)},s^{(i)})
\end{align} 
for some {\em path parameter} $\mu^{(i)}\in \R_{\geq 0}$ and {\em weight function} $\tau(x^{(i)},s^{(i)})\in \R^m_{\geq 0}$. (Recall from \Cref{sec:preliminaries} that  $x^{(\ell)}s^{(\ell)}$ is an element-wise multiplication.) 

The intuition behind this approach is that for many weight functions, e.g. any constant positive vector, the set of primal-dual pairs $(x_\mu, s_\mu) \in \R^m_{\geq} \times \R^m_{\geq0}$, that are \emph{feasible}, i.e. satisfy $\ma^\top x_\mu = b$ and $\ma y + s = c$ for some $y \in \R^n$, and are \emph{$\mu$-centered}, i.e. $x s = \mu \tau(x,s)$, form a continuous curve from solutions to \eqref{eq:over:primal_dual}, at $\lim_{\mu \rightarrow 0} (x_\mu, s_\mu)$, to a type of center of the primal and dual polytopes (in the case they are bounded), at $\lim_{\mu \rightarrow \infty} (x_\mu, s_\mu)$. This curve is known as the \emph{central path} and consequently these methods can be viewed as maintaining approximately centrality to approximately follow the central path. 

Our methods follow a standard step-by-step approach (similar to \cite{blss20}) to reduce solving a linear program to efficiently following the central path, i.e. maintaining \eqref{eq:approx_centered} for changing $\mu$: 
(i) modify the linear program to have trivial feasible centered initial points 
(ii) following the path towards the interior (i.e. increase $\mu$) 
(iii) show that the resulting points are on the central path for the original linear program 
(iv) follow the path towards optimal points (i.e. decreasing $\mu$). 
(v) show that the resulting points are near optimal and can be rounded to exactly optimal solutions (depending on the size of the weight function and the particular LP). 
See \Cref{sec:initial_point} for how to perform steps (i), (iii), and (v).

Where the IPMs of \cite{cls19,lsz19,b20,blss20,JiangSWZ20} and ours all differ is in what weight function is used and how the central path is followed. 
There is a further complication in some of these methods in that in some cases feasibility of $x$ is not always maintained exactly. 
In some, linear systems can only be solved to high-precision, however this can be handled by natural techniques, see e.g. \cite{ds08,ls14}. 
Further, in \cite{blss20}, to allow for approximate linear system solves in the iterations and thereby improve the iteration costs, feasibility of $x$ was maintained more crudely through complicated modifications to the steps. 
A key contributions of this paper, is a simple sampling-based IPM that also maintains approximately feasible $x$ to further decrease the iteration costs of \cite{blss20}. 

\paragraph{Weight function} %

In this paper we provide a general IPM framework that we instantiate with our sampling-based techniques on two different weight functions $\tau(x^{(i)},s^{(i)})$. While there are many possible weight functions we restrict our attention to
$\taulog$ induced by the standard logarithmic barrier and $\tauLS$ a regularized variant of the weights induced by the Lee-Sidford barrier function \cite{ls19} (also used in \cite{blss20}) defined as follows:
\begin{align*}
\taulog(x^{(i)},s^{(i)})\defeq\onevec \quad \mbox{ and } \quad\tauLS(x^{(i)},s^{(i)})=\sigma(x^{(i)}, s^{(i)}) + \frac{n}{m} \vones
\end{align*}
Above, $\sigma(x^{(i)}, s^{(i)})\in \R^m$ are {\em leverage scores} of  $\ma$ under a particular row re-weighting by $x^{(i)}$, and $s^{(i)}$ as used in, e.g., \cite{ls14,ls19,blss20} (see \Cref{def:LS_weight} for formal definition). 
Roughly, $\sigma(x^{(i)}, s^{(i)})$ measures the importance of each row of $\mA$ with respect to the current primal dual pair $x^{(i)}$ and $s^{(i)}$ in a way that the induced central path is still continuous and can be followed efficiently.

On the one hand, $\taulog$ is perhaps the simplest weight function one could imagine. The central path it induces is the same as the one induced by penalizing approach the constraints of $(P)$ and $(D)$ with a logarithmic barrier function, see \eqref{eq:log_barrier}. Starting with the seminal work of \cite{Renegar88} there have been multiple $\otilde(\sqrt{m})$ iteration IPMs induced by $\taulog$. 
On the other hand, $\tauLS$ is closely related to the Lewis weight barrier or Lee-Sidford barrier given in \cite{ls19} and its analysis is more complex. However, in \cite{blss20} it was shown that this weight function induces a $\otilde(\sqrt{n})$ iteration IPM. (See \cite{ls19,blss20} for further explanation and motivation of $\tauLS$): 

Though the bounds achieved by $\tauLS$ in this paper are never worse than those achieved by $\taulog$ (up to logarithmic factors), we consider both weight functions for multiple reasons. First, the analysis of $\taulog$ in this paper is simpler than that for $\taulog$ and yet is sufficient to still obtain $\tilde{O}(n \sqrt{m}) = \tilde{O}(n^2)$ time algorithms for the matching problems considered in this paper (and consequently on a first read of this paper one might want to focus on the use of $\taulog$). Second, the analysis of the two weight functions is very similar and leverage much common algorithmic and analytic machinery. Consequently, considering both barriers demonstrates the versatility of our sampling-based IPM approach. 

\paragraph{Centrality potentials} 
To measure whether  $x^{(i)}s^{(i)}\approx \mu^{(i)}\tau(x^{(i)},s^{(i)})$ (for $\tau\in \{\taulog,\tauLS\}$) and design our steps, as with previous IPM advances \cite{ls19, cls19,lsz19,b20,blss20,JiangSWZ20} we use the softmax potential function $\Phi : \R^m \rightarrow \R$ defined for all vectors $v$ by
\[
\Phi(v) \defeq \sum_{i \in [n]} \phi(v_i) \text{ where } \phi(v_i) \defeq \exp(\lambda(v_i - 1)) + \exp(-\lambda(v_i - 1))
\]
for some parameter $\lambda$. We then define the {\em centrality measures} or {\em potentials} as 
\[
\Phi(x^{(i)}, s^{(i)}, \mu^{(i)})  \defeq \Phi\left(
\frac{ x^{(i)} s^{(i)} }{ \mu^{(i)} \tau(x^{(i)}, s^{(i)}) }    
\right)
\]
where $\tau\in \{\taulog,\tauLS\}$ depending on which weight function is used.
$\Phi$ intuitively measure how far $x^{(i)} s^{(i)}$ is from $\mu^{(i)}\tau(x^{(i)},s^{(i)})$, i.e. how far $x^{(i)}$ and $s^{(i)}$ are from being centered and having \eqref{eq:approx_centered} hold. Observe that $\Phi(v)$  is small when $v=\onevec$ (thus $x^{(i)}s^{(i)}=\mu^{(i)}\tau(x^{(i)},s^{(i)})$)  and increases very quickly as $v$ deviates from $\onevec$. 

The centrality potential we consider has been leveraged extensively by previous IPM advances. In particular, $\Phi$ with $\tau = \taulog$ was used in  \cite{cls19,lsz19,b20,JiangSWZ20} and $\Phi$ with $\tau = \tauLS$ was used in \cite{blss20}. Where our method differs from prior work is in how we design our steps for controlling the value of this potential function a discussed in the next section.

\paragraph{Improvement Step (Short Step)}

Given the choice of weight function $\tau \in \{\taulog, \tauLS\}$ our IPM follows the central path by taking improvement steps (called {\em short steps}) defined as follows:
\begin{align}\label{eq:over:update_xs}
x^{(i+1)}= x^{(i)}+ \eta_x \delta_x^{(i)} \mbox{ , } s^{(i+1)}= s^{(i)}+ \eta_s \delta_s^{(i)}
\mbox{ , and }
\mu^{(i + 1)} = \mu^{(i)} + \delta_\mu^{(i )}
\end{align}
where $\eta_x, \eta_s$ are constants depending on whether we use $\taulog$ or $\tauLS$, and $\delta_x^{(i+1)}$, $\delta_s^{(i+1)}$, and $\delta_\mu^{(i + 1)}$ are defined next (See \Cref{alg:short_step_log} and \Cref{alg:short_step_LS} in \Cref{sec:ipm} for the complete pseudocode for the short-step method for  $\taulog$ and $\tauLS$ respectively.) Informally, these steps are defined as approximate projected Newton steps of $\Phi$ in the appropriate norm. Formally, $\delta_x^{(i)}$ and $\delta_s^{(i)}$ are given by the following equations
\begin{align}%
\delta_x^{(i)} = \omX^{(i)} \og^{(i)} - \mR^{(i)} [\omX^{(i)}(\omS^{(i)})^{-1}&\mA(\omH^{(i)})^{-1}\mA^\top \omX^{(i)} \og^{(i)} + \delta_c^{(i)}],   \label{eq:over:delta_x new} \\
\delta_s^{(i)} = &\mA(\omH^{(i)})^{-1}\mA^\top \omX^{(i)} \og^{(i)}. \label{eq:over:delta_s}
\end{align}
where the variables in \eqref{eq:over:delta_x new} and \eqref{eq:over:delta_s} are described below. 
\begin{enumerate}[noitemsep,label=(\Roman*)]
	\item \label{item:over:def ox os} $\ox^{(i)},\os^{(i)} \in \R^{m}$ are any entry-wise, multiplicative approximations of $x^{(i)}$ and $s^{(i)}$ and $\omX^{(i)}=\mdiag(\ox^{(i)})$ and $\omS^{(i)}=\mdiag(\os^{(i)})$. (See:  Line~\ref{line:log:approx_xs} of \Cref{alg:short_step_log} and Line~\ref{line:LSstep:xstau} of \Cref{alg:short_step_LS}.)

	\item \label{item:over:def og} $\og^{(i)}\in \R^{m}$: This is an approximate 
	steepest descent direction of $\Phi$ with respect to some norm 
	$\|\cdot\|$.
	Formally, for $\ov^{(i)}\in \R^m$ which is an element-wise approximation of 
		$\frac{x^{(i)} s^{(i)}}{\mu^{(i)}\tau(x^{(i)},s^{(i)})}$,
	we choose
	\begin{align}
	\og^{(i)} = \underset{ z \in \R^m: \|z\|\le 1 }{ \argmax } \langle \nabla\Phi(\ov^{(i)}), z \rangle \label{eq:over:og}
	\end{align}
	for some norm $\|\cdot\|$ that depends on whether we use $\taulog$ or $\tauLS$. (See Lines \ref{line:log:v}-\ref{line:log:g} of \Cref{alg:short_step_log} and Lines~\ref{line:LSstep:v}-\ref{line:LSstep:g} of \Cref{alg:short_step_LS}.)
	
	\item $\omH^{(i)}\in \R^{n\times n}$ is any matrix such that $\omH^{(i)} \approx 
	\mA^\top \omX^{(i)} (\omS^{(i)})^{-1} \mA$.
	(See Line~\ref{line:logstep:omh} of  \Cref{alg:short_step_log} and 
	Line~\ref{line:LSstep:omh} of \Cref{alg:short_step_LS}.)  

	\item \label{item:over:R} $\mR^{(i)}\in \R^{m\times m}$ is a randomly selected PSD 
	diagonal matrix chosen so that $\E[\mR^{(i)}] = \mI$,  $$\mA^\top \omX^{(i)} 
	(\omS^{(i)})^{-1} \mR^{(i)} \mA \approx \mA^\top \omX^{(i)} (\omS^{(i)})^{-1} 
	\mA,$$ and the second moments of $\delta_x$ are bounded. The number of non-zero 
	entries in $ \mR^{(i)}$ is $\tilde O(n+\sqrt{m})$ when we use $\taulog$ and 
	$\tilde O(n+m/\sqrt{n})$ when we use $\tauLS$. Intuitively $\mR$ randomly samples 
	some rows of the matrix following it in \eqref{eq:over:delta_x new} with 
	overestimates of importance measures of the row. 
	 (See Line~\ref{line:log:R} of \Cref{alg:short_step_log}  and Line~\ref{line:LSstep:R} of \Cref{alg:short_step_LS}.)
	\item \label{item:over:delta_c} $\delta_c^{(i)}\in \R^{m}$: a ``correction vector'' which (as discussed more below), helps control the infeasibility of $x^{(i + 1)}$. For a parameter 
	$\eta_c$ of value $\eta_c\approx 1$  (more precisely, 
	$\eta_c=1$ for $\taulog$ and $\eta_c=\frac{1}{1-1/O(\log n)}$ for $\tauLS$) this is defined as
	\begin{align}\label{eq:over:delta_c}
	\delta_c^{(i)} \defeq \eta_c  \omX^{(i)}(\omS^{(i)})^{-1}\mA (\omH^{(i)})^{-1} (\mA^\top x^{(i)} - b).
	\end{align}

\end{enumerate}

\noindent{\em Flexibility of variables}: Note that there is flexibility in choosing variables of the form $\bar\square$, i.e. $\ox^{(i)}$, $\os^{(i)}$, $\og^{(i)}$ and $\omH^{(i)}$. Further, our algorithms have flexibility in the choice of $\mR^{(i)}$, we just need too sample by overestimates. This flexibility gives us freedom in how we implement the steps of this method and thereby simplifies the data-structure problem of maintaining them. As in \cite{blss20}, this flexibility is key to our obtaining our runtimes.

\smallskip\noindent{\em Setting $\delta_\mu^{(i + 1)}$:} As we discuss more below, if $\mR^{(i)}=\mI$ and $\delta_c^{(i)}=0$ in \eqref{eq:over:delta_x new}, our short steps would be almost the same as those in the IPMs in \cite{cls19,b20,blss20,JiangSWZ20}.
For such IPMs, it was shown in \cite{cls19, b20,JiangSWZ20} (respectively in \cite{blss20}) that  $\delta_\mu^{(i)}$ can be set to be roughly $\tilde{O}(1/\sqrt{m}) \mu^{(i)}$ if we use $\taulog$ (respectively  $\tilde{O}(1/\sqrt{n}) \mu^{(i)}$ if we use $\tauLS$), leading to a method with $\tilde{O}(\sqrt{m})$ (respectively $\tilde{O}(\sqrt{n})$) iterations.

In this paper, 
 we can adjust the analyses in \cite{cls19,b20,blss20,JiangSWZ20} to show that our IPMs require the same number of iterations. 
In particular, we provide a general framework for IPMs of this type (\Cref{sec:lp_meta}) and show that by carefully choosing the distribution for $\mR^{(i)}$ (and restarting when necessary) we can preserve the typical convergence rates from \cite{cls19,b20,blss20,JiangSWZ20} for using $\taulog$ and $\tauLS$ while ensuring that the infeasibility of $x$ is never too large. Provided $\mr^{(i)}$ can be sampled efficiently, our new framework supports arbitrary crude polylogarithmic multiplicative approximations of $\omh^{(i)}$ to $\ma^\top \omx^{(i)} (\oms^{(i)})^{-1} \oma$, in contrast to the high precision approximations required by \cite{cls19,b20,JiangSWZ20} and a more complicated  approximation required in \cite{blss20}

\paragraph{Motivations and comparisons to previous IPMs} 
The IPM in  \cite{blss20} and ours share a common feature that they only {\em approximately} solve linear systems in each iteration, i.e. they apply $(\omH^{(i)})^{-1}$ to a vector for $\omH^{(i)} \approx \ma^\top \omX^{(i)} (\omS^{ (i) })^{-1} \ma$.
While \cite{blss20} carefully modified the steps to make them feasible in expectation, here we provide a new technique of simply sampling from $\delta_x$ to essentially \emph{sparsify} the change in $x$ so that we always know the infeasibility and therefore can better control it.
In particular, the short steps in  \cite{blss20} are almost the same as ours with $\mR^{(i)}=\mI$ and $\delta_c^{(i)}=0$ in \eqref{eq:over:delta_x new}.
This means that we modify the previous short steps in two ways. 
First, we sparsify the change in $x^{(i)}$ using a sparse random matrix $\mR^{(i)}$ 
defined in \ref{item:over:R}. Since $\E[\mR^{(i)}] = \mI$, in expectation the behavior 
of our IPM is similar to that in \cite{blss20}. However, since $\mR^{(i)}$ has 
$\tilde{O}(n+m/\sqrt{n})$ non-zero entries (and less for $\taulog$), we can quickly 
compute $\ma^{\top} x^{(i + 1)} - b$ from $\ma^{\top} x^{(i)} - b$ by looking at 
$\tilde O(n+m/\sqrt{n})$ rows of $\ma$.
This information is very useful in fixing the 
feasibility of $x^{(i + 1)}$ so that $\ma^{\top} x^{(i + 1)}=b$ in the LP. In 
particular,  while \cite{blss20} requires a complicated process to keep  $x^{(i+1)}$ 
feasible, we only need our second modification: a ``correction vector'' 
$\delta_c^{(i)}$. The idea is that we choose $\delta_c^{(i)}$ so that 
$x^{(i+1)}=x^{(i)}+\delta_x^{(i)}+\delta_c^{(i)}$ would be feasible if we use 
$\omH^{(i)} = \ma^\top \omX^{(i)} (\omS^{ (i) })^{-1} \ma$. Although we will still 
have $\omH^{(i)} \approx \ma^\top \omX^{(i)} (\omS^{ (i) })^{-1} \ma$ and so 
$x^{(i+1)}$ will be infeasible, the addition of $\delta_c^{(i)}$ fixes some of the 
previous induced infeasibility. This allows us to bypass the expensive infeasibility
fixing step in \cite{blss20} which takes $\tilde{O}(mn+n^{3})$ time, 
and improve the running time to $\tilde{O}(mn +n^{2.5})$, 
and even less when $\ma$ is an incidence matrix.

To conclude, we advance the state-of-the-art for IPMs by providing new methods which can tolerate crude approximate linear system solvers and gracefully handle the resulting loss of infeasibility. Provided certain sampling can be performed efficiently, our methods improve and simplifying aspects of \cite{blss20}. 
This new IPM framework, together with our new data structures (discussed next), 
allow $\taulog$ to be used to obtain an $\tilde{O}(n \sqrt{m})$-time matching algorithm 
and $\tauLS$ to be used to obtain our $\tilde{O}(m + n^{1.5})$-time matching algorithm. 
We believe that this framework is of independent interest and may find further 
applications. 

\subsection{A Graph-Algorithmic Perspective on our IPM}\label{sec:over:graph view}
Here we provide an overview of the IPM discussed in \Cref{sec:overview:newIPM}, specialized to the graph problems we consider, such as matching and min-cost flow. This subsection is intended to provide further intuition on both our IPM and the data structures we develop for implementing the IPM efficiently. For simplicity, we focus on our IPM with $\taulog$ in this subsection. 
In the case of graph problems, typically the natural choice of $\ma$ in the linear programming formulations is the incidence matrix $\mA \in \{-1,0,1\}^{E\times V}$ of a graph (see \Cref{sec:preliminaries}). The structure of this matrix ultimately enables our methods to have the graph interpretation given in this section and allows us to achieve more efficient data structures (as compared to the case of general linear programs). This interpretation is discussed here and the data structures are discussed in \Cref{sec:over:data structures,sec:overview:vector}. 

Note that incidence matrices are degenerate; the all-ones vector is always in the kernel and therefore $\ma$ is not full column rank (and $\ma^\top \ma$ is not invertible). Consequently, the algorithms in \Cref{sec:overview:newIPM} do not immediately apply. This can be fixed by standard techniques  (e.g. \cite{ds08}). In this paper we fix this issue by appending an identity  block at the bottom of $\mA$ (which can be interpreted as adding self-loops to the input graph; see \Cref{sec:degenerate}). For simplicity, we ignore this issue in this subsection.

\paragraph{Min-cost flow} We focus on the uncapacited min-cost flow (a.k.a.~transshipment) problem, where the goal is the find the flow satisfying nodes' demands (Eq.~\eqref{eq:intro:demand}).  Other graph problems can be solved by reducing to this problem (see \Cref{fig:reductions}). For simplicity, we focus on computing $\delta_x^{(i)}$ as in \eqref{eq:over:delta_x new}  
and assume that $\eta_x=\eta_s=1$. Below, entries of any $n$-dimensional (respectively $m$-dimensional) vectors are associated with vertices (respectively edges). 
After $i$ iterations of our IPM, we have 
\begin{itemize}[noitemsep]
	\item a flow $\ox^{(i)}\in \R^m$ that is an approximation of a flow $x^{(i)}$ (we do not explicitly maintain $x^{(i)}$ but it is useful for the analysis), 
	\item an approximated slack variable $\os^{(i)}\in \R^m$, and
	\item $\mA^\top x^{(i)}-b\in \R^n$ called {\em infeasibility} (a reason for this will be clear later).
\end{itemize}
We would like to improve the cost of $x^{(i)}$ by augmenting it with flow $\ox^{(i)} \og^{(i)}\in \R^m$, for some ``gradient'' vector $\og^{(i)}$. This corresponds to the first term in \eqref{eq:over:delta_x new} and gives us an intermediate $m$-dimensional flow vector
$$\dot{x}^{(i+1)} \defeq x^{(i)}+\ox^{(i)}   \og^{(i)}$$ 
Let us oversimplify the situation by assuming that $\og^{(i)}$ has $\tilde O(n)$ non-zero entries, so that computing $\ox^{(i)}   \og^{(i)}$ is not a bottleneck in our runtime. We will come back to this issue later.

\paragraph{Infeasibility} 
The main problem of $\dot{x}^{(i+1)}$ is that it might be {\em infeasible}, i.e. 
$\mA^\top \dot{x}^{(i+1)}\neq b$.  The infeasibility $\mA^\top \dot{x}^{(i+1)}-b$ is due to (i) the infeasibility of $x^{(i)}$ (i.e. $\mA^\top x^{(i)}-b$), and (ii) the excess flow of $\ox^{(i)}   \og^{(i)}$, which is $(\mA^\top \omX^{(i)} \og^{(i)})_v=\sum_{uv\in E} \ox^{(i)}_{uv}\og^{(i)}_{uv}-\sum_{vu\in E} \ox^{(i)}_{vu}\og^{(i)}_{vu}$ on each vertex $v$.
This infeasibility would be fixed if we {\em subtract} $\dot{x}^{(i+1)}$ with some ``correction'' flow $f^{(i)}_c$ that satisfies,  for every vertex $v$, the demand vector $d^{(i)}\in \R^n$ where
\begin{align}
d^{(i)}\defeq \mA^\top \dot{x}^{(i+1)}-b=\mA^\top \omX^{(i)} \og^{(i)}+(\mA^\top x^{(i)}-b). \label{eq:over:demand}
\end{align}
Note that given sparse $\og^{(i)}$ (as assumed above) and $\mA^\top x^{(i)}-b$, we can compute the demand vector $d^{(i)}$ in $\tilde O(n)$ time. 

\paragraph{Electrical flow} 
A  standard candidate for $f^{(i)}_c$ is an electrical flow on the input graph $G^{(i)}$ with resistance $r^{(i)}_e=\os^{(i)}_e/\ox^{(i)}_e$
on each edge $e$.  In a close form, such electrical flow is 
$$f^{(i)}_c=\omX^{(i)}(\omS^{(i)})^{-1}\mA(\mH^{(i)})^{-1} d^{(i)},$$ 
where $\mH^{(i)}$ is the Laplacian of $G^{(i)}$. (Note that $(\mH^{(i)})^{-1}$ does not exist. 
This issue can be easily fixed (e.g., \Cref{sec:degenerate}), so we ignore it for now.) Observe that $f^{(i)}_c$ is exactly the second term of \eqref{eq:over:delta_x new} (also see \eqref{eq:over:delta_c}) with $\mR^{(i)}=\mI$ and $\mH^{(i)}=\omH^{(i)}$. Such $f^{(i)}_c$ can be computed in $\tilde O(m)$ time in every iteration via {\em fast Laplacian solvers} (\Cref{lem:laplacian_solver}).\footnote{We use a $(1+\epsilon)$-approximation Laplacian solver. Its runtime depends logarithmically on $\epsilon^{-1}$, so we can treat it essentially as an exact algorithm.} Since known IPMs require $\Omega(\sqrt{n})$ iterations, this leads to $\tilde O(m\sqrt{n})$ total time at best. This is too slow for our purpose. 
The main contribution of this paper is a combination of new IPM and data structures that reduces the time per iteration to $\tilde O(n)$. 

\paragraph{Spectral sparsifier}
A natural approach to avoid $\tilde O(m)$ time per iteration is to approximate $f^{(i)}_c$ using a spectral approximation of $\mH^{(i)}$, denoted by  $\omH^{(i)}$. 
In particular, consider a new intermediate flow 
\begin{align}
\ddot{x}^{(i+1)}\defeq x^{(i)}+\ox^{(i)}   \og^{(i)}-\bar{f}^{(i)}_c, \mbox{ where } 
\bar{f}^{(i)}_c\defeq \omX^{(i)}(\omS^{(i)})^{-1}\mA(\omH^{(i)})^{-1} d^{(i)},\label{eq:over:correction}
\end{align} 
Note that the definition of $\bar{f}^{(i)}_c$ is exactly the second term of \eqref{eq:over:delta_x new} with $\mR^{(i)}=\mI$, and it differs from $f^{(i)}_c$ only in $\omH^{(i)}$. 
Given $d\in \R^n$, computing $(\omH^{(i)})^{-1}d\in \R^n$ is straightforward: a spectral sparsifier $\omH^{(i)}$ with $(1+\epsilon)$-approximation ratio and $\tilde O(n/\epsilon^2)$ edges can be maintained in $\tilde O(n/\epsilon^2)$ time per iteration
(under the change of resistances), either using the leverage scores \cite{blss20} or the dynamic sparsifier algorithm of \cite{BernsteinBNPSS20}. We then run a fast Laplacian solver on top of such sparsifier to compute $(\omH^{(i)})^{-1}d$. This requires  only $\tilde O(n)$ time per iteration.

\paragraph{Difficulties}
There are at least two difficulties
in implementing the above idea:
\begin{enumerate}[noitemsep]
	\item {\em Infeasibility}: An approximate electrical flow $\bar{f}^{(i)}_c$ might not satisfy the demand $d^{(i)}$, thus does not fix the infeasibility of $x^{(i)}$. 
 	\item {\em Time:} 
	Computing $\bar{f}^{(i)}_c\in \R^m$ explicitly requires $\Omega(m)$ time even just to output the result. 
\end{enumerate}

\paragraph{Bounding infeasibility and random correction}
For the first issue, it turns out that while we cannot keep each $x^{(i)}$ feasible, we can prove that the infeasibility remains small throughout. As a result, we can bound the number of iterations as if every $x^{(i)}$ is feasible (e.g.  $\tilde O(\sqrt{m})$ iterations using $\taulog$). 
To get around the second issue, we apply the correction flow  $\bar{f}^{(i+1)}_c$ only on $\tilde O(n)$ carefully {\em sampled and rescaled} edges\footnote{\label{foot:over:tauLS}If we use $\tauLS$, the number of edges becomes $\tilde O(m/\sqrt{n})$.}; i.e. our new (final) flow is
\begin{align}
x^{(i+1)}\defeq x^{(i)}+\ox^{(i)}   \og^{(i)}-\mR^{(i)}\bar{f}^{(i)}_c, 
\label{eq:over:new x}
\end{align} 
for some random diagonal matrix $\mR^{(i)}\in \R^{m\times m}$ with $\tilde O(n)$ non-zero entries; in other words, 
$x^{(i+1)}_e= x^{(i)}_e+\ox^{(i)}_e   \og^{(i)}_e-\mR^{(i)}_{e,e}(\bar{f}^{(i)}_c)_e$ for every edge $e$. 
Observe that \eqref{eq:over:new x} is equivalent to how we define $x^{(i+1)}$ in our IPM (\eqref{eq:over:update_xs} and \eqref{eq:over:delta_x new}). 
Since $\mR^{(i)}$ has $\tilde O(n)$ non-zero entries\textsuperscript{\ref{foot:over:tauLS}}, we can compute $h^{(i)}=\mR^{(i)}\bar{f}^{(i)}_c$ in $\tilde O(n)$ time.\footnote{Given $d^{(i)}\in \R^n$, we can compute $(\mH^{(i)})^{-1} d^{(i)}\in \R^n$ using  spectral sparsifiers and Laplacian solvers as discussed earlier. We can then compute $(\mR^{(i)}\bar{f}^{(i)}_c)_{uv}=\mR^{(i)}_{uv,uv}(\ox^{(i)}/\os^{(i)}) (h^{(i)}_v-h^{(i)}_u)$ for every edge $uv$ such that $\mR_{uv,uv}\neq 0$.}

Our sampled edges basically form an {\em enhanced} spectral sparsifier, $\mA^\top \mR^{(i)} \mA$.
For each edge $e$, let $p_e^{(i)}$ be a probability that is proportional to the effective resistance of $e$ and $(\bar{f}^{(i)}_c)_e$. With probability $p_e^{(i)}$, we set $\mR^{(i)}_{e,e}=1/p_e^{(i)}$ and zero otherwise. Without $(\bar{f}^{(i)}_c)_e$ influencing the probability, this graph would be a standard spectral sparsifier.
Our enhanced spectral sparsifier can be constructed in $\tilde O(n)$ time using our new data structure based on the dynamic expander decomposition data structure, called {\em heavy hitter} (discussed in \Cref{sec:overview:vector,sec:matrix_vector_product}).
Compared to a standard spectral sparsifier, it provides some new properties (e.g. $\|\mR f_c\|_\infty$ is small in some sense and some moments are bounded) that allow us to bound the number of iterations to be the same as when we do not have $\mR^{(i)}$. In other words, introducing $\mR^{(i)}$ does not create additional issues (though it does change the analysis and make the guarantees probabilistic), and helps speeding up the overall runtime.

\paragraph{Computing $\ox^{(i+1)}$,  $\os^{(i+1)}$ and $\mA^\top x^{(i+1)}-b$} Above, we show how to compute $x^{(i+1)}$ in $\tilde O(n)$ time under an oversimplifying assumption that $\og^{(i)}$ is sparse. In reality, $\og^{(i)}$ may be dense and we cannot afford to compute  $x^{(i+1)}$ explicitly. A more realistic assumption (although still simplified) is that we can guarantee that the number of non-zero entries in $\og^{(i)}-\og^{(i-1)}$ is $\tilde O(\sqrt{m})$.\footnote{The actual situations  are slightly more complicated. If we use  $\taulog$, we can guarantee that we know some $t^{(i)}\in \R$, for all $i$, such that $\sum_i\|\og^{(i)}-t^{(i)}\og^{(i-1)}\|_0=\tilde O(m)$; i.e. we can obtain $\og^{(i)}$ by rescaling $\og^{(i-1)}$ and change the values of amortized $\tilde O(\sqrt{m})$ non-zero entries. We will stick with the simplified version in this subsection.
Note further that if we use $\tauLS$, we can guarantee that entries of each $\og^{(i)}$ can be divided into $\polylog(n)$ buckets where entries in the same bucket are of the same value. For every $i$, we can describe the bucketing of $\og^{(i)}$ by describing $\polylog(n)$ entries in the buckets of $\og^{(i-1)}$ that move to different buckets in the bucketing of $\og^{(i)}$. Additionally, each bucket of $\og^{(i)}$ may take different values than its  $\og^{(i-1)}$ counterpart.} 
In this case we cannot explicitly compute $\ox^{(i)}   \og^{(i)}$, and thus $x^{(i+1)}$. 
Instead, we explicitly maintain  $\ox^{(i+1)}$ such that for each edge $e$, $\ox^{(i+1)}_e$ is within a constant factor of $x^{(i+1)}_e$. 
This means that, for any edge $e$, if $\ell_e(i)$ is the last iteration before iteration $i$ that we set $\ox^{(\ell_e(i))}_e=x^{(\ell_e(i))}_e$, and $|\sum_{t=\ell_e(i)}^i \og^{(t)}_e|=\Omega(1)$, then we have to set $\ox^{(i)}_e=x^{(i)}_e$. 
Using the fact that $\og^{(i)}$ is a unit vector, we can show that we do not have to do this often; i.e. there are $\tilde O(m)$ pairs of $(i,e)$ such that $|\sum_{t=\ell_e(i)}^i \og^{(t)}_e|=\Omega(1)$.
By exploiting the fact that $\og^{(i)}-\og^{(i-1)}$ contains $\tilde O(\sqrt{m})$ non-zero entries, we can efficiently detect entries of $\ox^{(i)}$ that need to be changed from $\ox^{(i-1)}$. 
Also by the same fact, we can maintain $d^{(i)}$, thus $\mR^{(i)}\bar{f}^{(i)}_c$, in $\tilde O(n)$ time per iteration. 
This implies that we can computed $\ox^{(i+1)}$ in $\tilde O(n+\sqrt{m})=\tilde O(n)$ amortized time per iteration.

We are now left with computing $\os^{(i+1)}$ and $\mA^\top x^{(i+1)}-b$. Observe that $\delta_s^{(i)}$ (Eq.~\eqref{eq:over:delta_s}) appear as part of $\delta_x^{(i)}$  in \eqref{eq:over:delta_x new}; so, intuitively, $\os^{(i)}$ can be computed in a similar way to $\ox^{(i)}$. 
Note that although $\mR^{(i)}$ does not appear in \eqref{eq:over:delta_s}, 
we can use our heavy hitter data structure (mentioned earlier and discussed in \Cref{sec:overview:vector,sec:matrix_vector_product})
to also detect edges $e$ 
where $\os^{(i)}_e$ is no longer a good approximation of $s^{(i)}_e$.
That is, when $j$ was the last iteration when we set $\os^{(j)}_e = s^{(i)}_e$
then we can use the heavy hitter data structure to detect 
when $|s^{(i)}_e - \os^{(i)}_e| = | s^{(i)}_e - s^{(j)}_e |$ grows too large,
because the difference $s^{(i)} - s^{(j)}$ can be interpreted as some flow again.
Finally, note that $\mA^\top x^{(i+1)}-b=(\mA^\top x^{(i)}-b)+\mA^\top\omX^{(i)}   \og^{(i)}-\mA^\top\mR^{(i)}\bar{f}^{(i)}_c$. The first term is given to us. The last term can be computed quickly due to the sparsity of $\mR^{(i)}\bar{f}^{(i)}_c$. The middle term can be maintained in $\tilde O(\sqrt{m})$ time by exploiting the fact that there are $\tilde O(\sqrt{m})$ non-zero entries in $\og^{(i)}-\og^{(i-1)}$.

\subsection{Data Structures}\label{sec:over:data structures}\label{sec:over:what to compute}

As noted earlier, our IPMs are analyzed assuming that the constraint matrix $\mA$ of the linear program is non-degenerate (i.e. the matrix $(\mA^\top \mA)^{-1}$ exists).
If $\mA$ is an incidence matrix, then this is not satisfied. We fix this by appending an identity block at the bottom of $\mA$.  For proving and discussing the data structures we will, however, assume that $\mA$ is just an incidence matrix without this appended identity block, as it results in a simpler analysis. %

Ultimately we would like to compute $x^{(\ell)}$ in the final iteration $\ell$ of the IPM. However, we do {\em not} compute $x^{(i)}$ or $s^{(i)}$ in iterations $i<\ell$ because it would take to much time. Instead, we implement efficient data structures to maintain the following information about \eqref{eq:over:delta_x new} and \eqref{eq:over:delta_s} in every iteration.

\begin{enumerate}[noitemsep,label=\roman*]
	\item \label{item:over:primal maintenance} \label{item:over:gradient maintenance} 
	\textbf{Primal and Gradient Maintenance}
	Maintain vectors $\og^{(i)}$, $\mA^\top \omX^{(i)} \og^{(i)}$ and $\ox^{(i)} \in \R^{m}$. 
	
	\item \label{item:over:dual maintenance}\textbf{Dual Vectors Maintenance:}
	Maintain vector $\os^{(i)} \in \R^{m}$.
	
	\item \label{item:over:sampling}\textbf{Row (edge) sampling:}  Maintain $\mR^{(i)}$. 
	
	\item \label{item:over:inverse maintenance}\textbf{Inverse Maintenance:} Maintain ({\em implicitly}) $(\omH^{(i)})^{-1}$. Given  $w\in \R^n$, return $(\omH^{(i)})^{-1}w$. 
	
	\item \label{item:over:leverage scores maintenance}\textbf{Leverage Scores Maintenance ($\bar\tauLS(x^{(i)},s^{(i)})$):} 
	When using the faster $\tilde{O}(\sqrt{n})$-iteration IPM with potential $\tauLS$,
	we must maintain an approximation $\bar{\tauLS}(x^{(i)},s^{(i)})$ of $\tauLS(x^{(i)},s^{(i)}) = \sigma(x^{(i)},s^{(i)}) + n/m$, 
	so we can maintain $\ov^{(i)} \approx \frac{x^{(i)} s^{(i)}}{\mu^{(i)}(\tauLS x^{(i)},s^{(i)})}$  which is needed for $\og^{(i)}$ (in \eqref{eq:over:og}). 
	
	\item \label{item:over:infeasibility maintenance}\textbf{Infeasibility Maintenance:} Maintain the $n$-dimensional vector $(\mA^\top x^{(i)} - b)$.
\end{enumerate}
The above values, except for $\og^{(i)}$ (\ref{item:over:gradient maintenance}) 
and $(\omH^{(i)})^{-1}$ (\ref{item:over:inverse maintenance}), 
are computed {\em explicitly}, meaning that their values are maintained in the working memory in every iteration.
The vector $\og^{(i)}$ is maintained in an implicit form,
and for $(\omH^{(i)})^{-1}$, we maintain a data structure that, 
given $w\in \R^n$, can quickly return $(\omH^{(i)})^{-1}w$. 

\paragraph{Implementing the IPM via \ref{item:over:primal maintenance}-\ref{item:over:infeasibility maintenance}} 
Below, we repeat  \eqref{eq:over:delta_x new}, \eqref{eq:over:delta_s} and \eqref{eq:over:delta_c} to summarize how we use our data structures to maintain the information in these equations.
\begin{align}
\delta_x^{(i)} = 
  \underbrace{\omX^{(i)} \og^{(i)}}_{\text{(\ref{item:over:primal maintenance})}} 
- \underbrace{\mR^{(i)}}_{\text{(\ref{item:over:sampling})}}
  \underbrace{\omX^{(i)}}_{\text{(\ref{item:over:primal maintenance})}}
  \underbrace{(\omS^{(i)})^{-1}}_{\text{(\ref{item:over:dual maintenance})}}
  \mA  
  & 
  \underbrace{(\omH^{(i)})^{-1}}_{\text{(\ref{item:over:inverse maintenance})}}
  \underbrace{\mA^\top \omX^{(i)} \og^{(i)}}_{\text{(\ref{item:over:gradient maintenance})}} 
+ \underbrace{\mR^{(i)}\delta_c^{(i)}}_{{\text{below}}}, 
\mbox{ and} \tag{$\ref*{eq:over:delta_x new}'$}\label{eq:over:delta_x new prime}\\
\delta_s^{(i)} = 
  \mA
  &
  \underbrace{(\omH^{(i)})^{-1}}_{\text{(\ref{item:over:inverse maintenance})}}
  \underbrace{\mA^\top \omX^{(i)} \og^{(i)}}_{\text{(\ref{item:over:gradient maintenance})}}. 
\tag{$\ref*{eq:over:delta_s}'$}\label{eq:over:delta_s prime}\\
\mR^{(i)}\delta_c^{(i)} = 
  \eta_c 
  \underbrace{\mR^{(i)}}_{\text{(\ref{item:over:sampling})}}
  \underbrace{\omX^{(i)}}_{\text{(\ref{item:over:primal maintenance})}}
  \underbrace{(\omS^{(i)})^{-1}}_{\text{(\ref{item:over:dual maintenance})}}
  \mA 
  & 
  \underbrace{(\omH^{(i)})^{-1}}_{\text{(\ref{item:over:inverse maintenance})}} 
  \underbrace{(\mA^\top x^{(i)} - b)}_{\text{(\ref{item:over:infeasibility maintenance})}}.
\tag{$\ref*{eq:over:delta_c}'$}%
\end{align}
Here we see, that all information required to compute $\delta_x$ and $\delta_s$ is provided by the data structures
\ref{item:over:primal maintenance}-\ref{item:over:inverse maintenance}.

\paragraph{Constructing the data structures} 
Next, we explain how to implement these data structures efficiently.
Our main contribution with respect to the data structures are for 
\ref{item:over:primal maintenance}, 
\ref{item:over:dual maintenance}, 
and \ref{item:over:sampling} 
(primal, dual and gradient maintenance and row sampling). 
These data structures are outlined in \Cref{sec:overview:vector}.

When $\ma$ is an incidence matrix, maintaining the inverse implicitly (\ref{item:over:inverse maintenance}) 
can be done by maintaining a sparse spectral approximation $\omH^{(i)}$ 
of the Laplacian $\mA^\top \omX^{(i)} (\omS^{(i)})^{-1} \mA$ 
and then running an existing approximate Laplacian system solver 
\cite{Vaidya90, SpielmanT03, SpielmanT04, KoutisMP10, KoutisMP11, KelnerOSZ13,lee2013efficient,CohenKMPPRX14, KyngLPSS16, KyngS16}. 
The spectral approximation $\omH^{(i)}$ can be maintained using existing tools 
such as the dynamic spectral sparsifer data structure from \cite{BernsteinBNPSS20} 
or sampling from the leverage scores upper bounds.

Maintaining the leverage scores (\ref{item:over:leverage scores maintenance}) 
is done via a data structure from \cite{blss20}
which reduces leverage scores maintenance to dual slack maintenance 
(\ref{item:over:dual maintenance}) with some overhead. 
The resulting complexity for maintaining the leverage scores,
when using our implementation of \ref{item:over:dual maintenance}
is analyzed in \Cref{sec:leverage_score}.

To maintain $\mA^\top x^{(i)} - b$ (\ref{item:over:infeasibility maintenance}), 
observe that $\mA^\top x^{(i)} - b = \mA^\top x^{(i-1)} - b + \mA^\top \delta_x^{(i-1)}$.
Here $\mA^\top x^{(i-1)} - b$ is known from the previous iteration and 
\begin{align*}
  \mA^\top \delta_x^{(i-1)}
=
  \mA^\top \omX^{(i-1)} \og^{(i-1)} 
- \mA^\top \mR^{(i-1)} [\omX^{(i-1)}(\omS^{(i-1)})^{-1}\mA(\omH^{(i-1)})^{-1}\mA^\top \omX^{(i-1)} \og^{(i-1)} + \delta_c^{(i-1)}]
\end{align*} 
can be computed efficiently because of the sparsity of $\mR^{(i-1)}$ 
and the fact that we know $\mA^\top \omX^{(i-1)}\og^{(i-1)}$ 
from gradient maintenance (\ref{item:over:gradient maintenance}).

\paragraph{Time complexities}

The total time to maintain the data structures \ref{item:over:primal maintenance},
\ref{item:over:dual maintenance}, \ref{item:over:inverse maintenance}, 
and \ref{item:over:infeasibility maintenance} over $\ell$ iterations is $\tilde O(m+n\ell)$
when using the slower $\sqrt{m}$-iteration IPM,
and $\tilde{O}(m+(n+m/\sqrt{n}) \ell)$ when using the faster $\sqrt{n}$-iteration IPM.
The exception is for the leverage scores $\bar\tauLS(x^{(i)},s^{(i)})$ (\ref{item:over:leverage scores maintenance})
which is only needed for the $\sqrt{n}$-iteration IPM, 
where we need $\tilde O(m+n\ell+\ell^2m/n)$ time. 
So, in total these data structures take $\tilde O(n\sqrt{m})$ time when we use $\taulog$ 
and $\tilde O(m+n\sqrt{n})$ time using $\tauLS$.

\subsection{Primal, Dual, and Gradient Maintenance and Sampling}
\label{sec:overview:vector}

We first describe how to maintain the approximation $\os^{(i)} \approx s^{(i)}$, 
i.e.~the data structure of \ref{item:over:dual maintenance}.
Via a small modification we then obtain a data structure for \ref{item:over:sampling}. 
Finally, we describe a data structure for \ref{item:over:gradient maintenance} 
which allows us to maintain the gradient $\og$ and the primal solution $\ox$.

\paragraph{Approximation of $s$ (See \Cref{sec:matrix_vector_product,sec:vector_maintenance})}

In order to maintain an approximation $\os^{(i)} \approx s^{(i)}$ (i.e. a data structure for \ref{item:over:dual maintenance}),
we design a data structures for the following two problems:
\begin{enumerate}[start=1,label={(D\arabic*):},noitemsep]
	\item Maintain the exact vector $s^{(i)} \in \R^m$ implicitly, 
	such that any entry can be queried in $O(1)$ time.
	\item Detect all indices $j \in [m]$ for which the current $\os^{(i)}_j$ 
	is no longer a valid approximation of $s^{(i)}_j$.
\end{enumerate}
Task (D1) can be solved easily and we explain further below how to do it.
Solving task (D2) efficiently is one of our main contributions and proven in \Cref{sec:matrix_vector_product},
though we also given an outline in this section further below.
Once we solve both tasks (D1) and (D2), we can combine these data structures 
to maintain a valid approximation of $s^{(i)}$ as follows (details in \Cref{sec:vector_maintenance}): 
Whenever some entry $s^{(i)}_j$ changed a lot 
so that $\os^{(i)}_j$ is no longer a valid approximation,
(which is detected by (D2))
then we simply query the exact value via (D1) 
and update $\os^{(i)}_j \leftarrow s^{(i)}_j$.
To construct these data structure, observe that by \eqref{eq:over:delta_s prime}, we have
\begin{align*}
s^{(i+1)} = s^{(i)} + 
  \mA 
  \underbrace{(\omH^{(i)})^{-1}}_{\text{(\ref{item:over:inverse maintenance})}}
  \underbrace{\mA^\top \omX^{(i)} \og^{(i)}}_{\text{(\ref{item:over:gradient maintenance})}}
  = s^{(i)} + \mA h^{(i)}. 
\tag{$\ref*{eq:over:delta_s}'$}
\end{align*}
Here the vector $h^{(i)} \in \R^n$ can be computed efficiently, 
thanks to \ref{item:over:inverse maintenance} and \ref{item:over:gradient maintenance}.
So we are left with the problem of maintaining 
\begin{align*}
\os^{(i+1)} \approx s^{(i+1)} = s^\init + \mA \sum_{k=1}^{i} h^{(k)}.
\end{align*}
Here we can maintain $\sum_{k=1}^{i} h^{(k)}$ in $O(n)$ time per iteration 
by simply adding the new $h^{(i)}$ to the sum in each iteration.
For any $j$ one can then compute $s^{(i+1)}_j$ in $O(1)$ time so we have a data structure that solves (D1).

To get some intuition for (D2), assume we have some $\os^{(i)}$ with $\os^{(i)} \approx s^{(i)}$.
Now if an entry $(\delta_s^{(i)})_j$ is small enough,
then we have $\os^{(i)}_j \approx s^{(i)}_j + (\delta_s^{(i)})_j = s^{(i+1)}_j$.
This motivates why we want to detect a set $J \subset [m]$ containing all $j$ where $|(\delta_s^{(i)})_j|$ is large,
and then update $\os^{(i)}$ to $\os^{(i+1)}$ by setting $\os^{(i+1)}_j \leftarrow s^{(i+1)}_j$ for $j \in J$.
So for simplicity we start with the simple case where we only need to detect entries of $s^{(i+1)}$ 
that changed by a lot within a {\em single} iteration of the IPM. 
That is, we want to find every index $j$ 
such that $|(\delta_s^{(i)})_j| = |s^{(i+1)}_j-s^{(i)}_j|> \epsilon s^{(i)}_j$ for some $\epsilon \in (0,1)$; 
equivalently, 
\begin{equation}
| (\mA h^{(i)} )_j |  > \epsilon s^{(i)}_j.\label{eq:over:big entries}
\end{equation}
We assume that in each iteration we are given the vector $h^{(i)}$. 
Since $\mA$ is an incidence matrix, 
index $j$ corresponds to some edge $(u,v)$ 
and \eqref{eq:over:big entries} is equivalent to 
\begin{align}
|h^{(i)}_v - h^{(i)}_u| > \epsilon s^{(i)}_{(u,v)}\label{eq:over:big edges}
\end{align}
where $s^{(i)}_{(u,v)} := s^{(i)}_j$.
Assume by induction $\os^{(i-1)}_j \approx  s^{(i-1)}_j$ for all $j$, 
so then finding all $j$'s where $s_j$ is changed by a lot in iteration $(i+1)$ reduces to the problem of finding all edges $(u,v)$ such that 
$|(h^{(i)}_v - h^{(i)}_u) / \os^{(i)}_{(u,v)}| > \epsilon' $ 
for some $\epsilon'=\Theta(\epsilon)$. 

To get the intuition of why we can efficiently find all such edges, 
start with the simplified case when the edges have uniform weights (i.e. $\os^{(i)}=\vec{1}$). 
Since we only care about the differences between entries in the vector $h$, 
we can shift $h$ by any constant vector $c\cdot\onevec$ in $O(n)$ time to make $h\bot d$, 
where $d$ is the vector of degrees of the nodes in the graph. 
For any edge $j=(u,v)$ to have $|h_u-h_v|\geq \epsilon'$, 
at least one of $|h_u|$ and $|h_v|$ has to be at least $\epsilon'/2$. 
Thus, it suffices to check the adjacent edges of a node $u$ only when $|h_u|$ is large, 
or equivalently $h^2_u\epsilon'^{-2}$ is at least $1/4$. 
Since checking the adjacent edges of any $u$ takes time $\deg(u)$, 
the time over all such nodes is bounded by $O(\sum_u \deg(u)h_u^2\epsilon^{-2})$, 
which is $O(h^\top \mD h \epsilon^{-2})$ where $\mD$ is the diagonal degree matrix. 
If the graph $G$ has conductance at least $\phi$ (i.e., $G$ is a $\phi$-expander), 
we can exploit the classic spectral graph theory result of Cheeger's inequality 
to bound the running time by $O(h^\top \mL h \phi_G^{-2}\epsilon^{-2})$, 
where $\mL=\mA^\top \mA$ is the graph Laplacian (see \Cref{lem:cheeger_based}). 
Here $h^\top \mL h = \|\mA h\|_2^2$ will be small due to properties of our IPM, 
and thus this already gives us an efficient implementation 
if our graph has large conductance $\phi = 1/\polylog(n)$.

To extend the above approach to the real setting where $\os$ is non-uniform and the graph is not an expander, 
we only need to partition the edges of $G$ so that these two properties hold in each induced subgraph. 
For the non-uniform $\os$ part, we bucket the edges by their weights in $\os^{(i)}$ 
into edge sets $E^{(i)}_k=\{(u,v) \mid \os^{(i)}(u,v)\in [2^{k}, 2^{k+1})\}$, 
so edges in each bucket have roughly uniform $\os$ weights. 
To get the large conductance condition, we further partition each $E^{(i)}_k$ into expander subgraphs, 
i.e., $E^{(i)}_{k,1}, E^{(i)}_{k,2}, \ldots$, each inducing an $(1/\polylog(n))$-expander. 
Note that edges move between buckets over the iterations as $\os^{(i)}$ changes, 
so we need to maintain the expander decompositions in a dynamic setting. 
For this we employ the algorithm of \cite{BernsteinBNPSS20} 
(building on tools developed for dynamic minimum spanning tree 
\cite{sw19,cglnps20,nsw17,ns17,w17}, especially \cite{sw19}).
Their dynamic algorithm can maintain our expander decomposition efficiently (in $\polylog(n)$ time per weight update). 
With the dynamic expander decomposition, 
we can essentially implement the method discussed above in each expander as follows. 
For each expander subgraph, we constrain the vector $h$ to the nodes of the expander. 
Then, we translate $h$ by the all-one vector so that $h$ is orthogonal to the degree vector of the nodes in the expander. 
To perform the translation on all expanders, 
we need the total size (in terms of nodes) of the induced expanders to be small for the computation to be efficient. 
We indeed get this property as the dynamic expander decomposition algorithm in \cite{BernsteinBNPSS20} guarantees $\sum_q |V(E^{(i)}_{k,q})| =O(n \log n)$. 
In total this will bound the running time in the $i^{th}$ iteration of the IPM to be
\[
\tilde O( (\epsilon')^{-2} \| (\omS^{(i)})^{-1} \mA h^{(i)}\|_2^2 + n \log W),
\]
where $W$ is a bound on the ratio of largest to smallest entry in $\os$.
By properties of the IPM, which bound the above norm, 
the total running time of our data structure over all $\tilde O(\sqrt{n})$ iterations of the IPM 
becomes $\tilde O(m + n^{1.5})$,
or $\tilde O(\sqrt{m}n)$ when using the slower $\tilde{O}(\sqrt{m})$ iteration IPM.

In the above we only consider detecting entries of $s^{(i)}$ undergoing large changes in a single iteration. 
In order to maintain $\os^{(i)}$, we also need to detect entries of $s^{(i)}$ 
that change slowly every iteration, but accumulate enough change across multiple iterations so our approximation is no longer accurate enough. 
This can be handled via a reduction similar to the one performed in \cite{blss20}, 
where we employ lazy update and batched iteration tracking. 
In particular, for every $k = 0, 1,\ldots,\lceil(\log n)/2\rceil$, 
we use a copy of (D2) to check every $2^k$ iterations of the IPM, 
if some entry changes large enough over the past $2^k$ iterations. 
This reduction only incurs a $\polylog(n)$ factor overhead in running time comparing to the method that only detects large single iteration changes, so the total running time is the same up to $\polylog$ factors.

\paragraph{Row Sampling (Details in \Cref{sec:matrix_vector_product})}

Another task is to solve data structure problem \ref{item:over:sampling} 
which is about constructing the random matrix $\mR^{(i)}$.
The desired distribution of $\mR^{(i)}$ is as follows. 
For some large enough constant $C>0$ let $q \in \R^m$ with
\begin{align*}
q_j \ge \sqrt{m} ((\delta_r^{(i)})^2_j / \|\delta_r\|_2^2 + 1/m) + C \cdot \sigma((\omX^{(i)})^{1/2} (\omS^{(i)})^{-1/2} \mA)_j \polylog n \\
\text{where}~\delta_r^{(i)} = \omX^{(i)} (\omS^{(i)})^{-1} \mA (\omH^{(i)})^{-1} \mA \omX^{(i)} \og^{(i)} + \delta_c^{(i)}
\end{align*}
then we have $\mR^{(i)}_{j,j} = (\min(q_i, 1))^{-1}$ with probability $\min(q_i, 1)$ and $0$ otherwise.

This sampling task can be reduced to the two tasks of 
(i) sampling according to $\sqrt{m} ((\delta_r)^2_j / \|\delta_r\|_2^2$ 
and (ii) sampling according to $C \cdot \sigma((\omX^{(i)})^{1/2} (\omS^{(i)})^{-1/2} \mA) \polylog n$.
The latter can be implemented easily as we have approximate leverage scores 
via data structure \ref{item:over:leverage scores maintenance}.
The former is implemented in a similar way as data structure (D2) of the previous paragraph.
Instead of finding large entries of some vector $(\omS^{(i)})^{-1}\mA h^{(i)}$ 
as in the previous paragraph (i.e.~\eqref{eq:over:big entries}),
we now want to sample the entries proportional to $\omX^{(i)} (\omS^{(i)})^{-1} \mA h'^{(i)}$
where $h'^{(i)} = (\omH^{(i)})^{-1} (\mA \omX^{(i)} \og^{(i)} + \mA^\top x^{(i)} - b)$.

This sampling can be constructed via a simple modification of the previous (D2) data structure.
Where (D2) tries to find edges $(u,v)$ with large $|((\omS^{(i)})^{-1} \mA h^{(i)})_{(u,v)}|$
by looking for nodes $v$ with large $|h^{(i)}_v|$,
we now similarly sample edges $(u,v)$ proportional to $(\omX^{(i)} (\omS^{(i)})^{-1} \mA h'^{(i)})^2_{(u,v)}$
by sampling for each node $v$ incident edges proportional to $(h^{(i)}_v)^2$.

\paragraph{Gradient Maintenance and Approximation of $x$ (\Cref{sec:gradient_maintenance}) }

For the primal solution $x$, again we aim to maintain a good enough approximation $\ox$ through our IPM algorithm. 
Consider the update to $x^{(i)}$ in \eqref{eq:over:delta_x new prime}, 
\begin{align*}
x^{(i+1)} = x^{(i)} + \omX^{(i)} \og^{(i)}
- \mR^{(i)} \omX^{(i)} (\omS^{(i)})^{-1} \mA (\omH^{(i)})^{-1} \mA^\top \omX^{(i)} \og^{(i)}
+ \mR^{(i)}\delta_c^{(i)} 
\end{align*}
the last two terms will be sparse due to the sparse diagonal sampling matrix $\mR^{(i)}$, 
so we can afford to compute that part of the updates explicitly. 
For the part of $\omX^{(i)}\og^{(i)}$
where 
\begin{align*}
\og^{(i)} = \underset{ z \in \R^m: \|z\|\le 1 }{ \argmax } \langle \nabla\Phi(\ov^{(i)}), z \rangle
\end{align*}
(see \eqref{eq:over:og})
we will show that $\og$ admits a low dimensional representation. %
Here by low dimensionality of $\og \in \R^m$ we mean that the $m$ indices in the vector can be put into $\tilde{O}(1)$ buckets, 
where indices $j,j'$ in the same bucket share the common value $\og_j = \og_{j'}$.
This allows us to represent the values of $\og$ as a $\tilde{O}(1)$ dimensional vector
so we can efficiently represent and do computations with $\og$ in a very compact way. 

For simplicity consider the case where we use $\| \cdot \|_2$ as the norm for the maximization problem that defines $\og^{(i)}$
(this norm is used by the $\sqrt{m}$-iteration IPM, while the $\sqrt{n}$-iteration IPM uses a slightly more complicated norm).
In that case $\og^{(i)} = \nabla \Phi(\ov^{(i)})/\|\nabla \Phi(\ov^{(i)})\|_2$
and the way we construct the $\tilde{O}(1)$ dimensional approximation is fairly straightforward. 
We essentially discretize $\ov^{(i)}$ by rounding each entry down to the nearest multiple of some appropriate granularity to make it low dimensional. 
Once $\ov^{(i)}$ is made to be $\tilde{O}(1)$ dimensional, 
it is simple to see from the definition of the potential function $\Phi(\cdot)$ 
that $\nabla\Phi(\ov^{(i)})$ will also be in $\tilde{O}(1)$ dimension. 
For the faster $\sqrt{n}$ iteration IPM 
where a different norm $\|\cdot\|$ is used,
we can show that the low dimensionality of $\nabla\Phi(\ov^{(i)})$ 
also translates to the maximizer being in low dimensional (see \Cref{lem:projected_flat}).

Once we compute the low dimensional updates, 
we still need to track accumulated changes of 
$\omX^{(i)} \og^{(i)}$ over multiple iterations. 
Because of properties of the IPM, we have that on average
any index $j$ switches its bucket (of the low dimensional representation of $\og$) only some $\polylog(n)$ times.
Likewise, the value of any entry $\omX^{(i)}_j$ changes only some $\polylog(n)$ number of times.
Thus the rate in which $\sum_{k=1}^i \omX^{(k)}_j \og^{(k)}_j$ changes, stays the same for many iterations.
This allows us to (A) predict when $\ox^{(i)}_j$ is no longer a valid approximation of $x^{(i)}_j$,
and (B) the low dimensionality allows us to easily compute any $x^{(i)}_j$.
In the same way as $\os^{(i)}$ was maintained via (D1) and (D2),
we can now combine (A) and (B) to maintain $\ox^{(i)}$.

\subsection{Putting it all Together}

In summary, we provide the IPMs in \Cref{sec:ipm}: 
the first one in \Cref{sec:log_barrier_method} requires $\tilde{O}(\sqrt{m})$ iterations,
while the second IPM in \Cref{sec:ls_barrier_method} requires only $\tilde{O}(\sqrt{n})$ iterations.
These IPMs require various approximations which we maintain efficiently via data structures.
The approximation for the slack of the dual solution is maintained via the data structure presented in \Cref{sec:vector_maintenance}
and the approximation of the gradient and the primal solution is handles in \Cref{sec:gradient_maintenance}.
The $\tilde{O}(\sqrt{n})$ iteration IPM also requires approximate leverage scores
which are maintained via the data structure of \Cref{sec:leverage_score}.

As IPMs only improve an initial (non-optimal) solution, we also have to discuss how to construct such an initial solution.
This is done in \Cref{sec:initial_point:initial_point}.
Further, the IPMs only yield a fractional almost optimal and almost feasible solution,
which is why in \Cref{sec:initial_point:feasible} we prove how to convert the solution to a truly feasible solution.
Rounding this fractional almost optimal feasible solution to a truly optimal integral solution is done in \Cref{sec:initial_point:integral}.

In \Cref{sec:matching} we combine all these tools and data structures to obtain algorithms
for minimum weight perfect bipartite matching,
where \Cref{sec:matching:slow} obtains an $\tilde{O}(n\sqrt{m})$-time algorithm via the $\tilde{O}(\sqrt{m})$ iteration IPM
and \Cref{sec:matching:fast} obtains an $\tilde{O}(m + n^{1.5})$-time algorithm via the faster $\tilde{O}(\sqrt{n})$ iteration IPM.
At last, we show in \Cref{sec:app} how this result can be extended to problems 
such as uncapacitated minimum cost flows (i.e.~transshipment) and shortest paths.

\newpage

\section{IPM}
\label{sec:ipm}

Throughout this section we let $\ma\in\R^{m\times n}$ denote a non-degenerate
matrix, let $b\in\R^{n}$ and $c\in\R^{m}$, and consider the problem
of solving the following linear program $(P)$ and its dual $(D)$
\begin{equation}
(P) \defeq \min_{x\in\R_{\geq0}^{m}:\ma^{\top}x=b}c^{\top}x
\enspace \text{ and } \enspace
(D) \defeq \max_{y\in\R^{n}:\ma y\leq c}b^{\top}y\,.\label{eq:primal_dual}
\end{equation}
First, in \Cref{sec:lp_meta} we provide our general IPM framework for solving \eqref{eq:primal_dual}. Then, in \Cref{sec:log_barrier_method} we show how to instantiate this framework with the logarithmic barrier function and thereby enable $\tilde{O}(n \sqrt{m})$ time minimum-cost perfect matching algorithms. Next, in \Cref{sec:ls_barrier_method} we show how to instantiate this framework with the Lee-Sidford barrier \cite{ls14} as used in \cite{blss20} and obtain an improved $\tilde{O}(m + n^{1.5})$-time minimum-cost perfect matching algorithm. The analysis in both \Cref{sec:log_barrier_method} and \Cref{sec:ls_barrier_method} hinge on technical properties of a potential function deferred to \Cref{sec:potential_function}. Finally, additional properties of our IPM needed to obtain our final runtimes are given in \Cref{sec:stability_fixing}.

\subsection{Linear Programming Meta-algorithm}
\label{sec:lp_meta}

Here we present our core routine for solving \eqref{eq:primal_dual}. This procedure, Algorithm~\ref{alg:meta},
is a general meta algorithm for progressing along the central path
in IPMs. This general routine maintains primal-dual, approximately
feasible points induced by a weight function $\tau(x,s)$ defined
as follows.
\begin{definition}[$\epsilon$-Centered Point]
\label{def:central_point} For fixed weight function $\tau: \R_{>0}^{m} \times \R_{>0}^{m} \rightarrow \R_{>0}^{m}$, and $\gamma \in [0,1)$ we call ($x,s,\mu)\in\R_{>0}^{m}\times\R_{>0}^{m}\times\R_{>0}$,
\emph{$\epsilon$-centered }if $\epsilon\in[0,1/80)$ and the following hold
\begin{itemize}
	\item (Approximate Centrality) $w\approx_{\epsilon}\tau(x,s)$ where $w\defeq w(x,s,\mu)\in\R^{m}$
	and $[w(x,s,\mu)]_{i}\defeq\frac{x_{i}s_{i}}{\mu}$.
	\item (Dual Feasibility) There exists $y\in\R^{n}$ such that $\ma y+s=c$.
	\item (Approximate Primal Feasibility) $\|\ma^{\top}x-b\|_{(\ma^{\top}\mx\ms^{-1}\ma)^{-1}}\leq\epsilon\gamma\sqrt{\mu}$.
\end{itemize}
\end{definition}

Under mild assumptions on the weight function $\tau$, the set of
$0$-centered $(x,s,\mu)$ form a continuous curve, known as the \emph{central
path}, from a center of the polytope (at $\mu=\infty$) to the solution
to the linear program (at $\mu=0$). This path changes depending on
the weight function and in Section~\ref{sec:log_barrier_method}
and Section~\ref{sec:ls_barrier_method} we analyze two different
paths, the one induced by the standard logarithmic barrier and the
one induced by the LS-barrier \cite{ls19} respectively. The former
is simpler and yields $\otilde(n\sqrt{m})$ time matching algorithms
whereas the second is slightly more complex and yields a $\otilde(m+n^{1.5})$
time matching algorithm.

Here, we provide a unifying meta algorithm for leveraging both of
these paths. In particular, we provide Algorithm~\ref{alg:meta}
which given centered points for one value of $\mu$, computes centered
points for a target value of $\mu$. Since we provide different algorithms
depending on the choice of weight function, we present our algorithm
in a general form depending on how a certain $\textsc{ShortStep}$
procedure is implemented. Algorithm~\ref{alg:meta} specifies how
to use $\textsc{ShortStep}$ and then we implement this procedure
in Section~\ref{sec:log_barrier_method} for the logarithmic barrier
and Section~\ref{sec:ls_barrier_method} for the LS-barrier.

To analyze and understand this method we measure the quality of $(x,s,\mu)\in\R_{>0}^{m}\times\R_{>0}^{m} \times \R_{>0}$ by a potential function we call the \emph{centrality potential }.
\begin{definition}[Centrality Potential]
\label{def:central_potential} For all $(x,s,\mu)\in\R_{>0}^{m}\times\R_{>0}^{m}\times\R_{>0}$
we define the \emph{centrality potential }as follows
\begin{align*}
\phicent(x,s,\mu) \defeq \Phi(w(x,s,\mu))\defeq\sum_{i\in[n]} \phi\left(\frac{[w(x,s,\mu)]_{i}}{[\tau(x,s)]_{i}}\right)
\end{align*}
where for all $v\in\R^{m}$ we let $\phi(v)_{i}\defeq \phi(v_i) \defeq \exp(\lambda(v_{i}-1))+\exp( - \lambda( v_i - 1 ) )$
for all $i\in[n]$ and $\lambda\geq 1$ is a parameter we chose later.
Further, we let $\phi'(w),\phi''(w) \in \R^{n}$ be the vectors with
$[\phi'(w)]_{i} \defeq \phi'(w_{i})$ and $[\phi''(w)]_{i} \defeq \phi''(w_{i})$
for all $i\in[n]$.
\end{definition}

This centrality potential is the same as the one used in several recent
IPM advances, it was key to the IPM advances in \cite{ls19}, was
used extensively in \cite{cls19,lsz19,b20,blss20}, and most recently was
used in the IPM of \cite{JiangSWZ20}.
Correspondingly, the analysis of
this meta-algorithm is primarily a clean abstraction for these
papers and properties of the potential function are deferred to \Cref{sec:potential_function}. The novelty in our IPM is the attainment of efficient implementations
of $\textsc{ShortStep}$ in Section~\ref{sec:log_barrier_method}
and Section~\ref{sec:ls_barrier_method}. 

In the remainder of this section we define a valid $\textsc{ShortStep}$
procedure (Definition~\ref{lem:path_following_master}), provide
$\textsc{PathFollowing}$ (Algorithm~\ref{alg:meta}) which uses
it, and analyze this algorithm (Lemma~\ref{lem:path_following_master}).

\newcommand{\PathFollowing}{\textsc{PathFollowing}}

\begin{algorithm2e}[H]

\caption{Path Following Meta Algorithm}\label{alg:meta}

\SetKwProg{Globals}{global variables}{}{}

\SetKwProg{Proc}{procedure}{}{}

\SetKwProg{VirtualProc}{virtual procedure}{}{}

\Proc{$\PathFollowing(x^{\init}\in\R_{>0}^{m},s^{\init}\in\R_{>0}^{m},\mu^{\init}>0,\mu^{\target}>0)$}{

\LineComment{Assume $\textsc{ShortStep}$ is a $(\epsilon,r,\lambda,\gamma,\tau)$-short
step procedure}

\textbf{\LineComment{Assume $x^{\init}s^{\init}\approx_{\epsilon/2}\mu^{\init}\cdot\tau(x^{\init},s^{\init})$
}}

\textbf{\LineComment{Assume $\|\ma x^{\init}-b\|_{(\ma^{\top}\mx^{\init}\ms^{\init-1}\ma)^{-1}}\le\epsilon\gamma\sqrt{\mu^{\init}}$
}}

$x\leftarrow x^{\init},s\leftarrow s^{\init},\mu\leftarrow\mu^{\init}$

\While{$\mu\neq\mu^{\target}$}{

$x^{(0)}\leftarrow x,s^{(0)}\leftarrow s,\mu^{(0)}\leftarrow\mu$

\For{$i=1,2,\cdots,\frac{\epsilon}{r}$}{ \label{line:for_loop_start}

  $\mu^{\text{(new)}}\leftarrow\text{median}(\mu^{\target},(1-r)\mu,(1+r)\mu)$

  $(x,s)\leftarrow\textsc{ShortStep}(x,s,\mu,\mu^{(\text{new})})$

  $\mu \leftarrow \mu^{\text{(new)}}$

  Compute $p_{1}\approx_{1/3}\phicent(x,s,\mu),p_{2}\approx_{1/3}\norm{\ma^{\top}x-b}_{(\ma^{\top}\mx\ms^{-1}\ma)^{-1}}$. 

\lIf{$p_{1}>e^{-1/3}\exp(\frac{3\lambda\epsilon}{4})$ or $p_{2}>\epsilon\gamma e^{-1/3}\sqrt{\mu}$}{\textbf{break}\label{line:meta:fail}}

}

\If{$p_{1}>e^{-1/3}\exp(\frac{\lambda\epsilon}{4})$ or $p_{2}>\epsilon\gamma e^{-1/3}\sqrt{\mu}$}{

$x\leftarrow x^{(0)},s\leftarrow s^{(0)},\mu\leftarrow\mu^{(0)}$\label{line:failed_meta}

}

}
\Return$(x,s)$

} \label{line:for_loop_end}

\end{algorithm2e}
\begin{definition}[Short Step Procedure]
\label{def:short_step} We call a $\textsc{ShortStep}(x,s,\mu,\mu^{\new})$
a \emph{$(\epsilon,r,\lambda,\gamma,\tau)$-short step procedure}
if $\epsilon \in (0,1/80)$, $r \in (0,1/2)$, $\lambda\geq 12 \epsilon^{-1} \log(16m/r) $,
$\tau:\R_{>0}^{m}\times\R_{>0}^{m}\rightarrow\R_{>0}^{m}$, and $\gamma\in[0,1)$
and given input $x,s\in\R_{>0}^{m}$ and $\mu,\mu^{\new}\in\R_{>0}$
such that $(x,s,\mu)$ is $\epsilon$-centered and $|\mu^{\new}-\mu|\leq r\cdot\mu$
the procedure outputs random $x^{\new},s^{\new}\in\R_{>0}^{m}$ such
that
\begin{enumerate}
\item $\ma y^{\new}+s^{\new}=c$ for some $y^{\new}\in\R^{n}$,
\item $\E
[ \phicent( x^{\new},s^{\new},\mu^{\new} ) ] \leq(1-\lambda r) \phicent(x,s,\mu)+\exp( \lambda \epsilon /6 )$
(for $\Phi$ defined by $\tau$ and $\lambda)$
\item $\P\left[\norm{\ma^{\top}x^{\new}-b}_{(\ma^{\top}\mx^{\new}(\ms^{\new})^{-1}\ma)^{-1}}\leq(\epsilon/2)\gamma\sqrt{\mu^{\new}}\right]\geq 1 - r.$
\end{enumerate}
\end{definition}

\begin{lemma}
\label{lem:path_following_master} Let $x^{\textrm{\ensuremath{\mathrm{(init)}}}},s^{\textrm{\ensuremath{\mathrm{(init)}}}}\in\R_{>0}^{m}$
and $\mu^{\textrm{\ensuremath{\mathrm{(init)}}}},\mu^{\textrm{\ensuremath{\mathrm{(end)}}}}>0$
satisfy $\ma y^{\init}+s^{\init}=c$ for some $y^{\init}\in\R^{n}$
and satisfy 
\begin{align*}
x^{\init}s^{\init}\approx_{\epsilon/2}\mu^{\init}\cdot\tau(x^{\init},s^{\init})\enspace\text{ and }\enspace\|\ma^{\top} x^{\init}-b\|_{(\ma^{\top}\mx^{\init}\ms^{\init-1}\ma)^{-1}}\le\epsilon\gamma\sqrt{\mu^{\init}}\,.
\end{align*}
 $\textsc{PathFollowing}(x^{\init},s^{\init},\mu^{\init},\mu^{\target})$
(Algorithm~\ref{alg:meta}) outputs $(x^{\target},s^{\target})$
such that 
\begin{align*}
x^{\target}s^{\target}\approx_{\epsilon/3}\mu^{\target}\cdot\tau(x^{\target},s^{\target})\enspace\text{ and }\enspace\|\ma^{\top} x^{\target}-b\|_{(\ma^{\top}\mx^{\target}\ms^{\target-1}\ma)^{-1}}\le\epsilon\gamma\sqrt{\mu^{\target}}
\end{align*}
in expected time $O(\frac{1}{r}\mathcal{T})\cdot|\log(\mu^{\textrm{\ensuremath{\mathrm{(end)}}}}/\mu^{\textrm{\ensuremath{\mathrm{(init)}}}})|$
where $O(\mathcal{T})$ expected time to compute $p_{1}$, $p_{2}$
such that $p_{1}\approx_{1/3}\Phi(x,s,\mu)$ and $p_{2}\approx_{1/3}\norm{\ma^{\top}x-b}_{(\ma^{\top}\mx\ms^{-1}\ma)^{-1}}$
and to invoke $\textsc{ShortStep}$ once on average. Further, the $(x,s,\mu)$ at each \Cref{line:for_loop_start} and the termination of the algorithm are $\epsilon$-centered.
\end{lemma}

\begin{proof}
Note that the algorithm divides into phases of size $\frac{\epsilon}{r}$
(Line~\ref{line:for_loop_start} to Line~\ref{line:for_loop_end}).
In each phase, if $\Phi$ or $\|\ma^{\top}x-b\|_{(\ma^{\top}\mx\ms^{-1}\ma)^{-1}}$
is too large, then Line \ref{line:failed_meta} occurs and the phase
is repeated. To bound the runtime, we need to show that Line \ref{line:failed_meta}
does not happen too often. Notationwise, we use $x^{(i)},s^{(i)},\mu^{(i)}$ to denote the state of $x,s,\mu$ at the end of the $i$-th iteration in a phase, and $x^{(0)},s^{(0)},\mu^{(0)}$ denote the vectors at the beginning of a phase which are stored in case we need to revert to those vectors.

First, we claim that throughout the algorithm, $\Phi$ and $\|\ma^{\top}x-b\|_{(\ma^{\top}\mx\ms^{-1}\ma)^{-1}}$
are small at the beginning of each phase, i.e. \Cref{line:for_loop_start}.
\begin{align*}
\Phi( x^{(0)},s^{(0)},\mu^{(0)} ) \leq \frac{r}{8} \exp( \frac{3}{4} \lambda\epsilon )\text{ and }\norm{\ma^{\top}x^{(0)}-b}_{(\ma^{\top}\mx^{(0)}\ms^{(0)-1}\ma)^{-1}}\leq\epsilon\gamma\sqrt{\mu^{(0)}}\,.
\end{align*}
At the beginning of the algorithm, we have by assumption that 
\begin{align*}
xs\approx_{\epsilon/2}\mu\tau(x,s)\text{ and }\norm{\ma^{\top}x-b}_{(\ma^{\top}\mx\ms^{-1}\ma)^{-1}} \leq \epsilon\gamma\sqrt{\mu}\,.
\end{align*}
Since $\epsilon \in [0,1/8]$, we have $\left|\frac{x_{i}s_{i}}{\mu\cdot\tau(x,s)_{i}} - 1\right|\leq\frac{2}{3}\epsilon$.
Using $\lambda \geq  12 \epsilon^{-1} \log(16m/r) $ and Lemma~\ref{lem:smoothing:helper}
yields that
\begin{align*}
\Phi(x,s,\mu) \leq 2m \cdot \exp \left(\frac{2}{3}\lambda\epsilon\right)\leq2\left(\frac{r}{16} \exp \left(\frac{1}{12} \lambda\epsilon \right)\right)\cdot \exp \left( \frac{2}{3}\lambda\epsilon \right)  \leq \frac{r}{8}\exp\left(\frac{3}{4}\lambda\epsilon\right).
\end{align*}
Further, when we replace $(x^{(0)},s^{(0)},\mu^{(0)})$ (i.e., when the last phase is successful), we
have $p_{1} \leq e^{-1/3}\exp(\frac{\lambda\epsilon}{4})$ and $p_{2}\leq\epsilon\gamma e^{-1/3}\sqrt{\mu}$.
By the definition of $p_{1}$ and $p_{2}$, and that $(1/2) \leq \exp(-1/3)$
this implies the same desired bound on $\Phi$ and $\|\ma^{\top}x-b\|_{(\ma^{\top}\mx\ms^{-1}\ma)^{-1}}$
and the claim holds. Similarly, as long as the phase continues (i.e. the conditions at the end of each iteration hold), we have that for all iterations $i$ 
\begin{align*}
\Phi(x^{(i)},s^{(i)},\mu^{(i)})\leq\exp\left(\frac{3}{4}\lambda\epsilon\right)\text{ and }\norm{\ma^{\top}x^{(i)}-b}_{(\ma^{\top}\mx^{(i)}\ms^{(i)-1}\ma)^{-1}}\leq\epsilon\gamma\sqrt{\mu}.
\end{align*}
Note that $\Phi(x,s,\mu)\leq\exp(\frac{3}{4}\lambda\epsilon)$
implies $\left|\frac{x_{j}s_{j}}{\mu\cdot\tau(x,s)_{j}}-1\right|\leq\frac{3}{4}\epsilon$ for all coordinates
by Lemma~\ref{lem:smoothing:helper}. Using $\epsilon\in[0,1/8]$,
this implies that $x^{(i)}s^{(i)}\approx_{\epsilon}\mu^{(i)}\tau(x^{(i)},s^{(i)})$. Combined with
property 1 of a short step procedure, this implies that $(x^{(i)},s^{(i)},\mu^{(i)})$
is $\epsilon$-centered for  each iteration $i$.

Now, we show that Line \ref{line:failed_meta} does not happen too
often. Note that there are three reasons Line \ref{line:failed_meta}
happens, one is that it exceeds the threshold for $\|\ma^{\top}x^{(i)}-b\|_{(\ma^{\top}\mx^{(i)}\ms^{(i)-1}\ma)^{-1}}$
at some iteration, one is that it exceeds the threshold for $\Phi$ at some
iteration, another is that it exceeds the threshold for $\Phi$
at the end of the phase.

For $\| \ma^{\top} x^{(i)} - b \|_{(\ma^{\top}\mx^{(i)}\ms^{(i)-1}\ma)^{-1}}$: conditioned on the phase reaching iteration $i$, we can use the guarantee of $\textsc{ShortStep}$ to get
\begin{align*}
\norm{\ma^{\top}x^{(i)}-b}_{(\ma^{\top}\mx^{(i)}\ms^{(i)-1}\ma)^{-1}}\leq\frac{\epsilon\gamma}{2}\sqrt{\mu^{(i)}}
\end{align*}
with probability $1-r$, in which case $p_{2}\leq(\epsilon/2)e^{1/3}\gamma\sqrt{\mu^{(i)}}\leq\epsilon e^{-1/3}\gamma\sqrt{\mu^{(i)}}$. Thus, the probability that the phase fails at iteration $i$ due to the check on $p_2$ is at most $r$, and union bound over $\epsilon/r$ iterations gives that with probability at least $1-\epsilon$, the algorithm passes per iteration threshold
on $\| \ma^{\top} x^{(i)} - b \|_{(\ma^{\top}\mx^{(i)}\ms^{(i)-1}\ma)^{-1}}$ for all iterations in a phase.

For $\Phi$: again conditioned on the phase reaching iteration $i$, the guarantee of $\textsc{ShortStep}$ yields
\begin{align*}
\E [ \Phi( x^{(i)}, s^{(i)}, \mu^{(i)} ) ] & \leq (1-\lambda r) \Phi(x^{(i-1)},s^{(i-1)},\mu^{(i-1)})+\exp( \frac{1}{6} \lambda\epsilon )\\
 & \leq (1-\lambda r) \Phi(x^{(i-1)},s^{(i-1)},\mu^{(i-1)})+\frac{r^2}{8} \exp( \frac{3}{4}\lambda\epsilon ).
\end{align*}
where the last inequality is because of the $\log(240m/\epsilon^2)$ term in the definition of $\lambda$. We know at the beginning of the phase, $\Phi(x^{(0)},s^{(0)},\mu^{(0)})\leq\frac{r}{8}\exp(\frac{3}{4}\lambda\epsilon)$
by the first claim, and knowing the phase reaches iteration $i$ won't make $\E [ \Phi(x^{(k)},s^{(k)},\mu^{(k)})]$ increase for any $k\in[0,i]$. Together with the bound above, we know $\E[ \Phi(x^{(i)},s^{(i)},\mu^{(i)}) ] \leq \frac{r}{8}\exp(\frac{3}{4}\lambda\epsilon)$. By Markov's inequality we have that $\Phi(x^{(i)},s^{(i)},\mu^{(i)})  \leq \frac{1}{8}\exp(\frac{3}{4}\lambda\epsilon)$
with probability at least $1-r$, and in this case $p_{1} \leq \frac{e^{1/3}}{8}\exp(\frac{3}{4}\lambda\epsilon)\leq e^{-1/3}\exp(\frac{3}{4}\lambda\epsilon)$. Thus, the (a priori) probability that the algorithm fails the check at the end of iteration $i$ is at most $r$, and union bound over the $\epsilon/r$ iterations yields that with probability at least $1-\epsilon$ the check for $\Phi$ is passed in each iteration of a given phase. 

Furthermore, by expanding the guarantee of $\textsc{ShortStep}$ over the $\epsilon/r$ iterations, we know at the end of the phase
\begin{align*}
\E [ \Phi(x^{(\epsilon/r)},s^{(\epsilon/r)},\mu^{(\epsilon/r)}) ] & \leq(1-\lambda r)^{\frac{\epsilon}{r}}\Phi(x^{(0)},s^{(0)},\mu^{(0)})+\frac{1-(1-\lambda r)^{\frac{\epsilon}{r}}}{\lambda r}\exp\left(\frac{1}{6}\lambda\epsilon\right)\\
 & \leq e^{-\lambda\epsilon}\frac{r}{8}\exp \left(\frac{3}{4}\lambda\epsilon\right)+\frac{1}{\lambda r}\exp\left(\frac{1}{6}\lambda\epsilon\right)\\
 & \leq\frac{1}{8}\exp\left(\frac{1}{4}\lambda\epsilon\right)
\end{align*}
where we used that $\lambda \geq  12 \epsilon^{-1} \log(16m/r) $. With probability at least $3/4$, $\Phi(x,s,\mu) \leq \frac{1}{2}\exp(\frac{1}{4}\lambda\epsilon)$ at the end of the phase, and hence $p_{1}\leq\frac{e^{1/3}}{2} \exp (\frac{1}{4}\lambda\epsilon) \leq e^{-1/3}\exp(\frac{1}{4}\lambda\epsilon)$.
Namely, it passes end of phase check for $\Phi$.

Union bounds over all three triggers, each phases success without
reset with probability at least $\frac{1}{4}$. Hence, in expected
$O(1/r)$ iterations, $\mu$ increases or decreases by $\Theta(1)$.
This yields the runtime $O(\frac{1}{r}\mathcal{T})\cdot|\log(\mu^{\textrm{\ensuremath{\mathrm{(end)}}}}/\mu^{\textrm{\ensuremath{\mathrm{(init)}}}})|$.

Finally, we note that when the algorithm ends, we have $p_{1}\leq e^{-1/3}\exp(\frac{1}{4}\lambda\epsilon)$
which implies $\Phi \leq \exp(\frac{1}{4}\lambda\epsilon)$ and $x^{\target}s^{\target}\approx_{\epsilon/3}\mu^{\target}\cdot\tau(x^{\target},s^{\target})$.
\end{proof}

\subsection{Short-step Method for Log Barrier Weights}
\label{sec:log_barrier_method}

In this section we provide and analyze a short step when the weight
function is constant, i.e. $\tau(x,s)\defeq\vones\in\R^{m}$ for all
$x,s\in\R_{>0}^{m}$. The central path induced by this weight function
is also the one induced by the logarithmic barrier function, i.e.
$(x_{\mu},s_{\mu},\mu)$ is $0$-centered if and only if
\begin{align}
\label{eq:log_barrier}
x_{\mu} =  ~ \underset{x\in\R_{\geq0}^{m} : \ma^{\top}x=b}{\argmin} c^{\top} x - \mu\sum_{i\in[m]}\log(x_{i})
\text{ and }
s_{\mu} = ~ \underset{s\in\R_{\geq0}^{m},y\in \R^{n} : \ma y+s=c}{\argmax} b^{\top} y + \mu\sum_{i\in[m]}\log(s_{i})\,.
\end{align}

Our algorithm using this weight function is given as Algorithm~\ref{alg:short_step_log}.
 The algorithm simply computes an approximate gradient of $\Phi(x,s,\mu)$
for the input $(x,s,\mu)$ and then takes a natural Newton step to
decrease the value of this potential as $\mu$ changes and make $x$
feasible. The key novelty in our step over previous IPMs is that we
sparsify the projection onto $x$ by random sampling. We show that steps on $x$ can be made $\otilde(\sqrt{m}+n)$-sparse,
ignoring the effect of a known additive term, 
by carefully picking a sampling distribution. Further, we show that the amount
of infeasibility this sampling induces can be controlled and  $\otilde(\sqrt{m})$ iterations of the method suffice. Ultimately, this leads to a $\otilde(n\sqrt{m})$ time
algorithm for matching.

\begin{algorithm2e}[th!]

\caption{Short Step (Log Barrier Weight Function)}
\label{alg:short_step}
\label{alg:short_step_log}

\SetKwProg{Proc}{procedure}{}{}

\Proc{$\textsc{ShortStep}(x \in \R^m_{>0},s \in \R^m_{>0},\mu \in \R_{>0},\mu^{\new} \in \R_{>0})$  }{

\LineComment{Fix $\tau(x,s)\defeq\vones\in\R^{m}\text{ for all }x\in\R_{>0}^{m},s\in\R_{>0}^{m}$
in Definition~\ref{def:central_potential}}

\LineComment{Let $\lambda\defeq 36\epsilon^{-1} \log(240m/\epsilon^{2}) $, $\gamma\defeq \epsilon / ( 100\lambda )$, $r \defeq \epsilon\gamma/ \sqrt{m}$, $\epsilon \leq 1/80$. }

\textbf{\LineComment{Assume $(x,s,\mu)$ is $\epsilon$-centered
}}

\textbf{\LineComment{Assume $\delta_{\mu}\defeq\mu^{\new}-\mu$ satisfies
$|\delta_{\mu}|\leq r\mu$}}

\vspace{0.2in}

\LineComment{\textbf{Approximate the iterate and compute step direction}}

\State Find $\ox$ and $\os$ such that $\ox\approx_{\epsilon}x$
and \textbf{$\os\approx_{\epsilon}s$} \label{line:log:approx_xs} \label{line:logstep:xstau}

\State Find $\ov\in\R_{>0}^{m}$ such that $\norm{\ov - w}_{\infty}\leq\gamma$ \label{line:log:v} for $w \defeq w(x,s,\mu)$ 

\State $g\leftarrow-\gamma\nabla\Phi(\overline{v})/\norm{\nabla\Phi(\overline{v})}_{2}$ \label{line:log:g}

\vspace{0.2in}

\LineComment{\textbf{Sample rows}}

\State $\oma\defeq\omx^{1/2}\oms^{-1/2}\ma,\ow\defeq w(\ox,\os,\mu),\omx\defeq\mdiag(\ox),\oms\defeq\mdiag(\os),\omw\defeq\mdiag(\ow)$ \label{line:log:define_shorthand}

\State Find $\omh$ such that $\omh\approx_{\gamma}\oma^{\top}\oma$\label{line:logstep:omh}

\State Let $\delta_{p}\defeq\omw^{-1/2}\oma\omh^{-1}\oma^{\top}\omw^{1/2}g$,
$\delta_{c}\defeq\mu^{-1/2}\cdot\omw^{-1/2}\oma\omh^{-1}(\ma^{\top}x-b)$,
and $\delta_{r}\defeq\delta_{p}+\delta_{c}$ \label{line:log:delta_r}\label{line:log:delta_c}

\State Find $q\in\Rn_{>0}$ such that $q_{i}\geq\sqrt{m}\cdot\left((\delta_{r})_{i}^{2}/\|\delta_{r}\|_{2}^{2}+1/m\right)+C\cdot\sigma_{i}(\oma)\log(m/(\epsilon r))\gamma^{-2}$
for some large enough constant $C$.

\State Sample a diagonal matrix $\mr\in\R^{n\times n}$ where each
$\mr_{ii}=\min\{q_{i},1\}^{-1}$ independently with probability $\min\{q_{i},1\}$
and $\mr_{ii}=0$ otherwise \label{line:log:R}

\vspace{0.2in}

\LineComment{\textbf{Compute projected step}\footnote{This step arises as a combination of an approximate solution to $\ms\delta_{x}+\mx\delta_{s}=(\mx\ms)h$,
$\ma^{\top}x=b$, and $\ma y+s=b$ and an additional step on $x$
to fix the current infeasibility. This definition of $h$ is analogous
to what we do in the tall dense paper and makes $h$ multiplicative
and unit-less.}}

\State $\delta_{x}\leftarrow\omx(g-\mr\delta_{r})$ \label{line:log:delta_x}

\State $\delta_{s}\leftarrow\oms\delta_{p}$ \label{line:log:delta_s}

\State $x^{\new}\leftarrow x+\delta_{x}$ and $s^{\new}\leftarrow s+\delta_{s}$ \label{line:log:move_xs}

\Return$(x^{\new},s^{\new})$

}

\end{algorithm2e}

Our main result of this section is the following lemma which shows that \Cref{alg:short_step_log}
is a valid $\textsc{ShortStep}$ procedure for the meta-algorithm of \Cref{sec:lp_meta}. In this section
we simply prove this lemma. Later, we show how to efficiently implement
\Cref{alg:short_step_log} to obtain efficient linear programming
and matching algorithms and improve using the Lee-Sidford
barrier.
\begin{lemma}[Log Barrier Short Step]
\label{lem:main_lem} \Cref{alg:short_step_log} is an $(\epsilon, r, \lambda, \gamma, \tau)$
$\textsc{ShortStep}$ procedure (see \Cref{def:short_step}) for all $\epsilon \in (0, 1/80)$
with weight function $\tau(x,s)\defeq\vones\in\R^{m}\text{ for all }x,s\in\R_{>0}^{m}$
and parameters $\lambda\defeq 36 \epsilon^{-1} \log(240m/\epsilon^{2}) $,
let $\gamma\defeq \epsilon / (100\lambda)$, and $r \defeq \epsilon\gamma/\sqrt{m}$,
i.e.
\begin{enumerate}
\item \label{lem:log_main_claim_1} $\lambda \geq 12 \epsilon^{-1} \log(16m/r) $,
\item \label{lem:log_main_claim_2} $\ma y^{\new}+s^{\new}=c$ for some $y^{\new}\in\R^{n}$,
\item \label{lem:log_main_claim_3} $\E [ \phicent(x^{\new},s^{\new},\mu^{\new}) ] \leq (1-\lambda r)\phicent(x,s,\mu)+\exp( \lambda \epsilon / 6)$
\item \label{lem:log_main_claim_4} $\P\left[\norm{\ma^{\top}x^{\new}-b}_{(\ma^{\top}\mx^{\new}(\ms^{\new})^{-1}\ma)^{-1}}\leq\frac{\epsilon}{2}\sqrt{\mu^{\new}}\right]\geq1-r\,.$
\end{enumerate}
\end{lemma}

We prove \Cref{lem:main_lem} at the end of this
section in Section~\ref{sec:putting_it_all_together}. 
There we prove claims \ref{lem:log_main_claim_1} and \ref{lem:log_main_claim_2}, 
which follow fairly immediately by straightforward calculation. 
The bulk of this section is dedicated to proving claim \ref{lem:log_main_claim_3}, 
i.e. \Cref{lem:centrality_main} in \Cref{sec:centrality}, 
which essentially amounts to proving
that the given method is an $\otilde(\sqrt{m})$-iteration IPM. 
This proof leverages several properties of the centrality potential,
$\Phi$ given in Section~\ref{sec:potential_function}. 
Finally, claim \cref{lem:log_main_claim_4} of the lemma follows by careful calculation 
and choice of the sampling distribution $q$ and is given
towards the end of this section as Lemma~\ref{lem:infeasibility_bound}
in Section~\ref{sec:feasibility}.

\subsubsection{Centrality\label{sec:centrality}}

Here we analyze the effect of Algorithm~\ref{alg:short_step_log} on
the centrality potential $\Phi$ to prove claim \ref{lem:log_main_claim_3} of Lemma~\ref{lem:main_lem}.
Our proof leverages \Cref{lem:expand_potential} and \Cref{lem:centrality_main_general} which together 
reduces bounding the potential $\phicent(x^{\new},s^{\new},\mu^{\new})$ to carefully analyzing
bounding the first and second moments of $\delta_{s}$, $\delta_{x}$,
and $\delta_{\mu}$ (as well as there multiplicative stability). 

More precisely, \Cref{lem:expand_potential} shows that provided  $x^{\new}$, $s^{\new}$, and $\mu^{\new}$ do not change too much multiplicatively, i.e. $\norm{\mx^{-1} \delta_x}_\infty$, $\norm{\ms^{-1}  \delta_s}_\infty$, and $|\delta_{\mu} / \mu|$ are sufficiently small, then
\begin{align}
\E [\phicent(x^{\new},s^{\new},\mu^{\new})]
& \leq\Phi(x,s,\mu)+\phi'(w)^{\top}\mw\left[(\mx^{-1}\E [ \delta_{x} ] +\ms^{-1}\delta_{s}-\mu^{-1}\delta_{\mu}\vones)\right]\nonumber \\
& \enspace\enspace\enspace + O\left(\E \norm{\mx^{-1}\delta_{x}}_{\phi''(w)}^{2}+\norm{\ms^{-1}\delta_{x}}_{\phi''(w)}^{2}+\norm{\mu^{-1}\delta_{\mu}\vones}_{\phi''(w)}^{2}\right) 
\label{eq:expansion_phinew}
\,.
\end{align}
(To see this, pick 
$u^{(1)} = x$, $\delta^{(1)} = \delta_{x}$,
$c^{(1)} = 1$, $u^{(2)} = s$, $\delta^{(2)} = \delta_{s}$, $c^{(2)}=1$,
$u^{(3)} = \mu\vones$, $\delta^{(3)} = \delta_{\mu}\vones$, and $c^{(3)}=-1$ in  \Cref{lem:expand_potential} as we do in the proof of \Cref{lem:log_exp_progress} later in this section).  \Cref{lem:expand_potential} follows by careful analysis of the second order Taylor expansion of $\Phi$ and multiplicative stability of $\phi$. 

Now, \eqref{eq:expansion_phinew} implies that to bound that the centrality potential $\Phi$ decreases in expectation, it suffices to show the following:
\begin{itemize}
	\item \textbf{First-order Expected Progress Bound:} The first-order contribution of $\delta_x$ and $\delta_s$ in expectation, i.e. $\omw [\omx^{-1} \E[\delta_x] + \oms^{-1} \delta_s]$ is sufficiently well-aligned with $-\phi'(w)$, i.e. $g$ (up to scaling). (We show this in \Cref{lem:expect_bound}.)
	\item \textbf{Multiplicative Stability and Second Order $s$ Bound:} $s$ does not change too much multiplicatively, i.e. $\norm{\oms^{-1} \delta_s}_\infty$ is small, and its second order contribution to \eqref{eq:expansion_phinew} is small, i.e. $\norm{\ms^{-1}\delta_{x}}_{\phi''(w)}^{2}$ is small. (We show this in \Cref{lem:log_s_second}.)
	\item \textbf{Multiplicative Stability and Second Order $x$ Bound:} $x$ does not change too much multiplicatively, i.e. $\norm{\omx^{-1} \delta_x}_\infty$ is small with probability $1$, and its expected second order contribution to \eqref{eq:expansion_phinew} is small, i.e. $\E \norm{\mx^{-1}\delta_{x}}_{\phi''(w)}^{2}$ is small. (We show this in \Cref{lem:stoch_x_bounds}.)
\end{itemize}
We provide these results one at a time, after first providing some basic bounds on the components of a step, i.e. $\delta_p$ and $\delta_c$, in \Cref{lem:projection_sizes}. With these results in hand and some basic analysis of $\delta_\mu$, in \Cref{lem:log_exp_progress}, we can show that
\[
\E  [\phicent(x^{\new},s^{\new},\mu^{\new})]
 \leq\Phi(x,s,\mu) + \phi'(w)^{\top} g + O(\norm{\phi'(w)}_2) +O(\norm{\phi''(w)}_2)
\]
for some sufficiently small constants hidden in the $O(\cdot)$ notation above. 
We obtain the desired \Cref{lem:centrality_main} by applying \Cref{lem:centrality_main_general} 
which leverages properties of $\phi$ 
to show that progress of this forms
yields to a sufficiently large decrease to $\Phi$ 
whenever it is sufficiently large.

\begin{lemma}[Basic Step Size Bounds]
\label{lem:projection_sizes} In Algorithm~\ref{alg:short_step_log}
we have that
\begin{equation}
\norm{\delta_{p}}_{2}\leq\exp(4\epsilon)\gamma\text{ and }\norm{\delta_{c}}_{2}\leq\epsilon\exp(4\epsilon)\gamma\label{eq:step_size_bounds}
\end{equation}
and consequently, as $\epsilon<1/80$
\begin{align*}
\norm{\ms^{-1}\delta_{s}}_{2}\leq\exp(5\epsilon)\gamma\text{ and }\norm{\mx^{-1}\E[\delta_{x}]}_{2}\leq3\gamma\,.
\end{align*}
\end{lemma}

\begin{proof}
Recall that $\delta_{p}=\omw^{-1/2}\oma\omh^{-1}\oma^{\top}\omw^{1/2}g$
and $\omh\approx_{\gamma}\oma^{\top}\oma$ for $\gamma \leq \epsilon$. Since $(x,s,\mu)$ is
$\epsilon$-centered we have $\ow\approx_{3\epsilon}\vones$ ($\ow\approx_{2\epsilon}w$
and $w\approx_{\epsilon}1$). Consequently, $\omw^{-1} \preceq \exp(-3\epsilon) \mi$ and
\begin{equation}
\label{eq:log_barrier:deltap_1}
\norm{\delta_{p}}_{2} 
= \norm{\oma\omh^{-1}\oma^{\top}\omw^{1/2}g}_{\omw^{-1}} 
\leq \exp(3\epsilon/2)\norm{\omh^{-1}\oma^{\top}\omw^{1/2}g}_{\oma^{\top}\oma} ~.
\end{equation}
Now, applying $\omh\approx_{\gamma}\oma^{\top}\oma$ and $\gamma \leq \epsilon$ twice we have
\begin{equation}
\label{eq:log_barrier:deltap_2}
\norm{\omh^{-1}\oma^{\top}\omw^{1/2}g}_{\oma^{\top}\oma}
\leq \exp(\epsilon/2)  \norm{\oma^{\top}\omw^{1/2}g}_{\omh^{-1}}
\leq \exp(\epsilon)  \norm{\omw^{1/2}g}_{\oma(\oma^{\top}\oma)^{-1}\oma^{\top}} ~.
\end{equation}
Further, since $\oma(\oma^{\top}\oma)^{-1}\oma^{\top}$ is a projection matrix we have $\mzero \preceq \oma(\oma^{\top}\oma)^{-1}\oma^{\top} \preceq \mI$ and
\begin{equation}
\label{eq:log_barrier:deltap_3}
 \norm{\omw^{1/2}g}_{\oma(\oma^{\top}\oma)^{-1}\oma^{\top}}
\leq \norm{g}_{\ow} 
\leq  \exp(3\epsilon/2)\norm g_{2}
= \exp(3\epsilon/2) \gamma ~,
\end{equation}
where in the second to last step we again used $\ow\approx_{3\epsilon}\vones$  and in the last step we used $\norm g_{2}=\gamma$. Combining \eqref{eq:log_barrier:deltap_1}, \eqref{eq:log_barrier:deltap_2}, and \eqref{eq:log_barrier:deltap_3} yields that $\norm{\delta_p}_2 \leq \exp(4\epsilon) \gamma$ as desired.

Next, as $\ox\approx_{\epsilon}x$ and $\os\approx_{\epsilon}s$ we
have $\oma^{\top}\oma\approx_{2\epsilon}\ma^{\top}\mx\ms^{-1}\ma$. Since $\omh\approx_{\gamma}\oma^{\top}\oma$ for $\gamma \leq \epsilon$
this implies $\omh\approx_{3\epsilon}\ma^{\top}\mx\ms^{-1}\ma$. Since, $\delta_{c}=\mu^{-1/2}\omw^{-1/2}\oma\omh^{-1}(\ma^{\top}x-b)$
and again $\ow\approx_{3\epsilon}\vones$ we have 
\begin{align*}
\norm{\delta_{c}}_{2} 
&= \frac{1}{\sqrt{\mu}}\norm{\oma\omh^{-1}[\ma^{\top} x - b]}_{\omw^{-1}} 
\leq  \frac{\exp(3\epsilon/2)}{\sqrt{\mu}}\norm{\omh^{-1}[\ma^{\top}x-b]}_{\oma^{\top}\oma} \\ 
&\leq  \frac{\exp(2\epsilon)}{\sqrt{\mu}}\norm{\ma^{\top}x-b}_{\omh^{-1}}
\leq  \frac{\exp(7\epsilon/2)}{\sqrt{\mu}}\norm{\ma^{\top}x-b}_{(\ma^{\top}\mx\ms^{-1}\ma)^{-1}} \
\leq \epsilon\exp(4\epsilon)\gamma.
\end{align*}
where in the last step we used that $\norm{\ma^{\top}x-b}_{(\ma^{\top}\mx\ms^{-1}\ma)^{-1}}\leq\epsilon\gamma\sqrt{\mu}$
since $(x,s,\mu)$ is $\epsilon$-centered.

Since, $\norm{\ms^{-1}\delta_{s}}_{2}\leq\exp(\epsilon)\norm{\oms^{-1}\delta_{s}}_{2}=\exp(\epsilon)\norm{\delta_{p}}_{2}$
this yields the desired bound on $\norm{\oms^{-1}\delta_{s}}_{2}$
as well. The remainder of the result follows as by definition and
triangle inequality we have
\begin{align*}
\norm{\omx^{-1}\E[\delta_{x}]}_{2}=\norm{g-\delta_{p}-\delta_{c}}_{2}\leq\norm g+\norm{\delta_{p}}_{2}+\norm{\delta_{c}}_{2}
\end{align*}
and we have $(1+(1+\epsilon)\exp(3\epsilon))\leq3$ for $\epsilon\in[0,1/80]$.
\end{proof}

We now want to prove in \Cref{lem:expect_bound} 
that $\omw [\omx^{-1} \E[\delta_x] + \oms^{-1} \delta_s]$ 
is sufficiently well-aligned with $g$.
For that we will use the following \Cref{lem:mult_to_add} which 
will be proven in \Cref{sec:potential_function}
and bounds additive $\ell_2$-error with respect to multiplicative spectral-error.
\begin{restatable*}{lemma}{multToAdd}
	\label{lem:mult_to_add} If $\mm\approx_{\epsilon}\mn$ for symmetric
	PD $\mm,\mn\in\R^{n\times n}$ and $\epsilon\in[0,1/2)$ then 
	\begin{align*}
	\norm{\mn^{-1/2}(\mm-\mn)\mn^{-1/2}}_{2}\leq\epsilon+\epsilon^{2}\,.
	\end{align*}
	Consequently, if $x\approx_{\epsilon}y$ for $x,y\in\R_{>0}$ and
	$\epsilon\in[0,1/2)$ then $\norm{\mx^{-1}y-\vones}_{\infty}\leq\epsilon+\epsilon^{2}$.
\end{restatable*}

\begin{lemma}[First-order Expected Progress Bound]
\label{lem:expect_bound} In Algorithm~\ref{alg:short_step_log} we
have
\begin{align*}
\norm{\mw(\mx^{-1}\E[\delta_{x}]+\ms^{-1}\E[\delta_{s}])-g}_{2}\leq7\epsilon\gamma\,.
\end{align*}
\end{lemma}

\begin{proof}
Note that 
\begin{align*}
\mx^{-1}\E[\delta_{x}]+\ms^{-1}\E[\delta_{s}] & =\omx^{-1}\E[\delta_{x}]+\oms^{-1}\E[\delta_{s}]+(\omx^{-1}\mx-\mi)\mx^{-1}\E[\delta_{x}]+(\oms^{-1}\ms-\mi)\ms^{-1}\E[\delta_{s}]\,.
\end{align*}
Further, by the design of $\E[\delta_{x}]$ and $\E[\delta_{s}]$
we have $\omx^{-1}\E[\delta_{x}]+\oms^{-1}\E[\delta_{s}]=g-\delta_{c}$.
Also, since $\omx\approx_{\epsilon}\mx$ and $\oms\approx_{\epsilon}\ms$,
Lemma~\ref{lem:mult_to_add} implies that $\norm{\omx^{-1}\mx-\mi}_{2}\leq\epsilon+\epsilon^{2}\leq\epsilon\exp(\epsilon)$
and similarly, $\norm{\oms^{-1}\ms-\mi}_{2}\leq\epsilon\exp(\epsilon)$.
Therefore Lemma~\ref{lem:projection_sizes} implies
\begin{align*}
&~
\norm{\mw(\mx^{-1}\E[\delta_{x}]+\ms^{-1}\E[\delta_{s}])-\mw g}_{2} \\
\leq & ~
\norm{\delta_{c}}_{\mw^{2}}+\norm{(\omx^{-1}\mx-\mi)\omx^{-1}\E[\delta_{x}]}_{\mw^{2}}+\norm{(\oms^{-1}\ms-\mi)\oms^{-1}\E[\delta_{s}]}_{\mw^{2}}\\
\leq & ~
\exp(\epsilon)\left(\norm{\delta_{c}}_{2}+\epsilon \exp(\epsilon) \left(\norm{\omx^{-1}\E[\delta_{x}]}_{2}+\norm{\oms^{-1}\E[\delta_{s}]}_{2}\right)\right)\\
\leq & ~
\exp(\epsilon)(\epsilon\exp(4\epsilon)+\epsilon\exp(\epsilon) (3+\exp(5\epsilon))\gamma \\
= & ~ \epsilon \gamma \cdot \Big( \exp(\epsilon) \cdot ( \exp(4\epsilon) + \exp(\epsilon) ) \cdot (3 + \exp(5 \epsilon)) \Big).
\end{align*}
Further, as $w\approx_{\epsilon}\mi$, Lemma~\ref{lem:mult_to_add}
implies
\begin{align*}
\norm{\mw g-g}_{2}\leq\norm{\mw-\mi}_{2}\norm g_{2}\leq(\epsilon+\epsilon^{2})\gamma = \epsilon \gamma \cdot ( 1 + \epsilon ).
\end{align*}
The result follows as $\exp(\epsilon)(\exp(4\epsilon)+(\exp(\epsilon))(3+\exp(5\epsilon)))+(1+\epsilon)\leq7$
for $\epsilon\in[0,1/80]$.
\end{proof}

\begin{lemma}[Multiplicative Stability and Second Order $s$ Bound]
\label{lem:log_s_second}In Algorithm~\ref{alg:short_step_log} we have
\begin{align*}
\norm{\ms^{-1}\delta_{s}}_{\infty}\leq\exp(5\epsilon)\gamma\text{ and }\norm{\ms^{-1}\delta_{s}}_{\phi''(v)}^{2}\leq\exp(10\epsilon)\gamma^{2}\norm{\phi''(v)}_{2}\,.
\end{align*}
\end{lemma}

\begin{proof}
The result follows from Lemma~\ref{lem:projection_sizes}, $\norm{\ms^{-1}\delta_{s}}_{\infty}\leq\norm{\ms^{-1}\delta_{s}}_{2}$,
and 
\begin{align*}
\norm{\ms^{-1}\delta_{s}}_{\phi''(v)}^{2}\leq\norm{\ms^{-1}\delta_{s}}_{\infty}\cdot\norm{\ms^{-1}\delta_{s}}_{2}\cdot\norm{\phi''(v)}_{2}\,.
\end{align*}
\end{proof}
\begin{lemma}[Multiplicative Stability and Second Order $x$ Bound]
\label{lem:stoch_x_bounds}In Algorithm~\ref{alg:short_step_log} we
have 
\begin{align*}
\mathbb{P} \Big[ \norm{\mx^{-1}\delta_{x}}_{\infty}\leq\exp(6\epsilon)\gamma \Big] = 1 {~and~} \E\Big[ \norm{\mx^{-1}\delta_{x}}_{\phi''(v)}^{2} \Big]\le7\gamma^{2}\norm{\phi''(v)}_{2}\,.
\end{align*}
\end{lemma}

\begin{proof}
Recall that $\delta_{x}=\omx g-\omx\mr\delta_{r}$ and $x\approx_{\epsilon}\ox\in\R^{m}$.
Consequently,
\begin{align*}
\norm{\mx^{-1}\delta_{x}}_{\infty} & \leq\exp(\epsilon)\cdot(\norm g_{\infty}+\norm{\mr\delta_{r}}_{\infty})\,.
\end{align*}
Now,  $\mr_{ii} \leq [\min\{1,q_{i}\}]^{-1}=\max\{1,q_{i}^{-1}\}$ by definition and therefore
\begin{align*}
\norm{\mr\delta_{r}}_{\infty} 
& \leq \max_{i\in[m]}\max\left\{ |[\delta_{r}]_{i}|\,,\,\frac{\left|[\delta_{r}]_{i}\right|}{\sqrt{m}\left( [\delta_{r}]_{i}^{2} / \norm{\delta_{r}}_{2}^{2} + 1/m \right)}\right\} 
\leq \max \left\{ \norm{\delta_{r}}_{\infty}\,,\,\frac{1}{2}\norm{\delta_{r}}_{2}\right\} 
\leq\norm{\delta_{r}}_{2}
\end{align*}
where we used that $q_i \geq \sqrt{m} \left( [\delta_{r}]_{i}^{2} / \norm{\delta_{r}}_{2}^{2} + 1/m \right)$ by assumption and
\begin{align}
\label{eq:entry_to_l2}
\left|[\delta_{r}]_{i}\right|
\leq \frac{\sqrt{m}}{2\|\delta_{r}\|_{2}}[\delta_{r}]_{i}^{2} + \frac{\|\delta_{r}\|_{2}}{2\sqrt{m}}
= \frac{\sqrt{m}}{2} \|\delta_{r}\|_{2} \cdot ( {[\delta_{r}]_{i}^{2}} / {\|\delta_{r}\|_{2}^{2}} + 1/m )\,.
\end{align}
Since by definition and triangle inequality $\norm{\delta_{r}}_{2}\leq\norm{\delta_{p}}_{2}+\norm{\delta_{c}}_{2}$
the desired bound on $\norm{\mx^{-1}\delta_{x}}_{\infty}$ follows
from Lemma~\ref{lem:projection_sizes} and $\exp(5\epsilon)(1+\epsilon)\leq2$
for $\epsilon\in[0,1/80]$.

Now, since $\ox\approx_{\epsilon}x$ the definition of $\delta_{x}$
and $\norm{a+b}_{2}^{2}\leq2\norm a_{2}^{2}+2\norm b_{2}^{2}$ imply
\begin{align*}
\norm{\mx^{-1}\delta_{x}}_{\phi''(v)}^{2}\leq\exp(2\epsilon)\norm{g-\mr\delta_{r}}_{\phi''(v)}^{2}\leq2\exp(2\epsilon)\norm g_{\phi''(v)}^{2}+2\exp(2\epsilon)\norm{\mr\delta_{r}}_{\phi''(v)}^{2}\,.
\end{align*}
Now, 
\begin{align*}
\norm g_{\phi''(v)}^{2}\leq\norm{\phi''(v)}_{2}\norm g_{\infty}\norm g_{2}\leq\norm{\phi''(v)}_{2}\norm g_{2}^{2}=\norm{\phi''(v)}_{2}\gamma^{2}\,.
\end{align*}
Further, the definition of $\mr$ and that%
\begin{align*}
[ \min\{1,q_{i}\} ]^{-1} \leq 1 + \frac{1}{q_{i}} \leq 1 + \frac{\norm{\delta_{r}}_{2}^{2}}{\sqrt{m}[\delta_{r}]_{i}^{2}}
\end{align*}
imply that 
\begin{align*}
\E_{\mr}\left[\normFull{\mr\delta_{r}}_{\phi''(v)}^{2}\right] 
= & ~ \sum_{i\in[m]}[\phi''(v)]_{i}\cdot\frac{[\delta_{r}]_{i}^{2}}{\min\{q_{i},1\}} \\
\leq & ~ \sum_{i\in[m]}[\phi''(v)]_{i}\cdot\left[[\delta_{r}]_{i}^{2}+\frac{1}{\sqrt{m}}\norm{\delta_{r}}_{2}^{2}\right]\\
\leq & ~ \norm{\phi''(v)}_{2}\norm{\delta_{r}}_{\infty}\norm{\delta_{r}}_{2}+\norm{\phi''(v)}_{2}\norm{\delta_{r}}_{2}^{2} \\
= & ~ 2\norm{\phi''(v)}_{2}\norm{\delta_{r}}_{2}^{2}\,.
\end{align*}
The result follows again from $\norm{\delta_{r}}_{2}\leq\norm{\delta_{p}}_{2}+\norm{\delta_{c}}_{2}$,
Lemma~\ref{lem:projection_sizes}, and $2\exp(2\epsilon)+4\exp(2\epsilon)((1+\epsilon)\exp(4\epsilon))^{2}\leq7$
for $\epsilon\in[0,1/80]$.
\end{proof}

With the previous bounds on first and second order error,
we now want to bound in \Cref{lem:log_exp_progress} how much our potential $\Phi$ changes.
For that we first require the following expansion of $\Phi$ obtained via Taylor-approximation.
This expansion is proven in \Cref{sec:potential_function}.

\begin{restatable*}{lemma}{expandPotential}
\label{lem:expand_potential} For all $j\in[k]$ let vector $u^{(j)}\in\Rn_{>0}$,
exponent $c_{j}\in\R$ and change $\delta^{(j)}\in\Rn$ induce vectors
$v,v^{\new}\in\R^{n}$ defined by
\begin{align*}
v_{i}\defeq\prod_{j\in[k]}(u_{i}^{(j)})^{c_{j}}\text{ and }v_{i}^{\new}\defeq\prod_{j\in[k]}(u_{i}^{(j)}+\delta_{i}^{(j)})^{c_{j}}\text{ for all }i\in[n]\,.
\end{align*}
Further, let $\mv\defeq\mdiag(v)$ and $\mU_{j}\defeq\mdiag(u^{(j)})$
and suppose that $\norm{\mU_{j}^{-1}\delta^{(j)}}_{\infty}\leq\frac{1}{12(1+\lambda\norm c_{1})}$
for all $j\in[k]$ and $v\leq\frac{13}{12}\vones$. Then $v^{\new}\approx_{1/(11\lambda)}v$
and
\begin{align*}
\Phi(v^{\new})\leq & \Phi(v)+\sum_{j\in[k]}c_{j}\phi'(v)^{\top}\mv\mU_{j}^{-1}\delta^{(j)}+2(1+\norm c_{1})\sum_{j\in[k]}|c_{j}|\cdot\norm{\mU^{-1}\delta^{(j)}}_{\phi''(v)}^{2}.
\end{align*}
\end{restatable*}

We now have everything to bound the effect of a step in terms of $\norm{\phi'(w)}_{2}$
and $\norm{\phi''(w)}_{2}$.

\begin{lemma}[Expected Potential Decrease]
\label{lem:log_exp_progress} In Algorithm~\ref{alg:short_step_log}
we have
\begin{align*}
\E [ \phicent(x^{\new},s^{\new},\mu^{\new}) ] \leq \phicent(x,s,\mu) + \phi'(w)^{\top}g + 9\epsilon\gamma\norm{\phi'(w)}_{2}+80\gamma^{2}\norm{\phi''(w)}_{2} .
\end{align*}
\end{lemma}

\begin{proof}
We apply Lemma~\ref{lem:expand_potential} with $u^{(1)} = x$, $\delta^{(1)} = \delta_{x}$,
$c^{(1)} = 1$, $u^{(2)} = s$, $\delta^{(2)} = \delta_{s}$, $c^{(2)}=1$,
$u^{(3)} = \mu\vones$, $\delta^{(3)} = \delta_{\mu}\vones$, and $c^{(3)}=-1$.
Note that, $v$ in the context of Lemma~\ref{lem:expand_potential}
is precisely $w(x,s,\mu)$ and $v^{\new}$ in the context of Lemma~\ref{lem:expand_potential}
is precisely $w( x + \delta_{x} , s + \delta_{s} , \mu + \delta_{\mu} )$.

Now, by Lemma~\ref{lem:log_s_second}, the definition of $\gamma$,
and $\epsilon\in[0,1/80]$ we have
\begin{align*}
\norm{\ms^{-1}\delta_{s}}_{\infty} \leq 2 \gamma \leq \frac{1}{50\lambda} \leq \frac{1}{12(1+\norm c_{1}\lambda)}\,.
\end{align*}
Similarly, Lemma~\ref{lem:stoch_x_bounds} implies
\begin{align*}
\norm{\mx^{-1}\delta_{x}}_{\infty} \leq 2\gamma \leq \frac{1}{50\lambda} \leq \frac{1}{12(1+\norm c_{1}\lambda)}\,.
\end{align*}
Further, by definition of $r$ we have $\mu^{-1}|\delta_{\mu}| \leq \gamma\leq(12(1+\norm c_{1}\lambda))^{-1}$l. Thus, Lemma~\ref{lem:expand_potential} implies 
\begin{align}
\E[\phicent(x^{\new},s^{\new},\mu^{\new})]
& \leq\Phi(x,s,\mu)+\phi'(w)^{\top}\mw\left[(\mx^{-1}\E [ \delta_{x} ] +\ms^{-1}\delta_{s}-\mu^{-1}\delta_{\mu}\vones)\right]\nonumber \\
 & \enspace\enspace\enspace+8\left(\norm{\E \mx^{-1}\delta_{x}}_{\phi''(w)}^{2}+\norm{\ms^{-1}\delta_{x}}_{\phi''(w)}^{2}+\norm{\mu^{-1}\delta_{\mu}\vones}_{\phi''(w)}^{2}\right)\,.\label{eq:step}
\end{align}

Now, since $w\approx_{\epsilon}\vones$, the assumption on $\delta_{\mu}$
implies 
\begin{align*}
|(\delta_{\mu}/\mu)\phi'(w)^{\top}\mw\vones|\leq|(\delta_{\mu}/\mu)|\exp(\epsilon)\norm{\phi'(w)}_{2}\sqrt{m}\leq2\epsilon\gamma\norm{\phi'(w)}_{2}
\end{align*}
and
\begin{align*}
\norm{(\delta_{\mu}/\mu)\vones}_{\phi''(w)}^{2}=|(\delta_{\mu}/\mu)|^{2}\norm{\phi''(w)}_{1}\leq ( \epsilon\gamma / \sqrt{m} )^{2}\sqrt{m}\norm{\phi''(w)}_{2}\leq\gamma^{2}\norm{\phi''(w)}_{2}\,.
\end{align*}
The result follows by applying
Lemmas~\ref{lem:expect_bound},~\ref{lem:log_s_second},~and~\ref{lem:stoch_x_bounds} to the  expectation of (\ref{eq:step}).
\end{proof}

To obtain our desired bound on $\E [ \phicent(x^{\new},s^{\new},\mu^{\new}) ] $
we will use the following general \Cref{lem:centrality_main_general}
to obtain bounds on $\norm{\phi'(w)}_{2}$
and $\norm{\phi''(w)}_{2}$.
\Cref{lem:centrality_main_general} is proven in \Cref{sec:potential_function}.

\begin{restatable*}{lemma}{centralityMainGeneral}
\label{lem:centrality_main_general}
	Let $U\subset\R^{m}$ be an axis-symmetric convex set\footnote{A set $U \subset\R^{m}$ is axis-symmetric if $x \in U$ implies $y \in U$ for all $y \in \R^m$ with $y_i \in \{-x_i, x_i\}$ for all $i \in [m]$.}
	which is contained in the $\ell_{\infty}$ ball of
	radius $1$ and contains the $\ell_{\infty}$ ball of radius $u\leq 1$, i.e. $x \in U$ implies $\norm{x}_\infty \leq 1$ and for some $\norm{x}_\infty \leq u$ implies $x \in U$. Let
	$w,v\in\R^{m}$ such that $\|w-v\|_{\infty}\leq\delta\leq\frac{1}{5\lambda}$.
	Let $g=-\gamma\argmax_{z\in U}\left\langle \nabla\Phi(v),z\right\rangle $
	and $\norm h_{U}\defeq\max_{z\in U}\left\langle h,z\right\rangle $
	with $\gamma\geq0$ and
	\begin{align*}
	\delta_{\Phi}\defeq\phi'(w)^{\top}g+c_{1}\gamma\norm{\phi'(w)}_{U}+c_{2}\gamma^{2}\norm{\phi''(w)}_{U}
	\end{align*}
	for some $c_{1},c_{2}\geq0$ satisfying $\lambda\gamma\leq\frac{1}{8}$,
	$2\lambda\delta+c_{2}\lambda\gamma\leq 0.5 ( 1 - c_{1} )$. Then, 
	\begin{align*}
	\delta_{\Phi}\leq & - 0.5 (1-c_{1} ) \lambda\gamma u\Phi(w)+m
	\end{align*}
\end{restatable*}

Finally, we can prove the desired bound on $\E [ \phicent(x^{\new},s^{\new},\mu^{\new}) ] $
by combining the preceding lemma with known bounds on $\norm{\phi'(w)}_{2}$
and $\norm{\phi''(w)}_{2}$. We first provide a very general lemma
for this which we can use in both the log barrier and LS barrier analysis
and then we use it to prove our desired bound.

\begin{lemma}[Centrality Improvement of Log Barrier Short-Step] \label{lem:centrality_main}
In Algorithm~\ref{alg:short_step_log} we have 
\begin{align*}
\E [ \phicent(x^{\new},s^{\new},\mu^{\new}) ] \leq(1-\lambda r)\phicent(x,s,\mu)+\exp(\lambda\epsilon / 6 )\,.
\end{align*}
\end{lemma}

\begin{proof}
Consider an invocation of \Cref{lem:centrality_main_general} with $\delta = \gamma$, $c_{1}=9\epsilon$, $c_{2}=80$, $U$ set to be the unit $\ell_2$ ball. Note that the assumptions of \Cref{lem:centrality_main_general} hold as $\gamma=\epsilon/(100\lambda)$. Further, as $U$ contains the $\ell_{\infty}$ ball of radius $\frac{1}{\sqrt{m}}$, $\|\ov-w\|_{\infty}\leq \gamma$, and $g = -\gamma\nabla\Phi(\overline{v})/\norm{\nabla\Phi(\overline{v})}_{2}$ we see that \Cref{lem:centrality_main_general}  applied to the claim of \Cref{lem:log_exp_progress} yields that 
\begin{align*}
\E [ \phicent(x^{\new},s^{\new},\mu^{\new}) ] & \leq\phicent(x,s,\mu)-\left(\frac{1-c_{1}}{2}\right)\frac{\lambda\gamma}{\sqrt{m}}\phicent(x,s,\mu)+m\\
& \leq \left(1-\frac{\lambda\gamma}{4\sqrt{m}} \right) \phicent(x,s,\mu)+m.
\end{align*}
The result follows from $\exp( \lambda \epsilon / 6)\geq m$
and $r=\frac{\epsilon\gamma}{\sqrt{m}}$.

\end{proof}

\subsubsection{Feasibility \label{sec:feasibility}}

Here we prove item 3 of \cref{lem:main_lem}, analyzing the effect
of Algorithm~\ref{alg:short_step_log} on approximate $x$ feasibility. Our proof proceeds by first showing that if instead of taking the step we take in \Cref{alg:short_step_log}, we instead took an idealized step where $\omh$ was replaced with $\oma^\top \oma$ and $\mr$ was replaced with $\mI$, then the resulting $x$ would have exactly met the linear feasibility constraints. We then bound the difference between linear feasbility of the actual step and this idealized step. Ultimately, the proof reveals that so long as $\oma^{\top}\mr\oma$ approximates $\oma^{\top}\oma$, and $\omh$ approximates $\oma^\top \oma$ well enough and the step is not too large, then the infeasibility is not too large either.

\begin{lemma}[Feasibility Bound]\label{lem:infeasibility_bound} 
In \Cref{alg:short_step_log}
if $\oma^{\top}\mr\oma\approx_{\gamma}\oma^{\top}\oma$ then
\begin{equation}\label{eq:log_feas_bound}
\norm{\ma^{\top}x^{\new}-b}_{\ma^{\top}\mx^{\new}(\ms^{\new})^{-1}\ma}\leq 0.5 \epsilon\gamma ( \mu^{\new} )^{1/2}
\end{equation}
and consequently, this holds with probability at least $1-r$.
\end{lemma}

\begin{proof}
Recall that $x^{\new}=x+\delta_{x}$ where $\delta_x = \omx(g - \mr \delta_r)$ and $\delta_{r}\defeq\delta_{p}+\delta_{c}$. Further, recall that
\[
\delta_{p} = \omw^{-1/2}\oma\omh^{-1}\oma^{\top}\omw^{1/2}g 
\text{ and }
\delta_{c} = \mu^{-1/2}\cdot\omw^{-1/2}\oma\omh^{-1}(\ma^{\top}x-b)
~.
\]
Consequently, we can rewrite the step as
\[
\delta_r = \omw^{-1/2} \oma \omh^{-1} d 
\text{ where }
d \defeq \oma^\top \omw^{1/2} g +  \mu^{-1/2} (\ma^\top x - b) ~.
\]

Now, consider the idealized step $x^*$ where there is no matrix approximation error,  i.e. $\omh = \oma^\top \oma$, and no sampling error, i.e. $\mI = \mR$. Formally, $x^* \defeq x + \delta_x^*$ where $\delta_x^* \defeq \omx(g - \delta_r^*)$, and $\delta_r^* = \omw^{-1/2} \oma (\oma^\top \oma)^{-1} d$. Now since, $\ma^{\top}\omx=\sqrt{\mu}\oma^{\top}\omw^{1/2}$  we have
\[
\ma^\top x^*
= \ma^\top x
+ \sqrt{\mu} \oma^\top \left( \mw^{1/2} g - \oma (\oma^\top \oma)^{-1} d \right)
= \ma^\top x - (\ma^\top x - b) = b ~.
\]
Thus, we see that $x^*$ obeys the linear constraints for feasibility. Consequently, it suffices to bound the error induced by matrix approximation error and sampling error, i.e. 
\begin{align}
\ma^\top x - b
&= \ma^\top (x - x^*)
= - \ma^\top \omx 
(\mR \delta_r - \delta_r^*)
=  \sqrt{\mu} \oma^\top \left(
\oma (\oma^\top \oma)^{-1} - \mr \oma \omh^{-1} \right) d
\nonumber
\\
&= \sqrt{\mu} \left(\mI - \oma^\top \mR \oma \omh^{-1}\right) d
\label{feasibility:equality}
\end{align}
where we used again that $\ma^{\top}\omx=\sqrt{\mu}\oma^{\top}\omw^{1/2}$  and that diagonal matrices commute.

Now, since $\oma^{\top}\mr\oma\approx_{\gamma}\oma^{\top}\oma\approx_{\gamma}\omh$
we have
\begin{align*}
\norm{(\oma^{\top}\oma)^{-1/2}(\mi-\oma^{\top}\mr\oma\omh^{-1})\omh^{1/2}}_{2} & =\norm{(\oma^{\top}\oma)^{-1/2}(\omh-\oma^{\top}\mr\oma)\omh^{-1/2}}_{2}\\
 & \leq\exp(\gamma)\norm{\omh^{-1/2}(\omh-\oma^{\top}\mr\oma)\omh^{-1/2}}_{2}\leq3\gamma
\end{align*}
where we used Lemma~\ref{lem:mult_to_add} and $(2\gamma+4\gamma^{2})\exp(\gamma)\leq3\gamma$
for $\gamma\leq1/80$. Consequently, combining with \eqref{feasibility:equality} yields that\footnote{For further intuition regarding this proof, note that, ignoring the term $\norm{\oma^{\top}\omw^{-1/2}g}_{\omh^{-1}}$ and difference between the $\ma^{\top}\mx^{\new}(\ms^{\new})^{-1}\ma$ and $\oma^\top \oma$ norm, the inequality below shows that the feasibility of the primal solution is improved by a factor of $3\gamma < 1$. This improvement comes from the term $\delta_c$ defined in Line~\ref{line:log:delta_c}. The remainder of the proof after this step is therefor restricted to bounding the impact from the term  $\norm{\oma^{\top}\omw^{-1/2}g}_{\omh^{-1}}$ which is the impact of $\delta_p$ defined in Line~\ref{line:log:delta_c}.}
\begin{align*}
\norm{\ma^{\top}x^{\new}-b}_{(\oma^{\top}\oma)^{-1}} 
&=
\sqrt{\mu}\normFull{(\oma^{\top}\oma)^{-1/2}(\mi-\oma^{\top}\mr\oma\omh^{-1})\omh^{1/2}\omh^{-1/2}
d}_{2}\\
&\leq 
3\gamma\sqrt{\mu} \norm{d}_{\omh^{-1}}
\leq
3\gamma\sqrt{\mu}\left(\norm{\oma^{\top}\omw^{-1/2}g}_{\omh^{-1}}+\mu^{-1/2}\cdot\norm{\ma^{\top}x-b}_{\omh^{-1}}\right)\,.
\end{align*}
Further, by design of $g$ and that $\omw \approx_{3\epsilon} \mI$ and $3\epsilon \leq \gamma$
\begin{align*}
\norm{\oma^{\top}\omw^{-1/2}g}_{\omh^{-1}} 
\leq & ~ \exp(\gamma/2)\norm{\omw^{-1/2}g}_{\oma(\oma^{\top}\oma)^{-1}\oma^{\top}} 
\leq  \exp(\gamma/2)\norm{\omw^{-1/2}g}_{2} 
\leq  \exp(2\gamma)\norm g_{2}\leq2\gamma
\end{align*}
and by the approximate primal feasibility of $(x,s,\mu)$ we have
\begin{align*}
\norm{\ma^{\top}x-b}_{\omh^{-1}}\leq\exp(\gamma/2)\norm{\ma^{\top}x-b}_{\ma^{\top}\mx\ms^{-1}\ma}\leq2\epsilon\gamma\sqrt{\mu}\,.
\end{align*}
Combining yields that 
\begin{align*}
\norm{\ma^{\top}x^{\new}-b}_{(\oma^{\top}\oma)^{-1}}\leq3\gamma\sqrt{\mu}\left(2\gamma+2\epsilon\gamma\right)\leq9\gamma\sqrt{\mu}\gamma\,\leq 0.25 \epsilon\gamma \sqrt{\mu}
~.
\end{align*}

Now, since $\norm{\ms^{-1}\delta_{s}}_{\infty}\leq2\gamma$ and $\norm{\mx^{-1}\delta_{x}}_{\infty}\leq2\gamma$
by Lemma~\ref{lem:log_s_second} and Lemma~\ref{lem:stoch_x_bounds}
and $\gamma\leq\epsilon/100$ we have that $\ma^{\top}\mx^{\new}(\ms^{\new})^{-1}\ma\approx_{\epsilon}\ma^{\top}\mx\ms\ma$.
Combining with the facts that $\ma^{\top}\mx\ms\ma\approx_{2\epsilon}\oma^{\top}\oma$,
and $\mu^{\new}\approx_{2\epsilon\gamma}\mu$ yields (\ref{eq:log_feas_bound}).
The probability bound follows by the fact that $q_{i}\geq C\sigma_{i}(\oma)\log(n/(\epsilon r))\gamma^{-2}$
and Lemma 4 in \cite{clmmps15}. That is, the matrix $\mr$ corresponds to sampling each row of $\ma$ with probability at least proportional to its leverage score. 
\end{proof}

\subsubsection{Proof of Lemma~\ref{lem:main_lem}}

\label{sec:putting_it_all_together}

We conclude by proving Lemma~\ref{lem:main_lem}, the main claim
of this section.
\begin{proof}[Proof of Lemma~\ref{lem:main_lem}]

(1) Note that $r=\frac{\epsilon\gamma}{\sqrt{m}}=\frac{\epsilon^{2}}{100\lambda\sqrt{m}}=\frac{\epsilon^{3}}{3600\log(240\beta)\sqrt{m}}$
	for $\beta=\frac{m}{\epsilon^{2}}$. Thus, 
\begin{align*}
	12 \epsilon^{-1} \log(16m/r) 
	\leq 12 \epsilon^{-1} \log(16\beta^{3/2}\cdot3600\log(240\beta)) 
	\leq 12 \epsilon^{-1} \log((240)^{3}\beta^{2})
	\leq\lambda
\end{align*}
	where we used that $\sqrt{x}\log x\leq x$ for $x\geq1$, $(240)^{3}=16\cdot3600\cdot240$,
	and $\beta\geq1$
	
	(2) Note that $\delta_{s}=\mu^{1/2}\ma\omh^{-1}\oma^{\top}\omw^{1/2}g$. Consequently,
	$\delta_{s}\in\im(\ma)$ and this follows from the fact that $\ma y+s=c$
	for some $y\in\R^{n}$ by definition of $\epsilon$-centered.
	
	(3) and (4) these follow immediately from Lemma~\ref{lem:centrality_main}
	and Lemma~\ref{lem:infeasibility_bound} respectively.
\end{proof}

\subsection{Short-step Method for Lee-Sidford Weights}
\label{sec:ls_barrier_method}

In this section we show how to efficiently implement a short step
when the weight function is regularized variant of the weights induced by Lee-Sidford barriers \cite{ls19}, i.e. it is defined as follows and for brevity, we refer to it as the Lee-Sidford (LS)-weight function. 

\begin{definition}[Lee-Sidford Weight Function]
\label{def:LS_weight} We define the \emph{Lee-Sidford (LS) weight function}
for all $x,s \in \R_{>0}^{m}$ as $\tau(x,s) \defeq \sigma(x,s) + \frac{n}{m} \vones$
where $\sigma(x,s) \defeq \sigma( \mx^{1/2-\alpha} \ms^{-1/2-\alpha} \ma )$ for  $\alpha \defeq \frac{1}{4\log(4m/n)}$, and $\sigma$ defined as in \Cref{sec:preliminaries}.
\end{definition}

This weight function is analogous to the one used in \cite{blss20} which in turn is inspired by the weight functions and barriers used in \cite{ls19} (see these works for further intuition).

Our short-step algorithm for this barrier (\Cref{alg:short_step_log}) is very similar to the one for the logarithmic barrier (\Cref{alg:short_step_LS}). The primary difference is what $g$ is set to in the two algorithms. In both algorithms $g$ is chosen to approximate the negation of the gradient of the centrality potential while being small in requisite norms for the method. However, by choosing a different weight function both the gradient and the norm change. To explain the new step (and analyze our method)  throughout this section we make extensive use of the following \emph{mixed norm} and its \emph{dual vector} (variants of which were defined in \cite{ls19,blss20}).

\begin{definition}[Mixed Norm and Dual Vector]\label{def:mixed norm}
\label{def:mixed_norm} We define the \emph{mixed norm} $\| x \|_{\tau+\infty} \defeq \| x \|_{\infty} + \cnorm \| x \|_{\tau}$
with $\cnorm = \frac{10}{\alpha}$ and we define the \emph{dual vector} of
$g$ as
\begin{equation}
\label{eq:flat_op}
g^{\flat(w)} \defeq \argmax_{ \|x\|_{w + \infty} = 1} g^{\top} x.
\end{equation}
Further we define the dual norm $\|g\|_{\tau+\infty}^* := \langle g, g^{\flat(\tau)} \rangle$, i.e. the maximum attained by \eqref{eq:flat_op}.
\end{definition}

With this notation in mind, recall that in the short-step algorithm for the logarithmic barrier (\Cref{alg:short_step_LS})) we chose $g$ to approximate the direction which has best inner product with $- \nabla \Phi(w)$ and has $\norm{\cdot}_2$ at most $\gamma$. Formally, we picked $\ov$ with 
$\norm{\ov - w}_{\gamma}$ and let $g = - \gamma \nabla(\Phi(\ov)) / \norm{\nabla(\Phi(\ov))}_2$. Note that $g = - \gamma g^{\flat(\ov)}$ if in \eqref{eq:flat_op} we replaced  $\|x\|_{w + \infty}$ with $\|x\|_{2}$. 
Since when using the Lee-Sidford weight function (\Cref{def:LS_weight}) we instead wish to decrease $\Phi(w/\tau(x,s))$ we instead pick $\ov\in\R_{>0}^{m}$ such that $\norm{\ov- w  / \tau(x,s)}_{\infty}\leq\gamma$  and let $g \defeq - \gamma \nabla\Phi(\overline{v})^{\flat(\overline{\tau})}$ where $\flat$ is as defined in \Cref{def:mixed_norm}.
 
 This setting of $g$ is the primary difference between the algorithm in this section (\Cref{alg:short_step_LS}) and the one in previous section (\Cref{alg:short_step_log}). The other difference are induced by this change or done for technical or analysis reasons. We change the setting of $r$ and the requirements for $q$ to account for the fewer iterations we take and the different norm in which we analyze our error. Further, 
 we slightly change the step sizes for $\delta_x$ and $\delta_s$, i.e. in \Cref{alg:short_step_LS} we let
 \[
 \delta_{x} = \omx ((1+2\alpha)g - \mR \delta_r)
 \text{ and }
 \delta_s = (1 - 2\alpha) \oms \delta_p
 \]
(whereas in \Cref{alg:short_step_log} we set $\delta_x$ and $\delta_s$ as above with $\alpha = 0$). This change was also performed in \cite{blss20} to account for the effect of $\tau$ on the centrality potential and ensure sufficient progress is made.  %

\begin{algorithm2e}[h]

\caption{Short Step (LS Weight Function)}

\label{alg:short_step_LS}

\SetKwProg{Proc}{procedure}{}{}

\Proc{$\textsc{ShortStep}$$(x \in \R^m_{>0},s \in \R^m_{>0},\mu \in \R_{>0},\mu^{\new} \in \R_{>0})$}{

\LineComment{Fix $\tau(x,s) \defeq \sigma(x,s) + \frac{n}{m} \vones \in \R^{m} \text{ for all } x \in \R_{>0}^{m},s \in \R_{>0}^{m}$}

\LineComment{Let $\lambda = \frac{36\log(400m/\epsilon^{2})}{\epsilon}$,
$\gamma = \frac{\epsilon}{100\lambda}$, $r = \frac{\epsilon\gamma}{\sqrt{n}}$,
$\epsilon = \frac{\alpha^2}{240}$, $\alpha = \frac{1}{4\log(\frac{4m}{n})}$}

\textbf{\LineComment{Assume $(x,s,\mu)$ is $\epsilon$-centered
}}

\textbf{\LineComment{Assume $\delta_{\mu}\defeq\mu^{\new}-\mu$ satisfies
$|\delta_{\mu}|\leq r\mu$}}

\vspace{0.2in}

\LineComment{\textbf{Approximate the iterate and compute step direction}}

\State  Find $\ox,\os,\overline{\tau}$ such that $\ox \approx_{\epsilon}x$,
$\os\approx_{\epsilon}s$, $\overline{\tau}\approx_{\epsilon} \tau$ for $\tau \defeq \tau(x,s)$ \label{line:LSstep:xstau}

\State Find $\ov\in\R_{>0}^{m}$ such that $\norm{\ov- w  / \tau}_{\infty}\leq\gamma$   
 for $w \defeq w(x,s,\mu)$
\label{line:LSstep:v} %

\State $g= - \gamma \nabla\Phi(\overline{v})^{\flat(\overline{\tau})}$
where $h^{\flat(w)}$ is defined in \Cref{def:mixed_norm} \label{line:LSstep:g}

\vspace{0.2in}

\LineComment{\textbf{Sample rows}}

\State $\oma \defeq \omx^{1/2}\oms^{-1/2}\ma,\ow \defeq w(\ox,\os,\mu),\omx \defeq \mdiag(\ox),\oms \defeq \mdiag(\os),\omw \defeq \mdiag(\ow)$

\State Find $\omh$ such that $\omh\approx_{\gamma}\oma^{\top}\oma$\label{line:LSstep:omh}

\State Let $\delta_{p} \defeq \omw^{-1/2}\oma\omh^{-1}\oma^{\top}\omw^{1/2}g$,
$\delta_{c} \defeq \mu^{-1/2} \cdot \omw^{-1/2} \oma \omh^{-1} ( \ma^{\top}x - b )$,
and $\delta_{r} \defeq ( 1 + 2\alpha ) \delta_{p} + \delta_{c}$\label{line:LSstep:delta_c}

\State Find $q\in\Rn_{>0}$ such that $q_{i} \geq \frac{m}{\sqrt{n}} \cdot \left((\delta_{r})_{i}^{2}/\|\delta_{r}\|_{2}^{2}+1/m\right)+C\cdot\sigma_{i}(\oma)\log(m/(\epsilon r))\gamma^{-2}$
for some large enough constant $C$ 

\State Sample a diagonal matrix $\mr\in\R^{m\times m}$ where each
$\mr_{i,i} = [\min\{q_{i},1\}]^{-1}$ independently with probability $\min \{q_{i},1\}$
and $\mr_{i,i} = 0$ otherwise\label{line:LSstep:R}

\vspace{0.2in}

\LineComment{\textbf{Compute projected step}}

\State $\delta_{x}\leftarrow\omx((1+2\alpha)g-\mr\delta_{r})$

\State $\delta_{s}\leftarrow(1-2\alpha)\oms\delta_{p}$

\State $x^{\new}\leftarrow x+\delta_{x}$ and $s^{\new}\leftarrow s+\delta_{s}$

\Return$(x^{\new},s^{\new})$

}

\end{algorithm2e}

We analyze the performance of this short-step procedure for the Lee-Sidford barrier analogously to how we analyzed the short step procedure for the log barrier in \Cref{sec:log_barrier_method}. We again apply \Cref{lem:expand_potential} and \Cref{lem:centrality_main_general}  and together these lemmas reduce bounding the potential $\phicent(x^{\new},s^{\new},\mu^{\new})$ to carefully analyzing
bounding the first moments, second moments, and multiplicative stability induced by $\delta_{s}$, $\delta_{x}$,
$\delta_{\mu}$ as well as the change from $\tau(x,s)$ to $\tau(x^{\new}, s^{\new})$, which we denote by $\delta_{\tau} \defeq \tau(x^{\new}, s^{\new}) -  \tau(x,s)$. 

Key differences in the analysis of this section in the previous section are in how these bounds are achieved and in analyzing $\delta_\tau$. Whereas the analysis of the log barrier short-step method primarily considered the $\ell_2$ norm, here we consider the mixed norm, \Cref{def:mixed norm} (variants of which were analyzed in \cite{ls19,blss20}). For $\delta_x$ and $\delta_s$ this analysis is performed in \Cref{sub:ls:deltas}. For analyzing $\delta_{\tau}$ new analysis, though related to calculations in \cite{blss20}, is provided in \Cref{sub:delta_tau}. This analysis is then leveraged in \Cref{sec:ls:potential_decrease} to provide the main result of this section, that \Cref{alg:short_step_LS} is a valid short-step procedure. 

\subsubsection{Bounding $\delta_{x}$ and $\delta_{s}$}
\label{sub:ls:deltas}

Here we provide basic bounds on $\delta_x$, $\delta_s$, and related quantities.  First, we provide \Cref{lem:sigma_approx} and \Cref{lem:proj_bounds} which are restatements of lemmas from \cite{blss20} which immediately apply to our algorithm. Leveraging these lemmas, we state and prove \Cref{lem:proj_bounds_2} which bounds the effect of our feasibility correction step (which was not present in the same way \cite{blss20}). Leveraging these lemmas we prove  \Cref{lem:projection_sizes_ls}, which provides first-order bounds and then we prove \Cref{lem:ls_second} which provides second-order and multiplicative bounds.

\begin{lemma}[{\cite[Lemma 18]{blss20}}]
	\label{lem:sigma_approx}In Algorithm~\ref{alg:short_step_LS}, we
	have
	\begin{align*}
	\frac{1}{2}\sigma(\oms^{-1/2-\alpha}\omx^{1/2-\alpha}\ma) \leq \sigma(\oms^{-1/2}\omx^{1/2}\ma) \leq 2\sigma(\oms^{-1/2-\alpha}\omx^{1/2-\alpha}\ma).
	\end{align*}
\end{lemma}

\begin{lemma}[{\cite[Lemma 19]{blss20}}]
	\label{lem:proj_bounds}In Algorithm~\ref{alg:short_step_LS}, 
	$\mq \defeq \oma\omh^{-1}\oma^{\top}$ satisfies $\mq\preceq e^{3\epsilon}\mproj(\oma)$. 
	Further, for all $h\in\R^{m}$,
	\begin{align*}
	\norm{\omw^{-1/2}\mq\omw^{1/2}h}_{\tau}\leq e^{4\epsilon}\norm h_{\tau}\text{ and }\norm{\omw^{-1/2}(\mi-\mq)\omw^{1/2}h}_{\tau}\leq e^{3\epsilon}\norm h_{\tau}\,.
	\end{align*}
	Further, $\norm{\omw^{-1/2}\mq\omw^{1/2}h}_{\infty} \leq 2 \norm h_{\tau}$
	and therefore $\|\omw^{-1/2}\mq\omw^{1/2}\|_{\tau+\infty} \leq e^{4\epsilon} + 2/\cnorm$.
\end{lemma}

\begin{lemma}
	\label{lem:proj_bounds_2} In \Cref{alg:short_step_LS}, 
	\[
	\| \omw^{-1/2} \oma \omh^{-1}g\|_{\tau} \leq e^{2\epsilon} \| g \|_{\omh^{-1}} 
	\text{ and }
	\| \omw^{-1/2}\oma\omh^{-1}g \|_{\infty} \leq 2 \|g\|_{\omh^{-1}}
	~.
	\]
\end{lemma}

\begin{proof}
	Let $\mq = \oma\omh^{-1}\oma^{\top}$.
	Since $\omh \succeq  e^{-\gamma} \oma^{\top} \oma \succeq e^{-\epsilon} \oma^{\top} \oma$, we have 
	\[
	\|\oma\omh^{-1/2}\|_{2}^{2}=\|\oma \omh^{-1} \oma^{\top}\|_{2}\leq e^{\epsilon} \norm{\oma(\oma^{\top}\oma)^{-1}\oma^{\top}}_{2}\leq e^{\epsilon} ~.
	\]
	Since $\gamma \leq \epsilon$ we have $\ow \approx_{2\epsilon} w$ and $w \approx_{\epsilon} \tau$. This implies $\ow \approx_{3\epsilon} \tau$ and
	\begin{align*}
	\|\omw^{-1/2}\oma\omh^{-1}g\|_{\tau}\leq e^{1.5\epsilon}\|\oma\omh^{-1}g\|_{2}\leq e^{1.5\epsilon}\|\oma\omh^{-1/2}\|_{2}\|\omh^{-1/2}g\|_{2}\leq e^{2\epsilon}\|g\|_{\omh^{-1}}.
	\end{align*}
	Further, by Cauchy Schwarz, we have
	\begin{align*}
	\norm{\omw^{-1/2}\oma\omh^{-1}g}_{\infty}^{2} & =\max_{i\in[m]}\left[e_{i}^{\top}\omw^{-1/2}\oma\omh^{-1}g\right]^{2}\leq\max_{i\in[m]}\left[\omw^{-1/2}\mq\omw^{-1/2}\right]_{i,i}\cdot\left[g^{\top}\omh^{-1}g\right]\,
	\end{align*}
	$\mq\preceq e^{3\epsilon}\mproj(\oma)$ (Lemma
	\ref{lem:proj_bounds}) and $\sigma(\oma)_{i}\leq2\sigma(\oms^{-1/2-\alpha}\omx^{1/2-\alpha}\ma)_{i}\leq3\tau_{i}$
	(Lemma~\ref{lem:sigma_approx}), implies 
	\begin{align*}
	\max_{i\in[m]}\left[\omw^{-1/2}\mq\omw^{-1/2}\right]_{i,i} \leq\max_{i\in[m]}\frac{e^{3\epsilon}\sigma(\oma)_{i}}{\ow_{i}}
	\leq e^{6\epsilon}\max_{i\in[m]}\frac{\sigma(\oma)_{i}}{\tau_{i}}
	\leq3e^{6\epsilon}\,.
	\end{align*}
\end{proof}

\begin{lemma}[Basic Step Size Bounds]
\label{lem:projection_sizes_ls}In Algorithm~\ref{alg:short_step_LS},
we have that
\begin{equation}
\norm{\delta_{p}}_{\tau+\infty}\leq\frac{10}{9}\gamma\text{ and }\norm{\delta_{c}}_{\tau+\infty}\leq\frac{11\epsilon\gamma}{\alpha}\leq\frac{2}{7}\gamma\label{eq:step_size_bounds_ls}
\end{equation}
and consequently
\begin{align*}
\norm{ \mx^{-1} \E[ |\delta_{x}| ] }_{\tau+\infty}\leq2\gamma\text{ and }\norm{\ms^{-1}\delta_{s}}_{\tau+\infty}\leq2\gamma .
\end{align*}
\end{lemma}

\begin{proof}
We first bound $\delta_{p}$. Let $\mq \defeq \oma\omh^{-1}\oma^{\top}$
and note that 
\begin{equation}
\delta_{p} = \omw^{-1/2} \mq \omw^{1/2} g.\label{eq:delta_p_formula}
\end{equation}
By Lemma \ref{lem:proj_bounds}, we have $\|\omw^{-1/2}\mq\omw^{1/2}g\|_{\tau+\infty}\leq(e^{4\epsilon}+2/\cnorm)\|g\|_{\tau+\infty}\leq\frac{10}{9}\|g\|_{\tau+\infty}$.

Now, we bound $\delta_{c}$. Recall that
\begin{align*}
\delta_{c} = \mu^{-1/2} \cdot \omw^{-1/2}\oma\omh^{-1}(\ma^{\top}x-b).
\end{align*}
By Lemma \ref{lem:proj_bounds_2}, we have
\begin{align*}
\|\mu^{-1/2}\cdot\omw^{-1/2}\oma\omh^{-1}(\ma^{\top}x-b)\|_{\tau} 
 \leq \frac{e^{2\epsilon}}{\sqrt{\mu}}\norm{\ma^{\top}x-b}_{\omh^{-1}} 
 \leq \frac{e^{3\epsilon}}{\sqrt{\mu}}\norm{\ma^{\top}x-b}_{(\ma^{\top}\mx\ms^{-1}\ma)^{-1}} ~.
\end{align*}
Hence, we have $\norm{\delta_{c}}_{\tau}\leq e^{3\epsilon}\gamma\epsilon$.
For the sup norm, Lemma \ref{lem:proj_bounds_2}, shows that
\begin{align*}
\|\mu^{-1/2}\cdot\omw^{-1/2}\oma\omh^{-1}(\ma^{\top}x-b)\|_{\infty} 
 \leq \frac{2}{\sqrt{\mu}}\norm{\ma^{\top}x-b}_{\omh^{-1}}
 \leq \frac{2e^{\epsilon}}{\sqrt{\mu}}\norm{\ma^{\top}x-b}_{(\ma^{\top}\mx\ms^{-1}\ma)^{-1}} ~.
\end{align*}
Hence, we have $\norm{\delta_{c}}_{\infty}\leq2e^{\epsilon}\gamma\epsilon$ by definition of $\epsilon$-centered. 
In particular, as $\epsilon \leq \alpha / 80$, we have
\begin{align*}
\|\delta_{c}\|_{\tau+\infty} \leq 2e^{\epsilon}\gamma\epsilon+\cnorm e^{3\epsilon}\gamma\epsilon=\left(2e^{\epsilon}+\frac{10e^{3\epsilon}}{\alpha}\right)\gamma\epsilon\leq\frac{11\gamma\epsilon}{\alpha}\leq\frac{2}{7}\gamma.
\end{align*}
For the bound of $\ms^{-1}\delta_{s}$, we note that $\delta_{s}=(1-2\alpha)\oms\delta_{p}$
and hence
\begin{align*}
\|\ms^{-1}\delta_{s}\|_{\tau+\infty}\leq\|\ms^{-1}\oms\delta_{p}\|_{\tau+\infty}\leq1.1\cdot\|\delta_{p}\|_{\tau+\infty}\leq2\gamma.
\end{align*}
For the bound of $\mx^{-1}\E[|\delta_{x}|]$, note that $\mx^{-1}\E[|\delta_{x}|]=\mx^{-1}\omx|(1+2\alpha)(g-\delta_{p})-\delta_{c}|$ and (\ref{eq:delta_p_formula}) yields 
\begin{align*}
g-\delta_{p}=\omw^{-1/2}(\mi-\mq)\omw^{1/2}g.
\end{align*}
By Lemma \ref{lem:proj_bounds}, we have $\|g-\delta_{p}\|_{\tau}\leq e^{3\epsilon}\|g\|_{\tau}$
and $\|g-\delta_{p}\|_{\infty}\leq\|g\|_{\infty}+\|\omw^{-1/2}\mq\omw^{1/2}g\|_{\infty}\leq\|g\|_{\infty}+2\|g\|_{\tau}$.
Hence, we have
\begin{align*}
\|g-\delta_{p}\|_{\tau+\infty} & =\|g-\delta_{p}\|_{\infty}+\cnorm\|g-\delta_{p}\|_{\tau}\\
 & \leq\|g\|_{\infty}+2\|g\|_{\tau}+\cnorm e^{3\epsilon}\|g\|_{\tau}\\
 & \leq \left(e^{3\epsilon}+\frac{2}{\cnorm}\right)\|g\|_{\tau+\infty}\leq1.1\gamma.
\end{align*}
Finally, we have
\begin{align*}
\| \mx^{-1}\E[|\delta_{x}|] \|_{\tau+\infty} 
\leq & ~ 1.1\|(1+2\alpha)(g-\delta_{p})-\delta_{c}\|_{\tau+\infty}
\leq (1.1)^{2}(1+2\alpha)\gamma+\frac{2}{7}\gamma\\
\leq & ~ (1.1)^{2}\cdot(1+2/5)\gamma+\frac{2}{7}\gamma
\leq 2 \gamma.
\end{align*}
where the third step follows $\alpha\in[0,1/5]$.
\end{proof}
\begin{lemma}[Multiplicative Stability and Second-order Analysis of  $x$ and $s$]
\label{lem:ls_second} In Algorithm~\ref{alg:short_step_LS},
\begin{itemize}
\item $\| \mx^{-1} \delta_{x} \|_{\infty} \leq 2\gamma$ and $\norm{ \ms^{-1} \delta_{s} }_{\infty} \leq 2\gamma$
with probability $1$.
\item $\| \E[ (\mx^{-1} \delta_{x})^{2} ] \|_{\tau+\infty} \leq 12\gamma^{2}$
and $\| ( \ms^{-1} \delta_{s} )^{2} \|_{\tau+\infty} \leq 4\gamma^{2}$
\end{itemize}
In particular, we have $\E\left[\normFull{\mx^{-1}\delta_{x}}_{\phi''(v)}^{2}\right]\le12\gamma^{2}\norm{\phi''(v)}_{\tau+\infty}^{*}$.
and $\normFull{\ms^{-1}\delta_{s}}_{\phi''(v)}^{2}\le4\gamma^{2}\norm{\phi''(v)}_{\tau+\infty}^{*}$.
\end{lemma}

\begin{proof}
For sup norm for $s$, it follows from $\norm{\ms^{-1}\delta_{s}}_{\infty}\leq\norm{\ms^{-1}\delta_{s}}_{\tau+\infty}$
and Lemma \ref{lem:projection_sizes_ls}.

For sup norm for $x$, recall that
\begin{align*}
\mx^{-1} \delta_{x} = \mx^{-1} \omx( ( 1 + 2 \alpha ) g - \mr \delta_{r} ).
\end{align*}
There are two cases. 

Case 1) $i$-th row is sampled with probability $1$. We have
\begin{align*}
( \mx^{-1} \delta_{x} )_{i} = \E[ ( \mx^{-1} \delta_{x} )_{i} ]
\end{align*}
By Lemma \ref{lem:projection_sizes_ls}, we have $| ( \mx^{-1} \delta_{x} )_{i} | \leq 2\gamma$.

Case 2) $i$-th row is sampled with probability at least $\frac{m}{\sqrt{n}}\cdot((\delta_{r})_{i}^{2}/\|\delta_{r}\|_{2}^{2}+\frac{1}{m}) \leq 1$ and this quantity is at most $1$. In this case, \eqref{eq:entry_to_l2} of the the proof too the analogous \Cref{lem:stoch_x_bounds}, still applies to yield that $| [\delta_r]_i | \leq (\sqrt{m} / 2) \norm{\delta_r}_2 \cdot ([\delta_r]_i^2 / \norm{\delta_r}_2^2 + 1/m)$
and therefore
\begin{align*}
|(\mr\delta_{r})_{i}| & \leq\frac{|(\delta_{r})_{i}|}{q_{i}}\leq\frac{\sqrt{n}}{m}\frac{|(\delta_{r})_{i}|}{(\delta_{r})_{i}^{2}/\|\delta_{r}\|_{2}^{2}+\frac{1}{m}}\leq\frac{\sqrt{n}}{m}\cdot\frac{\sqrt{m}}{2}\|\delta_{r}\|_{2}\\
 & \leq\frac{\sqrt{n}}{m}\cdot\frac{\sqrt{m}}{2}\cdot\sqrt{\frac{m}{n}}\|\delta_{r}\|_{\tau}=\frac{1}{2}\|\delta_{r}\|_{\tau}
\end{align*}
where we used $\tau \geq \frac{n}{m}$ in the last inequality. Using
$\ox\approx_{\epsilon}x$, we have
\begin{align}
\|\mx^{-1}\delta_{x}\|_{\infty} & \leq\frac{3}{2}\|g\|_{\infty}+\frac{2}{3}\|\delta_{r}\|_{\tau}.\label{eq:X_ls_norm}
\end{align}
Finally, by $\delta_{r}=(1+2\alpha) \delta_{p} + \delta_{c}$ and Lemma
\ref{lem:projection_sizes_ls}, we have
\begin{equation}
\|\delta_{r}\|_{\tau+\infty} \leq 1.4 \norm{\delta_{p}}_{\tau+\infty} + \norm{\delta_{c}}_{\tau+\infty} \leq \frac{13}{7}\gamma.\label{eq:delta_r_upper}
\end{equation}
Now, using $\cnorm \geq 50$, we have
\begin{equation}\label{eq:delta_r_upper_tau} %
\|\delta_{r}\|_{\tau}\leq\frac{1}{\cnorm}\|\delta_{r}\|_{\tau+\infty}\leq\frac{1}{50}\|\delta_{r}\|_{\tau+\infty}\leq\frac{\gamma}{40}.
\end{equation}
This gives the bound $\| \mx^{-1} \delta_{x} \|_{\infty}\leq2\gamma$.

For the second order term for $s$, leveraging the generalization of Cauchy Schwarz to arbitrary norms, i.e. that for every norm $\norm{\cdot}$, its dual $\norm{\cdot}^*$, and all $x,y$ we have $|x^\top y| \leq \norm{x} \norm{y}^*$, yields
\begin{align*}
\normFull{\ms^{-1}\delta_{s}}_{\phi''(v)}^{2} 
&=  \sum_{i \in [m]} \phi''(v_{i})(\ms^{-1}\delta_{s})_{i}^{2}
\leq  \|(\ms^{-1}\delta_{s})^{2}\|_{\tau+\infty}\|\phi''(v)\|_{\tau+\infty}^{*}\\
& \leq \|\ms^{-1}\delta_{s}\|_{\infty}\|\ms^{-1}\delta_{s}\|_{\tau+\infty}\|\phi''(v)\|_{\tau+\infty}^{*}
\leq  4\gamma^{2}\|\phi''(v)\|_{\tau+\infty}^{*} ~.
\end{align*}

For the second order term for $x$, we note that
\begin{align*}
\E\left[\norm{\mx^{-1}\delta_{x}}_{\phi''(v)}^{2}\right] = \sum_{i \in [m]} \phi''(v_{i})\E[(\mx^{-1}\delta_{x})_{i}^{2}]\leq\|\E[(\mx^{-1}\delta_{x})^{2}]\|_{\tau+\infty}\|\phi''(v)\|_{\tau+\infty}^{*}.
\end{align*}
Hence, it suffices to bound $\| \E [ (\mx^{-1}\delta_{x})^{2} ] \|_{ \tau + \infty }$.
Using that $\delta_{x}=\omx((1+2\alpha)g-\mr\delta_{r})$, we have
\begin{align*}
(\mx^{-1}\delta_{x})_{i}^{2}\leq2((1+2\alpha)\mx^{-1}\omx g)_{i}^{2}+2(\mx^{-1}\omx\mr\delta_{r})_{i}^{2}
\end{align*}
Using $\ox\approx_{\epsilon}x$ and $\alpha\leq\frac{1}{5}$, we have
\begin{align}
(\mx^{-1}\delta_{x})_{i}^{2} & \leq4g_{i}^{2}+2.1(\mr\delta_{r})_{i}^{2}.\label{eq:Ex2_expand}
\end{align}
For the second term, we note that $i$-th row is sampled with probability
at least $\min(1,\frac{m}{\sqrt{n}}\cdot(\delta_{r})_{i}^{2}/\|\delta_{r}\|_{2}^{2})$.
Hence, we have
\begin{align*}
\E[(\mr\delta_{r})_{i}^{2}]\leq(\delta_{r})_{i}^{2}+\frac{(\delta_r)_{i}^{2}}{\frac{m}{\sqrt{n}}\cdot(\delta_{r})_{i}^{2}/\|\delta_{r}\|_{2}^{2}}\leq(\delta_{r})_{i}^{2}+\frac{\sqrt{n}}{m}\cdot\|\delta_{r}\|_{2}^{2}\leq(\delta_{r})_{i}^{2}+\frac{1}{\sqrt{n}}\|\delta_{r}\|_{\tau}^{2}
\end{align*}
Hence, we have
\begin{align*}
\|\E[(\mr\delta_{r})^{2}]\|_{\tau+\infty}
\leq
\|(\delta_{r})^{2}\|_{\tau+\infty}+\frac{1}{\sqrt{n}}\|\delta_{r}\|_{\tau}^{2}\cdot\|1\|_{\tau+\infty}.
\end{align*}
Since $\|1\|_{\tau+\infty}=\|1\|_{\infty}+\cnorm\|1\|_{\tau}=1+\sqrt{2n}\cnorm$,
we have
\begin{align*}
\|\E[(\mr\delta_{r})^{2}]\|_{\tau+\infty} & \leq\|\delta_{r}\|_{\tau+\infty}^{2}+\frac{5}{2}\cnorm\|\delta_{r}\|_{\tau}^{2}
  \leq \left(1+\frac{5}{2\cnorm}\right)\|\delta_{r}\|_{\tau+\infty}^{2}\\
 & \leq \left(1+\frac{5}{100}\right)\cdot\left(\frac{13}{7}\right)^{2}\gamma^{2}
  \leq3.7\gamma^{2}
\end{align*}
where we used (\ref{eq:delta_r_upper}) at the third inequalities.

Putting this into (\ref{eq:Ex2_expand}) gives
\begin{align*}
\|\E[(\mx^{-1}\delta_{x})^{2}]\|_{\tau+\infty} & \leq4\|g^{2}\|_{\tau+\infty}+(2.1)(3.7)\gamma^{2}\leq4\|g\|_{\tau+\infty} \|g\|_{\infty}+(2.1)(3.7)\gamma^{2}\leq12\gamma^{2}.
\end{align*}
\end{proof}

\subsubsection{Bounding $\delta_{\tau}$}
\label{sub:delta_tau}

Here we bound the first two moments of the change in the LS weight function.
First we leverage \Cref{lem:Dsigma}, leverage score related derivative calculations from \cite{ls19}, to provide general bounds on how $\tau$ changes in \Cref{lem:expand_sigma}. Then, applying a \Cref{lem:P2sample}, a technical lemma about how the quadratic form of entrywise-squared projection matrices behave under sampling from \cite{blss20}, and \Cref{lem:p2_stability}, a new lemma regarding the stability $\mproj(\mw \ma)^{(2)}$ under changes to $\mw$,  in \Cref{lem:Edelta_tau_bound} we bound $\delta_\tau \defeq \tau(x^{\new},s^{\new}) - \tau(x,s)$ in \Cref{alg:short_step_LS}. 

\begin{lemma}[{\cite[Lemma 49]{ls19}}]
\label{lem:Dsigma}Given a full rank $\ma\in\R^{m\times n}$ and
$w\in\R_{>0}^{m}$, we have
\begin{align*}
D_{w}\mproj(\mw\ma)[h]=\mDelta\mproj(\mw\ma)+\mproj(\mw\ma)\mDelta-2\mproj(\mw\ma)\mDelta\mproj(\mw\ma)
\end{align*}
where $\mw=\mdiag(w)$, $\mDelta=\mdiag(h/w)$, and $D_w f[h] \defeq \lim{t \rightarrow 0} \frac{1}{t}[f(w + th) - f(w)]$ denotes the directional derivative of $f$ with respect to $w$ in direction $h$. In particular,
we have $D_{w}\sigma(\mw\ma)[h]=2\mLambda(\mw\ma)\mw^{-1}h$.
\end{lemma}

\begin{lemma} \label{lem:expand_sigma} Let $w \in \R^m_{> 0}$ be arbitrary and let $w^{\new}$ be a random vector such that $w^{\new}\approx_{\gamma}w$ with probability $1$ for some $0\leq\gamma\leq\frac{1}{80}$.
Let $\tau=\sigma(\mw\ma)+\eta$ and $\tau^{\new}=\sigma(\mw^{\new}\ma)+\eta$
for some $\eta \in \R^m_{> 0}$. Let $\delta_{w}=w^{\new}-w$, $\delta_{\tau}=\tau^{\new}-\tau$
and $\mw_{t}=(1-t)\mw+t\mw^{\new}$. Then, for 
\[
K \defeq 10\|\E_{w^{\new}} [ (\mw^{-1}\delta_{w})^{2} ] \|_{\sigma(\mw\ma)+\infty}+13\gamma\max_{t\in[0,1]}\E_{w^{\new}}[ \|\mw^{-1}|\delta_{w}|\|_{\mproj^{(2)}(\mw_{t}\ma) +\infty}]
\]
where $\|x\|_{M + \infty} \defeq \|x\|_\infty + C \|x\|_{M}$ for some fixed constant $C \geq 1$ and any norm $\norm{\cdot}_M$, the following hold
\begin{itemize}
\item $\|\mathbf{T}^{-1}\delta_{\tau}\|_{\infty}\leq5\gamma$,
\item $\|\E_{w^{\new}} [ (\mathbf{T}^{-1}\delta_{\tau})^{2} ] \|_{\tau +\infty} \leq K$
\item $\|\E_{w^{\new}} [ \mathbf{T}^{-1}(\tau^{\new}-\tau-2\mLambda(\mw \ma)\mw^{-1}\delta_{w}) ] \|_{\tau +\infty}\leq K$.
\end{itemize}
\end{lemma}

\begin{proof}
Let $w_{t}=w+t\delta_{w}$, $\sigma_{t}=\sigma(\mw_{t}\ma)$, $\tau_{t}=\sigma(\mw_{t}\ma)+\eta$.
For simplicity, we let $\sigma \defeq \sigma_0$ and $r_{t} \defeq \mw_{t}^{-1}\delta_{w}$. 

For the first result, recall that
$\sigma_{t}=\mdiag(\mw_{t}\ma(\ma^{\top}\mw_{t}^{2}\ma)^{-1}\ma^{\top}\mw_{t})$. Since, $w_{t}\approx_{\gamma}w$,
we have that $\sigma_{t}\approx_{4\gamma}\sigma$ and hence $\|\mathbf{T}^{-1}\delta_{\tau}\|_{\infty}\leq5\gamma$.

For the second result, Jensen's inequality shows that 
\begin{align}
(\mathbf{T}^{-1}\delta_{\tau})_{i}^{2} 
& = \left(\mathbf{T}^{-1}\int_{0}^{1} \left[\frac{ \mathrm{d} }{ \mathrm{d} t}\sigma_{t}\right] \mathrm{d} t\right)_{i}^{2}\leq\int_{0}^{1} \left(\mathbf{T}^{-1} \left[\frac{ \mathrm{d} }{ \mathrm{d} t}\sigma_{t}\right] \right)_{i}^{2}dt\label{eq:Texpand1}
\end{align}
Using that $\frac{ \mathrm{d} }{ \mathrm{d} t }\sigma_{t} =2\mLambda_{t} r_t$ (Lemma
\ref{lem:Dsigma}) and $\mLambda_{t}=\mSigma_{t}-\mproj_{t}^{(2)}$,
we have
\begin{equation}\label{eq:Texpand2}
\tau_{i}^{-2} \left( \left[\frac{ \mathrm{d} }{ \mathrm{d} t}\sigma_{t}\right]_i \right)^{2} = 4\tau_{i}^{-2}(\mLambda_{t} r_t)_{i}^{2} \leq 8 \tau_{i}^{-2}\left((\mSigma_{t} r_t)_{i}^{2}+(\mproj_{t}^{(2)} r_t)_{i}^{2}\right).
\end{equation}
Substituting (\ref{eq:Texpand2}) in (\ref{eq:Texpand1}), we have
\begin{align}
\|\E(\mathbf{T}^{-1}\delta_{\tau})^{2}\|_{\tau + \infty}  & \leq 8\int_{0}^{1}\left\Vert \mathbf{T}^{-2}\E\left((\mSigma_{t} r_t)^{2}+(\mproj_{t}^{(2)} r_t)^{2}\right)\right\Vert _{\tau + \infty}dt\nonumber \\
 & \leq 8 \max_{t\in[0,1]}\left(\left\Vert \mathbf{T}^{-2}\E(\mSigma_{t} r_t)^{2}\right\Vert _{\tau + \infty}+\left\Vert \mathbf{T}^{-2}\E(\mproj_{t}^{(2)} r_t)\right\Vert _{\tau + \infty}\right)\nonumber \\
 & \leq 10 \max_{t\in[0,1]}\left(\left\Vert \E r_t^{2}\right\Vert _{\sigma+\infty}+
 \left\Vert \E(\mSigma_{t}^{-1}\mproj_{t}^{(2)} r_t)^{2}\right\Vert _{\sigma_{t}+\infty} \right)
 \label{eq:Texpand3}
\end{align} 
where we used $w_{t}\approx_{\gamma}w$, $\sigma_{t}\approx_{4\gamma}\sigma$
and $\tau\geq\sigma$ at the end. 

For the second term in (\ref{eq:Texpand3}), using $\|\mSigma_{t}^{-1}\mproj_{t}^{(2)} r_t\|_{\infty}\leq\| r_t\|_{\infty}\leq\frac{9}{8}\gamma$
(Lemma \ref{lem:P2bound}), $\mSigma_{t}\succeq\mproj_{t}^{(2)}$
(Lemma \ref{lem:P2bound}) and that $\mproj_{t}^{(2)}$ is entry-wise non-negative, we have
\begin{align}
\left\Vert \E(\mSigma_{t}^{-1}\mproj_{t}^{(2)} r_t)^2\right\Vert _{\sigma_{t}+\infty} & \leq
\frac{9}{8}\gamma\E\left\Vert \mSigma_{t}^{-1}\mproj_{t}^{(2)}| r_t|\right\Vert _{\sigma_{t}+\infty}\nonumber \\
 & \leq \frac{9}{8} \gamma\E\|| r_t|\|_{\mproj_{t}^{(2)}+\infty}\leq
 1.3\gamma\E\|\mw^{-1}|\delta_{w}|\|_{\mproj_{t}^{(2)}+\infty}\label{eq:sigmaP2term}
\end{align}
Putting this into (\ref{eq:Texpand3}) gives the second result.

For the third result, Taylor expansion shows that
\begin{align*}
\sigma_{1} = \sigma_{0} + \frac{ \mathrm{d} }{ \mathrm{d} t } \sigma_{t} \Big|_{t =0} +\int_{0}^{1}(1-t) \left[\frac{ \mathrm{d}^{2}}{ \mathrm{d} t^{2} }\sigma_{t}  \right]  \mathrm{d} t.
\end{align*}
Using again that $\frac{ \mathrm{d} }{ \mathrm{d} t }\sigma_{t} =2\mLambda_{t} r_t$ (Lemma
\ref{lem:Dsigma}) we have that $\frac{ \mathrm{d} }{ \mathrm{d} t } \sigma_{t} \Big|_{t = 0}  = 2 \mLambda \mw^{-1} \delta_w$ and therefore,
\begin{align}
(*)\defeq & \|\E_{w^{\new}} [ \mathbf{T}^{-1}(\tau_{1}-\tau_{0}-2\mLambda\mw^{-1}\delta_{w}) ] \|_{\tau+\infty}\nonumber \\
\leq & \normFull{\int_{0}^{1}(1-t)\E_{w^{\new}} \left[ \mathbf{T}^{-1}\frac{ \mathrm{d}^{2}}{ \mathrm{d} t^{2}}\sigma_{t}(t) \right]  \mathrm{d} t }_{\tau+\infty}\nonumber \\
\leq & \frac{1}{2}\int_{0}^{1} \normFull{ \E_{w^{\new}} \left[ \mathbf{T}^{-1}\frac{ \mathrm{d}^{2} }{\mathrm{d} t^{2}}\sigma_{t}(t) \right] }_{\tau+\infty} \mathrm{d} t.\label{eq:Texpand4}
\end{align}
Hence, it suffices to bound $\frac{d^{2}}{dt^{2}}\sigma_{t}(t)$. 

Note that 
\begin{equation}
\frac{ \mathrm{d}^{2}}{ \mathrm{d} t^{2}}\sigma_{t}(t) = 2 \left[ \frac{ \mathrm{d} }{ \mathrm{d} t }\mLambda_{t} \right]  r_t - 2 \mLambda_{t}  r_t^{2}.\label{eq:dsigma2}
\end{equation}
Let $\mDelta_{t} \defeq \mdiag( r_t)$. Recall that $\mLambda_{t} = \mSigma_{t}-(\mw_{t}\ma(\ma^{\top}\mw_{t}^{2}\ma)^{-1}\ma^{\top}\mw_{t})^{(2)}$,
Lemma \ref{lem:Dsigma} shows
\begin{align*}
\frac{ \mathrm{d} }{ \mathrm{d} t}\mLambda_{t}
= & ~ \frac{ \mathrm{d} }{ \mathrm{d} t}\mSigma_{t}-2\mproj_{t} \circ \frac{ \mathrm{d} }{ \mathrm{d} t }(\mw_{t}\ma(\ma^{\top}\mw_{t}^{2}\ma)^{-1}\ma^{\top}\mw_{t})\\
= & ~ 2 \mdiag( \mLambda_{t}  r_t ) - 2 \mDelta_{t} \mproj_{t}^{(2)} - 2\mproj_{t}^{(2)} \mDelta_{t} + 4\mproj_{t} \circ ( \mproj_{t} \mDelta_{t} \mproj_{t} ).
\end{align*}
Using this in (\ref{eq:dsigma2}) gives
\begin{align*}
\frac{ \mathrm{d}^{2} }{ \mathrm{d} t^{2} }\sigma_{t}(t)
= & ~ 4\mdiag(\mLambda_{t} r_t)\cdot r_t-4\mDelta_{t}\mproj_{t}^{(2)} r_t
  -4\mproj_{t}^{(2)}\mDelta_{t} r_t+8(\mproj_{t}\circ(\mproj_{t}\mDelta_{t}\mproj_{t})) r_t-2\mLambda_{t} r_t^{2}\\
= & ~ 4\mDelta_{t}\mLambda_{t} r_t-4\mDelta_{t}\mproj_{t}^{(2)} r_t-4\mproj_{t}^{(2)} r_t^{2}
  +8(\mproj_{t}\circ(\mproj_{t}\mDelta_{t}\mproj_{t})) r_t-2\mSigma_{t} r_t^{2}+2\mproj_{t}^{(2)} r_t^{2}\\
= & ~ 2\mSigma_{t} r_t^{2}-8\mDelta_{t}\mproj_{t}^{(2)} r_t-2\mproj_{t}^{(2)} r_t^{2}+8(\mproj_{t}\circ(\mproj_{t}\mDelta_{t}\mproj_{t})) r_t.
\end{align*}
Putting this into (\ref{eq:Texpand4}) and using $w_{t}\approx_{\gamma}w$,
$\sigma_{t} \approx_{4\gamma} \sigma$, and $\tau \geq \sigma$, we have
\begin{align}
(*)\leq & \max_{t}\left[\| \E [ \mathbf{T}^{-1} \mSigma_{t} r_t^{2}\|_{\tau+\infty} ] + 4\|\E [ \mathbf{T}^{-1}\mDelta_{t}\mproj_{t}^{(2)} r_t ] \|_{\tau+\infty}\right]\nonumber \\
 & +\max_{t}\left[\|\E [ \mathbf{T}^{-1}\mproj_{t}^{(2)} r_t^{2} ] \|_{\tau+\infty}+4\|\E [ \mathbf{T}^{-1}(\mproj_{t}\circ(\mproj_{t}\mDelta_{t}\mproj_{t})) r_t ] \|_{\tau+\infty}\right]\nonumber \\
\leq & \max_{t}\left[1.2 \|\E [ (\mw^{-1}\delta_{w})^{2} ] \|_{\sigma+\infty}+5\| r_t\|_{\infty}\|\E [ \mathbf{\mSigma}_{t}^{-1}\mproj_{t}^{(2)} r_t ] \|_{\sigma_{t}+\infty}\right]\nonumber \\
 & +\max_{t}\left[1.2 \|\E [ \mathbf{\mSigma}_{t}^{-1}\mproj_{t}^{(2)} r_t^{2} ] \|_{\sigma_{t}+\infty}+5\|\E [ \mathbf{\mSigma}_{t}^{-1}(\mproj_{t}\circ(\mproj_{t}\mDelta_{t}\mproj_{t})) r_t ] \|_{\sigma_{t}+\infty}\right].\label{eq:Texpand5}
\end{align}
For the second term in (\ref{eq:Texpand5}), we follow a similar argument
as (\ref{eq:sigmaP2term}) and get
\begin{align*}
\|\E [ \mathbf{\mSigma}_{t}^{-1}\mproj_{t}^{(2)} r_t ] \|_{\sigma_{t}+\infty}
\leq \E [ \| \mathbf{\mSigma}_{t}^{-1}\mproj_{t}^{(2)} r_t\|_{\sigma_{t}+\infty} ] 
\leq \E \Big[ \|| r_t|\|_{\mproj_{t}^{(2)}\mSigma_{t}^{-1}\mproj_{t}^{(2)}+\infty} \Big]
\leq \E \Big[ \|| r_t|\|_{\mproj_{t}^{(2)}+\infty} \Big].
\end{align*}
Similarly, the third term in (\ref{eq:Texpand5}), we have
\begin{align*}
\|\E [ \mathbf{\mSigma}_{t}^{-1} \mproj_{t}^{(2)}  r_t^{2} ] \|_{\sigma_{t}+\infty} \leq \E \Big[ \| r_t^{2}\|_{\mproj_{t}^{(2)}+\infty} \Big] .
\end{align*}
For the fourth term, note that
\begin{align*}
e_i^\top (\mproj_{t}\circ(\mproj_{t}\mDelta_{t}\mproj_{t})) r_t 
& =\sum_{j \in [m]} (\mproj_{t})_{i,j} (\mproj_{t}\mDelta_{t}\mproj_{t})_{i,j} [ r_t]_j
  = e_i^\top \mproj_{t}\mDelta_{t}\mproj_{t}\mDelta_{t}\mproj_{t} e_i \\
 & \leq e_{i}^{\top}\mproj_{t}\mDelta_{t}^{2}\mproj_{t}e_{i}=e_{i}^{\top}\mproj_{t}^{(2)} r_t^{2}.
\end{align*}
Hence, we can bound the term again by $\E\| r_t^{2}\|_{\mproj_{t}^{(2)}}^{2}$.
In summary, we have
\begin{align*}
(*) 
& \leq 1.2 \|\E [ (\mw^{-1}\delta_{w})^{2} ] \|_{\sigma+\infty}+ 5\| r_t\|_{\infty}\E [ \|| r_t|\|_{\mproj_{t}^{(2)}+\infty} ] + 6.2\E [ \| r_t^{2}\|_{\mproj_{t}^{(2)}+\infty} ] \\
& \leq 1.2 \|\E [ (\mw^{-1}\delta_{w})^{2} ] \|_{\sigma+\infty}+11.2\| r_t\|_{\infty}\E [ \|| r_t|\|_{\mproj_{t}^{(2)}+\infty} ] \\
& \leq 1.2 \|\E [ (\mw^{-1}\delta_{w})^{2} ] \|_{\sigma+\infty}+13\gamma\E [ \|\mw^{-1}|\delta_{w}|\|_{\mproj_{t}^{(2)}+\infty} ]
\end{align*}
where we used that $\mproj_{t}^{(2)}$ is entry-wise non-negative
and $\| r_t\|_{\infty}\leq\frac{9}{8}\gamma$ at the end.
This gives the third result.
\end{proof}

\begin{lemma}[{\cite[Appendix A]{blss20}}]
\label{lem:P2sample}Let $\ma\in\R^{m\times n}$ be non-degenerate,
$\delta\in\R^{m}$ and $p_{i}\geq\min\{1,k\cdot\sigma_{i}(\ma)\}$.
Further, let $\tilde{\delta}\in\R^{m}$ be chosen randomly where for
all $i\in[m]$ independently $\tilde{\delta}_{i}=\delta_{i}/p_{i}$
with probability $p_{i}$ and $\tilde{\delta_{i}}=0$ otherwise. Then
\begin{align*}
\E [ \tilde{\delta} ] = \delta \quad \text{and} \quad 
\E\|\tilde{\delta}\|_{\mproj^{(2)}(\ma)}^{2}\leq(1+(1/k))\cdot\|\delta\|_{\sigma(\ma)}^{2}.
\end{align*}
\end{lemma}

\begin{lemma}\label{lem:P2order}\label{lem:p2_stability}
for all non-degenerate  $\ma\in\R^{m\times n}$, $w,w'\in\R_{>0}^{m}$ with
$w'\approx_{\gamma}w$ for $\gamma>0$, and any non-negative vector
$v\in\R_{\geq0}^{m}$, we have
\[
\|v\|_{\mproj(\mw'\ma)^{(2)}}\leq e^{4\gamma}\|v\|_{\mproj(\mw\ma)^{(2)}}.
\]
\end{lemma}

\begin{proof}
Let $\md \defeq \mw'(\mw)^{-1}$. Then, we have that
\begin{align*}
0\preceq \mproj(\mw'\ma) & =\md\mw\ma(\ma^{\top}\mw'^{2}\ma)^{-1}\ma^{\top}\mw\md\\
 & \preceq e^{2\gamma}\md\mw\ma(\ma^{\top}\mw^{2}\ma)^{-1}\ma^{\top}\mw\md\\
 & = e^{2\gamma}\md\mproj(\mw\ma)\md.
\end{align*}
By Schur product theorem (that the entrywise product of two PSD matrices is PSD), we see that for any PSD $\mm, \mn \in \R^{m \times m}$ with $\mm \succeq \mn$ it is the case that $\mm^{(2)} - \mn^{(2)} = (\mm - \mn) \circ (\mm + \mn)$ is PSD (where we recall that $\circ$ denotes entrywise product and use the shorthand $\mm^{(2)} \defeq \mm \circ \mm$). Thus, $\mproj(\mw'\ma) \preceq e^{2\gamma}\md\mproj(\mw\ma)\md$ implies
\[
\mproj(\mw'\ma)^{(2)} \preceq  \left[e^{2\gamma}\md\mproj(\mw\ma)\md\right] \circ \left[e^{2\gamma}\md\mproj(\mw\ma)\md\right] = e^{4\gamma}\md^{2}\mproj(\mw\ma)^{(2)}\md^{2}.
\]
Hence, we have 
\begin{align*}
\|v\|_{\mproj(\mw'\ma)^{(2)}}^{2}  \leq e^{4\gamma}\cdot\|\md^{2}v\|_{\mproj(\mw\ma)^{(2)}}^{2}
  \leq e^{8\gamma}\cdot\|v\|_{\mproj(\mw\ma)^{(2)}}^{2}
\end{align*}
where we used that $v$ and $\mproj(\mw\ma)^{(2)}$ are coordinate-wise
non-negative in the last step.
\end{proof}

Now, we can bound the change of $\tau$ by bounding the change of
$w$.
\begin{lemma}
\label{lem:Edelta_tau_bound} In \Cref{alg:short_step_LS} the vector $\delta_\tau \defeq \tau(x^{\new},s^{\new}) - \tau(x,s)$ satisfies the following
\begin{itemize}
\item $\|\mT^{-1}\delta_{\tau}\|_{\infty}\leq25\gamma$,
\item $\delta_{\tau}=\mLambda((1-2\alpha)\mx^{-1}\delta_{x}-(1+2\alpha)\ms^{-1}\delta_{s})+\eta$
with $\|\mT^{-1}\E[\eta]\|_{\tau+\infty}\leq2000\gamma^{2}$, and
\item $\|\E[(\mT^{-1}\delta_{\tau})^{2}]\|_{\tau+\infty}\leq2000\gamma^{2}.$
\end{itemize}
In particular, we have $\E\left[\normFull{\mT^{-1}\delta_{\tau}}_{\phi''(v)}^{2}\right]\le500\gamma^{2}\norm{\phi''(v)}_{\tau+\infty}^{*}$
\end{lemma}

\begin{proof}
Let $\tilde{w} \defeq x^{\frac{1}{2}-\alpha}s^{-\frac{1}{2}-\alpha}$, $\tilde{w}^{\new} \defeq  [x^\new]^{\frac{1}{2}-\alpha} [s^\new]^{-\frac{1}{2}-\alpha}$, $\tilde{\mw} \defeq \mdiag(\tilde{w})$, and $\tilde{\mw}_t \defeq \mdiag((1-t)w + t w^{\new}))$ for all $t \in [0, 1]$.
Using $\|\mx^{-1}\delta_{x}\|_{\infty}\leq2\gamma$,
$\|\ms^{-1}\delta_{s}\|_{\infty}\leq2\gamma$ (Lemma \ref{lem:ls_second})
and $\alpha \in[0, \frac{1}{5}]$, we have that $w^{\new}\approx_{5\gamma}w$. Consequently, we can invoke \Cref{lem:expand_sigma} with $w$, $w^{\new}$, $\mw_t$, and $\gamma$ set to be $\tilde{w}$, $\tilde{w}^{\new}$, $\tilde{\mw}_t$, and $5\gamma$  respectively. This gives the first result immediately and for the second and third result, for $\delta_{\tilde{w}} \defeq \tilde{w}^{\new} - \tilde{w}$ it suffices
to bound
\begin{align*}
K \defeq 10\|\E [ (\tilde{\mw}^{-1}\delta_{\tilde{w}})^{2} ] \|_{\sigma + \infty}+65 \gamma \max_{t\in[0,1]}\E_{\tilde{w}^{\new}} [ \|\tilde{\mw}^{-1}|\delta_{\tilde{w}}|\|_{\mproj^{(2)}(\tilde{\mw}_{t}\ma)+\infty} ] .
\end{align*}

Now, for all $t \in [0,1]$ let $x_t \defeq x + t \delta_x$ and $s_t \defeq x = t \delta s$. We have
\begin{align*}
\frac{d}{dt} \ln\left( x_t^{\frac{1}{2} - \alpha } s_t^{- \frac{1}{2} - \alpha} \right)
&= \mx_t^{-\frac{1}{2} + \alpha } \ms_t^{\frac{1}{2} + \alpha} \left[ \left(\frac{1}{2} - \alpha\right) \ms_t^{- \frac{1}{2} -\alpha} \mx_t^{\frac{1}{2} - \alpha - 1}   \delta_x
+ \left(- \frac{1}{2} - \alpha\right) \mx_t^{\frac{1}{2} -\alpha} \ms_t^{- \frac{1}{2} - \alpha - 1} \delta_s
\right] \\
&=  \left(\frac{1}{2} - \alpha\right) \mx_t^{-1}   \delta_x
- \left(\frac{1}{2} + \alpha\right) \ms_t^{-1} \delta_s
\end{align*}
and therefore we have
\begin{align*}
\norm{ \ln(\tilde{w}^{\new}) - \ln( \tilde{w} ) }_\infty &\leq \int_{0}^{1} \normFull{
 \left(\frac{1}{2} - \alpha\right) \mx_t^{-1}   \delta_x
- \left(\frac{1}{2} + \alpha\right) \ms_t^{-1} \delta_s
} dt \\
&\leq 
\left(\frac{1}{2} - \alpha\right) \int_{0}^{1}  \normFull{\mx_t^{-1}   \delta_x }_\infty dt
+ \left(\frac{1}{2} + \alpha\right)  \int_{0}^{1}  \normFull{\ms_t^{-1} \delta_s}_\infty dt \\
&\leq \frac{1}{1 - (1/40)} \left[
\left(\frac{1}{2} - \alpha\right) \int_{0}^{1}  \normFull{\mx^{-1}   \delta_x }_\infty dt
+ \left(\frac{1}{2} + \alpha\right)  \int_{0}^{1}  \normFull{\ms^{-1} \delta_s}_\infty dt \right]
\end{align*}
where in the last step we used that $\|\mx^{-1}\delta_{x}\|_{\infty}\leq\frac{1}{40}$ and
$\|\ms^{-1}\delta_{s}\|_{\infty}\leq\frac{1}{40}$. Since $\alpha \in [0, \frac{1}{5}]$ this implies that $\norm{ \ln(\tilde{w}^{\new}) - \ln( \tilde{w} ) }_\infty \leq 1/20$, i.e. $w^{\new} \approx_{1/20} w$. Consequently, \Cref{lem:mult_to_add} yields
\begin{align}
\norm{{\tilde{\mw}}^{-1}\delta_{\tilde{w}}}_\infty 
&\leq 
\frac{1 + (1/20)}{1 - (1/40)} \left[
\left(\frac{1}{2} - \alpha\right) \int_{0}^{1}  \normFull{\mx^{-1}   \delta_x }_\infty dt
+ \left(\frac{1}{2} + \alpha\right)  \int_{0}^{1}  \normFull{\ms^{-1} \delta_s}_\infty dt \right]
\nonumber
\\
&\leq \|\mx^{-1}\delta_{x}\|_\infty
+\norm{\ms^{-1}\delta_{s}}_\infty ~. \label{eq:mw_split}
\end{align}
Consequently, for the first term in $K$, using $\|\E[(\mx^{-1}\delta_{x})^{2}]\|_{\tau+\infty}\leq12\gamma^{2}$,
$\|(\ms^{-1}\delta_{s})^{2}\|_{\tau+\infty}\leq4\gamma^{2}$ (Lemma
\ref{lem:ls_second}), $\tau\geq\sigma$, and \eqref{eq:mw_split} we have
\begin{align*}
\|\E [ (\tilde{\mw}^{-1}\delta_{\tilde{w}})^{2} ] \|_{\sigma + \infty}
\leq & ~ 2\|\E(\mx^{-1}\delta_{x})^{2}\|_{\sigma + \infty}+2\|(\ms^{-1}\delta_{s})^{2}\|_{\sigma+\infty}\\
\leq & ~ 2(12\gamma^{2})+2(4\gamma^{2})=32\gamma^{2}.
\end{align*}

For the second term in $K$, we have again by \eqref{eq:mw_split} that
\begin{align*}
\E[ \|\tilde{\mw}^{-1}|\delta_{\tilde{w}}|\|_{\mproj^{(2)}(\mw_{t}\ma)+\infty} ]
& \leq \E [ \|\mx^{-1}|\delta_{x}|\|_{\mproj^{(2)}(\tilde{\mw}_{t}\ma)+\infty}+\|\ms^{-1}|\delta_{s}|\|_{\mproj^{(2)}(\tilde{\mw}_{t}\ma)+\infty} ] \\
& \leq 1.1 \E [ \|\mx^{-1}|\delta_{x}|\|_{\mproj^{(2)}(\tilde{\mw}\ma)+\infty}+\|\ms^{-1}|\delta_{s}|\|_{\mproj^{(2)}(\tilde{\mw}\ma)+\infty} ]
\end{align*}
where we used Lemma \ref{lem:P2order}, $\tilde{\mw}_{t} \approx_\gamma \tilde{\mw}$ and $\gamma<1/80$ at the end.

Since we sample coordinates $i$ of $\delta_{x}$ independently with
probability at least $10\sigma_{i}(\tilde{\mw}\ma)$, \Cref{lem:P2sample}  shows that
\begin{align*}
\E [ \| \mx^{-1} | \delta_{x} | \|_{\mproj^{(2)} ( \tilde{\mw} \ma ) } ]^{2} 
\leq
\E [ \| \mx^{-1} | \delta_{x} | \|_{\mproj^{(2)} ( \tilde{\mw} \ma ) }^{2} ] 
\leq 1.2 \cdot \| \E [ \mx^{-1} | \delta_{x} | ] \|_{\sigma(\tilde{\mw}\ma)}^{2} 
\end{align*}
Using this, we have
\begin{align*}
\E [ \|\tilde{\mw}^{-1}|\delta_{\tilde{w}}|\|_{\mproj^{(2)}(\tilde{\mw}_{t}\ma) + \infty} ] 
& \leq 1.4 \| \E [ \mx^{-1} | \delta_{x} | ] \|_{\sigma+\infty} + 1.1 \|\ms^{-1}\delta_{s}\|_{\sigma+\infty}\\
& \leq 2.5 \gamma+2.5\gamma=5\gamma
\end{align*}
where we used $\|\mx^{-1}\E [ |\delta_{x}| ] \|_{\tau + \infty} \leq 2 \gamma$, $\|\ms^{-1}\delta_{s}\|_{\tau + \infty}\leq2\gamma$
(Lemma \ref{lem:projection_sizes_ls}).

Hence, we have
\begin{align*}
K\leq 10(32\gamma^2)+65(5\gamma)(24\gamma^{2})\leq2000\gamma^2.
\end{align*}
\end{proof}

\subsubsection{Potential Decrease}
\label{sec:ls:potential_decrease}

Here we combine the analysis from \Cref{sub:ls:deltas} and \Cref{sub:delta_tau} to analyze our short-step method for the LS weight function.  We first provide two technical lemmas (\Cref{lem:P2bound} and \Cref{lem:LS_op_bound} respectively) from from \cite{ls19} and \cite{blss20} that we use in our analysis and then provide a sequence of results analogous to those we provided for the short-step method for the log barrier. In \Cref{lem:expect_bound_ls} we bound our first order expected progress from a step  (analogous to \Cref{lem:expect_bound}), in \Cref{lem:log_exp_progress_LS} we bound the expected decrease in the potential (analogous to \Cref{lem:log_exp_progress}), and in \Cref{lem:centrality_mainLS} we give our main centrality improvement lemma (analogous to \Cref{lem:centrality_main}. Finally, in \Cref{lem:infeasibility_bound_ls} we analyze the feasibility of a step (analagous to \Cref{lem:infeasibility_bound}) and put everything together to prove that \Cref{alg:short_step_LS} is a valid short step procedure in \Cref{lem:main_lem_LS} (analogous to \Cref{lem:main_lem}). 

\begin{lemma}[Lemma 47 in \cite{ls19}]
	\label{lem:P2bound}For any non-degenerate $\ma\in\R^{m\times n}$,
	we have
	\[
	\|\mSigma(\ma)^{-1}\mproj^{(2)}(\ma)\|_{\infty} \leq 1
	\text{ and }
	\|\mSigma(\ma)^{-1}\mproj^{(2)}(\ma)\|_{\tau(\ma)} \leq 1 ~.
	\] 
	Consequently,  $\|\mT(\ma)^{-1}\mproj^{(2)}(\ma)\|_{\infty} \leq 1$
	and $\|\mT(\ma)^{-1}\mproj^{(2)}(\ma)\|_{\tau(\ma)} \leq 1$.
\end{lemma}

\begin{lemma}[Lemma 26 in \cite{blss20}]
	\label{lem:LS_op_bound}In Algorithm~\ref{alg:short_step_LS}, $\|\Psi\|_{\tau+\infty}\leq1-\alpha$ where 
	$\Psi \defeq ((1-4\alpha^{2})\mT^{-1}\mLambda-2\alpha \mI)\omw^{-1/2}(\mI - 2\mq)\omw^{1/2}$,
	$\mLambda \defeq \mLambda(\oma)$ and $\mq =\oma\omh^{-1}\oma^{\top}$.
\end{lemma}

\begin{lemma}[First-order Expected Progress Bound]
\label{lem:expect_bound_ls} In Algorithm~\ref{alg:short_step_LS}
we have
\begin{align*}
\norm{\mx^{-1}\E[\delta_{x}]+\ms^{-1}\E[\delta_{s}]-g-\mT^{-1}\E[\delta_{\tau}]}_{\tau+\infty}\leq(1-\alpha+ 36\epsilon/\alpha)\gamma.
\end{align*}
\end{lemma}

\begin{proof}
Let $\mq \defeq \oma\omh^{-1}\oma^{\top}$. Note that
\begin{align*}
\E[\mx^{-1}\delta_{x}]=(1+2\alpha)(g-\omw^{-1/2}\mq\omw^{1/2}g)+\eta_{x}-\E[\delta_c]
\end{align*}
where $\eta_{x}=\E[\mx^{-1}\delta_{x}-\omx^{-1}\delta_{x}]$. Since $\omx \approx_{\epsilon} \mx$, \Cref{lem:mult_to_add}  and
 \Cref{lem:Edelta_tau_bound} show that
\[
\|\eta_{x}\|_{\tau+\infty}\leq
\norm{(\mi - \omx^{-1} \mx ) \mx^{-1} \E \delta_x}_{\tau + \infty}
\leq 2\epsilon \|\mx^{-1}\E\delta_{x}\|_{\tau+\infty}\leq4\epsilon\gamma
\]
Further, \Cref{lem:projection_sizes_ls} shows that $\|\delta_c\|_{\tau+\infty}\leq 11\epsilon\gamma/\alpha$.

Similarly, we have
\begin{align*}
\ms^{-1}\delta_{s} = ( 1 - 2 \alpha ) \omw^{-1/2} \mq \omw^{1/2}g + \eta_{s}
\end{align*}
with $\| \eta_{s} \|_{\tau+\infty} \leq 4 \epsilon\gamma$. Lemma \ref{lem:Edelta_tau_bound}
shows that
\begin{align*}
\mT^{-1}\delta_{\tau} = \mT^{-1}\mLambda((1-2\alpha)\mx^{-1}\delta_{x}-(1+2\alpha)\ms^{-1}\delta_{s})+\eta_{\tau}
\end{align*}
with $\|\E \eta_{\tau}\|_{\tau+\infty}\leq2000\gamma^{2}$. Combining
the three inequalities above, we have
\begin{align*}
 & \E\left[\mT^{-1}\delta_{\tau}-\mx^{-1}\delta_{x}-\ms^{-1}\delta_{s}\right]\\
= & \mT^{-1}\mLambda((1-2\alpha)\E[\mx^{-1}\delta_{x}]-(1+2\alpha)\ms^{-1}\delta_{s})+\eta_{\tau}-\E[\mx^{-1}\delta_{x}]-\ms^{-1}\delta_{s}\\
= & \mT^{-1}\mLambda((1-2\alpha)(1+2\alpha)(g-\omw^{-1/2}\mq\omw^{1/2}g)-(1+2\alpha)(1-2\alpha)\omw^{-1/2}\mq\omw^{1/2}g)\\
 & -(1+2\alpha)(g-\omw^{-1/2}\mq\omw^{1/2}g)-(1-2\alpha)\omw^{-1/2}\mq\omw^{1/2}g\\
 & +\eta_{\tau}+(1-2\alpha)\mT^{-1}\mLambda\eta_{x}-\eta_{x}-(1+2\alpha)\mT^{-1}\mLambda\eta_{s}-\eta_{s}\\
= & (\Psi-\mi)g+\eta
\end{align*}
with
\begin{align*}
\Psi & =((1-4\alpha^{2})\mT^{-1}\mLambda-2\alpha)\omw^{-1/2}(\mI-2\mq)\omw^{1/2},\\
\eta & =\eta_{\tau}+(1-2\alpha)\mT^{-1}\mLambda\eta_{x}-\eta_{x}-(1+2\alpha)\mT^{-1}\mLambda\eta_{s}-\eta_{s} + \E[\delta_c].
\end{align*}
By Lemma \ref{lem:LS_op_bound}, we have $\|\Psi\|_{\tau+\infty}\leq1-\alpha$.
By the fact that $\|\mT^{-1}\mLambda\|_{\tau+\infty}\leq2$ (Lemma
\ref{lem:P2bound}), we have that
\begin{align*}
\| \eta \|_{\tau+\infty} & \leq 2000\gamma^{2} + 24 \epsilon\gamma + 11 \epsilon\gamma/\alpha \leq 36 \epsilon\gamma/\alpha
\end{align*}
where we used $\gamma\leq\frac{\epsilon}{500}$ and $\alpha < 1/4$. Hence, we have
\begin{align*}
\|\E[\mx^{-1}\delta_{x}-\ms^{-1}\delta_{s}-\mT^{-1}\delta_{\tau}]-g\|_{\tau+\infty}
\leq  \|\Psi g\|_{\tau+\infty}+\|\eta\|_{\tau+\infty}
\leq  (1-\alpha+36\epsilon/\alpha)\gamma .
\end{align*}
\end{proof}
We now have everything to bound the effect of a step in terms of $\norm{\phi'(w)}_{\tau+\infty}$
and $\norm{\phi''(w)}_{\tau+\infty}$.
\begin{lemma}[Expected Potential Decrease]
\label{lem:log_exp_progress_LS}In Algorithm~\ref{alg:short_step_LS}
we have
\begin{align*}
\E[\phicent(x^{\new},s^{\new},\mu^{\new})]
\leq&~
\phicent(x,s,\mu)+\phi'(v)^{\top}g+(1-\alpha+20\epsilon/\alpha)\gamma\norm{\phi'(v)}_{\tau+\infty}^{*}\\
& +6000\gamma^{2}\norm{\phi''(v)}_{\tau+\infty}^{*}.
\end{align*}
\end{lemma}

\begin{proof}
We apply Lemma~\ref{lem:expand_potential} with $u^{(1)}=x$, $\delta^{(1)}=\delta_{x}$,
$c^{(1)}=1$, $u^{(2)}=s$, $\delta^{(2)}=\delta_{s}$, $c^{(2)}=1$,
$u^{(3)}=\mu\vones$, $\delta^{(3)}=\delta_{\mu}\vones$, $c^{(3)}=-1$,
$u^{(4)}=\tau(x,s)$, $\delta^{(4)}=\delta_{\tau}$, $c^{(4)}=-1$.
Note that, $v$ in the context of Lemma~\ref{lem:expand_potential}
is precisely $\frac{xs}{\mu\tau}$ and $v^{\new}$ in the context
of Lemma~\ref{lem:expand_potential} is precisely $\frac{x^{\new}s^{\new}}{\mu^{\new}\tau^{\new}}$.

Now, by Lemma~\ref{lem:ls_second}, the definition of $\gamma$,
and $\epsilon\in[0,1/80]$ we have
\begin{align*}
\norm{\ms^{-1}\delta_{s}}_{\infty} \leq 2\gamma \leq \frac{1}{50\lambda} \leq \frac{1}{12(1+\norm c_{1}\lambda)} .
\end{align*}
Similarly, Lemma~\ref{lem:ls_second} implies
\begin{align*}
\norm{\mx^{-1}\delta_{x}}_{\infty}\leq2\gamma\leq\frac{1}{50\lambda}\leq\frac{1}{12(1+\norm c_{1}\lambda)} .
\end{align*}
Further, by definition of $r$ we have $\mu^{-1}|\delta_{\mu}|\leq\gamma\leq(12(1+\norm c_{1}\lambda)^{-1})$. Thus, Lemma~\ref{lem:expand_potential} implies 
\begin{align}
&~
\E \left[\phicent(x^{\new},s^{\new},\mu^{\new})\right] \notag\\
 \leq&~
 \Phi(x,s,\mu)+\phi'(v)^{\top}\mv\left[(\mx^{-1} \E[\delta_{x}]+\ms^{-1} \delta_s-\mu^{-1}\delta_{\mu}\vones-\mT^{-1}\E[\delta_{\tau}])\right]\nonumber \\
 &~ +10\left(\E[\norm{\mx^{-1}\delta_{x}}_{\phi''(v)}^{2}]+\norm{\ms^{-1}\delta_{x}}_{\phi''(v)}^{2}+\norm{\mu^{-1}\delta_{\mu}\vones}_{\phi''(v)}^{2}+\E\left[\norm{\mT^{-1}\delta_{\tau}]}_{\phi''(v)}^{2}\right]\right)\,.\label{eq:step-1}
\end{align}

Now, since $v\approx_{\epsilon}\vones$, the assumption on $\delta_{\mu}$
implies 
\begin{align*}
|\mu^{-1}\delta_{\mu}\phi'(v)^{\top}\mv\vones|\leq\mu^{-1}|\delta_{\mu}|\exp(\epsilon)\norm{\phi'(v)}_{\tau+\infty}^{*}\| \vones \|_{\tau+\infty}\leq6\epsilon\gamma\norm{\phi'(v)}_{\tau+\infty}^{*}
\end{align*}
and
\begin{align*}
\norm{\mu^{-1}\delta_{\mu}\vones}_{\phi''(v)}^{2}=|\mu^{-1}\delta_{\mu}|^{2}\norm{\phi''(v)}_{1}\leq\left(\epsilon\gamma/\sqrt{n}\right)^{2}\norm{\phi''(v)}_{\tau+\infty}^{*}\|1\|_{\tau+\infty}\leq\gamma^{2}\norm{\phi''(v)}_{\tau+\infty}^{*}\,.
\end{align*}
Applying
Lemmas~\ref{lem:expect_bound_ls},~\ref{lem:ls_second},~and~\ref{lem:Edelta_tau_bound}
yields the result.
\end{proof}
Finally, we can prove the desired bound on $\E[\phicent(x^{\new},s^{\new},\mu^{\new})]$
by combining the preceding lemma with \Cref{lem:centrality_main_general} and bounds on $\norm{\cdot}_{\tau + \infty}$.
\begin{lemma}[Centrality Improvement of LS Barrier Short-Step]
	\label{lem:centrality_mainLS}
In Algorithm~\ref{alg:short_step_LS}
we have 
\begin{align*}
\E[\phicent(x^{\new},s^{\new},\mu^{\new})]\leq(1-\lambda r)\phicent(x,s,\mu)+\exp(\lambda\epsilon/6)\,.
\end{align*}
\end{lemma}

\begin{proof}
Note that $\|\ov-w/\tau\|_{\infty}\leq \gamma$. Consequently, we will apply \Cref{lem:centrality_main_general}
with $\delta = \gamma$, $c_{1}=1-\alpha+20\epsilon/\alpha$, $c_{2}=6000$,
and $U=\{ x : \|x\|_{\tau+\infty} \leq 1 \}$. Using $\gamma = \epsilon/(100\lambda)$, $\lambda = 36 \log(400 m/ \epsilon^2) / \epsilon$, and $\epsilon=\frac{\alpha^2}{240}$, we have $\delta \leq 1/(5\lambda)$, $\lambda \gamma \leq 1/8$, and 
\begin{align*}
2\lambda\delta+c_{2}\lambda\gamma=6004\lambda\gamma=60.04\epsilon\leq0.26\alpha^2\leq\frac{\alpha^2}{3}\leq\frac{1-c_{1}}{2}.
\end{align*}
Hence, the assumption in Lemma \ref{lem:centrality_main_general}
is satisfied. Further,  if  $\norm{x}_{\infty} \leq \frac{\alpha}{16 \sqrt{n} }$ then
\[
\norm{x}_{\tau + \infty}
= \cnorm \norm{x}_{\tau} + \norm{x}_{\infty}
\leq \frac{\alpha}{16 \sqrt{n}} \left[ \frac{10}{\alpha} \cdot \sqrt{2 n} + 1 \right]
\leq 1 
\]
and consequently $U$ contains the $\norm{\cdot}_{\tau+\infty}$ ball of radius $\alpha/(16 \sqrt{n})$. Consequently, applying \Cref{lem:centrality_main_general} and using that $(1- c_1) \geq \alpha/2$
\begin{align*}
\E[\phicent(x^{\new},s^{\new},\mu^{\new})] 
& \leq \phicent(x,s,\mu)-\frac{(1 - c_1)\lambda\gamma u}{2}\phicent(x,s,\mu)+m\\
& \leq \phicent(x,s,\mu)-\frac{\alpha^2 \lambda\gamma}{16\sqrt{n}}\phicent(x,s,\mu)+m
~.
\end{align*}
The result follows from $\exp(\frac{1}{6}\lambda\epsilon)\geq m$
and $r=\frac{\epsilon\gamma}{\sqrt{n}}$.
\end{proof}

Next, we bound the infeasibility after a step (\Cref{lem:infeasibility_bound_ls}). This lemma is analogous to \Cref{lem:infeasibility_bound} and consequently the proof is abbreviated. (See \Cref{sec:feasibility} for further intuition.)

\begin{lemma}[Feasibility Bound]\label{lem:infeasibility_bound_ls} 
In \Cref{alg:short_step_LS}
	if $\oma^{\top}\mr\oma\approx_{\gamma}\oma^{\top}\oma$ then
	\begin{equation}\label{eq:log_feas_bound_ls}
	\norm{\ma^{\top}x^{\new}-b}_{\ma^{\top}\mx^{\new}(\ms^{\new})^{-1}\ma}\leq 0.5 \epsilon\gamma ( \mu^{\new} )^{1/2}
	\end{equation}
	and consequently, this holds with probability at least $1-r$.
\end{lemma}

\begin{proof} 
	Recall that $x^{\new}=x+\delta_{x}$ where $\delta_x = \omx(g - \mr \delta_r)$, $\delta_{r} \defeq (1 + 2\alpha) \delta_{p}+\delta_{c}$, and 
	\[
	\delta_{p} = \omw^{-1/2}\oma\omh^{-1}\oma^{\top}\omw^{1/2}g 
	\text{ and }
	\delta_{c} = \mu^{-1/2}\cdot\omw^{-1/2}\oma\omh^{-1}(\ma^{\top}x-b)
	~.
	\]
	Consequently, we can rewrite the step as
	\[
	\delta_r = \omw^{-1/2} \oma \omh^{-1} d 
	\text{ where }
	d \defeq (1 + 2\alpha) \oma^\top \omw^{1/2} g +  \mu^{-1/2} (\ma^\top x - b) ~.
	\]
	
	Now, consider the idealized step $x^* \defeq x + \delta_x^*$ where $\delta_x^* \defeq \omx(g - \delta_r^*)$, and $\delta_r^* = \omw^{-1/2} \oma (\oma^\top \oma)^{-1} d$. Since, $\ma^{\top}\omx=\sqrt{\mu}\oma^{\top}\omw^{1/2}$  we have
	\[
	\ma^\top x^*
	= \ma^\top x
	+ \sqrt{\mu} \oma^\top \left( \mw^{1/2} g - \oma (\oma^\top \oma)^{-1} d \right)
	= \ma^\top x - (\ma^\top x - b) = b ~.
	\]
	 Consequently,
	\begin{align}
	\ma^\top x - b
	&= \ma^\top (x - x^*)
	= - \ma^\top \omx 
	(\mR \delta_r - \delta_r^*)
	=  \sqrt{\mu} \oma^\top \left(
	\oma (\oma^\top \oma)^{-1} - \mr \oma \omh^{-1} \right) d
	\nonumber
	\\
	&= \sqrt{\mu} \left(\mI - \oma^\top \mR \oma \omh^{-1}\right) d ~.
	\label{feasibility:equality:ls}
	\end{align}
	
	Since $\oma^{\top}\mr\oma\approx_{\gamma}\oma^{\top}\oma\approx_{\gamma}\omh$
	we have
	\begin{align*}
	\norm{(\oma^{\top}\oma)^{-1/2}(\mi-\oma^{\top}\mr\oma\omh^{-1})\omh^{1/2}}_{2} & =\norm{(\oma^{\top}\oma)^{-1/2}(\omh-\oma^{\top}\mr\oma)\omh^{-1/2}}_{2}\\
	& \leq\exp(\gamma)\norm{\omh^{-1/2}(\omh-\oma^{\top}\mr\oma)\omh^{-1/2}}_{2}\leq3\gamma
	\end{align*}
	where we used Lemma~\ref{lem:mult_to_add} and $(2\gamma+4\gamma^{2})\exp(\gamma)\leq3\gamma$
	for $\gamma\leq1/80$. Consequently, combining with \eqref{feasibility:equality} yields that
	\begin{align*}
	\norm{\ma^{\top}x^{\new}-b}_{(\oma^{\top}\oma)^{-1}} 
	&=
	\sqrt{\mu}\normFull{(\oma^{\top}\oma)^{-1/2}(\mi-\oma^{\top}\mr\oma\omh^{-1})\omh^{1/2}\omh^{-1/2}
		d}_{2}
	\leq 
	3\gamma\sqrt{\mu} \norm{d}_{\omh^{-1}}
	\\
	&
	\leq
	3\gamma\sqrt{\mu}\left( (1+2\alpha) \norm{\oma^{\top}\omw^{-1/2}g}_{\omh^{-1}}+\mu^{-1/2}\cdot\norm{\ma^{\top}x-b}_{\omh^{-1}}\right)\,.
	\end{align*}
	Next, note that by design of $g$ we have $\norm{g}_{\otau} \leq (\gamma/\cnorm)$ for $\cnorm = 10/\alpha$. Further, we have $\omw \approx_{4\epsilon} \otau$ as $\ow \approx_{2\gamma} w$, $w \approx_{\epsilon} \tau$, $\tau \approx_{\gamma} \otau$, and $\gamma \leq \epsilon$). Since $\epsilon \leq 1/80$ and $\alpha \leq 1/4$  this  implies
	\begin{align*}
	 \norm{\oma^{\top}\omw^{-1/2}g}_{\omh^{-1}} 
	&\leq  \exp(\gamma/2)\norm{\omw^{-1/2}g}_{\oma(\oma^{\top}\oma)^{-1}\oma^{\top}} 
	\leq  \exp(\gamma/2)\norm{\omw^{-1/2}g}_{2} \\
	&\leq  \exp(5 \epsilon / 2) \norm g_{\otau} \leq 
	\frac{\exp(5\epsilon / 2) \alpha \gamma}{10}
	\leq \frac{2 \gamma}{1 + 2\alpha}  ~.
	\end{align*}
	Further, by the approximate primal feasibility of $(x,s,\mu)$ we have
	\begin{align*}
	\norm{\ma^{\top}x-b}_{\omh^{-1}}\leq\exp(\gamma/2)\norm{\ma^{\top}x-b}_{\ma^{\top}\mx\ms^{-1}\ma}\leq2\epsilon\gamma\sqrt{\mu}\,.
	\end{align*}
	Combining yields that 
	\begin{align*}
	\norm{\ma^{\top}x^{\new}-b}_{(\oma^{\top}\oma)^{-1}}\leq3\gamma\sqrt{\mu}\left(2\gamma+2\epsilon\gamma\right)\leq9\gamma\sqrt{\mu}\gamma\,\leq 0.25 \epsilon\gamma \sqrt{\mu}
	~.
	\end{align*}
	
	Now, since $\norm{\ms^{-1}\delta_{s}}_{\infty}\leq2\gamma$ and $\norm{\mx^{-1}\delta_{x}}_{\infty}\leq2\gamma$
	by \Cref{lem:ls_second}
	and $\gamma\leq\epsilon/100$ we have that $\ma^{\top}\mx^{\new}(\ms^{\new})^{-1}\ma\approx_{\epsilon}\ma^{\top}\mx\ms\ma$.
	Combining with the facts that $\ma^{\top}\mx\ms\ma\approx_{2\gamma}\oma^{\top}\oma$,
	and $\mu^{\new}\approx_{2\epsilon\gamma}\mu$ yields \eqref{eq:log_feas_bound_ls}.
	The probability bound follows by the fact that $q_{i}\geq C\sigma_{i}(\oma)\log(n/(\epsilon r))\gamma^{-2}$
	and Lemma 4 in \cite{clmmps15}. That is, the matrix $\mr$ corresponds to sampling each row of $\ma$ with probability at least proportional to its leverage score. 
\end{proof}

Next we prove the main result of this section, that \Cref{alg:short_step_LS} is a valid $\textsc{ShortStep}$ procedure (see Definition~\ref{def:short_step}).

\begin{lemma}[LS Short Step]
\label{lem:main_lem_LS}
Algorithm~\ref{alg:short_step_LS} is a valid
$\textsc{ShortStep}$ procedure (\Cref{def:short_step})
with weight function $\tau(x,s)\defeq\sigma(\mx^{1/2-\alpha}\ms^{-1/2-\alpha}\ma)+\frac{n}{m}\vones \in\R^{m}\text{ for all }x,s\in\R_{>0}^{m}$
and parameters $\lambda \defeq 36 \epsilon^{-1} \log(400m/\epsilon^{2}) $,
$\gamma\defeq\frac{\epsilon}{100\lambda}$, $r\defeq\frac{\epsilon\gamma}{\sqrt{n}}$,
$\epsilon=\frac{\alpha^2}{240}$, and $\alpha=\frac{1}{4\log(\frac{4m}{n})}$,
i.e.
\begin{enumerate}
\item $\lambda\geq\frac{12\log(16m/r)}{\epsilon}$,
\item $\ma y^{\new}+s^{\new}=b$ for some $y^{\new}\in\R^{n}$,
\item $\E[\phicent(x^{\new},s^{\new},\mu^{\new})]\leq(1-\lambda r)\phicent(x,s,\mu)+\exp\left(\lambda\epsilon/6\right)$
\item $\P\left[\norm{\ma^{\top}x^{\new}-b}_{(\ma^{\top}\mx^{\new}(\ms^{\new})^{-1}\ma)^{-1}}\leq\frac{\epsilon}{2}\sqrt{\mu^{\new}}\right]\geq1-r\,.$
\end{enumerate}
\end{lemma}

\begin{proof}
We provide the proofs of all the four parts as follows:

(1) Note that $r=\frac{\epsilon\gamma}{\sqrt{n}}=\frac{\epsilon^{2}}{100\lambda\sqrt{n}}=\frac{\epsilon^{3}}{3600\log(400\beta)\sqrt{n}}$
for $\beta=\frac{m}{\epsilon^{2}}$. Thus, 
\begin{align*}
12\log(16m/r)\cdot\epsilon^{-1}\leq12\log(16\beta^{3/2}\cdot9600\log(400\beta))\cdot\epsilon^{-1}\leq12\log((400)^{3}\beta^{2})\cdot\epsilon^{-1}\leq\lambda
\end{align*}
where we used that $\sqrt{x}\log x \leq x$ for $x\geq1$, $(400)^{3}\geq 16\cdot9600\cdot400$,
and $\beta \geq 1$.

(2) Note that $\delta_{s} = \ma\omh^{-1} \oma^{\top}\omw^{1/2} g$. Consequently,
$\delta_{s}\in\im(\ma)$ and this follows from the fact that $\ma y+s=b$
for some $y\in\R^{n}$ by definition of $\epsilon$-centered.

(3) This follows immediately from Lemma~\ref{lem:centrality_mainLS}.

(4) This follows immediately from \Cref{lem:infeasibility_bound_ls}.
\end{proof}

\subsection{Potential Function and Analysis Tools}
\label{sec:potential_function}

Here we give general mathematical facts and properties of the potential function we used throughout this section in analyzing our IPM. First, we give a simple, standard technical lemma allowing us to relate different notions of multiplicative approximation \Cref{lem:mult_to_add}.
Then, in the remainder of this section we give facts about the potential function given in \Cref{def:central_potential}.

\multToAdd

\begin{proof}
	These follow from the fact that for all $\alpha\in\R$ with $|\alpha|\leq1/2$
	we have $1+\alpha\leq\exp(\alpha)\leq1+\alpha+\alpha^{2}$ and therefore
	$|\exp(\alpha)-1|\leq|\alpha|+|\alpha|^{2}$. 
\end{proof}

\begin{lemma}
	\label{lem:potential_simple_bounds} For all $w\in\R$ we have $0\leq\phi''(w)\leq\lambda\cdot(|\phi'(w)|+2\lambda)$ and for all $z_{0},z_{1}$ with $|z_{1}-z_{0}|\leq\delta$ we have
	$\phi''(z_{1})\leq\exp(\lambda\delta)\phi''(z_{0})$.
\end{lemma}

\begin{proof}
	Direct calculation shows that 
	\begin{align*}
	\phi'(w)=\lambda\cdot\left(\exp(\lambda(w-1))-\exp(\lambda(1-w))\right)\text{ and }\phi''(w)=\lambda^{2}\cdot(\exp(\lambda(w-1))+\exp(\lambda(1-w)))\,.
	\end{align*}
	and therefore, clearly $\phi''(w)\geq0$ and the bounds on $\phi''(w)$ follow
	from
	\begin{align*}
	\left|\phi'(w)\right| & =\lambda\cdot(\exp(\lambda|w-1|)-\exp(-\lambda|w-1|))\\
	& =\frac{1}{\lambda}\phi''(w)-2\lambda\exp(-\lambda|w-1|))\geq\frac{1}{\lambda}\phi''(w)-2\lambda\,.
	\end{align*}
	The bound on $\phi''(z_1)$ is then immediate from $\phi''(z_1)=\lambda^{2}\cdot(\exp(\lambda(z_1-))+\exp(\lambda(1-z_1)))$.
\end{proof}

\expandPotential

\begin{proof}
For all $t\in[0,1]$ let $z_{t}\in\R^{n}$ be defined for all $i\in[n]$
by $[z_{t}]_{i}\defeq\prod_{j\in[k]}(u_{i}^{(j)}+t\cdot\delta_{i}^{(j)})^{c_{j}}$
and $f(t)\defeq\Phi(z_{t})$. Note that $z_{0}=v$ and $z_{1}=v^{\new}$
and consequently $f(1)=\Phi(v^{\new})$ and $f(0)\defeq\Phi(v)$.
Further, by Taylor's theorem, $f(1)=f(0)+f'(0)+\frac{1}{2}f''(\zeta)$
for some $\zeta\in[0,1]$. Consequently, it suffices to
show that $z_{t}\approx_{1/(11\lambda)}z_{0}$, compute $f'(0)$ and
upper bound $f''(\zeta)$. 

First, by direct calculation, we know that for all $i\in[n]$ and
$t\in[0,1]$
\begin{align*}
\frac{d}{dt}[z_{t}]_{i} & =\sum_{j\in[k]}\left[\left(\prod_{j'\in[k]\setminus j}(u_{i}^{(j')}+t\cdot\delta_{i}^{(j')})^{c_{j'}}\right)\cdot\left(\frac{d}{dt}(u_{i}^{(j)}+t\cdot\delta_{i}^{(j)})^{c_{j}}\right)\right]\\
 & =\sum_{j\in[k]}\left[\left([z_{t}]_{i}\cdot(u_{i}^{(j)}+t\cdot\delta_{i}^{(j)})^{-c_{j}}\right)\cdot c_{j}(u_{i}^{(j)}+t\cdot\delta_{i}^{(j)})^{c_{j}-1}\cdot\delta_{i}^{(j)}\right]\\
 & =[z_{t}]_{i}\sum_{j\in[k]}\frac{c_{j}\cdot\delta_{i}^{(j)}}{u_{i}^{(j)}+t\cdot\delta_{i}^{(j)}}
\end{align*}
and 
\begin{align*}
\frac{d^{2}}{dt^{2}}[z_{t}]_{i}=\left[\frac{d}{dt}[z_{t}]_{i}\right]\sum_{j\in[k]}\frac{c_{j}\cdot\delta_{i}^{(j)}}{u_{i}^{(j)}+t\cdot\delta_{i}^{(j)}}-[z_{t}]_{i}\sum_{j\in[k]}\frac{c_{j}(\delta_{i}^{(j)})^{2}}{(u_{i}^{(j)}+t\cdot\delta_{i}^{(j)})^{2}}\,.
\end{align*}
Thus, $\mz_{t}\defeq\mdiag(z_{t})$, $u_{t}^{(j)}\defeq u^{(j)}+t\delta^{(j)}$,
$\mU_{t}^{(j)}\defeq\mdiag(u_{t}^{(j)})$, and $\tilde{\delta}_{t}^{(j)}\defeq[\mU_{t}^{(j)}]^{-1}\delta^{(j)}$
satisfy 
\begin{equation}
\frac{d}{dt}z_{t}=\mz_{t}\sum_{j\in[k]}c_{j}\tilde{\delta}_{t}^{(j)}\text{ and }\frac{d^{2}}{dt^{2}}z_{t}=\mz_{t}\left(\left[\sum_{j\in[k]}c_{j}\tilde{\delta}_{t}^{(j)}\right]^{2}-\sum_{j\in[k]}c_{j}[\tilde{\delta}_{t}^{(j)}]^{2}\right)\label{eq:zt_compute}
\end{equation}
where for $a\in\R^{n}$ we define $a^{2}\in\R^{n}$ with $[a^{2}]_{i}\defeq[a_{i}]^{2}$.

Applying chain rule and (\ref{eq:zt_compute}) yields that
\begin{align*}
f'(t) & =\sum_{i\in[n]}\phi'([z_{t}]_{i})\frac{d}{dt}[z_{t}]_{i}=\phi'(z_{t})^{\top}\mz_{t}\sum_{j\in[k]}c_{j}\tilde{\delta}_{t}^{(j)}\,.
\end{align*}
Since $z_{0}=v$, $\mU_{0}^{(j)}=\mU_{j}$, and $\tilde{\delta}_{0}^{(j)}=\mU_{j}^{-1}\delta^{(j)}$
we have $f'(0)=\sum_{j\in[k]}c_{j}\phi'(v)^{\top}\mv\mU^{-1}\delta^{(j)}$
as desired.

Further application of chain rule and (\ref{eq:zt_compute}) yields
that
\begin{align*}
f''(t)= & \sum_{i\in[n]}\phi''([z_{t}]_{i})\left(\frac{d}{dt}[z_{t}]_{i}\right)^{2}+\sum_{i\in[n]}\phi'([z_{t}]_{i})\frac{d^{2}}{dt^{2}}z_{t}\\
= & \sum_{i\in[n]}\phi''([z_{t}]_{i})\cdot\left[\mz_{t}\sum_{j\in[k]}c_{j}\tilde{\delta}_{t}^{(j)}\right]_{i}^{2}+\sum_{i\in[n]}\phi'([z_{t}]_{i})\cdot\left([z_{t}]_{i}\left[\sum_{j\in[k]}c_{j}\tilde{\delta}_{t}^{(j)}\right]_{i}^{2}-[z_{t}]_{i}\sum_{j\in[k]}[\tilde{\delta}_{t}^{(j)}]_{i}^{2}\right)\\
= & \normFull{\sum_{j\in[k]}c_{j}\tilde{\delta}_{t}^{(j)}}_{\mz_{t}\phi''(z_{t})\mz_{t}+\mz_{t}\phi'(z_{t})}^{2}-\sum_{j\in[k]}c_{j}\normFull{\tilde{\delta}_{t}^{(j)}}_{\mz_{t}\phi'(z_{t})}^{2}.
\end{align*}
Now, since by Cauchy Schwarz for all $i\in[n]$ we have
\[
\sum_{j\in[k]}\left[c_{j}\tilde{\delta}_{t}^{(j)}\right]_{i}^{2}\leq\left(\sum_{j\in[k]}|c_{j}|\right)\cdot\left(\sum_{j\in[k]}|c_{j}|\cdot\left[\tilde{\delta}_{t}^{(j)}\right]_{j}^{2}\right)
\]
we have that
\begin{align}
f''(t) & \leq\norm c_{1}\sum_{j\in[k]}|c_{j}|\normFull{\tilde{\delta}_{t}^{(j)}}_{\mz_{t}\phi''(z_{t})\mz_{t}+\mz_{t}|\phi'(z_{t})|}^{2}-\sum_{j\in[k]}c_{j}\normFull{\tilde{\delta}_{t}^{(j)}}_{\mz_{t}\phi'(z_{t})}^{2}\nonumber \\
 & \leq\norm c_{1}\sum_{j\in[k]}|c_{j}|\normFull{\tilde{\delta}_{t}^{(j)}}_{\mz_{t}\phi''(z_{t})\mz_{t}}^{2}+\left(\norm c_{1}+1\right)\sum_{j\in[k]}|c_{j}|\cdot\normFull{\tilde{\delta}_{t}^{(j)}}_{\mz_{t}|\phi'(z_{t})|}^{2}\nonumber \\
 & \leq(\norm c_{1}+1)\sum_{j\in[k]}|c_{j}|\cdot\norm{\tilde{\delta}_{t}^{(j)}}_{(\mz_{t}^{2}+\mz_{t})\phi''(z_{t})}^{2}\label{eq:second_deriv_bound}
\end{align}
where in the last line we used that $\phi''(z_{t})\geq\lambda|\phi'(z_{t,})|\geq|\phi'(z_{t})|$
by Lemma~\ref{lem:potential_simple_bounds} and $\lambda\geq1$.

Now, by construction and assumption we have that $v$, $v^{\new},$$u^{(j)}$,
and $u_{t}^{(j)}$ are all entrywise positive. Since $\norm{\mU_{j}^{-1}\delta^{(j)}}_{\infty}\leq\frac{1}{12(1+\lambda\norm c_{1})}$
for all $j\in[k]$ we have that for all $t\in[0,1]$ and $i\in[n]$
\begin{align*}
[u_{t}^{(j)}]_{i}\leq\left(1-\frac{1}{12(1+\lambda\norm c_{1})}\right)^{-|c_{j}|}[u_{t}^{(j)}]_{0}\leq\exp\left(\frac{|c_{j}|}{11(1+\lambda\norm c_{1})}\right)[u_{t}^{(j)}]_{0}
\end{align*}
where we used that $1-x\leq\exp(-x)$ for all $x$ yields and for
all $x\in(0,1)$ and $c\geq0$ it is the case that $(1-x)^{-c}\geq\max\left\{ (1+x)^{c},(1+x)^{-c},(1+x)^{-c}\right\} $.
Consequently, $\mU_{t}^{(j)}\preceq\exp(\frac{1}{11\lambda})\mU$
and for all $t\in[0,1]$ and $i\in[n]$ we have
\begin{align*}
[z_{t}]_{i}\leq\prod_{j\in[k]}[u_{t}^{(j)}]_{i}\leq\prod_{j\in[k]}\exp\left(\frac{|c_{j}|}{11(1+\lambda\|c\|_{1})}\right)[u^{(j)}]_{i}\leq\exp(\frac{1}{11\lambda})[z_{0}]_{i}\,.
\end{align*}
Thus, $\mz_{t}\preceq\exp(\frac{1}{11\lambda})\mv$ and by analogous
calculation $\mU_{t}^{(j)}\succeq\exp(-\frac{1}{11\lambda})\mU^{(j)}$
and $\mz_{t}\succeq\exp(-\frac{1}{11\lambda})\mv$. 

By the preceding paragraph, $u_{t}^{(j)}\approx_{1/(11\lambda)}u^{(j)}$and
$z_{t}\approx_{1/(11\lambda)}v$. Further, $v\leq\frac{13}{12}$ and
Lemma~\ref{lem:mult_to_add} implies
\begin{align*}
\norm{z_{t}-v}_{\infty}=\norm{(\mv^{-1}\mz_{t}-\vones)\mv}_{2}\leq\norm{\mv^{-1}\mz_{t}-\vones}_{2}\norm{\mv}_{2}\leq\left(\frac{1}{11\lambda}+\left(\frac{1}{11\lambda}\right)^{2}\right)\frac{13}{12}\leq\frac{0.11}{\lambda}\,.
\end{align*}
Consequently, by Lemma~\ref{lem:potential_simple_bounds} we have
that $\phi''(z_{t})\le\exp(0.11)\phi''(v)$. Combining these facts
with (\ref{eq:second_deriv_bound}) then yields that 
\begin{align*}
f''(t)
\leq
&~ (\norm c_{1}+1)\sum_{j\in[k]}|c_{j}|\cdot\norm{\mU^{-1}\delta^{(j)}}_{\phi''(v)}\\
&~ \cdot\left(\left[\exp\left(\frac{1}{11\lambda}\right)\frac{13}{12}\right]^{2}+\left[\exp\left(\frac{1}{11\lambda}\right)\frac{13}{12}\right]\right)\exp\left(\frac{1}{11\lambda}\right)\exp(0.11)
\end{align*}
and the result follows from $\lambda\geq1$ and $f(1)=f(0)+f'(0)+\frac{1}{2}f''(\zeta)$
for some $\zeta\in[0,1]$.
\end{proof}

\begin{lemma}[Lemma~54 of \cite{ls19}]
\label{lem:smoothing:helper}For all $x\in\R^{m}$, we have
\begin{equation}
e^{\lambda\|x\|_{\infty}}\leq\Phi(x)\leq2me^{\lambda\|x\|_{\infty}}\enspace\text{ and }\enspace\lambda\Phi(x)-2\lambda m\leq\|\grad\Phi(x)\|_{1}\label{eq:smoothing:pot_prop_1}
\end{equation}
Furthermore, for any symmetric convex set $U\subseteq\R^{m}$ and
any $x\in\R^{m}$, let $x^{\flat}\defeq\argmax_{y\in U}\left\langle x,y\right\rangle $
and $\norm x_{U}\defeq\max_{y\in U}\left\langle x,y\right\rangle $.
Then for all $x,y\in\R^{m}$ with $\|x-y\|_{\infty}\leq\delta\leq\frac{1}{5\lambda}$
we have
\begin{align}
e^{-\lambda\delta}\norm{\grad\Phi(y)}_{U}-\lambda\|\grad\Phi(y)^{\flat}\|_{1} & \leq\left\langle \grad\Phi(x),\grad\Phi(y)^{\flat}\right\rangle \nonumber \\
 & \leq e^{\lambda\delta}\norm{\grad\Phi(y)}_{U}+\lambda e^{\lambda\delta}\|\grad\Phi(y)^{\flat}\|_{1}.\label{eq:smoothing:pot_prop_2}
\end{align}
If additionally $U$ is contained in a $\ell_{\infty}$ ball of radius
$R$ then
\begin{equation}
e^{-\lambda\delta}\norm{\grad\Phi(y)}_{U}-\lambda mR\leq\norm{\grad\Phi(x)}_{U}\leq e^{\lambda\delta}\norm{\grad\Phi(y)}_{U}+\lambda e^{\lambda\delta}mR.\label{eq:smoothing:pot_prop_3}
\end{equation}
\end{lemma}

\centralityMainGeneral

\begin{proof}
	Since entrywise $|\phi''(w)| \leq \lambda\cdot(|\phi'(w)|+2\lambda)$ by Lemma~\ref{lem:potential_simple_bounds},
	we have that 
	\begin{align*}
	\norm{\phi''(w)}_{U}\leq\lambda\norm{\phi'(w)}_{U}+2\lambda^{2}\|1\|_{U}\leq\lambda\norm{\phi'(w)}_{U}+2\lambda^{2}m
	\end{align*}
	where in the first step we used triangle inequality and the axis-symmetry of $U$ and in the second step we used that $U$ is contained in a $\ell_{\infty}$ ball of
	radius $1$.
	
	Now, we bound the term $\phi'(w)^{\top}g$. Applying \eqref{eq:smoothing:pot_prop_2}
	and \eqref{eq:smoothing:pot_prop_3} in \Cref{lem:smoothing:helper},
	we have
	\begin{align*}
	\frac{-\phi'(w)^{\top}g}{\gamma} & \geq e^{-\lambda\delta}\norm{\phi'(v)}_{U}-\lambda\|g\|_{1} 
	\tag{\eqref{eq:smoothing:pot_prop_2} of \Cref{lem:smoothing:helper}}
	\\
	& \geq e^{-2\lambda\delta}\norm{\phi'(w)}_{U}-e^{-\lambda\delta}\lambda m -\lambda\|g\|_{1}
	\tag{\eqref{eq:smoothing:pot_prop_3}  \Cref{lem:smoothing:helper}}
	\\
	& \geq e^{-2\lambda\delta}\norm{\phi'(w)}_{U}-2\lambda m.
	\tag{$U$ is contained in a ball of radius $1$}
	\end{align*}
	Hence, as $e^{-2\lambda \delta} \geq 1 - 2 \lambda$, we have 
	\begin{align*}
	\delta_{\Phi}\leq & -(1-2\lambda\delta-c_{1}-c_{2}\lambda\gamma)\gamma\norm{\phi'(w)}_{U}+2\lambda m\gamma+2c_{2}\gamma^{2}\lambda^{2}m\\
	\leq & - 0.5 ( 1-c_{1} ) \gamma\norm{\phi'(w)}_{U}+ m/2
	\end{align*}
	where we used $\lambda \gamma \leq 1 / 8$ and $2\lambda\delta+c_{2}\lambda\gamma\leq\frac{1-c_{1}}{2}$.
	Since $U$ contains the $\ell_{\infty}$ ball of radius $u$, we have
	$\|\phi'(w)\|_{U}\geq u\|\phi'(w)\|_{1}$. Using this and (\ref{eq:smoothing:pot_prop_1})
	in Lemma \ref{lem:smoothing:helper} yields
	\begin{align*}
	\delta_{\Phi} & \leq- 0.5( 1-c_{1} ) \gamma u\norm{\phi'(w)}_{1}+0.5m\\
	& \leq- 0.5( 1 - c_{1} ) \lambda \gamma u \Phi(w) + \lambda \gamma mu+ 0.5m 
	\end{align*}
	and the result follows from the facts that $\lambda \gamma \leq 1/8$ and $u \leq 1$.
\end{proof}

\subsection{Additional Properties of our IPMs}
\label{sec:stability_fixing}

Here we provide additional properties of our IPMs. 
First, we show that although the $x$ iterates of Algorithm~\ref{alg:short_step_LS} 
may change more than in \cite{blss20} in a single step due to sampling, 
they do not drift too far away from a more stable sequence of points. 
We use this lemma to efficiently maintain approximations to regularized leverage scores, 
which are critical for this method. 

\begin{lemma}\label{lem:close_to_stable_sequence}
Suppose that $q_{i} \geq \frac{m}{n}\frac{T}{\beta^{2}}\log(mT)\cdot((\delta_{r})_{i}^{2}/\|\delta_{r}\|_{2}^{2}+1/m)$
for all $i\in[m]$ where $\beta\in(0,0.1)$ and $T\geq\sqrt{n}$.
Let $(x^{(k)},s^{(k)})$ be the sequence of points found by Algorithm~\ref{alg:short_step_LS}.
With probability $1-m^{-10}$, there is a sequence of points $\widehat{x}^{(k)}$
from $1 \leq k \leq T$ such that
\begin{enumerate}
\item $\widehat{x}^{(k)}\approx_{\beta}x^{(k)}$,
\item $\|(\widehat{\mx}^{(k)})^{-1}(\widehat{x}^{(k+1)}-\widehat{x}^{(k)})\|_{\tau(\widehat{x}^{(k)},s^{(k)})}\leq5\gamma$,
\item $\|\mT(\widehat{x}^{(k)},s^{(k)})^{-1}(\tau(\widehat{x}^{(k+1)},s^{(k+1)})-\tau(\widehat{x}^{(k)},s^{(k)}))\|_{\tau(\widehat{x}^{(k)},s^{(k)})} \leq 80 \gamma$.
\end{enumerate}
\end{lemma}

\begin{proof}
We define the sequence $\widehat{x}^{(k)}$ recursively by $\widehat{x}^{(1)}=x^{(1)}$
and
\begin{align*}
\ln\widehat{x}^{(k)} \defeq \ln \widehat{x}^{(k-1)} + \E_{k}[ \ln x^{(k)} ] - \ln x^{(k-1)}
\end{align*}
where $\E_{k} [ \ln x^{(k)} ] \defeq \E_{\mr_{k}}[ \ln x^{(k)}|x^{(k-1)} ]$
where $\mr_{k}$ is the random choice of $\mr$ at the $k$-th iteration.
Let $\Delta^{(k)}\defeq\ln\widehat{x}^{(k)}-\ln x^{(k)}$ and note
that this implies that
\begin{equation}
\Delta^{(k)}=\Delta^{(k-1)}+\E_{k}[\ln x^{(k)}]-\ln x^{(k)}\,.\label{eq:delta_recursion}
\end{equation}
Consequently, is a martingale, i.e. $\E_{k} [ \Delta^{(k)} ] = \Delta^{(k-1)}.$
Since $\widehat{x}^{(1)} = x^{(1)}$, we have $\Delta^{(1)}=0$ and
consequently to bound $\Delta^{(k)}$, it suffices to show that $\Delta^{(k)}-\Delta^{(k-1)}$
is small in the worst case and that the variance of $\Delta^{(k)}-\Delta^{(k-1)}$
is small.

For the worst case bound, note that by (\ref{eq:delta_recursion})
and Jensen's inequality
\begin{equation}
\left|\Delta_{i}^{(k)}-\Delta_{i}^{(k-1)}\right|=\left|\ln x_{i}^{(k)}-\E_{k}\ln x{}_{i}^{(k)}\right|\leq\sup_{x^{(k)},y^{(k)}}\left|\ln x_{i}^{(k)}-\ln y_{i}^{(k)}\right|\label{eq:delta_inf_bound}
\end{equation}
where $x_{i}^{(k)}$ and $y_{i}^{(k)}$ are two independent samples
using two different independent $\mr_{k}$ to take a step from $x^{(k-1)}$.
Lemma~\ref{lem:ls_second} shows that $x^{(k)}\approx_{0.1}x^{(k-1)}$
and $y^{(k)}\approx_{0.1}x^{(k-1)}$. Consequently, $x^{(k)}\approx_{0.2}y^{(k)}$.
Therefore, we have $\left|\Delta_{i}^{(k)}-\Delta_{i}^{(k-1)}\right|\leq\frac{3}{2}\frac{|x_{i}^{(k)}-y_{i}^{(k)}|}{x_{i}^{(k-1)}}$.
Using that $x^{(k)}-y^{(k)}=\omx\mr_{x,k}\delta_{r}-\omx\mr_{y,k}\delta_{r}$
(where $\mr_{x,k}$ and $\mr_{y,k}$ are two different independent
draws of $\mr$ in iteration $k$) and $\ox\approx_{0.1}x$, we can
simplify (\ref{eq:delta_inf_bound}) to
\begin{align*}
\left|\Delta_{i}^{(k)}-\Delta_{i}^{(k-1)}\right|\leq2|(\mr\delta_{r})_{i}-(\mr'\delta_{r})_{i}|.
\end{align*}
Recall that the $i$-th row of $\delta_{r}$ is sampled with probability
$\tilde{q}_{i}\defeq\min(q_{i},1)$. If the $i$-th row sampled with
probability $1$, then the difference is $0$. Hence, we can only
consider the case $\tilde{q}_{i}=q_{i}\geq\frac{m}{n}\frac{T}{\beta^{2}}\log(mT)\cdot((\delta_{r})_{i}^{2} / \|\delta_{r}\|_{2}^{2}+1/m)$.
In this case, we have
\begin{align*}
\left|\Delta_{i}^{(k)}-\Delta_{i}^{(k-1)}\right|\leq\frac{2}{q_{i}}|\delta_{r}|_{i}\leq\frac{2}{q_{i}}|\delta_{r}|_{i}=\frac{2|\delta_{r}|_{i}}{\frac{m}{n}\frac{T}{\beta^{2}}\log(mT)\cdot((\delta_{r})_{i}^{2}/\|\delta_{r}\|_{2}^{2}+1/m)}.
\end{align*}
One can check the worst case is given by $|\delta_{r}|_{i}=\|\delta_{r}\|_{2}/\sqrt{m}$
and we get
\begin{align*}
\left|\Delta_{i}^{(k)}-\Delta_{i}^{(k-1)}\right| & \leq\frac{2}{q_{i}}|\delta_{r}|_{i}\leq\frac{2}{q_{i}}|\delta_{r}|_{i}=\frac{\beta^{2}n}{\sqrt{m}T\log(mT)}\|\delta_{r}\|_{2}\\
 & \leq M\defeq\frac{\sqrt{n}\gamma\beta^{2}}{40T\log(mT)}.
\end{align*}
where we used $\tau(x^{(k)},s^{(k)})\geq\frac{n}{m}$ and $\|\delta_{r}\|_{\tau(x^{(k)},s^{(k)})}\leq\frac{\gamma}{40}$
(\ref{eq:delta_r_upper_tau}) at the last sentence.

For the variance bound, note that variance of a random variable $X$
can be written as $\E(X-\E X)^{2}=\frac{1}{2}\E(X-Y)^{2}$ where $Y$
is an independent sample from $\mx$. Hence, we have 
\begin{align*}
\E_{k}\left[[\Delta^{(k)}-\Delta^{(k-1)}]_{i}^{2}\right] & =\E_{k}\left[[\ln x^{(k)}-\E_{k}\ln x^{(k)}]_{i}^{2}\right]=\frac{1}{2}\E_{k}\left[[\ln x^{(k)}-\ln y^{(k)}]_{i}^{2}\right]\\
 & \leq\E_{k}\left[\frac{(x_{i}^{(k)}-y_{i}^{(k)})^{2}}{(x_{i}^{(k-1)})^{2}}\right]\leq2\E_{k}\left[(\mr_{x,k}\delta_{r}-\mr_{y,k}\delta_{r})_{i}^{2}\right]
\end{align*}
where we used that $x^{(k)}\approx_{0.2}y^{(k)}$, $x^{(k)}-y^{(k)}=\omx\mr_{x,k}\delta_{r}-\omx\mr_{y,k}\delta_{r}$
and $\ox\approx_{0.1}x$ as before at the last line. Again, if the
$i$-th row sampled with probability $1$, then the difference is
$0$. Hence, the last line is bounded by
\begin{align*}
\E_{k}\left[[\Delta^{(k)}-\Delta^{(k-1)}]_{i}^{2}\right]\leq\frac{2}{q_{i}}[\delta_{r}]_{i}^{2}\leq\frac{2n\|\delta_{r}\|_{2}^{2}}{mT\log(mT)}\leq\sigma^{2}\defeq\frac{\gamma^{2}\beta^{2}}{100T\log(mT)}
\end{align*}
where we used $q_{i}\geq\frac{m}{n}\frac{T}{\beta^{2}}\log(mT)\cdot(\delta_{r})_{i}^{2}/\|\delta_{r}\|_{2}^{2}$,
$\tau(x^{(k)},s^{(k)})\geq\frac{n}{m}$ and $\|\delta_{r}\|_{\tau(x^{(k)},s^{(k)})}\leq\frac{\gamma}{40}$
(\ref{eq:delta_r_upper_tau}) at the end.

Now, we apply Bernstein inequality and get that
\begin{align*}
\P(|\Delta_{i}^{(k)}|\leq\beta\text{ for any }k\in[T]) & \leq2T\exp(-\frac{\beta^{2}/2}{\sigma^{2}T+\frac{1}{3}M\beta})\\
 & =2T\exp(-\frac{\frac{1}{2}\beta^{2}}{\frac{\gamma^{2}\beta^{2}}{100T\log(mT)}T+\frac{1}{3}\cdot\frac{\sqrt{n}\gamma\beta^{2}}{40T\log(mT)}\cdot\beta})\\
 & \leq2T\exp(-25\log(mT))\leq\frac{1}{m^{10}}.
\end{align*}
This completes the proof of $\widehat{x}^{(k)}\approx_{\beta}x^{(k)}$
for all $k\leq T$.

For the Property 2, we note that $\widehat{x}^{(k+1)}\approx_{0.1}\widehat{x}^{(k)}\approx_{0.1}x^{(k)}$
and hence
\begin{align*}
\|(\widehat{\mx}^{(k)})^{-1}(\widehat{x}^{(k+1)}-\widehat{x}^{(k)})\|_{\tau(\widehat{x}^{(k)},s^{(k)})} & \leq2\|\ln\widehat{x}^{(k+1)}-\ln\widehat{x}^{(k)}\|_{\tau(x^{(k)},s^{(k)})}\\
 & =2\|\E\ln x^{(k+1)}-\ln x^{(k)}\|_{\tau(x^{(k)},s^{(k)})}.
\end{align*}
Using that $x^{(k+1)}\approx_{0.1}x^{(k)}$, we have 
\begin{align*}
\ln x_{i}^{(k+1)}=\ln x_{i}^{(k)}+((\mx^{(k)})^{-1}(x^{(k+1)}-x^{(k)}))_{i}\pm((\mx^{(k)})^{-1}(x^{(k+1)}-x^{(k)}))_{i}^{2}
\end{align*}
for all $i$ and hence
\begin{align*}
&~\|(\widehat{\mx}^{(k)})^{-1}(\widehat{x}^{(k+1)}-\widehat{x}^{(k)})\|_{\tau(\widehat{x}^{(k)},s^{(k)})} \\
\leq&~
2\|(\mx^{(k)})^{-1}(\E x^{(k+1)}-x^{(k)})\|_{\tau(x^{(k)},s^{(k)})}+2\|\E((\mx^{(k)})^{-1}(x^{(k+1)}-x^{(k)}))^{2}\|_{\tau(x^{(k)},s^{(k)})}.
\end{align*}
Lemma \ref{lem:projection_sizes_ls} shows that the first term is
bounded by $2\gamma$ and Lemma \ref{lem:ls_second} shows that the
second term bounded by $12\gamma^{2}$. Hence, the whole term is bounded
by $5\gamma$.

For the Property 3, we let $\delta_{\tau}=\tau(\widehat{x}^{(k+1)},s^{(k+1)})-\tau(\widehat{x}^{(k)},s^{(k)})$.
By Lemma \ref{lem:expand_sigma}, we have
\[
\delta_{\tau}=\mLambda(1-2\alpha)\widehat{\mx}^{-1}\delta_{\hat{x}}+6\eta
\]
where $\mLambda=\mLambda(\widehat{\mx}^{1/2-\alpha}\ms^{-1/2-\alpha}\ma)$
and 
\begin{align*}
|\eta_{i}|\leq & \sigma_{i}(e_{i}^{\top}\mSigma^{-1}\mproj^{(2)}\widehat{\mx}^{-1}\delta_{\hat{x}})^{2}+e_{i}^{\top}(\mSigma+\mproj^{(2)})(\widehat{\mx}^{-1}\delta_{\hat{x}})^{2}.
\end{align*}
Using $\|\mT^{-1}\mproj^{(2)}\|_{\tau}\leq1$, $\|\mT^{-1}\mLambda\|_{\tau}\leq2$
and $\|\mT^{-1}\mproj^{(2)}\|_{\infty}\leq1$ (Lemma \ref{lem:P2bound}),
we have
\begin{align*}
\|\mT^{-1}\delta_{\tau}\|_{\tau} & \leq\|\mT^{-1}\mLambda\widehat{\mx}^{-1}\delta_{\hat{x}}\|_{\tau}+6\|\mT^{-1}\mSigma^{-1}(\mproj^{(2)}\widehat{\mx}^{-1}\delta_{\hat{x}})^{2}\|_{\tau}+6\|\mT^{-1}(\mSigma+\mproj^{(2)})(\widehat{\mx}^{-1}\delta_{\hat{x}})^{2}\|_{\tau}\\
 & \leq2\|\widehat{\mx}^{-1}\delta_{\hat{x}}\|_{\tau}+6\|(\mT^{-1}\mproj^{(2)}\widehat{\mx}^{-1}\delta_{\hat{x}})^{2}\|_{\tau}+12\|(\widehat{\mx}^{-1}\delta_{\hat{x}})^{2}\|_{\tau}\\
 & \leq14\|\widehat{\mx}^{-1}\delta_{\hat{x}}\|_{\tau}+6\|\mT^{-1}\mproj^{(2)}\widehat{\mx}^{-1}\delta_{\hat{x}}\|_{\infty}\cdot\|\mT^{-1}\mproj^{(2)}\widehat{\mx}^{-1}\delta_{\hat{x}}\|_{\tau}\\
 & \leq14\|\widehat{\mx}^{-1}\delta_{\hat{x}}\|_{\tau}+6\|\widehat{\mx}^{-1}\delta_{\hat{x}}\|_{\infty}\cdot\|\widehat{\mx}^{-1}\delta_{\hat{x}}\|_{\tau}\\
 & \leq14\cdot5\gamma+6\cdot\frac{1}{3}\cdot5\gamma
\end{align*}
where we used $\|\widehat{\mx}^{-1}\delta_{\hat{x}}\|_{\tau}\leq5\gamma$
(property 2) and $\|\widehat{\mx}^{-1}\delta_{\hat{x}}\|_{\infty}\leq\frac{1}{3}$
 (since $x^{(k)}\approx_{0.1}x^{(k-1)}$).
\end{proof}

Our IPM maintains a point $x$ that is nearly feasible, 
that is a point with small norm $\|\mA^\top x - b\|_{(\mA^\top \mX \mS^{-1} \mA)^{-1}}^2$.
We now show that there is a truly feasible $x'$ with $x'_i \approx x_i$ for all $i$.

\begin{lemma}\label{lem:feasible_projection}
Let $(x,s)$ be a primal dual pair with $xs \approx_1 \tau(x,s) \cdot \mu$ for $\mu>0$, 
where $x$ is not necessarily feasible.
Let $\mH^{-1} \approx_\delta (\mA^\top \mX\mS^{-1} \mA)^{-1}$ for some $\delta \in [0,1/2]$.
Then the point 
$$
x' := x + \mX \mS^{-1} \mA \mH^{-1} (b - \mA^\top x)
$$
satisfies both
\begin{align}
\|\mA^\top x' - b\|_{(\mA^\top \mX \mS^{-1} \mA)^{-1}}^2
&\le
5\delta \cdot \|\mA^\top x - b\|_{(\mA^\top \mX \mS^{-1} \mA)^{-1}}^2
\text{, and} \label{eq:infeasibility_bound} \\
\|\mX^{-1} (x-x')\|_\infty^2
&\le 
\frac{5}{\mu}\| \mA x - b \|_{(\mA^\top \mX \mS^{-1} \mA)^{-1}}^2.
\label{eq:similarity_x_bound}
\end{align}
\end{lemma}

\begin{proof}
Let $\delta_x = \mX \mS^{-1} \mA (b - \mA^\top x)$.
We start by proving \eqref{eq:infeasibility_bound}.
\begin{align*}
\|\mA^\top x' - b\|_{(\mA^\top \mX \mS^{-1} \mA)^{-1}}^2
&=
\|\mA^\top (x + \delta_x) - b\|_{(\mA^\top \mX \mS^{-1} \mA)^{-1}}^2 \\
&=
\|(\mI - \mA^\top \mX \mS^{-1} \mA \mH^{-1})(\mA^\top x - b)\|_{(\mA^\top \mX \mS^{-1} \mA)^{-1}}^2
\end{align*}
By using $\mH^{-1} \approx_\delta (\mA\mX\mS^{-1}\mA)^{-1}$ we can bound
\begin{align*}
&~
(\mI - \mH^{-1}\mA^\top \mX \mS^{-1} \mA)(\mA^\top\mX\mS^{-1}\mA)^{-1}(\mI-\mA^{\top} \mX \mS^{-1}\mA\mH^{-1}) \\
=&~
(\mA^\top \mX \mS^{-1} \mA)^{-1} - 2\mH^{-1} + \mH^{-1} \mA^\top \mX \mS^{-1} \mA \mH^{-1} \\
\prec&~
(1-2\exp(-\delta) + \exp(2\delta)) \cdot (\mA^\top \mX \mS^{-1} \mA)^{-1} \\
\prec&~
5\delta (\mA^\top \mX \mS^{-1} \mA)^{-1}.
\end{align*}
This then implies
$$
\|\mA^\top x' - b\|_{(\mA^\top \mX \mS^{-1} \mA)^{-1}}^2
\le
5\delta \|\mA^\top x - b\|_{(\mA^\top \mX \mS^{-1} \mA)^{-1}}^2,
$$
so we have proven \eqref{eq:infeasibility_bound}.
Next, we prove \eqref{eq:similarity_x_bound}.
\begin{align*}
\|\mX^{-1} \delta_x \|_\infty 
&=
\max_{i \in [m]} \| e_i^\top \mS^{-1} \mA \mH^{-1} (b - \mA^\top x)\|_\infty \\
&\le
\max_{i \in [m]} \| e_i^\top \mS^{-1} \mA \mH^{-1} (b - \mA^\top x)\|_2 \\
&\le
\max_{i \in [m]} \| e_i^\top \mS^{-1} \mA \mH^{-1/2}\|_2 \|\mH^{-1/2} (b - \mA^\top x)\|_2 \\
&\le
\exp(2\delta) \cdot \max_{i \in [m]} \| e_i^\top \mS^{-1} \mA (\mA^\top \mX \mS^{-1} \mA)^{-1/2}\|_2 \|b - \mA^\top x\|_{(\mA^\top \mX \mS^{-1} \mA)^{-1}},
\end{align*}
where the third step uses Cauchy-Schwarz and the last step the approximation guarantee of $\mH^{-1}$.
The first norm of above term can be bounded by
\begin{align*}
\| e_i^\top \mS^{-1} \mA (\mA^\top \mX \mS^{-1} \mA)^{-1/2}\|_2^2 
&=
	e_i^\top \frac{1}{(\mX\mS)^{1/2}}\frac{\mX^{1/2}}{\mS^{1/2}} \mA 
	(\mA^\top \mX \mS^{-1} \mA)^{-1} 
	\mA^\top \frac{\mX^{1/2}}{\mS^{1/2}} \frac{1}{(\mX\mS)^{1/2}} e_i \\
&=
\frac{\sigma(\frac{\mX^{1/2}}{\mS^{1/2}})_i}{x_i s_i} 
\le
\frac{\sigma(\frac{\mX^{1/2-\alpha}}{\mS^{1/2+\alpha}})_i}{x_i s_i} 
\le
\frac{\tau(x,s)}{x_i s_i}
\le
\frac{3}{\mu},
\end{align*}
where the last step uses $xs \approx_1 \mu\cdot\tau(x,s)$
and $\sigma(x,s) \le \tau(x,s)$ for both possible choices of $\tau$
(i.e. the IPM of \Cref{sec:log_barrier_method} uses $\tau(x,s) = 1$ 
and \Cref{sec:ls_barrier_method} uses $\tau(x,s) = \sigma(x,s) + n/m$).
By using $\delta \le 1/2$ we then obtain
$$\|\mX^{-1} \delta_x\|_\infty^2 \le \frac{5}{\mu} \cdot \| \mA x - b \|_{(\mA^\top \mX \mS^{-1} \mA)^{-1}}^2.$$
\end{proof}

In the remaining of this section, we bound how much the iterates of our IPMs move from the initial point. 
Our proof relies on bounds on how fast the center
of a weighted self-concordant barrier moves. 
In particular we use the following lemma from \cite{ls19}
 specialized to the weighted log barrier.\footnote{This following \Cref{lem:stability_lemma} follows as a special case of Lemma~67 of  \cite{ls19} by choosing the $\phi_i(x)$ in that lemma to be $\phi_i \defeq  - \ln ([\ma_2^\top x - b_2]_i)$ noting that each $\phi_i$ is a $1$-self-concordant barrier for the set $\{x \in \R^{n_2} | [\ma_2^\top x]_i \geq [b_2]_i \}$ . }

\begin{lemma}[Special Case of \cite{ls19} Lemma~67]
	\label{lem:stability_lemma}
	\label{lem:slacks_bound} For arbitrary $\ma_{1}\in\R^{n \times m}$,
	$b_{1}\in\R^{m}$, $\ma_{2}\in\R^{n \times k}$, $b_{2}\in\R^{k}$,
	$b_3 \in \R^{n}$ and all $w \in\R_{>0}^{k}$, let
	$\Omega=\{x \in \R^{n} :\ma_{1}^{\top}x=b_{1},\ma_{2}^{\top}x > b_{2}\}$ let
	\[
	x_{w}\defeq\arg\min_{x\in\Omega} b_3^{\top}x-\sum_{i\in[k]}w_{i}\log([\ma_{2}^{\top}x-b_{2} ]_{i}).
	\]
	For all $w^{(0)},w^{(1)}\in\mathbb{R}_{>0}^{k}$, it holds that
	$x_{w^{(0)}}+t(x_{w^{(1)}}-x_{w^{(0)}})\in\Omega$ for all $t \in (-\gamma, 1 + \gamma)$ where  
	\[
	\gamma = \frac{\theta^2}{1 + 2\theta}
	\enspace \text{ and } \enspace
	\theta = 
	\frac{ \min_{i \in [k]}  \min\{ w_i^{(0)} , w_i^{(1)}  \}}
	{\sum_{i \in [k]}  |w_i^{(0)} - w_i^{(1)}  |}
	~.
	\]
\end{lemma}

Using the lemma above, we can bound how the fast primal-dual central
path moves by applying the lemma on both primal and dual separately.

\begin{lemma}\label{lem:size_bound_tool}
	Fix arbitrary non-degenerate $\ma\in\R^{m\times n}$, $b\in\R^{n}$,  and $c\in\R^{m}$. For all
	$w\in\R_{>0}^{m}$, we define $(x_{w},s_{w})\in\R_{>0}^{2m}$
	be the unique vector such that $\mx_{w}s_{w}=w$, $\ma^{\top}x_{w}=b$,
	$\ma y+s_{w}=c$ for some $y$. For any $w^{(0)},w^{(1)}\in\R_{>0}^{m}$,
	we have that entrywise 
	\[
	\frac{1}{\kappa}x_{w^{(1)}}\leq x_{w^{(0)}}\leq\kappa x_{w^{(1)}}
	\enspace \text{ and } \enspace
	\frac{1}{\kappa}s_{w^{(1)}}\leq s_{w^{(0)},}\leq\kappa s_{w^{(1)}}
	\]
	where $\kappa = 4 (W_{1} / W_{\min} )^2$ for 
	$W_{1} \defeq \|w^{(0)}\|_1 + \|w^{(1)}\|_1$ and $W_{\min} \defeq \min_{i \in [k]}  \min\{ w_i^{(0)} , w_i^{(1)} \}$.
\end{lemma}

\begin{proof}
	Apply Lemma~\ref{lem:slacks_bound} with $\ma_{1}=\ma$, $b_{1}=b$,
	$\ma_{2}=\mi$, $b_{2}=0$ and $b_{3}=c$ so that $x_w$ is as defined in \Cref{lem:slacks_bound}.  Lemma~\ref{lem:slacks_bound} implies that
	$x_{w^{(0)}}+ t (x_{w^{(1)}}-x_{w^{(0)}}) > 0 $ for all $t \in (-\gamma, 1 + \gamma)$ 
	with $\gamma$ and $\theta$ defined as in \Cref{lem:stability_lemma}. Consequently, $\frac{1 + \gamma}{\gamma} x_{w^{(1)}} \geq x_{w^(0)} \geq \frac{\gamma}{1 + \gamma} x_{w^{(1)}}$. Now, $\theta \geq \frac{W_{\min}}{W_{1}}$ as $|w_i^{(0)} - w_i^{(1)}| \leq |w_i^{(0)}| + |w_i^{(1)}|$. Since $\theta^2 / (1 + 2\theta)$
	increases monotonically for positive $\theta$ and $\frac{W_{\min}}{W_{1}} \leq 1$ we have $\gamma \geq \frac{1}{3} (\frac{W_{\min}}{ W_{1}} )^2$. Similarly, since $\gamma / (1 + \gamma)$ increases monotonically for positive $\gamma$ and $\frac{1}{3} (\frac{W_{\min}}{ W_{1}} )^2 \leq \frac{1}{3}$ we have $\frac{\gamma}{1 + \gamma} \geq \frac{1}{4} (\frac{W_{\min}}{W_1}) \geq   \frac{1}{\kappa}$. Thus,  $\frac{1}{\kappa}x_{w^{(0)}}\leq x_{w^{(1)}}\leq\kappa x_{w^{(0)}}$.
	
	Similarly, to bound $s_{w^{(0)}}$ and $s_{w^{(1)}}$ we apply  Lemma~\ref{lem:slacks_bound} with $\ma_{1}=0$,
	$b_{1}=0$, $\ma_{2}=\ma^{\top}$, $b_{2}=c$ and $b_{3}=b$ so that $s_w$ is as defined in \Cref{lem:slacks_bound}.  Lemma~\ref{lem:slacks_bound} implies that
	$s_{w^{(0)}}+ t (s_{w^{(1)}}-s_{w^{(0)}}) > 0 $ for all $t \in (-\gamma, 1 + \gamma)$. Since this is the same statement that was shown for $x_{w^{(0)}}$ and $x_{w^{(1)}}$ in the preceding paragraph we also have  $\frac{1}{\kappa}s_{w^{(0)}}\leq s_{w^{(1)}}\leq\kappa s_{w^{(0)}}$.
\end{proof}

This allows us to bound the number of bits of precision required to appropriately represent all $x$, $s$ 
as they occur throughout \Cref{alg:meta}, our main path following routine.

\begin{lemma}\label{lem:bitlength}
Let $x^\init,s^\init,\mu^\init,\mu^\target$ be the initial parameters 
for \Cref{alg:meta} as in \Cref{lem:path_following_master}.
Let $W$ be a bound on the ratio of largest to smallest entry in both $x^\init$ and $s^\init$.
Let $W'$ be a bound on the ratio of largest to smallest entry for all $x$ and $s$ at each \Cref{line:for_loop_start} and the algorithm's termination. Then $\log W' = \tilde{O}(\log W + |\log \mu^\init/\mu^\target|)$ for $\tau = \tauLS$ and $\tau = \taulog$.
\end{lemma}

\begin{proof}
Let $(x,s,\mu)$ be the state of the algorithm at either \Cref{line:for_loop_start} or the algorithm's termination. \Cref{lem:path_following_master} implies that $(x,s,\mu)$ is $\epsilon$-centered. Consequently, 
$\frac{1}{\sqrt{\mu}} \| \mA^\top x - b \|_{(\mA \mX \mS^{-1} \mA)^{-1}} = O(1)$.
By \Cref{lem:feasible_projection} this implies there is a feasible $x'$
with $x \approx_{O(1)} x'$ where for $w \defeq x' s$ and $x_w, s_w$ as defined in \Cref{lem:size_bound_tool}
we have $x_w = x'$ and $s_w = w$. Further, the definition of $\epsilon$-centered implies that  $w = x's \approx_{O(1)} xs \approx_{O(1)} \mu\tau(x,s)$. Now, for $\tau(x,s) = \taulog(x,s) = \vones$ we have
\begin{align*}
\mu\cdot\tau(x,s) \geq \min\{\mu^\init,\mu^\target\} 
\text{ and }
\norm{\mu  \cdot \tau(x,s)}_1 \leq m \cdot \max\{\mu^\init,\mu^\target\}.
\end{align*}
and for $\tau(x,s) = \tauLS(x,s) = \sigma(x,s)+\frac{n}{m} \vones$ we have
\begin{align*}
\mu \cdot \tau(x,s) \geq \frac{n}{m} \cdot \min\{\mu^\init,\mu^\target\} 
\text{ and }
\norm{\mu \cdot \tau(x,s)}_1 \leq 2n \cdot \max\{\mu^\init,\mu^\target\}.
\end{align*}
In either case, applying \Cref{lem:size_bound_tool} with 
\[
\kappa = O\left(\left(m \cdot  \frac{\max\{\mu^\init,\mu^\target\}}{\min\{\mu^\init,\mu^\target\}}\right)^2 \right)
= O\left( \left(m \cdot \max\left\{\frac{\mu^\init}{\mu^\target} , \frac{\mu^\target}{\mu^\init}\right\} \right)^2 \right)
\]
yields that
\[
\frac{1}{\kappa} x_{w^\init}
\leq  x_w \leq  \kappa x_{w^\init}
\text{ and }
\frac{1}{\kappa} s_{w^\init}
\leq  s_w \leq  \kappa s_{w^\init} ~.
\]
Since the largest ratio of the entries in $x^\init$ and $s^\init$ is bounded by $W$ this implies that 
$W'$ with $\log W' = \tilde{O}(\log W + |\log \mu^\init/\mu^\target|)$ is an upper bound on the largest ratio of entries in $x$ and $s$ as desired.
\end{proof}

\newpage

\section{Heavy Hitters}
\label{sec:matrix_vector_product}

This section concerns maintaining information 
about a matrix-vector product $\mA h$ 
for an incidence matrix $\mA\in \{-1,0,1\}^{m\times n}$ undergoing {\em row scaling}.
Formally, let $g\in \R^m$. 
Note that $\mdiag(g)\mA$ is the matrix obtained by 
scaling the $i^{th}$ row of $\ma$ by a factor of $g_i$. 
We want to be able to update entries of $g$ 
and compute certain information about $\mdiag(g)\mA h$ 
for query vector $h\in \R^n$. 
The data structure constructed in this section 
will be used in \Cref{sec:vector_maintenance}
to create a data structure for efficiently 
maintaining the slack of the dual solution inside our IPM method.

If $\mA$ is an incidence matrix representing some graph $G$, 
then $\mdiag(g)\mA h\in \R^m$ can be viewed as a vector of {\em values} of edges, 
defined as follows. 
View $h\in \R^n$ as a vector of {\em potentials} on nodes 
and view $g\in \R_{\ge0}^m$ as a vector of edge {\em weights} 
(node $v\in [n]$ has potential $h_v$ and edge $e\in[m]$ has weight $g_e$). 
Define the {\em value} of each edge $e=\{u,v\}$ 
as $|(\mdiag(g)\mA h)_e| = g_e |h(v)-h(u)|$. 
We would like to maintain certain information about the edge values 
when the potential and edge weights change.

The data structure of this section (\Cref{lem:large_entry_datastructure}) 
can report edges with large absolute values, i.e. large $|(\mdiag(g)\mA h)_e|$. 
It supports updating $g_e$ for some $e\in[m]$ ({\sc Scale} operation). 
Moreover, given $h\in \R^n$ and $\epsilon\geq 0$, 
it reports all edges (equivalently all entries in $\mdiag(g)\mA h$) 
with values at least $\epsilon$ ({\sc HeavyQuery} operation). 
Additionally, it can sample each edge with probability 
roughly proportional to its value {\em square} ({\sc Sample} operation). 
Thus, each edge $e=\{u,v\}$ is selected with probability roughly 
$\frac{(g_e (h_u - h_v))^2}{\|\mdiag(g) \mA h\|_2^2}.$ 

Our data structure requires near-linear preprocessing time 
and polylogarithmic time per {\sc Scale} and {\sc Sample} operation. 
It takes roughly $O(\|\mdiag(g) \mA h\|_2^2 \epsilon^{-2}+n)$ time to answer a query. 
(To make sense of this time complexity, 
note that reading the input vector $h$ takes $O(n)$ time,
and there can be at most and as many as 
$\|\mdiag(g) \mA h\|_2^2 \epsilon^{-2}$ edges reported: 
For any $x_1, x_2, \ldots, x_m$, there can be at most and as many as 
$(\sum_i x_i^2)/\epsilon^2$ indices $i$ such that $x_i\geq\epsilon$.) 
Note that the data structure is randomized but holds against an adaptive adversary. 
It can be made deterministic at the cost of $m^{1+o(1)}$ preprocessing time 
and $n^{o(1)}$ time per \textsc{Scale} (as opposed to $\polylog(n)$),
using \cite{cglnps20}.

We discuss our data structure to query all large entries (in terms of absolute value) of vector $\mdiag(g) \mA h$ and to sample entries of the vector proportional to the squared values. Formally, our data structure has the following guarantees.
\begin{lemma}[\ProjHeavyHitter]\label{lem:large_entry_datastructure}
	There exists a data structure \textsc{ProjectionHeavyHitter} %
	that supports the following operations.
	\begin{itemize}
		\item \textsc{Initialize($\mA \in \{-1,0,1\}^{m\times n}, g \in \R^m_{\geq 0}$)}: 
		The data structure is given the edge incidence matrix $\mA$ of a graph $G$ and a scaling vector $g$.
		It initializes in $O(m \log^9 n)$ time. 
		\item \textsc{Scale$(i \in [m], s \ge 0)$}: 
		Updates $g_i \leftarrow s$ in $O(\log^{14} n)$ amortized time.
		\item \textsc{HeavyQuery($h\in \R^n$, $\epsilon \in \R_{>0}$)}:
		With high probability, the data structure returns all $e\in [m]$ such that %
		$\left|\left(\mdiag(g) \mA h\right)_e\right| \geq \epsilon$
		in running time
		\begin{align*}
		O(\|\mdiag(g) \mA h\|_2^2 \epsilon^{-2} \log^8 n + n \log^2 n\log W),
		\end{align*}
		where 
		$W$ is the ratio of the largest to the smallest non-zero entries in $g$.
		\item \textsc{Sample}$(h \in \R^n, K \in \R_{>0})$:
		With high probability, in $O(K \log n + n \log^2 n \log W)$ time 
		the data-structure returns independently sampled indices of $\mdiag(g) \mA h\in \R^m$ (i.e., edges in graph $G$), where each edge $e=(u,v)$ is sampled with some probability $q_e$ which is at least
		\begin{align*}
		\min\left\{K \cdot \frac{(g_e (h_u - h_v))^2}{16 \|\mdiag(g) \mA h\|_2^2 \log^{8} n},1\right\},
		\end{align*}
		and with high probability there are at most $O(K \log n)$ entries returned.
		\item \textsc{Probability$(I\subset [m], h\in \R^n, K\in \R_{>0})$}: 
		Given a subset of edges $I$, this procedure returns for every $e\in I$ the probabilities $q_e$
		that $e$ would be sampled in the procedure \textsc{Sample$(h,K)$} . 
		The running time is $O(|I|+n\log^2n\log W)$.
		\item \textsc{LeverageScoreSample}$(K' \in \R_{>0})$:
		With high probability, in $O(K' n \log^{11} n \log W)$ time
		\textsc{LeverageScoreSample} returns a set of sampled edges 
		where every edge $e$ is included independently 
		with probability $p_e \ge K' \cdot \sigma(\mdiag(g)\mA)_e$, 
		and there are at most $O(K' n \log^{11} n \log W)$ entries returned.
		
		\item \textsc{LeverageScoreBound}$(K', I \subset [m])$:
		Given a subset of edges $I$, this procedure returns for every $e\in I$ the probabilities $p_e$ 
		that $e$ would be sampled in the procedure \textsc{LeverageScoreSample$(K')$}. 
		The running time is $O(|I|)$. %
	\end{itemize}
\end{lemma}
Our data structure exploits a classic spectral property of a graph, which is captured by the following simple known variant of Cheeger's inequality  (see, e.g.~\cite{ChuGPSSW18} for its generalization).

\begin{lemma}\label{lem:cheeger_based}
	Let $\mA \in \R^{m \times n}$ be an incidence matrix of an unweighted $\phi$-expander 
	and let $\mD \in \R^{n \times n}$ the degree matrix.
	Let $\mL = \mA^\top \mA$ be the corresponding Laplacian matrix, then for any $y\in \R^n$ such that $y \bot \mD\onevec_n$, we have
	\begin{align*}
	y^\top \mL y \ge 0.5 \cdot y^\top \mD y \cdot \phi^2.
	\end{align*}
\end{lemma}

\begin{proof}
Cheeger's inequality says that $\frac{\phi^2}{2} \leq \lambda_2(\md^{-1/2} \mL \md^{-1/2})$ (see e.g.~\cite{SpielmanT11}).
Since $\md^{-1/2} \mL \md^{-1/2} $ is PSD and $\md^{1/2} \vones$ is in its kernel, i.e.~it is an eigenvector of value $0$ we have
$$
\frac{\phi^2}{2}
\leq \lambda_2(\md^{-1/2} \mL \md^{-1/2})
= \min_{x \perp \md^{1/2} \vones} \frac{x^\top \md^{-1/2} \mL \md^{-1/2} x}{x^\top x}
= \min_{y \perp \md \vones} \frac{x^\top \mL x}{x^\top \md x}.
$$
\end{proof}

To see the intuition of why we can efficiently implement the \textsc{HeavyQuery} and \textsc{Sample} operations, consider for the simple case when there is no scaling (i.e. $g=\vec{1}$) and the graph $G$ is a $\phi_G$-expander. Since $(\mA h)_e = h_u-h_v$ for $e=(u,v)$, we can shift $h$ by any constant vector without changing $\mA h$, so we can shift $h$ to make $h \bot \mD\onevec_n$. For any edge $e=(u,v)$ to have $|h_u-h_v|\geq \epsilon$, at least one of $|h_u|$ and $|h_v|$ has to be at least $\epsilon/2$. Thus, it suffices to check the adjacent edges of a node $u$ only when $|h_u|$ is large, or equivalently $h^2_u\epsilon^{-2}$ is at least $1/4$. Since checking the adjacent edges of any $u$ takes time $\deg(u)$, the time over all such nodes is bounded by $\sum_u \deg(u)h_u^2\epsilon^{-2}$, which is $O(h^\top \mD h \epsilon^{-2})$. By \Cref{lem:cheeger_based} this is $O(h^\top \mL h \phi_G^{-2}\epsilon^{-2})$, and note $h^\top \mL h =  \|\mA h\|_2^2$. The intuition behind \textsc{Sample} is similar. Since we can approximate $\|\mA h\|_2^2$ (in the denominator of edge probabilities) within a factor of $\phi_G^2$ by using $h^\top \mD h$, and $(h_u-h_v)^2$ (in the numerator) is at most $2(h_u^2 + h_v^2)$, this allows us to work mostly in the $n$-dimensional node space instead of the $m$-dimensional edge space. 

The above discussion would give us the desired result when $g$ is a constant vector and $\phi_G$ is large, i.e. $\Omega(\log^{-c} n)$ for some small constant $c$. To extend the argument to the general case with scaling vector $g$ and graph of smaller conductance, we simply partition the edges of $G$ to divide it into subgraphs, where each induced subgraph has roughly the same $g$ values on its edges and large conductance. For the first part (i.e., roughly constant $g$), we can bucket the edges by their $g$ values and move edges between buckets when their $g$ values get updated. To get subgraphs with large conductance, we utilize a {\em dynamic expander decomposition} algorithm stated in \Cref{lem:dynamicExpanderDecomposition} as a blackbox to further partition the edges of each bucket into induced expanders and maintain these expanders dynamically as edges move between buckets. We use the dynamic expander decomposition described in \cite{BernsteinBNPSS20} (building on tools developed for dynamic minimum spanning tree \cite{sw19,cglnps20,nsw17,ns17,w17}, especially \cite{sw19}). 

\begin{lemma}[Dynamic expander decomposition, \cite{BernsteinBNPSS20}]
	\label{lem:dynamicExpanderDecomposition}
	For any $\phi = O(1/\log^4 n)$ there exists a dynamic algorithm against an adaptive adversary, 
	that preprocesses an unweighted graph $G$ with $m$ edges
	in $O(\phi^{-1} m \log^5 n)$ time (or $O(1)$ time, if the graph is empty). 
	The algorithm can then maintain\footnote{The algorithm lists the number of changes (i.e.~which edges are added/removed to/from any of the expanders). 
	The update time is thus also an amortized bound on the number of changes to the decomposition.} a decomposition of $G$ 
	into $\phi$-expanders $G_1,...,G_t$, 
	supporting both edge deletions and insertions to $G$ in $O(\phi^{-2} \log^6 n)$ amortized time.
	The subgraphs $(G_i)_{1 \le i \le t}$ partition the edges of $G$, 
	and we have $\sum_{i=1}^t |V(G_i)| = O(n \log^2 n)$. 
	The algorithm is randomized Monte-Carlo, i.e.~with high probability every $G_i$ is an expander.
	
\end{lemma}
Now we formally prove the guarantees of our data structure.

\begin{proof}[Proof of \Cref{lem:large_entry_datastructure}]
	We describe each of the operation independently.
	\paragraph{\textsc{Initialization}.}
	Let $G=(V, E)$ be the graph corresponding to the given incidence matrix $\mA$ with weight $g_e$ of each edge $e$.
	Partition the edges into subgraphs denoted as $G_i=(V, E_i)$, where 
	\begin{align}
	E_i=\{e \mid g_e\in [2^{i}, 2^{i+1})\} \label{eq:graph weight classified}
	\end{align}
	i.e. each $G_i$ is an unweighted subgraph of $G$ consisting of edges of roughly the same $g_e$ values. Let $G_{-\infty}$ be the subgraph induced by all zero-weight edges. %
	Next, we choose $\phi = 1/\log^4 n$ and initialize the expander decomposition algorithm in \Cref{lem:dynamicExpanderDecomposition}
	on each of the $G_i$, which results in $\phi$-expanders $G_{i,1},...,G_{i,t_i}$ for each $i$. Note this can be implemented so the running time does not depend on the ratio $W$ of largest and smallest non-zero entries in $g$ if we only spend time on the non-empty $G_i$'s (and if a later update of $g$ creates a new non-empty subgraph $G_i$, we attribute the time to initialize the expander decomposition of $G_i$ to the \textsc{Scale} operation. In total we spend $O(m \log^9 n)$ time for the initialization since the average time for each edge is $O(\log^9 n)$.
	\paragraph{\textsc{Scale}.}
	Changing $g_e \leftarrow s$ means the edge weight of $e$ is changed to $s$.
	Thus we may need to delete the given edge $e$ from its current graph $G_{i}$
	and insert it into some other $G_{i'}$, 
	so that \eqref{eq:graph weight classified} is maintained.
	By \Cref{lem:dynamicExpanderDecomposition} it takes 
	$O(\log^{14} n)$ amortized time 
	to update the expander decomposition of $G_i$ and $G_{i'}$.

	\paragraph{\textsc{HeavyQuery}.} For every $i\neq -\infty$ we iterate over each of the $\phi$-expander $G_{i,j}$ and do the following. Let $m'$ and $n'$ denote the number of edges and nodes in $G_{i,j}$ respectively.  
	Let $\mB$ be the incidence matrix of $G_{i,j}$; thus rows and columns of $\mB$ correspond to edges and nodes in $G_{i,j}$, respectively. 
	To simplify notations, we pretend rows of $\mb$ are in $\R^n$ instead of $\R^{n'}$ (i.e. keep a ghost column even for nodes in $G$ but not in $G_{i,j}$).\footnote{\label{foot:query:assumption}Without this assumption, $\mB$ might not contain all column in $\mA$ and we need to define vector $\hat{h}$ which contains $|V(G_{i,j})|$ entries of $h$ corresponding to nodes in $G_{i,j}$, which introduce unnecessary notations. However, it is crucial to actually work with $n'$-dimensional vectors for efficiency.}
	Note that $G_{i,j}$ is unweighted and each row of $\mB$ corresponds to an edge in $G_{i,j}$ and also appears as a row in $\mA\in \R^{m\times n}$; thus $\mB\in \R^{m'\times n}$ for some $m'\leq m$. 
	All $m'$-row/column (respectively $n$-row/column) matrices here have each row/column correspond to an edge (respectively node) in $G_{i,j}$, so we index them using edges $e$ and nodes $v$ in $G_{i,j}$; for example, we refer to a row of $\mB$ as $\mB_e$ and an entry of $h\in \R^n$ as $h_v$. 
	To answer the query (finding all $e$ such that $|(\mdiag(g)\mA h)_e|\geq \epsilon$), it suffices to find all edges $e$ in $G_{i,j}$ such that 
	\begin{align}
	|(\mdiag(g)\mB h)_e|\geq \epsilon \label{eq:query:original goal}
	\end{align}
	because rows of $\mB$ is a subset of rows of $\mA$. Finding $e$ satisfying \eqref{eq:query:original goal} is done as follows.\\
	{\em Step 1:} Shift $h$ so it is orthogonal to the degree vector of $G_{i,j}$: 
	\begin{align}
	h' \leftarrow h - 
	\onevec_n \cdot (\onevec_n^\top \mD h) / (\onevec_n^\top \mD \onevec_n),
	\label{eq:query:modify_h}
	\end{align} 
	where $\mD \in \R^{n \times n}$ is the diagonal degree matrix with $\mD_{v,v} = \deg_{G_{i,j}}(v)$.\\
	{\em Step 2:} Let $\delta=\epsilon/2^{i+1}$, and find all $e$ with 
	\begin{align}
	|(\mB h')_e|\geq \delta.\label{eq:query:new goal}
	\end{align}
	as follows.
	For every node $v$ with $|h'_v| \geq 0.5 \delta$,
	find all edges $e$ incident to $v$ that satisfies \eqref{eq:query:new goal}. 
	Among these edges, return those satisfying \eqref{eq:query:original goal}.   
	
	\smallskip\noindent{\em Correctness:} 
	To see we correctly return all edges satisfying \eqref{eq:query:original goal}, let $e$ be any such edge, then we have
	\begin{align*}
	|(\mB h')_e| &= |(\mB h)_e| &\mbox{(since $\mB\onevec_n = \zerovec_m$)}\\
	&\geq \epsilon/g_e &\mbox{(since $|(\mdiag(g)\mB h)_e|=g_e |(\mB h)_e| $)}\\
	& \ge \epsilon/2^{i+1}=\delta &\mbox{(since $g_e\leq 2^{i+1}$ for every edge $e$ in $G_{i,j}$).}
	\end{align*}
	Thus, $e$ must satisfy \eqref{eq:query:new goal}. 
	So if our algorithm discovers all edges satisfying \eqref{eq:query:new goal}, 
	we will find all edges satisfying \eqref{eq:query:original goal} as desired.
	
	It is left to show that our algorithm actually discovers all edges satisfying \eqref{eq:query:new goal}. 
	Note that an edge $e=(u,v)$ satisfies \eqref{eq:query:new goal} 
	only if $|h'(u)|\geq 0.5\delta$ or $|h'(v)|\geq 0.5\delta$: 
	if $|h'(u)|< 0.5\delta$ and $|h'(v)|< 0.5\delta$, 
	then  $|(\mB h')_e|=|h'(v)-h'(u)|\leq |h'(v)|+|h'(u)|< \delta$. 
	Since in step 2 we consider edges incident to every node $u$ 
	such that $|h'(u)|>0.5\delta$, 
	the algorithm discovers all edges satisfying \eqref{eq:query:new goal}.

	\smallskip\noindent{\em Time Complexity:} 
	Step 1 (computing $h'$) can be easily implemented to take $O(|V(G_{i,j})|)$ (see \Cref{foot:query:assumption}).
	For step 2 (finding $e$ satisfying  \eqref{eq:query:new goal}), 
	let $\hat{V}=\{v\in V(G_{i,j})\mid |h'_v| \geq 0.5 \delta\}$, 
	then the time we spend in step 2 is bounded by
	\begin{align}\label{eq:query:time complexity}
	O(\sum_{v \in \hat{V}} \deg_{G_{i,j}}(v))
	= & ~ O(\sum_{v \in \hat{V}} \deg_{G_{i,j}}(v) (h'_v)^2 / \delta^2 ) \notag \\
	\leq  & ~ O\left(\sum_{v \in V(G_{i,j})} \deg_{G_{i,j}}(v) (h'_v)^2 / \delta^2 \right) \notag \\
	= & ~ O((h')^\top \mD h' / \delta^2 ). 
	\end{align}
	Above, the first equality is because 
	$(h'_v)^2/\delta^2\geq 0.25$ for every $v\in \hat{V}$. 
	To bound this time complexity further,  observe that 
	\begin{align}
	h' \bot \mD \onevec_n \mbox{ and } \|\mB h\|_2^2=\|\mB h'\|_2^2,\label{eq:h' properties}
	\end{align}
	where the latter is because each row of the incidence matrix $\mB$ sums to zero, so $\mB h$ remains the same under constant shift.
	By \Cref{lem:cheeger_based}, %
	$$\|\mB h\|_2^2 = \|\mB h'\|_2^2= (\mB h')^\top (\mB h') 
	= (h')^\top \mL h' %
	\stackrel{\text{Lem~\labelcref{lem:cheeger_based}}}{\ge} (h')^\top \mD h' \phi^2,$$
	where here $\mL=\mB^\top \mB$ is the Laplacian matrix of $G_{i,j}$. 
	Thus, the time complexity in \eqref{eq:query:time complexity} can be bounded by
	$$O( \|\mB h\|_2^2 / (\phi^{2} \delta^2) ) = O(\|\mB h\|_2^2 \phi^{-2} 2^{2i} / \epsilon^2).$$
	Summing the time complexity over all maintained expanders $G_{i,j}$ 
	with $\mB_{i,j}$ being the relevant incidence matrix $\mB$%
	\footnote{Without the assumption discussed in \Cref{foot:query:assumption}, 
		we also need to define $h_{i,j}$ to be the corresponding $\hat{h}$ here as well. 
		Again, this does not change our calculation below.}, 
	the total time for answering a query is
	\begin{align*}
	& O\left( \sum_{i,j} \left(\underbrace{\left|V(G_{i,j})\right|}_{\text{step 1}} + \underbrace{\|\mB_{i,j} h\|_2^2 \phi^{-2} 2^{2i} / \epsilon^2}_{\text{step 2}}\right)\right)\\
	&= O\left(\left(\sum_{i,j}|V(G_{i,j})|\right)+\left(\sum_i \sum_{(u,v)\in E_i} \left(h(u)-h(v)\right)^22^{2i} \phi^{-2} \epsilon^{-2}  \right)\right)\\
	&=O\left(\left(\sum_{i,j}|V(G_{i,j})|\right)+\left(\sum_i \sum_{(u,v)\in E_i} g_{(u,v)}^2\left(h(u)-h(v)\right)^2 \phi^{-2} \epsilon^{-2}  \right)\right)\\
	&=O\left( n \log^2 n \log W + \|\mdiag(g) \mA h\|_2^2 \epsilon^{-2} \log^8 n\right).
	\end{align*}
	The second equality is because  $g_{(u,v)}\geq 2^i$ for all $(u,v)\in E_i$ 
	(by \eqref{eq:graph weight classified}). 
	The third is because $\sum_{j} |V(G_{i,j})| = O(n \log^2 n)$ for every $i$ 
	(by \Cref{lem:dynamicExpanderDecomposition}), and also because the $E_i$'s partition the edges in the original graph.

	\paragraph{\textsc{Sample}.}
	Again, let $G_{i,j}$ be an expander of the maintained decomposition
	and denote $\mD^{(i,j)}$ its diagonal degree matrix.
	Then we define $h^{(i,j)}$ as in \eqref{eq:query:modify_h} for each $G_{i,j}$
	which satisfies $(h_v - h_u)^2 = (h^{(i,j)}_v - h^{(i,j)}_u)^2$ and $h^{(i,j)} \bot \mD^{(i,j)} \onevec$.
	
	Next, let 
	$$
	Q = \frac{K}{
		\sum_{i,j} 2^{2i} \sum_{v \in V(G_{i,j})} 
		(h^{(i,j)}_v)^2 \deg_{G_{i,j}}(v)
	}
	$$ 
	then we perform the following procedure for each expander $G_{i,j}$ of our decomposition. For each node $v$ in $G_{i,j}$, we sample each edge incident to $v$ with probability
	$$
	\min\{ Q \cdot 2^{2i} (h^{(i,j)}_v)^2, 1\}.
	$$
	If an edge $(u,v)$ is sampled twice (i.e. once for $u$ and once for $v$), 
	then it is included only once in the output.
	
	Computing all $h^{(i,j)}$'s and $Q$ takes $O(n \log^2 n \log W)$ time 
	and the sampling of edges can be implemented 
	such that the complexity is bounded by the number of included edges.
	For example by first sampling a binomial for each node 
	and then picking the corresponding number of incident edges uniformly at random.
	
	The expected number of included edges can be bounded by
	\begin{align*}
	\sum_{i,j} \sum_{(u,v) \in E(G_{i,j})} Q \cdot 2^{2i} ((h^{(i,j)}_v)^2 + (h^{(i,j)}_u)^2)
	=
	2Q \sum_{i,j}  2^{2i} \sum_{u} \deg_{G_{i,j}}(u) (h^{(i,j)}_u)^2
	=
	2K
	\end{align*}
	Thus the expected runtime is $O(K + n \log^2 n \log W)$,
	and for a w.h.p bound this increases to $O(K\log n + n \log^2 n \log W)$. %
	
	We are left with proving the claim, i.e., each edge $(u,v)$ is sampled with probability at least
	$$
	\min \left\{ K \cdot \frac{g_{(u,v)}^2 (h_u - h_v)^2}{16 \|\mdiag(g) \mA h \|_2^2 \log^8 n},1 \right\}.
	$$
	If $Q \ge 2^{2i} (h^{(i,j)}_v)^{-2}$ or $Q \ge 2^{2i} (h^{(i,j)}_u)^{-2}$, 
	then this is clear since then the sampling probability is just $1$.
	So we consider the case of $Q \le \min(2^{2i} (h^{(i,j)}_v)^{-2},2^{2i} (h^{(i,j)}_u)^{-2})$.
	Then the probability is
	\begin{align*}
	q_{(u,v)} = &~
	Q \cdot 2^{2i} \left((h^{(i,j)}_v)^2 + (h^{(i,j)}_u)^2\right) - Q^2 2^{4i} (h^{(i,j)}_v)^2 (h^{(i,j)}_u)^2 \\
	\ge&~
	Q \cdot 2^{2i} \left((h^{(i,j)}_v)^2 + (h^{(i,j)}_u)^2 - |(h^{(i,j)}_v) (h^{(i,j)}_u)|\right) \\
	\ge&~
	0.5 Q \cdot 2^{2i} \left((h^{(i,j)}_v)^2 + (h^{(i,j)}_u)^2 \right) \\
	\ge&~
	0.25 Q \cdot 2^{2i} (h^{(i,j)}_v - h^{(i,j)}_u)^2 \\
	\ge&~
	0.25 Q \cdot g_{u,v}^2 (h_v - h_u)^2
	\end{align*}
	The first inequality is due to $Q \le 2^{2i} (h^{(i,j)}_v)^{-2}$ and $Q\le2^{2i} (h^{(i,j)}_u)^{-2}$.
	
	Further, we can bound $Q$ as
	\begin{align*}
	Q 
	= 
	\frac{K}{\sum_{i,j} 2^{2i} \sum_{v \in V(G_{i,j})} (h^{(i,j)}_v)^2 \deg_{G_{i,j}}(u)}
	\ge
	\frac{K\phi^2}{4\|\mdiag(g) \mA h\|_2^2},
	\end{align*}
	by applying the Cheeger-inequality to each $\sum_{v \in V(G_{i,j})} (h^{(i,j)}_v)^2 \deg_{G_{i,j}}(u)$.
	In summary we obtain that any edge $(u,v)$ is included with probability at least
	$$
	K \frac{g_{u,v}^2 (h_v - h_u)^2}{16 \|\mdiag(g) \mA h\|_2^2 \log^8 n}.
	$$
	\paragraph{\textsc{Probability}.}
	As discussed in the analysis of \textsc{Sample}, we can compute $Q$ and all $h^{(i,j)}$'s in $O(n\log^2n\log W)$ time, then for each edge $e\in I$, we can look up the subgraph $G_{i,j}$ it belongs to, and compute its probability $q_e$ used in \textsc{Sample} in $O(1)$ time, which takes $O(|I|)$ time in total for all edges in $I$.
	\paragraph{\textsc{LeverageScoreSample}.}
	We will sample the edges in each expander $G_{i,j}$ separately. Let $\mL^{(i,j)}$ denote the Laplacian of the unweighted subgraph $G_{i,j}$, and we consider it as in the original dimension $n\times n $. Since $g_e\in [2^i,2^{i+1}]$ if $e$ is an edge in $G_{i,j}$, we have \[
	\sum_{(i,j)} 2^{2i} \mL^{(i,j)}\preceq \hat{\mL} \preceq	\sum_{(i,j)} 2^{2i+2} \mL^{(i,j)}
	\]
	where $\hat{\mL}=\mA^\top \mdiag(g)^2 \mA$ is the weighted Laplacian of the entire graph $G$ with weights $g_e^2$ on each edge $e$. Since $\mL^{(i,j)}\succeq 0$ for all $i,j$, we know $2^{2i} \mL^{(i,j)}\preceq \hat{\mL}$ for all $i,j$, which gives $\hat{\mL}^\dagger \preceq 2^{-2i} (\mL^{(i,j)})^\dagger$ since the null space of $\mL^{(i,j)}$ contains the null space of $\hat{\mL}$. As a result, for an edge $e=(u,v)$ in the expander $G_{i,j}$, we know
	\[
	\sigma(\mdiag(g)\mA)_e = g_e^2 \chi_e^\top \hat{L}^\dagger \chi_e \leq g_e^2 2^{-2i}\chi_e^\top (\mL^{(i,j)})^\dagger \chi_e^\top \leq 4 \chi_e^\top (\mL^{(i,j)})^\dagger \chi_e,
	\]
	where $\chi_e$ is the row of $A$ corresponding to $e$, so it has only two non-zero entries, $1$ at $u$ and $-1$ at $v$. Thus, it suffices to have
	\[
	p_e\geq 4K' \chi_e^\top (\mL^{(i,j)})^\dagger \chi_e.
	\]
	As $G_{i,j}$ is a $\phi$-expander for $\phi=1/\log^4n$, 
	using Cheeger's inequality (i.e.~$\lambda_2(\mL^{(i,j)}) \ge \phi^2/2$) 
	we can derive %
	that for all $y\bot \onevec_n$, $y^\top (\mL^{(i,j)})^\dagger y \leq 2\phi^{-2}y^\top (\mD^{(i,j)})^{-1} y$, 
	where $\mD^{(i,j)}$ is the diagonal degree matrix of $G_{i,j}$. Thus, we know
	\[
		\chi_e^\top (\mL^{(i,j)})^\dagger \chi_e \leq 2\phi^{-2} \left(\frac{1}{\deg^{(i,j)}_u}+\frac{1}{\deg^{(i,j)}_v} \right).
	\]
	Similar to how we implement \textsc{Sample}, we can go through each node $v$ (with non-zero degree) in $G_{i,j}$ and sample each edge incident to $v$ with probability
	\[
	p_v = \min\left\{\frac{16K'\phi^{-2}}{\deg^{(i,j)}_u},1\right\},
	\]
	and only include an edge once if it is sampled twice at both endpoints. Again this can be implemented so that the complexity is bounded by the number of included edges, and the expected number of sampled edges can be bounded by
	\begin{align*}
	\sum_{i,j}\sum_{(u,v)\in E(G_{i,j})} \frac{16K'\phi^{-2}}{\deg^{(i,j)}_u} + \frac{16K'\phi^{-2}}{\deg^{(i,j)}_v} & = \sum_{i,j}\sum_u \deg^{(i,j)}_u\cdot \frac{16K'\phi^{-2}}{\deg^{(i,j)}_u} \\
	&= 16K'\phi^{-2}\sum_{i,j} |V(G_{i,j})| \leq O(K'n \log^{10} n \log W),
	\end{align*}
	and we can include additional $\log n$ factor to make the bound hold with high probability. To see the probability that an edge is sampled satisfies the requirement, note
	\begin{align*}
	&p_e \geq \max{p_u,p_v} \geq \frac{p_u+p_v}{2}\\
	&\geq \min\left\{8K'\phi^{-2}(\frac{1}{\deg^{(i,j)}_u}+\frac{1}{\deg^{(i,j)}_v}),1\right\}\\
	&\geq 4K'\chi_e (\mL^{(i,j)})^\dagger \chi_e \geq \sigma(\mdiag(g)\mA)_e 
	\end{align*}
	\paragraph{\textsc{LeverageScoreBound}.} We can get the exact probability an edge $e$ is sampled in the previous implementation of \textsc{LeverageScoreSample}. Suppose $e=(u,v)$ is in subgraph $G_{i,j}$, we can compute in constant time
	\[
	p_u = \min\left\{\frac{16K'\phi^{-2}}{\deg^{(i,j)}_u},1\right\}, p_v=\min\left\{\frac{16K'\phi^{-2}}{\deg^{(i,j)}_v},1\right\},
	\] and $p_e = p_u+p_v - p_up_v$. Going through all edges in $I$ takes $O(|I|)$ time.
\end{proof}

\newpage

\section{Dual Solution Maintenance}
\label{sec:vector_maintenance}
In this section we discuss how we efficiently maintain an approximation of the dual slack (i.e. $s$) in our IPM.
Recall that our dual slack vector starts with an initial value $s^{\init}$, 
and in each iteration $t$ accumulates an update with the generic form of $\mA\delta^{(t)}$ 
where $\mA \in \R^{m \times n}$ is the edge incidence matrix.
The actual computation $\delta^{(t)}$ relies on other data structures 
such as gradient maintenance,
which will be discussed in later sections. 
In this section, we focus on maintaining an approximation of the dual slack assuming $\delta^{(t)}$ is given to us.

\begin{theorem}[\VecMaintainer]
\label{thm:vector_maintenance}
There exists a data-structure that supports the following operations%
\begin{itemize}
\item \textsc{Initialize($\mA\in\R^{m\times n}, v^{\init}\in \R^m, \epsilon>0$)}
The data-structure stores the given edge incidence matrix $\mA \in \R^{m \times n}$,
the vector $v^{\init} \in \R^m$
and accuracy parameter $0 < \epsilon \le 1$
in $\tilde{O}(m)$ time.
\item \textsc{Add($h\in\R^n$)}:
Suppose this is the $t$-th time the {\sc Add} operations is called, 
and let $h^{(k)}$ be the vector $h$ given when the {\sc Add} operation is called for the $k^{th}$ time. 
Define $v^{(t)}\in\R^m$ to be the vector
$$
v^{(t)} = 
v^{\init} + \mA \sum_{k=1}^t h^{(k)}.
$$
Then the data structure returns a vector $\ov^{(t)} \in \R^m$ such that
$
\ov^{(t)} \approx_{\epsilon} v^{(t)}.
$
The output will be in a compact representation to reduce the size. 
In particular, the data-structure returns a pointer to $\ov$ 
and a set $I \subset [m]$ of indices $i$
where $\ov^{(t)}_i$ is changed compared to $\ov^{(t-1)}_i$, i.e., the result of the previous call to \textsc{Add}. 
The total time after $T$ calls to \textsc{Add} is
$$
\tilde{O}\left(
T n \log W + T\epsilon^{-2}\cdot\sum_{t=1}^T \| (v^{(t)}-v^{(t-1)})/v^{(t)}\|_2^2 
\right).
$$
\item \textsc{ComputeExact()}: 
Returns $v^{(t)}\in \R^m$ in $O(m)$ time, 
where $t$ is the number of times \textsc{Add} is called so far 
(i.e., $v^{(t)}$ is the state of the exact vector $v$ after the most recent call to \textsc{Add}).

\end{itemize}
\end{theorem}

\begin{algorithm2e}
\caption{\label{alg:vector_maintenance}Algorithm for \Cref{thm:vector_maintenance}}
\SetKwProg{Members}{members}{}{}
\SetKwProg{Proc}{procedure}{}{}
\Members{}{
$D,\hat{f}, \ov\in \R^m, t\in \N$\\
$D_j,f^{(j)} \in \R^n$ and $F_j\subset [m]$ for $0\le j\le \log n$\\
}
\Proc{\textsc{Initialize}$(\mA, v^{\init}, \epsilon)$}{
	$\ov \leftarrow v^{\init}$, $\hat{f} \leftarrow \zerovec_n$,
	$t \leftarrow 0$\\
	\For{$j=0,...,\log n$}{
		$D_j.\textsc{Initialize}(\mA,1/\ov,0.2 \epsilon/\log n)$ (\ProjHeavyHitter, \Cref{lem:large_entry_datastructure})\\
		$f^{(j)} \leftarrow \zerovec_n$,
		$F_j\leftarrow \emptyset$\\
	}
}
\Proc{\textsc{FindIndices}$(h\in\R^n)$}{
	$I \leftarrow \emptyset$\\
	\For{$j=\log n,...,0$}{
		$f^{(j)} \leftarrow f^{(j)} + h$ 
		\Comment{When $2^j | t$, then $f^{(j)} = \sum_{k=t-2^j+1}^t h^{(k)}$}\\
		\If{$2^j | t$}{
			$I \leftarrow I \cup D_j.\textsc{HeavyQuery}(f^{(j)})$ \label{line:detect_changes}\\
			$f^{(j)} \leftarrow \zerovec_n$
		}
	}
	\Return $I$
}
\Proc{\textsc{VerifyIndex}$(i)$}{
	\If{$|\ov_i - (v^\init + \mA \hat{f})_i| \ge 0.2\epsilon \ov_i/\log n$}{ \label{line:check_change}
		$\ov_i \leftarrow (v^\init + \mA \hat{f})_i$ \label{line:set_vi}\\
		\For{$j=0,...,\log n$}{
			$F_j \leftarrow F_j \cup \{ i \}$, 
			$D_j.\textsc{Scale}(i, 0)$ \label{line:reinsert_i} \Comment{Notify other $D_j$'s to stop tracking $i$.}
		}
		\Return True
	}
	\Return False
}
\Proc{\textsc{Add}$(h\in\R^n)$}{
	$t \leftarrow t + 1$, $\hat{f} \leftarrow \hat{f} + h$,
	$I \leftarrow \textsc{FindIndices}(h)$ \label{line:findIndices}\\
	$I \leftarrow \left\{i|i\in I \mbox{ and }\textsc{VerifyIndex}(i)=True\right\}$ \label{line:verify_I}\\
	\lFor{$j:2^j | t$}{
		$I \leftarrow I \cup \left\{i|i\in F_j \mbox{ and }\textsc{VerifyIndex}(i)=True\right\}$ \label{line:verify_F}
	}
	\For{$j:2^j | t$}{
		\For{$i\in I \cup F_j$}{
			$D_j.\textsc{Scale}(i, 1 / \ov_i)$ \label{line:reweight}
		}
		$F_j \leftarrow \emptyset$ \label{line:empty_F}\\
	}
	\Return $I$, $\ov$
}
\Proc{\textsc{ComputeExact}$()$}{
	\Return $v^\init + \mA \hat{f}$
}
\end{algorithm2e}

From the graph-algorithmic perspective, our data structure maintains a flow\footnote{Here, a flow refers any assignment of a real value to each edge.} $v\in \R^E$ on an unweighted oriented graph $G=(V,E)$ (initially $v=v^{\init}$).
This flow is changed by the \textsc{Add}($h$) operation, whose input is the vertex potential $h\in \R^E$ inducing a new electrical flow to be augmented; in particular, the new flow on each edge $(i,j)$ is $v_{(i,j)}\leftarrow v_{(i,j)}+h_j-h_i$.  The \textsc{Add} operation then returns necessary information so that the user can maintain an approximation of flow $v$. The user can query for the exact value of $v$ by calling \textsc{ComputeExact()}. Note that the graph $G$ never changes; though, our data structured can be easily modified to allow $G$ to have edge resistances that may change over time.

To make sense of the running time for the \textsc{Add} operations, note that reading the input over $T$ \textsc{Add} calls needs $Tn$ time, and in total there can be as many as $T\epsilon^{-2}\cdot\sum_{t=1}^T \| (v^{(t)}-v^{(t-1)})/v^{(t)}\|_2^2 $ changes to the output vector we need to perform to guarantee the approximation bound. Thus, ignoring the $\log W$ factor, the complexity in our theorem is optimal as it is just the bound on input and output size.

Throughout this section we denote $h^{(t)}$ the input vector $h$ 
of the $t$-th call to \textsc{Add} (or equivalently referred to as the $t$-th iteration),
and let $v^{(t)} = v^\init + \mA \sum_{k=1}^t h^{(k)}$ be the state of the exact solution $v$ 
(as defined in \Cref{thm:vector_maintenance}) 
for the $t$-th call to \textsc{Add}.

In our algorithm (see \Cref{alg:vector_maintenance}) we maintain a vector $\hat{f}$ 
which is the sum of all past input vectors $h$, 
so we can retrieve the exact value of $v^{(t)}_i = v^\init_i + (\mA \hat{f})_i$ for any $i$ efficiently.
This value is computed and assigned to $\ov_i$ whenever the approximation $\ov$ that we maintain no longer satisfies the error guarantee for some coordinate $i$. 
As to how we detect when this may happen, we know the difference between $v^{(t)}$ and the state of $v$ at an earlier $t'$-th \textsc{Add} call is
\[
v^{(t)} - v^{(t')} = \mA \left(\sum_{k=t'+1}^t h^{(t)}\right),
\]
and thus we can detect all coordinates $i$ 
that changes above certain threshold from $t'$ to $t$-th \textsc{Add} call 
using the \ProjHeavyHitter data structure of \Cref{lem:large_entry_datastructure} 
(by querying it with $\sum_{k=t'+1}^t h^{(t)}$ as the parameter $h$). 
Note since the error guarantee we want is multiplicative in $\ov$ 
(i.e., $\ov^{(t)}_i \in 1\pm \epsilon v^{(t)}_i$ for all $i$), 
while the threshold $\epsilon$ in \Cref{lem:large_entry_datastructure} is absolute and uniform, 
we give $1/\ov$ as the scaling vector to \ProjHeavyHitter to accommodate this. 

Since the most recent updates on $\ov_i$ for different indices $i$'s happen at different iterations, 
we need to track accumulated changes to $v_i$'s over different intervals 
to detect the next time an update is necessary for each $i$. 
Thus, it is not sufficient to just have one copy of \ProjHeavyHitter. 
On the other hand, keeping one individual copy of \ProjHeavyHitter for each $0 \le t' < t$ will be too costly in terms of running time. 
We handle this by instantiating $\log n$ copies of the \ProjHeavyHitter data structure $D_j$ for $j=0,...,\log n$, 
each taking charge of batches with increasing number of iterations. 
In particular, the purpose of $D_j$ is to detect all coordinates $i$ in $v$ 
with large accumulated change over batches of $2^j$ iterations 
(see how we update and reset $f^{(j)}$ in \textsc{FindIndices} in \Cref{alg:vector_maintenance}). 
Each $D_j$ has its local copy of a scaling vector, which is initialized to be $1/\ov$, 
we refer to it as $\hat{g}^{(j)}$, and the cost to query $D_j$ is proportional to $\|\mdiag(\hat{g}^{(j)}) \mA f^{(j)}\|_2^2$. 
Note $f^{(j)}$ accumulates updates over $2^j$ iterations, 
and $\|\sum_{k=1}^{2^j} h^{(k)}\|_2^2$ can be as large as $2^j\sum_{k=1}^{2^j} \|h^{(k)}\|_2^2$. 
Since we want to bound the cost of our data structure by the sum of the squares of updates 
(which can in turn be bounded by our IPM method) 
instead of the square of the sum of updates, 
querying $D_j$ incurs an additional $2^j$ factor overhead. 
Thus for efficiency purposes, if $v_i$ would take much less than $2^j$ iterations to accumulate a large enough change,
we can safely let $D_j$ stop tracking $i$ during its current batch 
since $v_i$'s change would have been detected by a $D_{j'}$ of appropriate (and much smaller) $j'$ 
so that $\ov_i$ would have been updated to be the exact value 
(see implementation of \textsc{VerifyIndex}). 
Technically, we keep a set $F_j$ to store all indices $i$ that $D_j$ stops tracking for its current batch of iterations and set $\hat{g}^{(j)}_i$ to $0$ so we don't pay for coordinate $i$ when we query $D_j$. 
At the start of a new batch of $2^j$ iterations for $D_j$, we add back all indices in $F_j$ to $D_j$ (\Cref{line:reweight}) and reset $F_j$. 
As a result, only those $i$'s that indeed would take (close to) $2^j$ iterations 
to accumulate a large enough change are necessary to be tracked by $D_j$, 
so we can query $D_j$ less often for large $j$ to offset its large cost. 
In particular, we query each $D_j$ every $2^j$ iterations (see \Cref{line:detect_changes}).

We start our formal analysis with the following lemma.

\begin{lemma}\label{lem:findIndices}
Suppose we perform the $t$-th call to \textsc{Add}, 
and let $\ell$ be the largest integer with $2^{\ell}|t$ (i.e. $2^{\ell}$ divides $t$). %
 Then the call to \textsc{FindIndices} in \Cref{line:findIndices}
returns a set $I \subset [m]$ containing all $i \in [m]$ 
such that there exists some $0\le j\le\ell$ satisfying both
$i \notin F_j$ and 
$|v^{(t-2^j)}_i - v^{(t)}_i| \ge 0.2 |\ov|_i \epsilon / \log n$.
\end{lemma}

\begin{proof}
Pick any $j\in [0,\ell]$, if
$|v^{(t-2^j)}_i - v^{(t)}_i| \ge 0.2 |\ov_i| \epsilon / \log n$,
then
$$
\left|
e_i^\top \mA \sum_{k=t-2^j+1}^t h^{(k)}
\right|
\ge 0.2 |\ov_i| \epsilon / \log n.
$$
We will argue why \textsc{FindIndices} detect all $i$'s satisfying this condition.

Note that we have $f^{(j)} = \sum_{k=t-2^j+1}^t h^{(k)}$, and thus by guarantee of \Cref{lem:large_entry_datastructure} 
when we call $D_j.$\textsc{HeavyQuery}$(f^{(j)})$ (in \Cref{line:detect_changes}), we obtain for every $j\in[0,\ell]$ all $i\in[m]$ with
$$
\left|\hat{g}^{(j)}_i e_i^\top  \mA \sum_{k=t-2^j+1}^t h^{(k)} \right| \ge 0.2 \epsilon / \log n.
$$
Here $\hat{g}^{(j)}_i = 0$ if $i \in F_j$ by \Cref{line:reinsert_i}, 
which happens whenever $\ov_i$ is changed in \Cref{line:set_vi}. 
Thus by \Cref{line:reweight} we have $\hat{g}^{(j)}_i = 1 / |\ov_i|$ for all $i \notin F_j$. 
Equivalently, we obtain all indices $i \notin F_j$ satisfying the following condition, which proves the lemma.
\begin{align*}
\left| e_i^\top \mA \left(\sum_{k=t-2^j+1}^t h^{(k)}\right) ~ \right|
\ge
0.2 |\ov_i| \epsilon / \log n
\end{align*}
\end{proof}
To guarantee that the approximation $\ov$ we maintain is within the required $\epsilon$ error bound of the exact vector $v$, we need to argue that the $D_j$'s altogether are sufficient to detect all potential events that would cause $\ov_i$ to become outside of $(1\pm\epsilon)v$. It is easy to see that if an index $i$ is included in the returned set $I$ of \textsc{FindIndices} (\Cref{line:findIndices}), then our algorithm will follow up with a call to \textsc{VerifyIndex$(i)$}, which will guarantee that $\ov_i$ is close to the exact value $v_i$ (or $\ov_i$ will be updated to be $v_i$). Thus, if we are in iteration $t$, and $\ot$ is the most recent time \textsc{VerifyIndex$(i)$} is called, we know $\ov_i\approx v^{(\ot)}_i$, the value of $\ov_i$ remains the same since iteration $\ot$, and the index $i$ is not in the result of \textsc{FindIndices} for any of the iterations after $\ot$. We will demonstrate the last condition is sufficient to show $v^{(\ot)}_i \approx v^{(t)}_i$, which in turn will prove $\ov_i \approx v^{(t)}_i$. To start, we first need to argue that for any two iterations $\ot < t$, the interval can be partitioned into a small number of batches such that each batch is exactly one of the batches tracked by some $D_j$. 
\begin{lemma}
\label{lem:transform_t}
Given any $\ot < t$, there exists a sequence $$t = t_0 > t_1 > ... > t_k = \overline{t}$$ 
such that $k\leq 2\log t$ and $t_{z+1} = t_z - 2^{\ell_z}$ where $\ell_z$ satisfies $2^{\ell_z} | t_z $ for all $z=0,\ldots,k-1$. 
\end{lemma}
\begin{proof}
We can find such a sequence $(t_z)_z$ of length at most $2\log n$ as follows. Write both $t,\ot$ as $\log t$-bit binary, and start from the most significant bit. $t$ and $\ot$ will have some common prefix, and the first bit where they differ must be a $1$ for $t$ and $0$ for $\ot$ since $t>\ot$. Suppose the bit they differ corresponds to $2^\ell$, then let $t'= \lfloor t/2^\ell\rfloor2^\ell$, i.e. $t'$ is $t$ but with everything to the right of $2^\ell$ zeroed out. We can create the sequence in two halves. First from $t$ to $t'$ by iteratively subtracting the current least significant $1$-bit. To go from $t'$ down to $\ot$, consider the sequence backwards, we start from $\ot$, and iteratively add a number equal to the current least significant $1$-bit. This will keep pushing the least significant $1$ to the left, and eventually arrive at $t'$. Clearly this gives a sequence satisfying the condition that $t_{z+1} = t_z - 2^{\ell_z}$ and $2^{\ell_z} | t_z $ for all $t_z$ created in our sequence. The length of the sequence is at most $2\log t$, since each half has length at most $\log t$. For an example see \Cref{fig:transform_t}.
\end{proof}
\begin{figure}
\centering
\begin{tabular}{l|l|l}
$t_k$ & {\small binary representation of $t_k$} & $l_k$ \\
\hline
$t=t_0$   & 10100110\textbf{1} & $l_0=0$ \\
$t_1$   & 101001\textbf{1}00 & $l_1=2$ \\
$t_2$   & 10100\textbf{1}000 & $l_2=3$ \\
$t_3$   & 101\textbf{0}00000 & $l_3=5$ \\
\end{tabular}
\begin{tabular}{l|l|l}
$t_k$ & {\small binary representation of $t_k$} & $l_k$ \\
\hline
$t_4$        & 100100\textbf{0}00 & $l_4=2$ \\
$t_5$        & 1000111\textbf{0}0 & $l_5=1$ \\
$t_6$        & 10001101\textbf{0} & $l_6=0$ \\
$\ot=t_{7}$ & 100011001
\end{tabular}
\caption{\label{fig:transform_t}
Example of transform $t$ to $\ot$ in $2\log t$ steps. $t'$ in the proof of \Cref{lem:transform_t} is $t_3$.
Here $2^{\ell_k} | t_k $ and $t_{k+1} = t_k - 2^{\ell_k}$ for all $k$, and the bit (i.e. $2^{\ell_k}$) is highlighted in each row. 
}
\end{figure}
Now we can argue $\ov$ stays in the desired approximation range around $v$.
\begin{lemma}[Correctness of \Cref{thm:vector_maintenance}]
Assume we perform the $t$-th call to \textsc{Add},
the returned vector $\ov$ satisfies 
$|v^{(t)}_i - \ov_i| \le \epsilon |v^{(t)}_i|$ for all $i\in[m]$, 
and $I$ contains all indices that have changed since the $(t-1)$-th \textsc{Add} call.
\end{lemma}
\begin{proof}
By $\hat{f} = \sum_{j=1}^t h^{(j)}$ we have $v^\init +\mA \hat{f} = v^{(t)}$ in \Cref{line:check_change}.
So after a call to $\textsc{VerifyIndex}(i)$ we know that $|\ov_i - v^{(t)}_i| < 0.2 \epsilon \ov_j / \log n$,
either because the comparison $|\ov_i - (v^\init + \mA \hat{f})_i | \ge 0.2 \epsilon \ov_i /\log n$ in \Cref{line:check_change} returned false,
or because we set $\ov_i \leftarrow (v^\init + \mA \hat{f})_i$ in \Cref{line:set_vi}. 
Note this is also the only place we may change $\ov_i$. 
So consider some time $\overline{t} \le t$ 
when $\textsc{VerifyIndex}(i)$ was called for the last time 
(alternatively $\ot = 0$).
Then $\ov_i$ has not changed during the past $t - \ot$ calls to \textsc{Add}, 
and we know $|\ov_i - v^{(\ot)}_i| \le 0.2  \epsilon \ov_i / \log n$.
We now want to argue that $|v^{(\ot)} - v^{(t)}| \le 0.2 \epsilon \ov_i$, 
which via triangle inequality would then imply $\ov_i \approx_\epsilon v^{(t)}_i$.

For $t = \ot$ this is obvious, so consider $\ot < t$. We know from \Cref{lem:transform_t} the existence of a sequence
$$t = t_0 > t_1 > ... > t_k = \overline{t}$$ 
with $2^{\ell_z} | t_z $ and $t_{z+1} = t_z - 2^{\ell_z}$. In particular, this means that the interval between iteration $t_{z+1}$ and $t_z$ correspond to exactly a batch tracked by $D_{l_z}$. Thus, at iteration $t_z$ when \textsc{FindIndices} is called, $D_{l_z}.\textsc{HeavyQuery}$ is executed in \Cref{line:detect_changes}. This gives us $|v^{(t_z)}_i-v^{(t_{z+1})}_i|< 0.2 \epsilon |\ov_i| / \log n$ for all $z$, because by \Cref{lem:findIndices} 
the set $I \cup (\bigcup_j F_j)$ contains all indices $i$ 
which might have changed by $0.2 |\ov_i| \epsilon / \log n$
over the past $2^\ell$ iterations for any $2^\ell | t$,
and because $\textsc{VerifyIndex}(i)$ is called for all $i \in I \cup (\bigcup_j F_j)$ 
in \Cref{line:verify_I} and \Cref{line:verify_F}.

Note we can assume $\log t \leq \log n$ by resetting the data-structure after $n$ iterations, and this bounds the length of the sequence $k\leq 2\log n$. This then yields the bound
\[
|v^{(\ot)}_i - v^{(t)}_i|  = |v^{(t_k)}_i - v^{(t_0)}_i| 
\le~ \sum_{z=1}^{k} |v^{(t_z)}_i - v^{(t_{z-1})}_i|
\le~ k\cdot 0.2 \epsilon |\ov_i| / \log n 
\le ~ 0.4\epsilon |\ov_i|
\]
Thus we have $|\ov_i - v^{(t)}_i| \le (0.4 \epsilon+0.2 \epsilon/ \log n) |\ov_i|$,
which implies $\ov_i \approx_{\epsilon} v_i$. 
It is also straightforward to check that when we return the set $I$ at the end of \textsc{Add}, 
$I$ contains all the $i$'s where $\textsc{VerifyIndex}(i)$ is called and returned true in this iteration, 
which are exactly all the $i$'s where $\ov_i$'s are changed in \Cref{line:set_vi}.
\end{proof}
Now we proceed to the complexity of our data structure. We start with the cost of \textsc{FindIndices}, which is mainly on the cost of querying $D_j$'s. As we discussed at the beginning, there can be a large overhead for large $j$, but this is compensated by querying large $j$ less frequently. 
\begin{lemma}\label{lem:dual_query_complexity}
After $T$ calls to \textsc{Add}, the total time spent in \textsc{FindIndices} and \textsc{VerifyIndex} is bounded by
\[
\tilde{O}\left(
T\epsilon^{-2}\sum_{t=1}^T \| (v^{(t)}-v^{(t-1)})/v^{(t)}\|_2^2 
+ Tn \log W
\right)
\]
\end{lemma}

\begin{proof}
We start with the cost of \textsc{FindIndices}. 
Every call to \textsc{Add} invokes a call to \textsc{FindIndices}, 
so we denote the $t$-th call to \textsc{FindIndices} as the one associated with the $t$-th \textsc{Add}. 
Fix any $j$ and consider the cost for $D_j$. 
We update $f^{(j)}$ once in each call, which takes $O(n)$. 
Every $2^j$ calls would incur the cost to $D_j.\textsc{HeavyQuery}(f^{(j)})$. 
Without loss of generality we consider the cost of the first time this happens (at iteration $2^j$) since the other batches follow the same calculation. 
We denote $\hat{g}^{(j)}$ as the scaling vector in $D_j$ when the query happens. 
We know $\hat{g}^{(j)}_i = 0$ if $i \in F_j$, 
and $\hat{g}^{(j)}_i = 1 / |\ov_i|$ otherwise. 
Note here we can skip the superscript indicating the iteration number, 
since if $i\notin F_j$ it must be that $\ov_i$ has not changed over the $2^j$ iterations. 
The cost to query $D_j$ in \Cref{line:detect_changes} can then be bounded by.
\begin{align*}
&~
\tilde O(\|\mdiag(\hat{g}^{(j)} \mA f^{(j)}_k \|_2^2 \epsilon^{-2}  + n \log W)) \\
\le&~
\tilde O(\|\sum_{t=1}^{2^j} \mdiag(1/v^{(t)}) \mA h^{(t)}\|_2^2 \epsilon^{-2} + n \log W)\\
\le&~
\tilde O(\left(\sum_{t=1}^{2^j}\| \mdiag(1/v^{(t)}) \mA h^{(t)}\|_2 \right)^2 \epsilon^{-2}  + n \log W)\\
\le&~
\tilde O(2^j\cdot\sum_{t=1}^{2^j}\| \mdiag(1/v^{(t)}) \mA h^{(t)}\|_2^2 \epsilon^{-2} + n \log W)
\end{align*}
The first line is by \Cref{lem:large_entry_datastructure}, and note we use $0.2\epsilon/\log n$ as the error parameter in the call. 
The first inequality is by looking at $\hat{g}^{(j)}_i$ for each index $i$ separately. 
The value is either $0$ so replacing it by $1/v^{(t)}_i$ only increase the norm, 
or we know $i\notin F_j$, so $\hat{g}^{(j)}_i$ is within a constant factor of $1/v^{(t)}_i$ 
since $v^{(t)}_i\approx v^{(1)}_i \approx \ov_i$ for all $t\in[1,2^j]$. 
The second inequality uses the triangle inequality, and the third inequality uses Cauchy-Schwarz.
The cost of all subsequent queries to $D_j$ follows similar calculation, 
and as this query is only performed once every $2^j$ iterations, the total time after $T$ iterations is
\begin{align*}
\tilde{O}\left(
T\epsilon^{-2}\sum_{t=1}^T \| \mdiag(1/v^{(t)}) \mA h^{(t)}\|_2^2 + T2^{-j}n \log W
\right)\\
=
\tilde{O}\left(
T\epsilon^{-2}\sum_{t=1}^T \| (v^{(t)}-v^{(t-1)})/v^{(t)}\|_2^2 + T2^{-j}n \log W
\right)
\end{align*}
Note that the equality follows the definition of $v^{(t)}$ in \Cref{thm:vector_maintenance}. 
We can then sum over the total cost for all $D_j$'s for $j=1,\ldots, \log n$ 
to get the final running time bound in the lemma statement.

As to the cost of \textsc{VerifyIndex}, 
each call computes $(v^\init + \mA \hat{f})_i$
which takes $O(1)$ time as each row of $\mA$ has only two non-zero entries. 
Further, the updates to $F_j$'s and calls to $D_j.\textsc{Scale}$ take $\tilde{O}(1)$ time. 
Now we need to bound the total number of times we call \textsc{VerifyIndex}, 
which can only happen in two cases. 
The first case (\Cref{line:verify_I}) is when $i$ is returned by \textsc{FindIndices} in \Cref{line:findIndices}, 
and the number of times we call \textsc{VerifyIndex} is bounded by the size of $I$, 
which is in turn bounded by the running time of \textsc{FindIndices}. 
So the total cost over $T$ iterations can be bounded by the total cost of \textsc{FindIndices}. 
The second case (\Cref{line:verify_F}) is when $i$ is in some $F_j$ 
because $\ov_i$ was updated due to $v_i$ changing by more than $0.2\epsilon/\log n$, 
and the number of times this occurs can be bounded by 
\[
\tilde{O}\left(
T\epsilon^{-2}\sum_{t=1}^T \| (v^{(t)}-v^{(t-1)})/v^{(t)}\|_2^2 
\right).
\]
Adding up the total cost of \textsc{VerifyIndex} and \textsc{FindIndices} proves the lemma.
\end{proof}

We proceed to prove the complexity bounds in \Cref{thm:vector_maintenance}.
\paragraph{\textsc{Initialize}.} 
The main work is to initialize the data structures $D_j$ for $j=1,...,\log n$, 
which takes $\tilde{O}(m)$ time in total by \Cref{lem:large_entry_datastructure}.
\paragraph{\textsc{Add}.} 
The cost not associated with any \textsc{FindIndices} and \textsc{VerifyIndex} is $O(n)$. 
Together with \Cref{lem:dual_query_complexity} gives the bound of total time of $T$ calls to \textsc{Add} 
\[
\tilde{O}\left(
T\epsilon^{-2}\sum_{t=1}^T \| (v^{(t)}-v^{(t-1)})/v^{(t)}\|_2^2 
+ Tn \log W
\right)
\]
\paragraph{\textsc{ComputeExact}.} 
This just takes $O(\nnz(A))$ to compute the matrix-vector product, 
which is $O(m)$ since $A$ is an edge incidence matrix.

\newpage

\section{Gradient and Primal Solution Maintenance}
\label{sec:gradient_maintenance}

In \Cref{sec:log_barrier_method} and \Cref{sec:ls_barrier_method} we provided two IPMs. 
These IPMs try to decrease some potential function $\Phi(v)$ for $v \in \R^m$, 
so they require the direction of steepest descent, which is typically given by the gradient $\nabla\Phi(v)$.
However, the two IPMs allows for approximations 
and works with different norms, so they actually require the maximizer
\begin{align}
\og := \argmax_{w \in \R^m: \|w\|\le 1} \langle \nabla\Phi(\ov), w\rangle,
\label{eq:norm_maximizer}
\end{align}
which gives the direction of steepest ascent with respect to some norm $\|\cdot\|$ and some $\ov \approx v$.
The IPM of \Cref{sec:log_barrier_method} uses the $\ell_2$-norm 
(in which case the solution of above problem is just $\og = \nabla\Phi(\ov)/\|\nabla\Phi(\ov)\|_2$)
whereas the IPM of \Cref{sec:ls_barrier_method} uses the norm $\|\cdot\|_{\tau+\infty}$.

In this section we describe a framework that is able to maintain the maximizer of \eqref{eq:norm_maximizer}
efficiently.
Note that the maximizer is an $m$-dimensional vector,
so writing down the entire vector in each iteration of the IPM would be too slow.
Luckily, the IPM only requires the vector $\mA^\top \omX \og$,
where $\ox \approx x$ is an approximation of the current primal solution
and $\mA \in \R^{m \times n}$ is the constraint matrix of the linear program.
This vector is only $n$ dimensional, so we can afford to write it down explicitly.
Hence our task is to create a data structure that can efficiently maintain $\mA^\top \omX \og$.
This will be done in \Cref{sec:gradient:reduction}.

Note that the vector $\og$ is closely related to the primal solution $x$.
In each iteration of our IPM, the primal solution $x$ changes by 
(see \eqref{eq:over:update_xs} and \eqref{eq:over:delta_x new} in \Cref{sec:overview:newIPM})
\begin{align*}
x^\new \leftarrow x + \eta \omX \og - \mR h,
\end{align*}
for some $h \in \R^n$,
$\eta \in \R$,
random sparse diagonal matrix $\mR \in \R^{m \times m}$, 
$\omX = \mdiag(\ox)$ for $\ox \approx x$, 
and $\og$ is as in \eqref{eq:norm_maximizer}.
Because of the sparsity of $\mR$,
we can thus say that during iteration $t$
the primal solution $x^{(t)}$ is of the form
\begin{align}
x^{(t)} = x^\init + \sum_{k = 1}^t \left( h'^{(k)} + \eta \omX^{(k)} \og^{(k)} \right), \label{eq:gradient:sum}
\end{align}
where $h'^{(k)}$ is the vector $\mR h$ during iteration number $k$,
and $\og^{(k)}$, $\ox^{(k)}$ are the vectors $\og$ and $\ox$ during iteration number $k$.

Thus in summary, the primal solution $x$ is just the sum of (scaled) gradient vectors $\og^{(k)}$
and some sparse vectors $h'^{(k)}$.
The main result of this section will be data structures
(\Cref{thm:gradient_maintenance,thm:gradient_maintenance_simple}) 
that maintain both $\mA^{\top} \omX \og$
and an approximation $\ox$ of the primal solution $x$.

We state the result with respect to the harder case, 
when using the $\|\cdot\|_{\tau+\infty}$-norm for \eqref{eq:norm_maximizer}.
At the end of \Cref{sec:gradient:accumulator}, when we finish this result,
we also state what the variant would look like for the easier case of using the $\ell_2$-norm,
where the vector $\og$ is given by the simple expression $\nabla\Phi(\ov)/\|\nabla\Phi(\ov)\|_2$.

Recall that for $v \in \R^m$ we defined
$\Phi(v):= \sum_{i=1}^m \exp ( \lambda (v_i - 1 ) ) + \exp ( -\lambda ( v_i - 1 ) )$ 
for some given fixed parameter $\lambda$ of value $\polylog n$. 
Our main result of this section is the following \Cref{thm:gradient_maintenance}.
When using \Cref{thm:gradient_maintenance} in our IPM,
we will use $g = \eta \ox$ for some scalar $\eta \in \R$, 
so that \textsc{QuerySum} returns the desired approximation of \eqref{eq:gradient:sum}.

\begin{theorem}
\label{thm:gradient_maintenance} 
There exists a deterministic data-structure
that supports the following operations
\begin{itemize}
\item $\textsc{Initialize }(\mA\in\R^{m\times n}, x^{\init} \in \R^m, g\in\R^{m}, \ttau\in\R^{m}, z\in\R^{m},\epsilon>0)$:
	The data-structure preprocesses the given matrix $\mA\in\R^{m\times n}$,
	vectors $x^{\init},g,\ttau,z\in\R^{m}$, and accuracy parameter $\epsilon>0$
	in $\tilde{O}(\nnz(\mA))$ time. We denote $\mG$ the diagonal matrix $\mdiag(g)$. 
	The data-structure assumes $0.5\le z\le2$ and $n/m\le\ttau\le2$.
\item $\textsc{Update}(i \in [m], a \in \R, b \in \R, c \in \R)$: 
	Sets $g_{i}\leftarrow a$, $\ttau_{i} \leftarrow b$ and $z_i \leftarrow c$ in $O(\|e_i^\top \mA\|_0)$ time. 
	The data-structure assumes $0.5\le z\le2$ and $n/m\le\ttau\le2$.
\item $\textsc{QueryProduct}()$: 
	Returns $\mA^{\top}\mG\nabla\Phi(\oz)^{\flat(\otau)} \in \R^n$ for some $\otau \in \R^m$, $\oz \in \R^m$ 
	with $\otau \approx_\epsilon \ttau$ and $\|\oz-z\|_{\infty}\le \epsilon$,
	where 
	\begin{align*}
	x^{\flat(\otau)} := \argmax_{\|w\|_{\otau + \infty}} \langle x, w \rangle.
	\end{align*}
	Every call to \textsc{QueryProduct} must be followed by a call to \textsc{QuerySum},
	and we bound their complexity together (see \textsc{QuerySum}).
\item $\textsc{QuerySum}(h \in \R^m)$:
	Let $v^{(\ell)}$ be the vector $\mG\nabla\Phi(\oz)^{\flat(\otau)}$ used for the result of the $\ell$-th call to \textsc{QueryProduct}.
	Let $h^{(\ell)}$ be the input vector $h$ given to the $\ell$-th call to \textsc{QuerySum}.
	We define 
	\begin{align*}
	  x^{(t)} := x^{\init} + \sum_{\ell=1}^{t} v^{(\ell)} + h^{(\ell)}.  
	\end{align*}
	Then the $t$-th call to \textsc{QuerySum} returns a vector $\ox \in \R^m$ with $\ox \approx_\epsilon x^{(t)}$.
	
	Assuming the input vector $h$ in a sparse representation (e.g. a list of non-zero entries), 
	then after $T$ calls to \textsc{QuerySum} and \textsc{QueryProduct} 
	the total time for all calls together is bounded by
	\begin{align*}
	O\left(
		T n \epsilon^{-2} \log^2 n
		+ \log n \cdot \sum_{\ell=0}^T \|h^{(\ell)}\|_0 
		+ T \log n \cdot \sum_{\ell=1}^T \|v^{(\ell)}/x^{(\ell-1)}\|_2^2 / \epsilon^2
	\right).
	\end{align*}
	The output $\ox \in \R^m$ is returned in a compact representation to reduce the size. In particular, the data-structure returns a pointer to $\ox$ 
and a set $J \subset [m]$ of indices which specifies which entries of $\ox$ have changed 
	between the current and previous call to \textsc{QuerySum}.
\item $\textsc{ComputeExactSum}()$:
	Returns the exact $x^{(t)}$ in $O(m\log n)$ time.
\item $\textsc{Potential}()$:
	Returns $\Phi(\oz)$ in $O(1)$ time for some  $\oz$ such that $\|\oz-z\|_{\infty}\le \epsilon$. 
\end{itemize}
\end{theorem}
The main idea is that, 
although the exact gradient $\nabla\Phi(z)^{\flat(\ttau)}\in \R^m$ 
can be indeed $m$-dimensional and thus costly to maintain, 
if we perturb $\ttau$ and $z$ slightly we can show that 
$\nabla\Phi(\oz)^{\flat(\otau)}$ will be in $O(\epsilon^{-2} \log n)$ dimension. 
Here we say dimension in the sense that the $m$ entries 
can be put into $O(\epsilon^{-2} \log n)$ buckets, 
and entries in the same bucket share a common value. 
The proof consists of two steps, and we employ two sub data structures. 
First, we construct the $O(\epsilon^{-2} \log n)$ dimensional approximation 
$\nabla\Phi(\oz)^{\flat(\otau)}$ to the exact gradient. 
Using its low dimensional representation we are able 
to efficiently maintain $\mA \mG \nabla\Phi(\oz)^{\flat(\otau)}$.
The data structure that maintains this low dimensional representation 
and $\mA \mG \nabla\Phi(\oz)^{\flat(\otau)}$ is discussed in \Cref{sec:gradient:reduction}.
The subsequent \Cref{sec:gradient:accumulator} presents a data structure 
that takes the low dimensional representation of the gradient as input 
and maintains the desired sum for \textsc{QuerySum}. %

\subsection{Gradient Maintenance}
\label{sec:gradient:reduction}

In this section we discuss the computation of the 
$O(\epsilon^{-2} \log n)$ dimensional representation 
of the gradient $\nabla\Phi(\oz)^{\flat(\otau)}$, 
which in turn provides a good enough approximation 
to the real gradient $\nabla\Phi(z)^{\flat(\ttau)}\in \R^m$.  
This representation is then used to efficiently maintain 
$\mA^{\top} \mG \nabla\Phi(\oz)^{\flat(\otau)}$. 
The formal result is as follows:
\begin{lemma}
\label{lem:gradient_reduction} 
There exists a deterministic data-structure
that supports the following operations
\begin{itemize}
\item $\textsc{Initialize }(\mA\in\R^{m\times n},g\in\R^{m},\ttau\in\R^{m},z\in\R^{m},\epsilon>0)$:
	The data-structure preprocesses the given matrix $\mA\in\R^{m\times n}$,
	vectors $g,\ttau,z\in\R^{m}$, and accuracy parameter $\epsilon>0$
	in $O(\nnz(\mA))$ time. 
	The data-structure assumes $0.5\le z\le2$ and $n/m\le\ttau\le2$.
	The output is a partition 
	$\bigcup_{k=1}^K I_k = [m]$ 
	with $K = O(\epsilon^{-2} \log n)$.
\item $\textsc{Update}(i \in [m], a \in \R, b \in \R, c \in \R)$: 
	Sets $g_{i}=a$, $\ttau_{i}=b$ and $z_i=c$ in $O(\|e_i^\top \mA)$ time. 
	The data-structure assumes $0.5\le z\le2$ and $n/m\le\ttau\le2$.
	The index $i$ might be moved to a different set, so the data-structure returns $k$ such that $i \in I_k$.
\item $\textsc{Query}()$: 
	Returns $\mA^{\top}\mG\nabla\Phi(\oz)^{\flat(\otau)} \in \R^n$ for some $\otau \in \R^m$, 
	$\oz \in \R^m$ with $\otau \approx_\epsilon \ttau$ 
	and $\|\oz-z\|_{\infty}\le \epsilon$,
	where $x^{\flat(\otau)} := \argmax_{\|w\|_{\otau + \infty} \le 1} \langle x, w \rangle$.
	The data-structure further returns the low dimensional representation $s\in\R^K$ such that
	\begin{align*}\sum_{k=1}^K s_k \mathbf{1}_{i \in I_k} = \left( \nabla\Phi(\oz)^{\flat(\otau)}\right)_i\end{align*}
	for all $i \in [m]$,
	in $O(n\epsilon^{-2}\log n )$ time. 
\item $\textsc{Potential}()$
	Returns $\Phi(z)$ in $O(1)$ time.
\end{itemize}
\end{lemma}

\begin{algorithm2e}[t!]
\caption{Algorithm for reducing the dimension of $\nabla\Phi(\oz)^{\flat}$ 
	and maintaining $\mA^{\top}\mG\nabla\Phi(\oz)^{\flat}$ 
	(\Cref{lem:gradient_reduction}) \label{alg:gradient_reduction}}
\SetKwProg{Members}{members}{}{}
\SetKwProg{Proc}{procedure}{}{}
\Members{}{
$I^{(k,\ell)}$ partition of $[m]$. \\
$w^{(k,\ell)}\in\R^{n}$ \tcp*{Maintained to be $\mA^{\top}\mG \mathbf{1}_{i \in I^{(k,\ell)}}$}
$g,z \in \R^m$, $p \in \R$ \tcp*{$p$ is maintained to be $\Phi(z)$}
}
\Proc{\textsc{Initialize}$(\mA\in\R^{m\times n},g\in\R^{m},\ttau\in\R^{m},z\in\R^{m},\epsilon>0)$}{
	$\mA\leftarrow\mA$, $g\leftarrow g$, $z\leftarrow z$, $p \leftarrow \Phi(z)$ \\
	$I^{(k,\ell)}\leftarrow \emptyset, w^{(k,\ell)}\leftarrow \zerovec \quad \forall k=1,\ldots,\log^{-1}_{(1-\epsilon)}(n/m)$ and $\ell=0,...,(1.5/\epsilon)$\\
	\For{$i\in[1,m]$}{
		Find $k,\ell$ such that $0.5+\ell\epsilon/2\le z_{i}<0.5+(\ell+1)\epsilon/2$ and $(1-\epsilon)^{k+1}\le\ttau_{i}\le(1-\epsilon)^{k}.$\\
		Add $i$ to $I^{(k,\ell)}$ and set $\otau_{i}\leftarrow(1-\epsilon)^{k+1}$ \\
		$w^{(k,\ell)}\leftarrow w^{(k,\ell)} + \mA^{\top}g_i e_i$ \\
		}
	
	\Return $(I^{(k,\ell)})_{k,\ell\ge0}$
}
\Proc{\textsc{Update}$(i\in[m],a\in\R,b\in\R,c\in\R)$}{
	 $p \leftarrow \exp(\lambda c) + \exp(-\lambda c) - (\exp(\lambda z_i) + \exp(-\lambda z_i))$ \\
	$z_i \leftarrow c$ \\
	Find $k,\ell$ such that $i \in I^{(k,\ell)}$, then remove $i$ from $I^{(k,\ell)}$. \\
	$w^{(k,\ell)}\leftarrow w^{(k,\ell)}-\mA^{\top} g_{i} e_i$ \\
	Find $k,\ell$ such that 
		$0.5+\ell\epsilon/2\le c<0.5+(\ell+1)\epsilon/2$
		and $(1-\epsilon)^{k+1}\le b\le(1-\epsilon)^{k}$, 
		then insert $i$ into $I^{(k,\ell)}$. \\
	$w^{(k,\ell)}\leftarrow w^{(k,\ell)}+\mA^{\top}a e_i$ \\
	$g_{i}\leftarrow a$ \\
	\Return $k,\ell$
}
\Proc{\textsc{Query}$()$}{
	\Comment{Construct scaled low dimensional representation of $\nabla\Phi(\oz)$}\\
	Let $x_{k,\ell} = |I^{(k,\ell)}|\left(\lambda \exp(\lambda (0.5+\ell\epsilon/2 - 1)) -\lambda \exp(-\lambda(0.5+\ell\epsilon/2-1))\right)$ \\
	Interpret $x$ as an $O(\epsilon^{-2} \log n)$ dimensional vector. \\
	\Comment{Construct scaled low dimensional representation of $\otau$, here $C$ is the constant when define $\|\cdot\|_{\tau+\infty}$}
	Let $v$ be the vector with $v_{k,\ell} = \sqrt{|I^{(k,\ell)}|(1-\epsilon)^{k+1}}/C$. \\
	$s \leftarrow \argmax_{y:\|vy\|_{2}+\|y\|_{\infty}\leq 1}
	\left\langle x,y\right\rangle $
	via \Cref{cor:compute_flat}. \\
	\Return $s$	and $\sum_{k,l} s_{k,\ell} w^{(k,\ell)}$
}
\Proc{\textsc{Potential}$()$}{
	\Return $p$
}
\end{algorithm2e}

Before describing the low dimensional representation of $\Phi(\oz))^{\flat(\otau)}$,
we first state the following result which we will use to compute $x^{\flat(\otau)}$. 
Note this is just the steepest ascent direction of $x$ 
with respect to a custom norm, 
and when $w,v$ are two vectors of the same dimension 
we refer to $wv$ as the entry-wise product vector.

\begin{lemma}[{\cite[Algorithm 8]{ls19}}]\label{lem:ball_projection}
Given $x \in \R^n$, $v \in \R^n$, we can compute in $O(n \log n)$ time
\begin{align*}
u = \underset{ \|w\|_2 + \|w v^{-1}\|_\infty \le 1 }{ \argmax } \langle x,w \rangle.
\end{align*}
\end{lemma}
We have the following simple corollary of \Cref{lem:ball_projection}.
\begin{corollary}\label{cor:compute_flat}
Given $x \in \R^n$, $v \in \R^n$, we can compute in $O(n \log n)$ time
\begin{align*}
u = \underset{ \|v w\|_2 + \|w\|_\infty \le 1 }{ \argmax } \langle x,w \rangle .
\end{align*}
\end{corollary}

\begin{proof}
We have
\begin{align*}
\max_{\|v w\|_2 + \|w\|_\infty \le 1} \langle x,w \rangle = \max_{\|wv\|_2 + \|w\|_\infty} \langle xv^{-1},wv \rangle
=
\max_{\|w'\|_2 + \|w'v^{-1}\|_\infty} \langle xv^{-1},w' \rangle
\end{align*}
by substituting $w' = wv$.
Thus for the maximizer we have
\begin{align*}
\underset{ \|v w\|_2 + \|w\|_\infty \le 1 }{\argmax} \langle x,w \rangle
=
v^{-1} \underset{ \|w'\|_2 + \|w'v^{-1}\|_\infty }{ \argmax } \langle xv^{-1},w' \rangle,
\end{align*}
where the maximizer of on the right hand side can be computed in $O(n \log n)$ via \Cref{lem:ball_projection}
and entry-wise multiplying by $v^{-1}$ takes only $O(n)$ additional time.
\end{proof}

Recall $x^{\flat(\tau)} := \argmax_{\|w\|_{\tau+\infty} \le 1} \langle x, w \rangle$ 
and $\|w\|_{\tau+\infty} := \|w\|_\infty + C\|w\sqrt{\tau}\|_2$ for some constant $C$. 
\Cref{cor:compute_flat} allows us to compute $x^{\flat(\tau)}$ 
(using the corollary with $v=C\sqrt{\tau}$) 
for any $x$ in its original dimension. 
However, if $x$ is an $m$ dimensional vector such as $\nabla\Phi(z)$, 
then \Cref{cor:compute_flat} is too slow to be used in every iteration of the IPM.

To address this, we first discretize $z$ and $\ttau$ 
by rounding the entries to values determined by the error parameter $\epsilon$ 
to get $\oz,\otau$ which are in low dimensional representation 
(i.e. put the indices into buckets). 
In particular, $\oz$ will be the vector $z$ 
with each entry rounded down to the nearest value 
of form $0.5+\ell\epsilon/2$ for non-negative integer $\ell$, 
and $\otau$ rounds down each entry of $\ttau$ 
to the nearest power of $(1-\epsilon)$. 
It is fairly straightforward to see from definition 
this in turn also gives low dimensional representation of $\nabla\Phi(\oz)$. 
In the next lemma we show when $x$ and $\tau$ are both low dimensional, 
$x^{\flat(\tau)}$ also admits a low dimensional representation 
that can be computed efficiently, 
which let us to compute $\nabla\Phi(\oz)^{\flat(\otau)}$. 
In the following lemma, when we use a set (i.e. a bucket) of indices as a vector,
we mean it's indicator vector with $1$ at indices belong to the set and $0$ elsewhere. 

\begin{lemma}\label{lem:projected_flat}
Given $\ox,\ov,\og \in \R^k$, $x,v \in \R^n$, 
and a partition $\bigcup_{i=1}^k I^{(i)} = [n]$,
such that
$x = \sum_{i=1}^k \ox_i I^{(i)}$,
$v = \sum_{i=1}^k \ov_i I^{(i)}$
and $\og_i = \sqrt{|I^{(i)}|}$, 
let $\ou \in \R^k$ be the maximizer of
\begin{align*}
\ou := \argmax_{\ow\in\R^k:\|(\og\ov)\ow\|_2 + \|\ow\|_\infty \le 1} \langle \ox~\og^2,\ow \rangle.
\end{align*}
We can define $u \in \R^n$ via $u = \sum_{i=1}^k \ou_i I^{(i)}$, then
\begin{align*}
\langle v, u \rangle = \max_{w\in\R^n:\|vw\|_2 + \|w\|_\infty\le 1} \langle x,w \rangle.
\end{align*}
\end{lemma}

\begin{proof}
Consider $\ow \in \R^k$ and $w \in \R^n$ with $w = \sum_{i=1}^k \ow_i I^{(i)}$.
Then $\|w\|_\infty = \|\ow\|_\infty$ and
\begin{align*}
\|wv\|_2^2 
= 
\sum_i (w_i v_i)^2 
= 
\sum_j \sum_{i\in I^{(j)}} (\ow_j \ov_j)^2 
= 
\sum_j |I^{(j)}| (\ow_j \ov_j)^2 
= 
\|(\og\ov)\ow\|_2^2.
\end{align*}
So $\ow$ satisfies $\|(\og\ov)\ow\|_2 + \|\ow\|_\infty \le 1$ iff $w$ satisfies $\|vw\|_2 + \|w\|_\infty\le 1$.
Further, 
\begin{align*}
\langle \ox~\og^2, \ow \rangle 
= 
\sum_{i=1}^k |I^{(i)}| \ox_i \ow_i
= 
\sum_{i=1}^n x_i w_i
= 
\langle x, w \rangle.
\end{align*}
Thus
\begin{align*}
\max_{\ow\in\R^k:\|(\og\ov)\ow\|_2 + \|\ow\|_\infty \le 1} \langle \ox~\og^2,\ow \rangle
\le
\max_{w\in\R^n:\|vw\|_2 + \|w\|_\infty \le 1} \langle x,w \rangle
\end{align*}
which in turn means that $u \in \R^k$ as defined in the lemma satisfies
\begin{align*}
\langle x, u \rangle = \langle \ox~\og^2, \ou \rangle \le \max_{w\in\R^n:\|vw\|_2 + \|w\|_\infty \le 1} \langle x,w \rangle.
\end{align*}
Let $S := \{w \in \R^n : \|vw\|_2 + \|w\|_\infty \le 1$ and $w_i = w_j$ for all $i,j\in I^{(\ell)}\}$,
i.e. the set of $n$ dimensional vectors $w$ 
that has a low dimensional representation via some $\ow \in\R^k$ 
and $w = \sum_{i=1}^k \ow_i I^{(i)}$.
Then
\begin{align*}
\langle x, u \rangle = \max_{w\in S} \langle x,w \rangle.
\end{align*}
by definition of $u$.
We claim that
\begin{align*}
\max_{w\in S} \langle x,w \rangle
=
\max_{w\in\R^n:\|vw\|_2 + \|w\|_\infty \le 1} \langle x,w \rangle,
\end{align*}
which would then conclude the proof of the lemma.

So assume there exists some maximizer $w\in\R^n$ 
with $\|vw\|_2 + \|w\|_\infty \le 1$ 
and there are indices $i,j,\ell$ 
with $i,j \in I^{(\ell)}$ and $w_i \neq w_j$, 
i.e. the vector $w \notin S$. 
Then define $w' \in \R^n$ 
with $w'_i = w'_j = (w_i + w_j) / 2$.
For this vector we have 
$\|w'\|_\infty \le \|w\|_\infty$ and
\begin{align*}
\|vw'\|_2^2 
= 
\sum_{k\notin \{i,j\}} v_k^2 w_k^2 + v_i^2 (w_i + w_j)^2 / 2 
\le 
\sum_{k\notin \{i,j\}} v_k^2w_k^2 + v_i^2 w_i^2 + v_j^2 w_j^2 
= 
\|vw\|_2^2,
\end{align*}
where the first equality is because $v_i = v_j$ for $i,j$ in same $I^{\ell}$. 
So we know $w'$ also satisfies $\|vw'\|_2 + \|w'\|_\infty \le 1$. On the other hand,
\begin{align*}
\langle x, w \rangle 
= 
\sum_{k \notin \{i,j\}} x_k w_k + x_i w_i + x_i w_j
= 
\sum_{k \notin \{i,j\}} x_k w_k + x_i (w_i+w_j)/2 + x_i (w_i+w_j)/2
= 
\langle x, w' \rangle.
\end{align*}
Thus $w'$ must also be a maximizer, and if we repeat this transformation we eventually have $w' \in S$.
This concludes the proof.
\end{proof}

We now have all tools to prove \Cref{lem:gradient_reduction}.

\begin{proof}[Proof of \Cref{lem:gradient_reduction}]
The data-structure is given by \Cref{alg:gradient_reduction}.
It is straightforward to check that \textsc{Initialize} and \textsc{Update} guarantees that
the data-structure maintains the following invariants:
\begin{itemize}
\item $\bigcup_{k,\ell} I^{(k,\ell)} = [m]$ is a partition of $[m]$.
\item $i \in I^{(k,\ell)}$ if and only if 
$0.5+\ell\epsilon/2\le z_{i}<0.5+(\ell+1)\epsilon/2$
and $(1-\epsilon)^{k+1}\le\ttau_{i}\le(1-\epsilon)^{k}$.
\item $w^{(k,\ell)} = \mA^{\top}\mG \cdot I^{(k,\ell)} \in \R^n$.
\item $p = \sum_i \exp(\lambda z_i) + \exp(-\lambda z_i) = \Phi(z)$.
\end{itemize}
Moreover, our rounding of $z,\ttau$ gives 
$O(\epsilon^{-1}\log_{1-\epsilon}^{-1}(n/m))$ dimensional representations 
because of the assumption $0.5\leq z\leq 2$ and $n/m\leq \ttau\leq 2$, 
and is simply $O(\epsilon^{-2}\log n)$. 

Next, let us analyze the function \textsc{Query}. 
We can interpret the decomposition $\bigcup_{k,\ell} I^{(k,\ell)} = [m]$ 
as a decomposition of the form $\bigcup_{t} I^{(t)} = [m]$ by renaming the indices. 
Recall $\oz,\otau$ are low dimensional approximations of $z,\ttau$ with each entry rounded. 
It is straightforward to see the vector $x$ constructed in the first line of \textsc{Query} 
is just the low dimensional representation of $\nabla\Phi(\oz)$ 
where each entry is scaled appropriately by the $|I^{(k,\ell)}|$'s, 
so the $x$ we constructed corresponds to $\ox\og$ in \Cref{lem:projected_flat}. 
Similarly, the $v$ vector constructed in \textsc{Query} 
is the low dimensional representation of $\sqrt{\otau}$ 
scaled appropriately to serve as the $\og\ov$ in \Cref{lem:projected_flat}. 
Thus by \Cref{lem:projected_flat} the vector $s$ obtained in \textsc{Query} is the vector
such that $s' := \sum_{t} s^{(t)} I^{(t)} \in \R^m$ satisfies
\begin{align*}
\langle s' , \nabla \Phi(\oz) \rangle
=
\max_{C\|\sqrt{\otau}y\|_2 + \|y\|_\infty \le 1} \langle \nabla \Phi(\oz), y\rangle,
\end{align*}
That is, $s'=\nabla\Phi(\oz)^{\flat(\otau)}$, 
and the returned $s$ is a low dimensional representation of it. 
Moreover, $\sum_{k,l} s_{k,\ell} w^{(k,\ell)}$ is exactly 
$\mA^{\top}\mG \nabla\Phi(\oz)^{\flat(\otau)}$ 
because we maintained 
$w^{(k,\ell)} = \mA^{\top}\mG \cdot I^{(k,\ell)} \in \R^n$.

\paragraph{Complexity}
The complexity of \textsc{Initialize} is dominated by the initialization of $w^{(k,\ell)}$'s, 
which takes $O(\nnz(A))$ since we just partition the (scaled) rows of $\mA$ 
into buckets according to our rounding and sum up the vectors in each bucket. 
The complexity of \textsc{Update}$(i, \cdot, \cdot, \cdot)$ is $O(\|e_i^\top \mA \|_0)$, 
i.e. the cost of maintaining 
$w^{(k,\ell)} = \mA^{\top}\mG \cdot I^{(k,\ell)} \in \R^n$. 
The time for \textsc{Potential} is $O(1)$ since it just returns the value $p$. 

As to the complexity of \textsc{Query}, 
as we work with $O(\epsilon^{-2}\log n)$ dimensional vectors in \textsc{Query}, 
the time complexity according to \Cref{lem:projected_flat} is 
$O(\epsilon^{-2} \log n \log(\epsilon^{-2} \log n))$ 
for computing the maximizer $s$.
Constructing $\mA^\top \mG (\nabla \Phi(\oz))^{\flat(\otau)} = \sum_{k,\ell} w^{(k,\ell)} s^{(k, \ell)}$ 
takes $O(n \epsilon^{-2} \log n)$ time, 
which subsumes the first complexity, 
assuming $\epsilon$ is polynomial in $1/n$.
Thus \textsc{Query} takes $O(n \epsilon^{-2} \log n)$ time.
\end{proof}
\subsection{Primal Solution Maintenance}
\label{sec:gradient:accumulator}

In this section we discuss the data structure that effectively maintains an approximation of the sum of all gradients computed so far, which is captured by \Cref{lem:gradient_accumulator}. 
The input to the data structure is the $O(\epsilon^{-2} \log n)$ dimensional representations of the gradients of each iteration,
which is computed by the data structure of the previous subsection. 
At the end of this subsection we will combine \Cref{lem:gradient_reduction} and the following \Cref{lem:gradient_accumulator} to obtain the main result \Cref{thm:gradient_maintenance}, 
which allows us to efficiently maintain an approximation $\ox$ of the primal solution.
\begin{lemma}
\label{lem:gradient_accumulator} 
There exists a deterministic data-structure
that supports the following operations
\begin{itemize}
\item $\textsc{Initialize }(x^{\init}\in\R^m,g\in\R^m,(I_k)_{1\le k \le K},\epsilon>0)$:
	The data-structure initialized on the given vectors $x^{\init},g\in\R^m$,
	the partition $\bigcup_{k=1}^K I_k = [m]$ where $K=O(\epsilon^{-2}\log n)$, and the accuracy parameter $\epsilon > 0$	in $O(m)$ time.
\item $\textsc{Scale}(i \in [m], a \in \R)$: 
	Sets $g_{i}\leftarrow a$ in $O(\log n)$ amortized time.
\item $\textsc{Move}(i \in [m],k\in[1,K])$: 
	Moves index $i$ to set $I_k$ in $O(\log n)$ amortized time.
\item $\textsc{Query}(s \in \R^K, h \in \R^m)$: 
	Let $g^{(\ell)}$ be the state of vector $g$ during the $\ell$-th call to \textsc{Query}
	and let $s^{(\ell)}$ and $h^{(\ell)}$ be the input arguments of the respective call. The vector $h$ will always be provided as a sparse vector so that we know where are the non-zeros in the vector. Define $y^{(\ell)} = \mG^{(\ell)} \sum_{k=1}^K I_k^{(\ell)} s_k^{(\ell)}$ and $x^{(t)} = x^{\init} + \sum_{\ell=1}^t h^{(\ell)} + y^{(\ell)}$,
	then the $t$-th call to \textsc{Query} returns a vector $\ox \approx_\epsilon x^{(t)}$.
	After $T$ calls to \textsc{Query}, the total time of all $T$ calls is bounded by
	\begin{align*}O\left(
		TK+\log n\cdot \sum_{\ell=0}^T \|h^{(\ell)}\|_0 + T \log n\cdot \sum_{\ell=1}^T \|y^{(\ell)} / x^{(\ell-1)}\|_2^2 / \epsilon^2
	\right).\end{align*}
	The vector $\ox \in \R^m$ is returned as a pointer 
	and additionally a set $J \subset [m]$ is returned 
	that contains the indices where $\ox$ changed compared to the result of the previous \textsc{Query} call.
\item $\textsc{ComputeExactSum}()$:
	Returns the current exact vector $x^{(t)}$ in $O(m\log n)$ time.
\end{itemize}
\end{lemma}

The implementation of the data structure is given in \Cref{alg:gradient_accumulator}. 
At a high level, each iteration $t$ (i.e. \textsc{Query} call) 
the vector $x$ accumulates an update of $h^{(t)} + y^{(t)}$. 
For the $h$ component, we always include the update to our approximation $\ox$ as soon as we see $h$. 
This is because the $h$ provided by our IPM will be sparse, 
and we can afford to spend $\tilde{O}(\|h\|_0)$ each iteration to carry out the update. 
For the $y$ component in the update, we use a "`lazy"' update idea.
In particular, we maintain a vector $\hat{\ell}\in \R^m$ 
where $\hat{\ell}_i$ stores the most recent iteration $t$ (i.e. \textsc{Query} call) 
when $\ox_i$ is updated to be the exact value $x^{(t)}_i$ for each index $i$. 
We update $\ox_i$ to be the exact value (in \textsc{ComputeX}) 
whenever $h_i$ is non-zero, $g_i$ is scaled, or $i$ is moved to a different set $I_k$. 
If none of these events happen, the accumulated update to $x_i$ since iteration $\hat{\ell}_i$ 
will have a very simple form, so we just store the accumulated updates up to each iteration $t$ in a vector $f^{(t)}$.
To detect when the accumulated update on $x_i$ has gone out of the approximation bound $\epsilon \ox_i$ 
since the last time $\ox_i$ is updated, 
we use $\Delta^{(high)}_i$ and $\Delta^{(low)}_i$ 
to store the upper and lower boundary of the approximation range around $\ox_i$, 
and update $\ox_i$ once the accumulated updates goes beyond the range. 
We proceed with the proof of our result.

\begin{algorithm2e}[t!]
\caption{Algorithm for accumulating $\mG \nabla\Phi(\ov)^{\flat}$ 
(\Cref{lem:gradient_accumulator}) \label{alg:gradient_accumulator}}
\SetKwProg{Members}{members}{}{}
\SetKwProg{Proc}{procedure}{}{}
\SetKwProg{Priv}{private procedure}{}{}
\Members{}{
	$I_1,...,I_K$ \tcp*{Partition $\bigcup_k I_k = [m]$}
	$t \in \N, \ox \in \R^m$ \tcp*{\textsc{Query} counter and approximation of $x^{(t)}$}
	$\hat{\ell}\in\N^{m}$ \tcp*{$\hat{\ell}_i$ is value of $t$ when we last update $\ox_i \leftarrow x_i$}
	$f^{(t)} \in \R^K$ \tcp*{Maintain $f^{(t)}=\sum_{k=1}^t s^{(k)}$}
	$\Delta^{(high)},\Delta^{(low)} \in \R^m$ \tcp*{Maintain $\Delta_i = f^{(\hat{\ell}_i)}_k\pm |\epsilon \ox_i / (10 g_i)| $ if $i\in I_k$}
}
\Proc{\textsc{Initialize}$(x^{\init} \in \R^m, g \in \R^m, (I_k)_{1\le k \le K})$}{
	$\ox \leftarrow x^{\init}$, 
	$(I_k)_{1\le k \le K} \leftarrow (I_k)_{1\le k \le K}$,
	$t \leftarrow 0$,
	$f^{(0)} \leftarrow \zerovec_K$,
	$w \leftarrow \zerovec_m$, 
	$g \leftarrow g$
}
\Priv{\textsc{ComputeX}$(i,h_i)$}{
	Let $k$ be such that $i \in I_k$ \\
	$\ox_i \leftarrow \ox_i + g_i \cdot (f^{(t)}_k -  f^{(\hat{\ell}_i)}_k)+ h_i$  \label{line:update_bin}\\
	$\hat{\ell}_i \leftarrow t$ \\
}
\Priv{\textsc{UpdateDelta}$(i)$}{
	Let $k$ be such that $i \in I_k$. \\
	$\Delta^{(high)}_i \leftarrow  f^{(\hat{\ell}_i)}_k+|\epsilon \ox_i / (10 g_i)| $ \\
	$\Delta^{(low)}_i \leftarrow  f^{(\hat{\ell}_i)}_k -|\epsilon \ox_i / (10 g_i)| $
}
\Proc{\textsc{Move}$(i\in[m], k)$}{
	\textsc{ComputeX}$(i,0)$ \\
	Move index $i$ to set $I_k$ \\
	$\textsc{UpdateDelta}(i)$
}
\Priv{\textsc{Scale}$(i, a)$}{
	\textsc{ComputeX}$(i,0)$ \\
	$g_i \leftarrow a$ \\
	$\textsc{UpdateDelta}(i)$
}
\Proc{\textsc{Query}$(s \in \R^K, h \in \R^m)$}{
	$t \leftarrow t + 1$, $J \leftarrow \emptyset$ \\
	$f^{(t)} \leftarrow f^{(t-1)} + s$\label{line:group_bin_ft} \\
	\For{$i$ such that $h_i \neq 0$}{
		\textsc{ComputeX}$(i,h_i)$, $\textsc{UpdateDelta}(i)$, $J \leftarrow J \cup \{ i\}$ \label{line:nonzero_hi}
	}
	\For{$k=1,\ldots,K$}{
		\For{$i \in I_k$ with $f^{(t)}_k > \Delta^{(high)}_i$ or $f^{(t)}_k < \Delta^{(low)}_i$}{
			\textsc{ComputeX}$(i,0)$, $\textsc{UpdateDelta}(i)$, $J \leftarrow J \cup \{ i\}$ \label{line:large_change_x_f}
		}
	}

	\Return $\ox$, $J$
}
\Proc{\textsc{ComputeExactSum}$()$}{
	\lFor{$i\in[m]$ and $\hat{\ell}_i<t$}{
		$\textsc{ComputeX}(i,0)$, $\textsc{UpdateDelta}(i)$
	}
	\Return $\ox$
}
\end{algorithm2e}

\begin{proof}[Proof of \Cref{lem:gradient_accumulator}]
	We start by analyzing the correctness.
	\paragraph{Invariant}
	Let $s^{(t)},h^{(t)},g^{(t)},I_k^{(t)}$ be the state of $s,h,g,I_k$ during the $t$-th call to \textsc{Query}
	and by definition of $x^{(t)}$ we have for any index $i$ that
	\begin{align*}
	x^{(t)}_i = x^{\init}_i +\sum_{\ell=1}^{t}g^{(\ell)}_i \left(\sum_{k=1}^K  s^{(\ell)}_k \mathbf{1}_{i\in I_k^{(\ell)}}\right) + h^{(\ell)}_i.
	\end{align*}
	It is easy to check that $\hat{\ell}_i$ always store the most recent iteration when $\ox_i$ is updated by \textsc{ComputeX}$(i,h_i)$. We first prove by induction that this update is always calculated correct, that is, the data-structure maintains the invariant
	$\ox_i = x^{(\hat{\ell}_i)}_i$.

We see from \Cref{line:group_bin_ft} that the data-structure maintains $f^{(t)} = \sum_{k=1}^t s^{(k)}$.
	Further note that \textsc{ComputeX}$(i,h_i)$ is called 
	whenever $h_i$ is non-zero (\Cref{line:nonzero_hi}), $g_i$ is changed, or $i$ is moved to a different $I_k$. Thus if $\hat{\ell}_i < t$, we know none of these events happened during iteration $\ell\in (\hat{\ell}_i,t]$ and the only moving part is the $s^{(\ell)}$'s over these iterations, which is exactly $f^{(t)}-f^{(\hat{\ell}_i)}$.
	Thus, if $k$ is the set $I_k$ where $i$ belongs to over iterations $(\hat{\ell}_i,t]$,   the execution of \Cref{line:update_bin} gives
	\begin{align*}
	&~
	g_i \cdot (f^{(t)}_k - f^{(\hat{\ell}_i)}_k) + h^{(t)}_i
	=
	g_i \sum_{\ell=\hat{\ell}_i+1}^t s^{(\ell)}_k + h^{(t)}_i 
	=
	h^{(t)}_i + \sum_{\ell=\hat{\ell}_i+1}^t g^{(\ell)}_i s^{(\ell)}_k \\
	=&~
	h^{(t)}_i + \sum_{\ell=\hat{\ell}_i+1}^t g^{(\ell)}_i \left(\sum_{k=1}^K  s^{(\ell)}_k \mathbf{1}_{i\in I_k^{(\ell)}}\right) 
	=
	\sum_{\ell=\hat{\ell}_i+1}^t g^{(\ell)}_i \left(\sum_{k=1}^K  s^{(\ell)}_k \mathbf{1}_{i\in I_k^{(\ell)}}\right) + h^{(\ell)}_i
	\end{align*}
	where the first equality uses $ f^{(t)}=\sum_{\ell=1}^t s^{(\ell)}$ and the second equality uses $g^{(\ell)}_i=g^{(t)}_i$ for all $\hat{\ell}_i < \ell \le t$,
	because \textsc{ComputeX}$(i,h_i)$ is called whenever $g_i$ is changed.
	The third equality is because \textsc{ComputeX}$(i,h_i)$ is called whenever $i$ is moved to a different set, 
	so $i \in I_k^{(\ell)}$ for the same $k$ for all $\hat{\ell}_i < \ell \le t$.
	The last equality uses $h^{(\ell)}_i = 0$ for $\hat{\ell}_i < \ell < t$, because \textsc{ComputeX}$(i,h_i)$ is called whenever $h_i$ is non-zero.
	Thus by induction over the number of calls to \textsc{ComputeX}$(i,h_i)$, when $\hat{\ell}_i$ is increased to $t$ we have
	\begin{align*}
	\ox_i = 	
	x^{(t)}_i = x^{\init}_i +\sum_{\ell=1}^{t}g^{(\ell)}_i \left(\sum_{k=1}^K  s^{(\ell)}_k \mathbf{1}_{i\in I_k^{(\ell)}}\right) + h^{(\ell)}_i	 = x^{(t)}_i,
	\end{align*}
	so the invariant is always maintained.
	\paragraph{Correctness of Query.}
	We claim that the function \textsc{Query} returns a vector $\ox$ such that for all $i$ 
	\begin{align*}\ox_i \approx_\epsilon x^{(t)}_i := x^{\init}_i + \sum_{\ell=1}^{t}g^{(\ell)}_i \left(\sum_{k=1}^K  s^{(\ell)}_k \mathbf{1}_{i\in I_k^{(\ell)}}\right) + h^{(\ell)}_i.\end{align*}
	Given the invariant discussed above, we only need to guarantee \textsc{ComputeExact}$(i,h_i)$ is called whenever the approximation guarantee is violated for some $i$. Moreover, same as when we proved the invariant above, we only need to guarantee this in the case that since iteration $\hat{\ell}_i$, $h_i$ is always $0$, $g_i$ remains constant and $i$ remains in the same $I_k$ for some $k$. Thus, the task is equivalent to detect whenever 
	\begin{align*}
	\left| \sum_{\ell=\hat{\ell}_i+1}^{t} g^{(\ell)}_i s^{(\ell)}_k \right| = |g_i \cdot (f^{(t)}_k - f^{(\hat{\ell}_i)}_k)|> \frac{\epsilon|\ox_i|}{10}.
	\end{align*}
	which is the same as 
	\begin{align*}
	f^{(t)}_k \notin [f^{(\hat{\ell}_i)}_k - |(\epsilon x_i)/(10 g_i)|,f^{(\hat{\ell}_i)}_k + |(\epsilon x_i)/(10 g_i)|] 
	\end{align*}
	Note that the lower and upper limits in the above range are exactly $\Delta^{(inc)}_i$ and $\Delta^{(dec)}_i$	as maintained by \textsc{UpdateDelta}$(i)$, which will be called whenever any of the terms involved in the calculation of these limits changes.
	Thus \Cref{line:large_change_x_f} makes sure that we indeed maintain
	\begin{align*}\ox_i \approx_\epsilon  x^{(t)}_i \qquad \forall i.\end{align*} Also it is easy to check the returned set $J$ contains all $i$'s such that $\ox_i$ changed since last \textsc{Query}.
	
	\paragraph{Complexity}
	The call to \textsc{Initialize} takes $O(m+K)$ as we initialize a constant number of $K$ and $m$ dimensional vectors, and this reduces to $O(m)$ since there can be at most $m$ non-empty $I_k$'s.
	A call to \textsc{ComputeX} takes $O(1)$ time.
	
	To implement \Cref{line:large_change_x_f} efficiently without enumerating all $m$ indices, we maintain for each $k\in[K]$ two sorted lists of the $i$'s in $I_k$, 
	sorted by $\Delta_i^{(high)}$ and $\Delta_i^{(low)}$ respectively.
	Maintaining these sorted lists results in $O(\log n)$ time per call for \textsc{UpdateDelta}.
	Hence \textsc{Move} and \textsc{Scale} also run in $O(\log n)$ time.
	To implement the loop for \Cref{line:large_change_x_f} we can go through the two sorted lists in order, but stop as soon as the check condition no longer holds. This bounds the cost of the loop by $O(K)$ plus $O(\log n)$ times the number of indices $i$ satisfying $f^{(t)}_k > \Delta^{(high)}_i$ or $f^{(t)} < \Delta^{(low)}_i$,
	i.e. $|f^{(t)}_k - f^{(\hat{\ell}_i)}_k| > \Theta(\epsilon x^{(\hat{\ell}_i)}_i)$.
	Note if a \textsc{ComputeX} and \textsc{UpdateDelta} is triggered by this condition for any $i$, $h_i$ must be $0$ during $(\hat{\ell}_i,t]$ iterations. Thus, let $z^{(t)} := x^{\init} + \sum_{\ell=1}^t \mG^{(\ell)} \sum_k I_k^{(\ell)} s_k^{(\ell)}$, we can rewrite that condition as
	$|z^{(t)}_i - z^{(\hat{\ell}_i)}_i| > \Theta(|\epsilon x^{(\hat{\ell}_i)}_i|)$.
	Throughout $T$ calls to \textsc{Query}, 
	we can bound the total number of times where $i$ satisfies 
	$|z^{(t)}_i - z^{(\hat{\ell}_i)}_i| > \Theta(|\epsilon x^{(\hat{\ell}_i)}_i|)$
	by
	\begin{align*}
	O\left(T \sum_{\ell=1}^T \| \mG^{(\ell)}(\sum_k I_k^{(\ell)} s_k^{(\ell)}) / x^{(\ell-1)} \|_2^2 / \epsilon^2\right).
	\end{align*}
	The number of times \textsc{ComputeX} and \textsc{UpdateDelta} are triggered due to $h^{(t)}_i\neq 0$ is $\|h^{(t)}\|_0$ each iteration, and updating $f^{(t)}$ takes $O(K)$ time. So the total time for $T$ calls to \textsc{Query} can be bounded by
	\begin{align*}
	O(TK + \log^n\cdot\sum_{\ell=0}^T \|h^{(\ell)}\|_0 + \log^n \cdot T\sum_{\ell=1}^T \|\mG^{(\ell)} (\sum_k s_k^{(\ell)I_k^{(\ell)} }) / x^{(\ell-1)}\|_2^2 / \epsilon^2).
	\end{align*}
	The time for \textsc{ComputeExactSum} is $O(m\log n)$ since it just calls \textsc{ComputeX} and \textsc{UpdateDelta} on all $m$ indices.
\end{proof}
\begin{proof}[Proof of \Cref{thm:gradient_maintenance}] 
The data-structure for \Cref{thm:gradient_maintenance} follows directly by combining
\Cref{lem:gradient_reduction} and \Cref{lem:gradient_accumulator}.
The result for \textsc{QueryProduct} is obtained from \Cref{lem:gradient_reduction},
and the result for \textsc{QuerySum} is obtained from \Cref{lem:gradient_accumulator} using the vector $s \in \R^K$ returned by \Cref{lem:gradient_reduction} 
as input to \Cref{lem:gradient_accumulator}.
\end{proof}

Note that for our IPM variant using the steepest descent direction in $\ell_2$ norm, we use a simple normalized $\nabla\Phi(\oz)/\|\nabla\Phi(\oz)\|_2$ instead of $(\nabla \Phi(\oz))^{\flat(\ttau)}$. 
Here 
$$
	\nabla\Phi(\oz)/\|\nabla\Phi(\oz)\|_2 = \argmax_{\|w\|_2 \le 1} \langle \nabla\Phi(\oz), w \rangle
$$ is analogous to 
$$
	(\nabla \Phi(\oz))^{\flat(\otau)}= \argmax_{\|w\|_{\otau+\infty} \le 1} \langle \nabla\Phi(\oz), w \rangle,
$$
as these are both steepest ascent directions of $\nabla\Phi(\oz)$ but in different norms. 
It is straightforward to see how to simplify and extend the data structure of \Cref{thm:gradient_maintenance} to work for this setting. We state without proof the following variant. Note in this case we only need to round $z$ so the low dimensional representation is in $O(\epsilon^{-1})$ instead of $O(\epsilon^{-2}\log n)$ dimension. 
\begin{theorem}
\label{thm:gradient_maintenance_simple} 
There exists a deterministic data-structure
that supports the following operations
\begin{itemize}
\item $\textsc{Initialize }(\mA\in\R^{m\times n},x^{\init} \in \R^m, g\in\R^{m},z\in\R^{m},\epsilon>0)$:
	The data-structure preprocesses the given matrix $\mA\in\R^{m\times n}$,
	vectors $x^{\init},g,z\in\R^{m}$, and accuracy parameter $\epsilon>0$
	in $\tilde{O}(\nnz(\mA))$ time. 
	The data-structure assumes $0.5\le z\le2$.
\item $\textsc{Update}(i \in [m], a \in \R, c \in \R)$: 
	Sets $g_{i}\leftarrow a$ and $z_i \leftarrow c$ in $O(\|e_i^\top \mA\|_0)$ time. 
	The data-structure assumes $0.5\le z\le2$.%
\item $\textsc{QueryProduct}()$: 
	Returns $\mA^{\top}\mG\nabla\Phi(\oz)/\|\nabla\Phi(\oz)\|_2 \in \R^n$ for some $\oz \in \R^m$ 
	with $\|\oz-z\|_{\infty}\le \epsilon$.
	Every call to \textsc{QueryProduct} must be followed by a call to \textsc{QuerySum},
	and we bound their complexity together (see \textsc{QuerySum}). 
\item $\textsc{QuerySum}(h \in \R^m)$:
	Let $v^{(\ell)}$ be the vector $\mG\nabla\Phi(\oz)/\|\nabla\Phi(\oz)\|_2$ used for the result of the $\ell$-th call to \textsc{QueryProduct}.
	Let $h^{(\ell)}$ be the input vector $h$ given to the $\ell$-th call to \textsc{QuerySum}.
	We define 
	\begin{align*}  x^{(t)} := x^{\init} + \sum_{\ell=1}^{t} v^{(\ell)} + h^{(\ell)}.  \end{align*}
	Then the $t$-th call to \textsc{QuerySum} returns a vector $\ox \in \R^m$ with $\ox \approx_\epsilon x^{(t)}$.
	
	Assuming the input vector $h$ in a sparse representation (e.g. a list of non-zero entries), 
	after $T$ calls to \textsc{QuerySum} and \textsc{QueryProduct} the total time for all calls together is bounded by
	\begin{align*}O\left(
		T n \epsilon^{-1} \log (1/\epsilon)
		+ \log n \cdot \sum_{\ell=0}^T \|h^{(\ell)}\|_0 
		+ T \log n \cdot \sum_{\ell=1}^T \|v^{(\ell)}/x^{(\ell-1)}\|_2^2 / \epsilon^2
	\right).\end{align*}
	
	The output $\ox \in \R^m$ is returned in a compact representation to reduce the size. In particular, the data-structure returns a pointer to $\ox$ 
and a set $J \subset [m]$ of indices which specifies which entries of $\ox$ have changed 
	between the current and previous call to \textsc{QuerySum}.
\item $\textsc{ComputeExactSum}()$
	Returns the exact $x^{(t)}$ in $O(m\log n)$ time.
\item $\textsc{Potential}()$
	Returns $\Phi(\oz)$ in $O(1)$ time.
\end{itemize}
\end{theorem}

\newpage

\section{Minimum Weight Perfect Bipartite $b$-Matching Algorithms}
\label{sec:matching}

In this section we prove the main result: 
a nearly linear time algorithm for minimum weight perfect bipartite $b$ matching.
The exact result proven in this section is the following theorem.

\begin{restatable}{theorem}{thmFastMatching}\label{thm:perfect_b_matching_fast}
For $b,c \in \Z^m$, we can solve minimum weight perfect bipartite $b$-matching for cost vector $c$ in
$\tilde{O}((m + n^{1.5}) \log^2 (\|c\|_\infty\|b\|_\infty))$ time.
\end{restatable}

We show how to combine the tools from 
\Cref{sec:ipm,sec:matrix_vector_product,sec:vector_maintenance,sec:gradient_maintenance}
to obtain our fast $b$-matching algorithm.
In \Cref{sec:log_barrier_method} and \Cref{sec:ls_barrier_method} we present two algorithms
that each encapsulate a variant of the IPM.
These algorithms only specify which computations must be performed,
but they do not specify how they must be implemented.
So our task is to use the data structures from 
\Cref{sec:matrix_vector_product,sec:vector_maintenance,sec:gradient_maintenance}
to implement these IPMs.

This section is split into five subsections:
We first consider the simpler, 
but slower IPMs from \Cref{sec:log_barrier_method},
based on the $\log$-barrier method.
This method is formalized via \Cref{alg:short_step}. 

\Cref{line:log:R} of  \Cref{alg:short_step},
requires random sampling to sparsify the solution of some linear system. 
Consequently, our first task (\Cref{sec:matching:sampling})
is to use the data structure from \Cref{sec:matrix_vector_product}
to implement this sampling procedure.

The next task is to implement the IPM of \Cref{alg:short_step},
which is done in \Cref{sec:matching:log}.

Note that \Cref{alg:short_step} only formalized a single step of the IPM,
and the complete method is formalized by combining \Cref{alg:short_step} with
\Cref{alg:meta}.
So the next task (\Cref{sec:matching:pathfollowing}) is to implement \Cref{alg:meta}
efficiently using our data structures.

We then use our implementation of the IPM in \Cref{sec:matching:slow}
to obtain a $\tilde{O}(n\sqrt{m})$ time algorithm.
At last, in \Cref{sec:matching:fast}, 
we show which small modification must be performed to
replace the slower $\log$-barrier based IPM
with the faster leverage score based IPM of \Cref{sec:ls_barrier_method}.
This is also where we prove \Cref{thm:perfect_b_matching_fast}.

\subsection{Primal Sampling}
\label{sec:matching:sampling}

Our first task is to implement the random sampling of \Cref{line:log:R}
of \Cref{alg:short_step}.
Let $\mA \in \R^{m \times n}$, $g,g' \in \R^m_{\ge0}$, and $h \in \R^n$, then for the random sampling 
we want to sample each $i\in[m]$ with some probability $q_i$ such that
$$
q_i \ge \min\{1, \sqrt{m} ((\mG \mA h)_i^2/\|\mG \mA h\|_2^2 + 1/m) + C \cdot \sigma(\mG' \mA)_i \log(m/(\epsilon r))\gamma^{-2}\},
$$
where $C$ is some large constant and $\gamma = 1/\polylog n$.
As this sampling is performed in each iteration of the IPM,
it would be prohibitively expensive to explicitly compute the leverage scores $\sigma(\mG'\mA)$ or the vector $\mG \mA h$,
as that would require $\Omega(m)$ time.
Consequently, instead we will use the data structure of \Cref{lem:large_entry_datastructure}
to efficiently perform this sampling.

\begin{algorithm2e}[t!]
\caption{Sampling Algorithm for the IPM \label{alg:sample_primal}}
\SetKwProg{Globals}{global variables}{}{}
\SetKwProg{Proc}{procedure}{}{}
\Proc{\textsc{SamplePrimal}$(K \in \R_{>0}, D^{(\sample, \sigma)}, D^{(\sample)}, h \in \R^n)$}{
	\LineComment{$D^{(\sample, \sigma)}, D^{(\sample)}$ instances of \Cref{lem:large_entry_datastructure}}
	$I_u \leftarrow D^{(\sample, \sigma)}.\textsc{LeverageScoreSample}(3C \log(m / ( \epsilon r )) \gamma^{-2})$ \\
	$I_v \leftarrow D^{(\sample)}.\textsc{Sample}(h, 3K\sqrt{m} \cdot (16\log^8(n)))$ \\
	$I_w \subset [m]$, where $\P[i \in I_w] = 3K/\sqrt{m}$ independently for each $i$. \\
	$I \leftarrow I_u \cup I_v \cup I_w$ \\
	Let $u,v,w$ be $\zerovec_m$. \\
	\LineComment{Set $v_i = \P[i \in I_v]$, $u_i = \P[i \in I_u]$, $w_i = \P[i \in I_w]$ for $i \in I$}
	$v_I \leftarrow D^{(\sample)}.\textsc{Probability}(I, h, 3K\sqrt{m} \cdot (16\log^8(n)))$ \\
	$u_I \leftarrow D^{(\sample, \sigma)}.\textsc{LeverageScoreBound}(3C \log(m / ( \epsilon r )) \gamma^{-2}, I)$ \\
	$w_i \leftarrow 3K/\sqrt{m}$ for $i \in I$ \\
	$\mR \leftarrow \mathbf{0}_{m \times m}$ \\
	\For{$i \in I$}{
		$\mR_{i,i} \leftarrow 1/\min\{1, u_i + v_i + w_i\}$ 
			with probability $\frac{\min\{1, u_i + v_i + w_i\}}{1-(1-u_i)(1-v_i)(1-w_i)}$
			\label{line:fix_probability}
	}
	\Return $\mR$
}
\end{algorithm2e}

\begin{lemma}\label{lem:sample_primal}
Consider a call to \textsc{SamplePrimal}$(K \in \R_{>0}, D^{(\sample, \sigma)}, D^{(\sample)}, h \in \R^n)$ (\Cref{alg:sample_primal}).
Let $\mG \mA$ be the matrix represented by $D^{(\sample)}$ (\Cref{lem:large_entry_datastructure}),
and $\mG' \mA$ be the matrix represented by $D^{(\sample, \sigma)}$ (\Cref{lem:large_entry_datastructure}),
then each $\mR_{i,i}$ is $1/q_i$ with some probability $q_i$ and zero otherwise, 
where
\begin{align*}
q_i \ge \min\{1, K\sqrt{m} ((\mG \mA h)_i^2/\|\mG \mA h\|_2^2 + 1/m) + C \cdot \sigma(\mG' \mA)_i \log(m/(\epsilon r))\gamma^{-2}\}
\end{align*}
for some large constant $C > 0$.
Further, w.h.p., 
there are at most %
\begin{align*}
\tilde{O} (n \log ( m / ( \epsilon r ) ) \log W + K \sqrt{m})
\end{align*}
non-zero entries in $\mR$ which is also a bound on the complexity of \textsc{SamplePrimal}.
Here $W$ is the ratio of largest to smallest (non-zero) entry in $\mG$ and in $\mG'$.
\end{lemma}

\begin{proof}
We start by analyzing the vector $u,v,w$. 
Afterwards we analyze the distribution of the matrix $\mR$.
At last we bound the number of non-zero elements in $\mR$
and the complexity of the algorithm.
\paragraph{Distribution of vectors $u,v,w$}
By guarantee of \Cref{lem:large_entry_datastructure} 
the methods $D^{(\sample)}.\textsc{Probability}$ 
and $D^{(\sample, \sigma)}.\textsc{LeverageScoreBound}$ 
return the sampling probabilities
used by $D^{(\sample)}.\textsc{Sample}$ and $D^{(\sample, \sigma)}.\textsc{LeverageScoreSample}$.
Thus the vectors $u,v$ satisfy
\begin{align*}
u_i = \P[i \in I_u] \text{ with probability } \P[i \in I] \text{ and $0$ otherwise.} \\
v_i = \P[i \in I_v] \text{ with probability } \P[i \in I] \text{ and $0$ otherwise.} 
\end{align*}
Further by \Cref{lem:large_entry_datastructure} we have
\begin{align*}
\P[i \in I_u] \ge 3C \cdot \sigma(\mG' \mA) \log(m/(\epsilon r))\gamma^{-2}
\text{~~~and~~~}
\P[i \in I_v] \ge 3K\sqrt{m} (\mG\mA h)_i^2 / \|\mG \mA h\|_2^2.
\end{align*}
For the vector $w$ note that we have $w_i = 3K/\sqrt{m}$ for each $i \in I$,
and $w_i = 0$ otherwise.

Here the set $I$ is distributed as follows
\begin{align*}
\P[i \in I] &= \P[i \in I_u \cup Iv \cup I_w] \\
&= 1-(1-\P[i \in I_u])(1-\P[i \in I_v])(1-\P[i \in I_w]) \\
&= 1-(1-u_i)(1-v_i)(1-w_i).
\end{align*}

\paragraph{Distribution of matrix $\mR$}

Note that we have
\begin{align*}
\P[i \in I] 
&\ge \min\{1, \max\{~ \P[i \in I_u],~ \P[i \in I_v],~ \P[i \in I_w] ~\} \} \\
&\ge \min\{1,~ (\P[i \in I_u] + \P[i \in I_v] + \P[i \in I_w])/3 ~\} \\
&= \min\{1,~ (u_i + v_i + w_i)/3 ~\},
\end{align*}
thus the probability used in \Cref{line:fix_probability} is well defined (i.e. at most $1$).
Hence \Cref{line:fix_probability} ensures that 
\begin{align*}
~\P \left[ \mR_{i,i} \neq 0 \right]
=& ~ \P[i \in I] \cdot \frac{\min\{1, (u_i + v_i + w_i)/3  \}}{1-(1-u_i)(1-v_i)(1-w_i)} \\
=& \min\{1, (u_i + v_i + w_i)/3  ~\}
\end{align*}
Further, the non-zero $\mR_{i,i}$ are set to exactly the inverse of that probability
and we have
\begin{align*}
&~\min\{1, (u_i + v_i + w_i)/3 \} \\
\ge& ~ \min \Big\{1,  C \sigma(\mG'\mA)_i \log(m / ( \epsilon r )) \gamma^{-2} + K / \sqrt{m}  + K\sqrt{m} (\mG\mA h)_i^2 / \|\mG \mA h\|_2^2  \Big\}
\end{align*}
so $\mR$ has the desired distribution.

\paragraph{Sparsity of $\mR$ and complexity bound}

The number of nonzero entries of $\mR$ can w.h.p. be bounded by
\begin{align*}
& ~ \tilde{O}(n\log(m / ( \epsilon r )) \gamma^{-2} \log W)
+
O(K \sqrt{m} \log n)
+
O(K \sqrt{m} \log n) \\
= & ~
\tilde{O}\left(n \log\left(m / ( \epsilon r )\right) \log W + K \sqrt{m}\right),
\end{align*}
via \Cref{lem:large_entry_datastructure} and $C = O(1)$.

The complexity is dominated by the two sample methods of \Cref{lem:large_entry_datastructure}
whose complexity is bounded by
$$\tilde{O}\left(n \log\left(m / ( \epsilon r )\right) \log W + K \sqrt{m}\right).$$
\end{proof}

\subsection{Log Barrier Short Step}
\label{sec:matching:log}

Here we show how to efficiently implement the log barrier based IPM,
formalized by \Cref{alg:short_step}. 
Note that, \Cref{alg:short_step} does not specify how to perform the computations, 
rather it simply specifies a sufficient set of conditions for the method to perform as desired. 
Consequently, in this section we show how to implement these computations using our data structures.

To be more precise,  consider \Cref{line:log:approx_xs} of \Cref{alg:short_step},
which specifies to find some $\os \approx_\epsilon s$.
The algorithm does not explain how to find such an approximation, 
so we must implement this ourself.
Note that in \Cref{alg:short_step}, 
$s$ is defined incrementally via
$$
s^\new 
\leftarrow 
s + \gamma \mA \omH^{-1} \mA^\top \omX \nabla\Phi(\ov)/\|\nabla\Phi(\ov)\|_2
$$
for any $\omH \approx \mA^\top \omX \omS^{-1} \mA$ and $\ov \approx xs/\mu$. However, here $h' := \gamma \mA^\top \omX \nabla\Phi(\ov)/\|\nabla\Phi(\ov)\|_2$
is exactly the return value of \textsc{QueryProduct}, 
promised by the data structure of \Cref{thm:gradient_maintenance_simple}.
So we can efficiently compute $h'$ in each iteration of the IPM by using
\Cref{thm:gradient_maintenance_simple}.

Next, we must compute $h'' := \omH^{-1} h'$.
This can be done by first sparsifying the weighted Laplacian matrix $\mL := (\mA^\top \mX \mS^{-1} \mA)$
via leverage score sampling\footnote{From a graph perspective this might also be known as sampling by the effective resistances.}. 
This sampling can be performed efficiently 
via the \textsc{LeverageScoreSample} method 
of the data structure 
of \Cref{lem:large_entry_datastructure}.
After sparsifying the Laplacian $\mL$
we can apply a Laplacian solver (\Cref{lem:laplacian_solver}) to compute $h''$ efficiently.
	
At last, we must compute $s + \mA h'$.
Note however, that \Cref{alg:short_step} never requires the exact vector $s$,
and it always suffices to just know some element-wise approximation $\os \approx s$.
Hence the problem we actually need to solve
is to maintain an element-wise approximation of the sum $s = s^\init + \sum_{t=1}^T \mA h''^{(t)}$ 
(this represents the vector $s$ after $T$ steps of the IPM),
for a sequence of vectors $h''^{(1)},...,h''^{(T)} \in \R^n$.
That is exactly the guarantee of the data structure of \Cref{thm:vector_maintenance}.

The only thing left to do
is to maintain a similar approximation $\ox \approx x$ 
of the primal solution via \Cref{thm:gradient_maintenance_simple}.
Then we have an efficient implementation of \Cref{alg:short_step}.

A formal description of this implementation can be found in \Cref{alg:implement:short_step_log_barrier}
and we are left with proving that this algorithm does indeed implement all steps of \Cref{alg:short_step}.
It is difficult to analyze the complexity per call to \Cref{alg:implement:short_step_log_barrier},
as many complexity bounds of the used data structures only bound the overall time over several iterations. 
Consequently, we defer the complexity analysis to the next section,
where we analyze the overall complexity of several consecutive calls to \Cref{alg:implement:short_step_log_barrier}.
For now we only prove the correctness of \Cref{alg:implement:short_step_log_barrier} 
in \Cref{lem:implement:short_step_log_barrier}.

\begin{algorithm2e}[t!]
\caption{Implementation of \Cref{alg:short_step} \label{alg:implement:short_step_log_barrier}}
\SetKwProg{Globals}{global variables}{}{}
\SetKwProg{Proc}{procedure}{}{}
\Globals{}{
	\LineComment{Imported from \Cref{alg:implement:path_following}}
	$D^{(x,\nabla)}$ instance of \Cref{thm:gradient_maintenance_simple},
	initialized on $(\mA, x^\init, \gamma x^\init, x^\init s^\init/\mu^\init,\gamma/16)$ \label{line:step:initDx} \label{line:step:first_assumption} \\
	$D^{(s)}$ instance of \Cref{thm:vector_maintenance},
	initialized on $(\mA, s^\init, \gamma/16)$ \label{line:step:initDs} \\
	$D^{(\sample)}$ instance of \Cref{lem:large_entry_datastructure},
	initialized on $(\mA, 1/s^\init)$\\
	$D^{(\sample, \sigma)}$ instance of \Cref{lem:large_entry_datastructure},
	initialized on $(\mA, \sqrt{x^\init / s^\init})$ \\
	$\Delta \in \R^n$ maintained to be $\mA^\top x - b$ \\
	$\mu,\omu \in \R$ progress parameter $\mu$ of the IPM and its approximation $\omu \approx_{\gamma/8} \mu$\\
	$\ox,\os \in \R^m$ element-wise approximations $\ox \approx_{\gamma/8} x$ and $\os \approx_{\gamma/8} s$ \label{line:step:last_assumption} \\
	$\gamma$ accuracy parameter
}
\Proc{\textsc{ShortStep}$(\mu^{\new}>0)$}{
	\LineComment{Update $\omu$ and data structures that depend on it.}
	\If{$\omu \not\approx_{\gamma/8} \mu^\new$}{ \label{line:step:mu}
		$\omu \leftarrow \mu^\new$ \\
		\lFor{$i \in [m]$}{$D^{(x,\nabla)}.\textsc{Update}(i, \gamma\ox_i, (\ox_i\os_i)/\omu)$}
	}
	\LineComment{Leverage Score Sampling to sparsify $(\mA^\top \mX\mS^{-1} \mA)$ with $\gamma/2$ spectral approximation}
	$I_v \leftarrow D^{(\sample,\sigma)}.\textsc{LeverageScoreSample}(c \gamma^{-2} \log n)$ for some large enough constant $c > 0$ \label{line:step:sample}\\
	$v \leftarrow D^{(\sample,\sigma)}.\textsc{LeverageScoreBound}(c \gamma^{-2} \log n, I_v)$ \label{line:step:bound}\\
	$v_i \leftarrow 1/v_i$ for $i$ with $v_i \neq 0$ \label{line:step:ls_sample}\\
	\LineComment{Perform \textsc{ShortStep} (\Cref{alg:short_step})}
	$h' \leftarrow D^{(x,\nabla)}.\textsc{QueryProduct()}$ 
			\Comment{$h' = \gamma \mA^\top \omX \nabla\Phi(\ov) / \| \nabla\Phi(\ov) \|_2$} 
			\label{line:step:gradient} \\
	$h'' \leftarrow$ solve $(\mA^\top \mV \mA)^{-1} (h' + \Delta)$ with $\gamma/2$ accuracy via Laplacian Solver. 
			\label{line:step:solve}
			\Comment{$h'' = \omH^{-1} (h' + (\mA^\top x - b))$, $\delta_r = \mS^{-1} \mA h''$} \\
	$\mR \leftarrow \textsc{SamplePrimal}(1, D^{(\sample)}, D^{(\sample, \sigma)}, h'')$ \label{line:step:sample_primal}\\
	$x^{\tmp}, I_x \leftarrow D^{(x,\nabla)}.\textsc{QuerySum}(-\mR \omX \omS^{-1} \mA h'')$ 
			\label{line:step:x_tmp}\\
	$s^{\tmp}, I_s \leftarrow D^{(s)}.\textsc{Add}((\mA^\top \mV \mA)^{-1} h', \gamma/16)$ 
			\label{line:step:s_tmp}\\
	$\Delta \leftarrow \Delta + h' - \mA^\top \mR \omX \omS^{-1} \mA h''$ 
			\label{line:step:delta}
			\Comment{Maintain $\Delta = \mA^\top x - b$}\\
	\LineComment{Update $\ox, \os$ and data structures that depend on them.}
	\For{$i \in I_x \cup I_s$}{ \label{line:step:update_x_s}
		\If{$x^{\tmp}_i \not\approx_{\gamma/16} \ox_i$ 
					or $s^{\tmp}_i \not\approx_{\gamma/16} \os_i$}{ \label{line:step:update_x_s_condition}
			$\ox_i \leftarrow x^{\tmp}_i$, $\os_i \leftarrow s^{\tmp}_i$ \\
			$D^{(x,\nabla)}.\textsc{Update}(i, \gamma\ox_i, (\ox_i\os_i)/\omu)$ \\
			$D^{(\sample)}.\textsc{Scale}(i, 1/\os_i)$ \\
			$D^{(\sample, \sigma)}.\textsc{Scale}(i, \sqrt{\ox_i/\os_i})$ \\
		}
	}
}
\end{algorithm2e}

\begin{lemma}\label{lem:implement:short_step_log_barrier}
\Cref{alg:implement:short_step_log_barrier} implements \Cref{alg:short_step},
assuming the data structures are initialized as stated in \Cref{alg:implement:short_step_log_barrier}. %
\end{lemma}

\begin{proof}
As \Cref{alg:implement:short_step_log_barrier} uses many different data structures
that share similar variable names 
(e.g variable $g$ occurs in \Cref{lem:large_entry_datastructure,%
thm:vector_maintenance,%
thm:gradient_maintenance_simple}),
we will use the notation $D.var$
to refer to variable $var$ of data structure $D$.
For example $D^{(\sample)}.g$ refers to variable $g$
of data structure $D^{(\sample)}$ 
(which is an instance of \Cref{lem:large_entry_datastructure}).

\paragraph{Assumptions}

By \Cref{line:step:first_assumption} to \Cref{line:step:last_assumption}
of \Cref{alg:implement:short_step_log_barrier},
we assume that the following assumptions hold true
at the start of the first call to \textsc{ShortStep}
(\Cref{alg:implement:short_step_log_barrier}).
\begin{align}
\ox \approx_{\gamma/8} x,~
\os &\approx_{\gamma/8} s,~
\omu \approx_{\gamma/8} \mu \label{eq:step:approx_xs}\\
\Delta &= \mA^\top x - b \label{eq:step:Delta}\\
D^{(x,\nabla)}.g = \gamma \ox,~
D^{(x,\nabla)}.z &= \gamma \ox\os/(\omu),~
D^{(x,\nabla)}.\epsilon = \gamma /16 \label{eq:step:Dxnabla}\\
D^{(sample, \sigma)}.g &= \sqrt{\ox\os^{-1}} \label{eq:step:Dsamplesigma} \\
D^{(sample)}.g &= \os^{-1} \label{eq:step:Dsample}
\end{align}
We prove further below that these assumptions also hold true for all further
calls to \textsc{ShortStep} (\Cref{alg:implement:short_step_log_barrier}).
However, first we prove that,
if these assumptions are satisfied,
then \textsc{ShortStep} (\Cref{alg:implement:short_step_log_barrier})
performs the computations required by the IPM of \Cref{alg:short_step}.

\paragraph{\Cref{alg:implement:short_step_log_barrier} implements \Cref{alg:short_step}}

We argue the correctness line by line of \Cref{alg:short_step}.
\Cref{line:log:approx_xs} (\Cref{alg:short_step})
is satisfied by assumption \eqref{eq:step:approx_xs}.
\Cref{line:log:v} and \Cref{line:log:g} (\Cref{alg:short_step})
asks us to compute
$$
g \leftarrow - \gamma\nabla\Phi(\ov)/\|\nabla\Phi(\ov)\|
$$
for $\ov \approx_\gamma xs/\mu$.
While we do not compute $g$,
our implementation does compute $\mA \omX g$ as follows:
\Cref{line:step:gradient} of \Cref{alg:implement:short_step_log_barrier}
computes
\begin{align*}
h' 
= 
\gamma \mA^\top \omX \nabla\Phi(\ov)/\|\nabla\Phi(\ov)\|_2
= 
-\mA^\top \omX g
\end{align*}
for some $\ov \approx_{\gamma/2} D^{(x,\nabla)}.z = \ox\os/(\omu) \approx_{\gamma/2} xs/\mu$
by the guarantees of \Cref{thm:gradient_maintenance_simple}
and assumptions \eqref{eq:step:Dxnabla}.

By \Cref{line:log:define_shorthand} and \Cref{line:logstep:omh}
of \Cref{alg:short_step}
we must obtain a matrix $\omH \approx_\gamma \mA^\top \omX \omS^{-1} \mA$.
In our implementation \Cref{alg:implement:short_step_log_barrier}
this is done in \Cref{line:step:ls_sample},
where a vector $v$ is constructed with
\begin{align*}
v_i = 1/p_i \text{ with probability } p_i \ge \min\{1, c \cdot \sigma(\omX^{1/2}\omS^{1/2} \mA)_i \log n\} \text{ and } v_i = 0 \text{ otherwise.}
\end{align*}
by the guarantee of \Cref{lem:large_entry_datastructure} and assumption \eqref{eq:step:Dsamplesigma}.
Thus with high probability $(\mA^\top \mV \mA) \approx_{\gamma/2} (\mA \omX\omS^{-1} \mA)$,
so by applying a $\gamma/2$-accurate Laplacian Solver (\Cref{lem:laplacian_solver})
(i.e. in \Cref{line:step:solve}) we are able to represents some 
\begin{align*}
\omH^{-1} \approx_{\gamma} (\mA^\top \omX\omS^{-1} \mA)^{-1}
\end{align*}

For $\omA := \omX^{1/2}\omS^{1/2} \mA$, $\mW := \mu^{-1}\omX\omS$,
\Cref{line:log:delta_r} of \Cref{alg:short_step}
wants us to compute
\begin{align*}
\delta_r
=&~
\omW^{-1/2} \omA\omH^{-1}\omA^\top\omW^{1/2} g \\
&~+
\mu^{-1/2}\omW^{-1/2}\omA\omH^{-1}(\mA^\top x - b)
\end{align*}
This is done implicitly in \Cref{line:step:solve} 
of our implementation \Cref{alg:implement:short_step_log_barrier}.
By assumption \eqref{eq:step:Delta} we have $\Delta = \mA^\top x - b$,
thus \Cref{line:step:solve} computes $h''$ with 
\begin{align*}
h'' = \omH^{-1}(h' + \mA^\top x-b),
\end{align*} 
so 
\begin{align*}
\omS^{-1}\mA h''
=&~
\omS^{-1}\mA \omH^{-1}(h' + \mA^\top x-b) \\
=&~
\omS^{-1}\mA \omH^{-1}(\gamma \mA \omX \nabla\Phi(\ov)/\|\nabla \Phi(\ov)\| + \mA^\top x-b) \\
=&~ 
(\omX\omS/\mu)^{-1/2}(\omX^{1/2}\omS^{-1/2}\mA)\omH^{-1}(\omX^{1/2}\omS^{-1/2}\mA)^\top (\omX \omS/\mu)^{1/2} \nabla\Phi(\ov)/\|\nabla \Phi(\ov)\| \\
&~+
\mu^{-1/2}(\omX\omS/\mu)^{-1/2}(\omX^{1/2}\omS^{-1/2}\mA)\omH^{-1}(\mA^\top x-b) \\
=&~
\mW^{-1/2}\omA\omH^{-1}\omA^\top \mW^{1/2} \nabla\Phi(\ov)/\|\nabla \Phi(\ov)\| \\
&~
+\mu^{-1/2}\mW^{-1/2}\omA\omH^{-1}(\mA^\top x-b) \\
=:&~
\delta_r.
\end{align*}
Thus we have access to $\delta_r$ via $\omS^{-1}\mA h''$.

\Cref{line:log:R} of \Cref{alg:short_step} wants us to construct
a random diagonal matrix $\mR$ where $\mR_{i,i} = 1/q_i$ with probability $q_i$ and $\mR_{i,i} = 0$ otherwise,
where
\begin{align*}
q_i \ge \min\{1, 
	\sqrt{m} ((\delta_r)_i^2/\|\delta_r\|_2^2 + 1/m) + C \cdot \sigma_i(\omX^{1/2}\omS^{-1/2}\mA) \log(m/(\epsilon r))\gamma^{-2}\}
\end{align*}
By assumption \eqref{eq:step:Dsample} and \eqref{eq:step:Dsamplesigma} 
we have $D^{(\sample)}.g = \os^{-1}$ and $D^{(\sample,\sigma)}.g =\sqrt{\ox\os^{-1}}$,
so by \Cref{line:step:sample_primal} 
\Cref{line:step:sample_primal} of \Cref{alg:implement:short_step_log_barrier} 
returns $\mR$ with the desired properties.

\Cref{line:log:delta_x} and \Cref{line:log:move_xs} of \Cref{alg:short_step} wants us to compute
$$x^\new \leftarrow x + \omX(g - \omR\delta_r).$$
By \Cref{thm:gradient_maintenance_simple}
and assumption \eqref{eq:step:Dxnabla}
\Cref{line:step:x_tmp} of our implementation \Cref{alg:implement:short_step_log_barrier}
computes $x^\tmp$ with
\begin{align*}
x^\tmp
\approx_{\gamma/16}
x + \omX g - \mR\omX\omS^{-1}\mA h'' 
=
x + \omX (g - \omR\delta_r).
\end{align*}

\Cref{line:log:delta_s} and \Cref{line:log:move_xs} of \Cref{alg:short_step} asks us to compute
\begin{align*}
s^\new \leftarrow s + \omS \delta_p \\
\delta_p := \omW^{-1/2}\omA\omH^{-1}\omA^\top\omW^{1/2}g
\end{align*}
By \Cref{thm:vector_maintenance}
and assumption \eqref{eq:step:approx_xs}
\Cref{line:step:s_tmp} of our implementation \Cref{alg:implement:short_step_log_barrier}
computes $s^\tmp$ with
\begin{align*}
s^\tmp
&\approx_{\gamma/16}
s + \mA \omH^{-1} h'
=
s + \mA \omH^{-1} \mA^\top \omX g \\
&=
s + \omS (\omX\omS/\mu)^{-1/2} (\omX^{1/2}\omS^{-1/2} \mA) \omH^{-1} (\omX^{1/2}\omS^{-1/2} \mA)^\top (\omX\omS/\mu)^{1/2} g \\
&= 
s + \omS~ \omW^{-1/2}\omA\omH^{-1}\omA^\top\omW^{1/2}g
= 
s + \omS \delta_p
\end{align*}

\paragraph{Assumption on $\ox$, $\os$, $\omu$:}
Previously we assumed $\ox \approx_{\gamma/8} x$, $\os \approx_{\gamma/8} s$, and $\omu \approx_{\gamma/16}$.
Here we show that this assumption is true by induction, that is, 
if it was true before calling \textsc{ShortStep}, 
then it is also true after executing \textsc{ShortStep}.

We previously argued that $x^{\tmp} \approx_{\gamma/16} x^{\new}$ 
and $s^{\tmp} \approx_{\gamma/16} s^{\new}$.
So by the FOR-loop of \Cref{line:step:update_x_s} 
we have $\ox \approx_{\gamma/16} x^{\tmp}$, 
and thus $\ox \approx_{\gamma/8} x^{\new}$.
Likewise $\os \approx_{\gamma/8} s^{\new}$.

The assumption $\omu \approx_{\gamma/16}$
is true by \Cref{line:step:mu}.

\paragraph{Assumption on $D^{(\sample,\sigma)}$, $D^{(\sample)}$, and $D^{(x,\nabla)}$:}

All data structures that depend on $\ox$ and $\os$ are updated, 
whenever an entry of $\ox$ or $\os$ changes 
(see \Cref{line:step:update_x_s}).
So the assumption that $D^{(\sample,\sigma)}.g = \sqrt{\ox/\os}$,
$D^{(\sample)}.g = 1/\os$, 
$D^{(x,\nabla)}.z = \ox\oz/\omu$ 
and $D^{(x,\nabla)}.g = \gamma \ox$,
are all still satisfied after the execution of \textsc{ShortStep}.

Likewise, whenever $\omu$ changes
we set $D^{(x,\nabla)}.z = \ox\oz/\omu$.

\paragraph{Assumption $\Delta = \mA^\top x - b$:}
If $\Delta = \mA^\top x - b$ initially, 
then after \Cref{line:step:delta} we have $\Delta = \mA^\top x^{\new} - b$ 
for $x^{\new} = x + \omX g - \mR\omX\omS^{-1}\mA h''$.
Thus we always maintain $\mA^\top x - b$, whenever $x$ changes.

\end{proof}

\subsection{Path Following Algorithm}
\label{sec:matching:pathfollowing}

Here we implement the path following procedure \Cref{alg:meta} using our data structures.
Note that \Cref{alg:meta} consists of essentially a single FOR-loop,
which calls \textsc{ShortStep} (\Cref{alg:implement:short_step_log_barrier}),
so the main task for the implementation
is just the initialization of the data structures
used in \Cref{alg:implement:short_step_log_barrier}.
\Cref{lem:complexity_pathfollowing} then analyzes the complexity of
\Cref{alg:implement:path_following}, i.e. our implementation of \Cref{alg:meta}.

\begin{algorithm2e}[t!]
\caption{Implementation of \Cref{alg:meta} \label{alg:implement:path_following}}
\SetKwProg{Globals}{global variables}{}{}
\SetKwProg{Proc}{procedure}{}{}
\Globals{}{
	$D^{(x,\nabla)}$ instance of \Cref{thm:gradient_maintenance_simple} \\
	$D^{(s)}$ instance of \Cref{thm:vector_maintenance} \\
	$D^{(\sample)}$ instance of \Cref{lem:large_entry_datastructure} \\
	$D^{(\sample, \sigma)}$ instance of \Cref{thm:leverage_score_maintenance} \\
	$\mA \in \R^{m\times n}$, $b \in \R^n$ linear constraints $\mA^\top x = b$ \\
	$\Delta \in \R^n$ maintained to be $\mA^\top x - b$ \\
	$\mu,\omu \in \R$ progress parameter $\mu$ of the IPM and its approximation $\omu$\\
	$\ox,\os \in \R^m$ element-wise approximations of $x$ and $s$ \\
	$\epsilon, \lambda, \gamma, r \in \R_+$ accuracy parameters
}
\Proc{\textsc{PathFollowing}$(x^{\init}, s^{\init}, \mu^{\init}, \mu^{\target})$}{
	$\epsilon \leftarrow \frac{\log n}{80}$, $\lambda \leftarrow \frac{36 \log(240m/\epsilon^2)}{\epsilon}$, $\gamma \leftarrow \frac{\epsilon}{100 \lambda}$, $r \leftarrow \frac{\epsilon\gamma}{\sqrt{m}}$\\
	$\Delta = \mA^\top x^{\init} - b$ \\
	$\mu \leftarrow \mu^{\init}$, $\omu \leftarrow \mu^{\init}$, %
	$\ox \leftarrow x^{\init}$, $\os \leftarrow \os^{\init}$, $\otau \leftarrow \tau$ \\
	\LineComment{Initialize data structures}
	$D^{(x,\nabla)}.\textsc{Initialize}(\mA, x^{\init}, x^{\init}, x^{\init} s^{\init} / \mu^{\init}, \gamma/16)$ \\
	$D^{(s)}.\textsc{Initialize}(\mA, s^{\init}, \gamma/16)$ \\
	$D^{(\sample)}.\textsc{Initialize}(\mA, 1/s^{\init})$ \\
	$D^{(\sample, \sigma)}.\textsc{Initialize}(\mA, \sqrt{x^{\init} / s^{\init}})$ \\
	\LineComment{Follow the Central Path}
	\While{$\mu \neq \mu^{\target}$}{
		\For{$i=1,2,...,\frac{\epsilon}{r}$}{ \label{line:following:for_loop}
			$\mu \leftarrow \median(\mu^{\target}, (1-r)\mu, (1+r)\mu)$ \\
			$\textsc{ShortStep}(\mu)$ \\
			$p_1 \leftarrow D^{(x, \nabla)}.\textsc{Potential}()$ \\
			$p_2 \leftarrow \Delta^\top (\mA^\top \mG \mA)^{-1} \Delta$, 
					where $g$ is a leverage score sample from $D^{(\sample, \sigma)}$. \\
			\lIf{$p_1 > \exp(\frac{3\lambda\epsilon}{4}-\frac{1}{3})$ or $p_2 > \gamma\epsilon\sqrt{\mu}\exp(-1/3)$}{
				\textbf{break}%
			}
		}
		\If{$p_1 > \exp(\frac{\lambda\epsilon}{4}-\frac{1}{3})$ or $p_2 > \gamma\epsilon\sqrt{\mu}\exp(-1/3)$}{
			Revert iterations of previous for loop
		}
	}
	\LineComment{Compute and return exact result}
	$x \leftarrow D^{(x, \nabla)}.\textsc{ComputeExactSum}()$, $s \leftarrow D^{(s)}.\textsc{ComputeExact}()$ \\
	\Return $x, s$
}
\end{algorithm2e}

\begin{lemma}\label{lem:complexity_pathfollowing}
The time complexity of \Cref{alg:implement:path_following} is
\begin{align*}
\tilde{O}\left(n\sqrt{m} \cdot |\log \mu^{\target} / \mu^{\init} |\cdot (\log W + |\log \mu^{\target} / \mu^{\init} | )\right),
\end{align*}
where $W$ is an upper bound on the ratio of largest to smallest entry in both $x^{\init}$ and $s^{\init}$.
\end{lemma}

\begin{proof}
We analyze the overall complexity by first bounding the time for the initialization.
Then we bound the total time spend on the most inner FOR-loop,
which also includes all calls to \textsc{ShortStep} (\Cref{alg:implement:short_step_log_barrier}).

\paragraph{Initialization}

\Cref{alg:implement:path_following} starts by computing 
$\Delta = \mA^\top x - b$, 
which takes $\tilde{O}(m)$ time.
The subsequent initializations of the data structures of 
\Cref{thm:gradient_maintenance,%
thm:vector_maintenance,%
lem:large_entry_datastructure},
require $\tilde{O}(m)$ as well.

\paragraph{Last line}

In the last line of \Cref{alg:implement:path_following},
we compute the exact $x$ and $s$ values,
which by \Cref{thm:gradient_maintenance_simple} and \Cref{thm:vector_maintenance}
require $\tilde{O}(m)$ time.

\paragraph{FOR-loop}

The most inner loop (\Cref{line:following:for_loop}) 
is repeated a total of 
\begin{align*}
K := O(r^{-1} |\log \mu^{\target} / \mu^{\init} |) = \tilde{O}(\sqrt{m} |\log \mu^{\target}  / \mu^{\init} |)
\end{align*}
times by \Cref{lem:path_following_master}.
We now analyze the cost per iteration.

The cost for $D^{(x,\nabla)}.\textsc{Potential}()$ is bounded by $O(1)$.
Computing $\Delta^\top (\mA^\top \mG \mA)^{-1} \Delta$ can be implemented via Laplacian Solver (\Cref{lem:laplacian_solver}),
so the complexity is bounded by obtaining the leverage score sample $g$, 
which also bound the number of nonzero entries.
By \Cref{lem:large_entry_datastructure} the complexity is bounded by $\tilde{O}(n \log W')$,
where $W'$ is the ratio of largest to smallest entry in $\sqrt{x/s}$.
By \Cref{lem:bitlength} we can bound $\log W'$ by $\Ot(\log W + |\log \mu^{\target}/\mu^{\init}|)$.
Thus computing $\Delta^\top (\mA^\top \mV \mA)^{-1} \Delta$ 
requires $\tilde{O}(n (\log W + |\log \mu^{\target}/\mu^{\init}|))$ time.

We are left with analyzing the cost of \textsc{ShortStep} (\Cref{alg:implement:short_step_log_barrier}).

\paragraph{Updating the data structures}

Performing $D^{(x,\nabla)}.\textsc{Update}$ for all $i \in [m]$ in \Cref{line:step:mu} of \Cref{alg:implement:short_step_log_barrier}
takes $O(m)$ time, but it happens only once every $O(r/\gamma) = \tilde{O}(\sqrt{m})$ calls to \textsc{ShortStep},
because $\mu$ changes by an $(1 \pm r)$ factor between any two calls to \textsc{ShortStep}.
The total cost of this call to $D^{(x,\nabla)}.\textsc{Update}$ over all calls to \textsc{ShortStep} together can thus be bounded by $\tilde{O}(\sqrt{m} K)$.

Next consider the calls to \textsc{Update} and \textsc{Scale} of our various data structures in
\Cref{line:step:update_x_s_condition} of \Cref{alg:implement:short_step_log_barrier}.
Note that throughout $K$ calls to \textsc{ShortStep}, 
the entries of $s$ can change at most $O(K^2/\gamma^2)$ times by a factor of $(1\pm\gamma/16)$,
because $\|\mS^{-1} \delta_1\|_2^2 < 1/2$ by \Cref{lem:projection_sizes}.
On the other hand, the entries of $x$ can change at most $\tilde{O}(Kn(\log(W) +|\log(\mu^{\init}/\mu^{\target})|) + K^2/\gamma^2)$ 
times by a factor of $(1\pm\gamma/16)$. 
This is because $x$ changes by $\gamma \omX \nabla \Phi(\ov)/\|\nabla\Phi(\ov)\|_2 - \mR\omX\omS^{-1} \mA h''$,
where for the first term we have $\|\mX^{-1} \gamma \omX \nabla \Phi(\ov)/\|\nabla\Phi(\ov)\|_2 \|_2 < 2\gamma$ %
and the second term has at most $\tilde{O}(n(\log(W) +|\log(\mu^{\init}/\mu^{\target})|))$ non-zero entries.

So the total number of times the condition in \Cref{line:step:update_x_s_condition} is satisfied
can be bounded by $\tilde{O}(Kn(\log(W) +|\log(\mu^{\init}/\mu^{\target})|) + K^2/\gamma^2)$.
Thus the cost for all calls to \textsc{Scale} and \textsc{Update} of the data structures is bounded by
$\tilde{O}(K^2 + Kn(\log(W) +|\log(\mu^{\init}/\mu^{\target})|))$, 
as each call to \textsc{Scale} and \textsc{Update} of any of the used data structures costs $\tilde{O}(1)$,
and $\gamma = \tilde\Omega(1)$.

\paragraph{$D^{(\sample)}$ and $D^{(\sample, \sigma)}$ data structure}

As argued before, 
the ratio of largest to smallest weight of $x$ and $s$ can be bounded via \Cref{lem:bitlength}.
Hence performing the leverage score sample to obtain $v$ in \Cref{line:step:sample} to \Cref{line:step:ls_sample}
of \Cref{alg:implement:short_step_log_barrier}
requires at most $\tilde{O}(n (\log(W) + |\log(\mu^{\init}/\mu^{\target})|))$ time
by \Cref{lem:large_entry_datastructure}.
Likewise, each call to $\textsc{SamplePrimal}$ in \Cref{line:step:sample_primal} of \Cref{alg:implement:short_step_log_barrier} requires at most
$\tilde{O}(n (\log(W) + |\log(\mu^{\init}/\mu^{\target})|))$ time
by \Cref{lem:sample_primal},
which is also a bound on the number of non-zero entries in matrix $\mR$.
This results in a total of 
$\tilde{O}(n K (\log(W) + |\log(\mu^{\init}/\mu^{\target})|))$ time
after all $K$ calls to \textsc{ShortStep}.

\paragraph{$D^{(x,\nabla)}$ data structure}

The total cost of $D^{(x.\nabla)}.\textsc{QueryProduct}$ 
and $\textsc{QuerySum}$ is bounded by
$$
\tilde{O}(Kn + Kn (\log(W) + |\log(\mu^{\init}/\mu^{\target})|) + K^2)
$$
by \Cref{thm:gradient_maintenance_simple}.
This is because (i) the number of non-zero entries in $-\mR\omX\omS^{-1}\mA h''$ is bounded by
$\tilde{O}(Kn (\log(W) + |\log(\mu^{\init}/\mu^{\target})|))$,
and (ii) the vector $v^{(\ell)}$ and $x^{(\ell-1)}$ in \Cref{thm:gradient_maintenance_simple} 
are exactly the vectors 
$\gamma \omX \nabla \Phi(\ov)/\|\nabla \Phi(\ov)\|_2$ 
and $x$ in \Cref{alg:implement:short_step_log_barrier} 
during the $\ell$-th call to \textsc{ShortStep}.
By $\ox \approx_{\gamma/8} x$ we have $\|\mX^{-1} \gamma \omX \nabla \Phi(\ov)/\|\nabla \Phi(\ov)\|_2 \|_2 = O(\gamma)$ which yields the complexity bound above.

\paragraph{$D^{(s)}$ data structure}

The total of all calls to $D^{(s)}.\textsc{Add}$ require at most 
\begin{align*}
&~ \tilde{O}(K^2 \|\mS^{-1} \delta_s \|_2^2 + K n (\log W + |\log(\mu^{\init}/\mu^{\target})|))\\
=&~
\tilde{O}(K^2 + K n (\log W + |\log(\mu^{\init}/\mu^{\target})|))
\end{align*}
by \Cref{thm:vector_maintenance} 
and $\|\mS^{-1} \delta_s \|_2 = O(1)$ 
by \Cref{lem:projection_sizes}.

\paragraph{Updating $\Delta$}

Updating $\Delta$ in \Cref{line:step:delta} takes $O(n + \|\mR\|_0) = \tilde{O}(n (\log(W) +|\log(\mu^{\init}/\mu^{\target})|))$ time,
where we bound $\|\mR\|_0$ via \Cref{lem:sample_primal}.
The same bound holds for computing $\mR\omX\omS^{-1} \mA h''$ in \Cref{line:step:x_tmp}.
After all calls to \textsc{ShortStep} this contributes a total cost of $\tilde{O}(Kn(\log(W) +|\log(\mu^{\init}/\mu^{\target})|))$.

\paragraph{Total complexity}

In summary, the time for all calls to \textsc{ShortStep} together is bounded by 
\begin{align*}
& ~ \tilde{O}(K^2 + K n (\log W + |\log(\mu^{\init}/\mu^{\target})|) \\
= & ~ \tilde{O}(n\sqrt{m} \cdot |\log ( {\mu^{\target}} / {\mu^{\init}} ) | \cdot (\log W + |\log ( {\mu^{\target}} / {\mu^{\init}} ) |)).
\end{align*}
This also subsumes the initialization cost.
\end{proof}

\subsection{$\tilde O(n\sqrt{m})$-Time Algorithm}
\label{sec:matching:slow}

In this section we use the implemented IPM (\Cref{alg:implement:path_following})
to obtain an algorithm that solves minimum weight bipartite $b$-matching 
in $\tilde{O}((n\sqrt{m}) \log^{2} W)$ time.
More accurately, we will prove the following theorem.

\begin{theorem}\label{thm:perfect_b_matching_slow}
\Cref{alg:perfect_b_matching} finds a minimum weight perfect bipartite $b$-matching.
On a graph with $n$ nodes, $m$ edges, and edge cost vector $c \in \Z^m$,
the algorithm runs in $\tilde{O}(n\sqrt{m} \log^2 (\|c\|_\infty\|b\|_\infty))$ time.
\end{theorem}

\begin{algorithm2e}[t!]
\caption{Minimum cost perfect $b$-matching algorithm \label{alg:perfect_b_matching}}
\SetKwProg{Globals}{global variables}{}{}
\SetKwProg{Proc}{procedure}{}{}
\Globals{}{

}
\Proc{\textsc{MinimumCostPerfectBMatching}$(G=(U,V,E), b \in \Z^{U \cup V}, c \in \Z^E)$}{
	\LineComment{Given minimum cost perfect $b$-matching instance on bipartite graph $G$ with cost vector $c$.}
	Construct $G' = (U \cup V \cup \{z\}, E')$, $x$, $s$, $c'$, $d'$ as in \Cref{lem:initial_point_simple}. \\
	Let $c'' \in \R^{E'}$ with $c''_{u,v} = c_{u,v}$ and $c''_{u,z} = c''_{z,v} = \|b\|_1\|c\|_\infty$ for all $u \in U$, $v \in V$. \\
	Modify cost vector $c''$ as in \Cref{lem:isolation_lemma} to create unique optimum. \\
	Let $\mA$ be the incidence matrix of $G'$ and rename $d'$ to $b$. \\
	\LineComment{Solve the linear program $\min c''^\top x$ for $\mA^\top x = b$ and $x \ge 0$.}
	$x, s \leftarrow \textsc{PathFollowing}(x, s, 1,~ 45 \|c'\|_\infty\|b\|_\infty)$ 
			\label{line:matching:follow_1}\\
	$s \leftarrow s + c'' - c'$ \label{line:matching:switch}\\
	$x, s \leftarrow \textsc{PathFollowing}(x, s, 45 \|c'\|_\infty\|b\|_\infty,~ \frac{1}{24n\cdot 12 m^{2}(\|b\|_\infty\|c''\|_\infty)^{3}})$ 
			\label{line:matching:follow_2}\\
	$x \leftarrow $ feasible solution via \Cref{lem:make_flow_feasible}. \\
	\LineComment{$x$ is now a feasible flow with $c''^\top x \le \OPT + \frac{1}{12m^{2}(\|b\|_\infty\|c''\|_\infty)^{3}}$.}
	Round entries of $x$ to the nearest integer. \\
	\Return $x$
}
\end{algorithm2e}

\paragraph{Outline of \Cref{thm:perfect_b_matching_slow}}

\Cref{alg:perfect_b_matching} can be outlined as follows:
first we reduce the problem to uncapacitated minimum cost flows,
so that we can represent the problem as a linear program of the form
\begin{align*}
\min_{ \mA^\top x = b, x \ge 0} c^\top x,
\end{align*}
where $\mA \in \R^{m \times n}$ is an incidence matrix 
and $b \in \R^n$ are is a demand vector.

We must then construct an initial feasible flow $x \in \R^m$ and
a feasible slack $s \in \R^m$ of a dual solution,
such that $x_i s_i \approx \mu$ for all $i\in[m]$ and some $\mu \in \R$.
Unfortunately, we are not able to construct such an initial pair $(x,s)$
for the given edge cost vector $c \in \R^m$.
So instead, we construct an initial pairs 
for some other cost vector $c' \in \R^m$ (\Cref{lem:initial_point_simple}).

We then use the IPM (\Cref{alg:implement:path_following})
to move this initial solution \emph{away} from the optimal solution towards the center of the polytope.
Once our intermediate solution pair $(x,s)$ is sufficiently close to the center,
then it does no longer matter which cost vector was used.%
So we are then able to switch the cost vector $c'$ with the actual cost vector $c$
(\Cref{lem:switch_cost}).
After switching the cost vector, 
we then use the IPM (\Cref{alg:implement:path_following}) again 
to move our solution pair $(x,s)$ towards the new optimal solution induced by $c$.

Note that the obtained primal solution $x$ is not feasible.
\Cref{lem:path_following_master} only guarantees that $x$ is close to a feasible solution.
So we then use \Cref{lem:make_flow_feasible} to turn $x$ 
into a nearby feasible solution $x'$.

The obtained solution $x'$ is an almost optimal \emph{fractional} solution.
Via the Isolation-Lemma \ref{lem:isolation_lemma},
we can simply round each entry of $x'$ to the nearest integer to obtain an optimal integral solution.

\paragraph{Reduction to uncapacitated minimum flow}

We start the formal proof of \Cref{thm:perfect_b_matching_slow}
by proving the reduction of minimum-cost bipartite $b$-matching 
to uncapacitated minimum cost flows.
Given a bipartite graph $G=(U,V,E)$ and vector $b \in \N^V$,
we construct the following flow instance.

\begin{definition}\label{def:starred_flow_graph}
Let $G=(U,V,E)$ be a bipartite graph
and $b \in \Z^{U \cup  V}$ a vector with $\sum_{u \in U} b_u - \sum_{v \in V} b_v = 0$
(i.e. a $b$-matching instance).

We define the corresponding \emph{starred flow graph} as $G' = (V', E')$ 
with nodes $V' = V \cup \{z\}$,
and directed edges $E' = E \cup \{ (u, z), (z,v) \mid u \in U, v \in V\}$
and demand vector $d' \in \R^{V'}$ with 
$d_v = b_v$ for $v \in V$, 
$d_u = - b_u$ for $u \in U$ 
and $d_z =0$.
\end{definition}

The reduction for minimum cost $b$-matching to uncapacitated minimum cost flow is now trivially given.
Given a cost vector $c \in \R^E$ for the bipartite $b$-matching instance,
we can define a cost vector $c' \in \R^{E'}$ for the starred flow graph $G'$,
where $c'_{u,v} = c_{u,v}$, $c'_{u,z} = c'_{z,v} = \|b\|_1\|c\|_\infty$ for all $u\in U$ and $v\in V$.
Then any integral minimum cost flow on $G'$ with cost less than $\|b\|_1\|c\|_\infty$ for cost vector $c'$
results in a minimum weight perfect $b$ matching on $G$.
On the other hand, if no such flow exists, then $G$ has no perfect $b$-matching.

\paragraph{Initial point}

Next, we use the following \Cref{lem:initial_point_simple} 
to construct an initial point $(x,s)$ for the IPM \Cref{alg:implement:path_following}
with $x_i s_i = \mu$ for some $\mu \in \R$ and all $i \in [m]$.
\Cref{lem:initial_point_simple} is proven in \Cref{sec:initial_point:initial_point}.

\begin{restatable}{lemma}{lemInitialPointSimple}\label{lem:initial_point_simple}
We are given a minimum weight perfect $b$-matching instance 
on graph $G=(U,V,E)$ with cost vector $c \in \R^E$.

Let $G' = (V', E')$ with demand vector $d'$ be the corresponding starred flow graph
(see \Cref{def:starred_flow_graph})
and let $n$ be the number of nodes of $G'$ and $m$ the number of edges.

For $\tau(x,s) := 1$,
we can construct in $O(m)$ time a cost vector $c' \in \R^{E'}$ for $G'$
and a feasible primal dual solution pair $(x,s)$ for the minimum cost flow problem on $G'$ with cost vector $c'$,
where the solution satisfies $xs = \tau(x,s)$.
The cost vector $c'$ satisfies
\begin{align*}
\|b\|_\infty^{-1} \le c' \le n.
\end{align*}

\end{restatable}

\paragraph{Switching the cost vector}

\Cref{lem:initial_point_simple} returns some cost vector $c'$ that differs from $c$.
As outlined before and as can be seen in \Cref{line:matching:switch} of \Cref{alg:perfect_b_matching},
we will switch the cost vectors once our point $(x,s)$ 
is far enough away from the optimum.
\Cref{lem:switch_cost}, which formally shows that this switching is possible,
is proven in \Cref{sec:initial_point:switch}.

\begin{restatable}{lemma}{lemSwitchCost}\label{lem:switch_cost}
Consider the function $\tau(x,s)$ and a uncapacitated min-cost-flow instance 
as constructed in \Cref{lem:initial_point} or \Cref{lem:initial_point_simple}
and let $\mA$ be the incidence matrix of the underlying graph $G'$.
Let $c$ be the cost vector constructed in \Cref{lem:initial_point}
and let $c'$ be any other cost vector with $\|c'\|_\infty \ge 1$.

Assume we have a primal dual solution pair $(x,s)$ 
for cost vector $c$ and demand $d$
with $\|\mA x - d \|_{(\mA \mX \mS^{-1} \mA)^{-1}} \le \sqrt{\mu}/4$
and $xs \approx_{\epsilon/3} \mu \tau(x,s)$ 
for any $\epsilon \le 1$ and $\mu \ge 45 \|c'\|_\infty \|d\|_\infty \epsilon$.

If we replace the cost vector $c$ with cost vector $c'$
then $(x,s+c'-c)$ is a primal dual solution pair 
for the new min-cost-flow instance with cost vector $c'$
and we have
$x(s+c'-c) \approx_{\epsilon/2} \mu \tau(x, s+c'-c)$.

\end{restatable}

\paragraph{Rounding to an optimal integral solution}

After running the IPM for the cost vector $c$
(i.e. \Cref{line:matching:follow_2} of \Cref{alg:perfect_b_matching})
we obtain a near optimal and near feasible fractional solution $x$.
The following \Cref{lem:make_flow_feasible} shows
that we can convert this near feasible solution to a truly feasible solution
that is still near optimal. The proof is deferred to \Cref{sec:initial_point:feasible}.

\begin{restatable}{lemma}{lemMakeFlowFeasible}\label{lem:make_flow_feasible}
Consider any $\epsilon > 0$ and an uncapacitated min-cost flow instance on a starred flow graph (\Cref{def:starred_flow_graph})
with cost vector $c \in \R^E$, demand vector $d \in \R^V$
and the property that any feasible flow $f$ satisfies $f \le \|d\|_\infty$.

Assume we are given a primal dual solution pair $(x,s)$,
with $xs \approx_{1/2} \mu \tau(x,s)$ and
\begin{align*}
\frac{1}{\mu}\cdot\|\mA x - d\|_{(\mA^\top \mX \mS^{-1}\mA)^{-1}}^2 \le \frac{1}{10}
\end{align*}
for $\mu \le \frac{\epsilon}{24 n}$.

Let $\delta = \frac{\epsilon^2}{12(\|c\|_\infty\|d\|_\infty m)^2}$, 
then in $\tilde{O}(m \log \delta^{-1})$ time 
we can construct a feasible flow $f$ with $c^\top f \le \OPT + \epsilon$,
where $\OPT$ is the optimal value of the min-cost flow instance.
\end{restatable}

At last, we convert the feasible near optimal fractional solution $x$
to a truly optimal integral solution via the Isolation-Lemma of \cite{ds08}. 
\Cref{lem:isolation_lemma} below shows that, 
if the flow instance is integral 
and the set of optimal flows has congestion at most 
$W$ (where the congestion of a flow refers to the maximum flow value over all edges), 
then we can obtain an optimal feasible flow by (i) slightly perturbing the
problem, (ii) solving the perturbed problem approximately with $1/\poly(mW)$
additive error, and (iii) rounding the flow on each edge to nearest
integer.
We show in \Cref{sec:isolation} how \Cref{lem:isolation_lemma} 
is obtained via the Isolation-Lemma of \cite{ds08,klivans2001randomness}.

\begin{restatable}{lemma}{lemRoundIntegral}%
	\label{lem:isolation_lemma}
	Let $\Pi=(G,b,c)$ be an instance
	for minimum-cost flow problem where $G$ is a directed graph with
	$m$ edges, the demand vector $b\in\{-W,\dots,W\}^{V}$ and the cost
	vector $c\in\{-W,\dots,W\}^{E}$. 
	Further assume that all optimal flows have congestion at most $W$.
	
	Let the perturbed instance $\tPi=(G,b,\tc)$
	be such that $\tc_{e}=c_{e}+z_{e}$ where $z_{e}$ is a random number
	from the set $\left\{ \frac{1}{4m^{2}W^{2}},\dots,\frac{2mW}{4m^{2}W^{2}}\right\} $.
	Let $\tf$ be a feasible flow for $\tPi$ whose cost is at most $\opt(\tPi)+\frac{1}{12m^{2}W^{3}}$
	where $\opt(\tPi)$ is the optimal cost for problem $\tPi$. 
	Let $f$	be obtained by rounding the flow $\tf$ on each edge to the nearest integer. 
	Then, with probability at least $1/2$, 
	$f$ is an optimal feasible flow for $\Pi$.
\end{restatable}

\paragraph{Proof of \Cref{thm:perfect_b_matching_slow}}

We start by proving the correctness of \Cref{alg:perfect_b_matching}
in \Cref{lem:perfect_b_matching_correct} by using the previously stated 
\Cref{lem:initial_point_simple,lem:switch_cost,lem:make_flow_feasible,lem:isolation_lemma}.
Then we analyze the complexity of \Cref{alg:perfect_b_matching} 
in \Cref{lem:perfect_b_matching_complexity_slow}.
\Cref{lem:perfect_b_matching_correct} and \Cref{lem:perfect_b_matching_complexity_slow}
together form the proof of \Cref{alg:perfect_b_matching}.

\begin{lemma}\label{lem:perfect_b_matching_correct}
With high probability \Cref{alg:perfect_b_matching} 
returns an optimal minimum weight perfect bipartite $b$-matching.
\end{lemma}

\begin{proof}
Given the bipartite graph $G$ and vector $b$, 
\Cref{lem:initial_point_simple} constructs the corresponding starred flow graph
with demand vector $d \in \R^n$,
a cost vector $c' \in \R^m$, 
and a feasible primal dual pair $(x,s)$ with $xs = 1$.
Thus \textsc{PathFollowing} in \Cref{line:matching:follow_1} returns $x,s$ 
with 
\begin{align*}
xs \approx_{\epsilon/3} 45 \|c'\|_\infty\|b\|_\infty =: \mu\text{, and} \\
\mu^{-1/2} \|\mA^\top x - b\|_{(\mA \mX \mS^{-1} \mA)^{-1}} \le \epsilon\gamma < 1/4
\end{align*}
by \Cref{lem:path_following_master} and \Cref{lem:main_lem}.
By \Cref{lem:switch_cost} we can switch the cost vector $c'$ with $c''$,
such that $xs \approx_{\epsilon/2} \tau(x,s)$.
Then \textsc{PathFollowing} in \Cref{line:matching:follow_2} 
returns $x,s$ with
\begin{align*}
xs \approx_{\epsilon/3} \frac{1}{24n\cdot 12 m^{2}(\|b\|_\infty\|c''\|_\infty)^{3}} =: \mu' \text{, and}
\end{align*}
\begin{align*}
\mu'^{-1/2} \|\mA^\top x - b\|_{(\mA \mX \mS^{-1} \mA)^{-1}} \le \epsilon\gamma < 1/10
\end{align*}
by \Cref{lem:path_following_master} and \Cref{lem:main_lem}.
\Cref{lem:make_flow_feasible} turns the given $x$ into a feasible flow on graph $G'$,
with 
\begin{align*}
c''^\top x \le \OPT + \frac{1}{12 m^{2}(\|b\|_\infty\|c''\|_\infty)^{3}}.
\end{align*}
Further we know that every optimal flow on that graph can have congestion at most $\|b\|_\infty$,
as the congestion is bounded by the demand of the incident nodes.
Thus by \Cref{lem:isolation_lemma}, 
we can round the entries of $x$ to the nearest integer 
and obtain an optimal solution for the minimum cost flow with cost vector $c''$.
Since $c''_{u,z} = c''_{z,v} = n \|c\|_\infty$, 
the optimal solution will not use any of the edges $(u,z)$ or $(z,v)$ 
for any $u \in U$ or $v \in V$.
Thus the obtained flow $x$ can be interpreted as a perfect $b$-matching, 
because the demands $d'_u = -b_u$ and $d'_v = b_v$ match the given $b$-vector.
Note that we only obtain the optimal solution with probability at least $1/2$,
but we can repeat the algorithm $O(\log n)$ times and return the minimum result,
which is the optimal solution with high probability.
\end{proof}

\begin{lemma}\label{lem:perfect_b_matching_complexity_slow}
\Cref{alg:perfect_b_matching} runs in $\tilde{O}(n\sqrt{m} \log^2 (\|c\|_\infty\|b\|_\infty))$ time.
\end{lemma}

\begin{proof}
Constructing the initial point via \Cref{lem:initial_point} takes $O(m)$ time.
Next, moving away from the optimal solution, 
switching the cost vector, 
and then moving towards the optimal solution 
takes $O(n\sqrt{m} \log^2 (\|c\|_\infty\|b\|_\infty))$ time by \Cref{lem:complexity_pathfollowing},
as both $|\log \mu^{\target}/\mu^{\init}|$ and $\log W$ are bounded by $\tilde{O}(\log(\|c\|_\infty\|b\|_\infty))$.
At last, we construct a feasible $x$ via \Cref{lem:make_flow_feasible} in $\tilde{O}(m \log(\|c\|_\infty\|b\|_\infty))$ time
and rounding each entry of $x$ to the nearest integer takes only $O(m)$ time.
\end{proof}

\subsection{Nearly Linear Time Algorithm for Moderately Dense Graphs}
\label{sec:matching:fast}

In this section we improve the previous 
$\tilde{O}(n\sqrt{m})$-time algorithm of \Cref{thm:perfect_b_matching_slow}
to run in $\tilde{O}(m+n^{1.5})$-time instead.
For that we will modify the \textsc{ShortStep} algorithm 
(\Cref{alg:implement:short_step_log_barrier}) 
to use the faster improved IPM of \Cref{sec:ls_barrier_method} (\Cref{alg:short_step_LS}).

Note that \Cref{alg:short_step_LS} (leverage score based IPM)
is very similar to \Cref{alg:short_step} (log barrier based IPM).
Hence when implementing \Cref{alg:short_step_LS},
we can start with the implemention of \Cref{alg:short_step} from \Cref{sec:matching:log}
and perform a few modifications.
The main difference is that \Cref{alg:short_step_LS} uses $(\nabla\Phi(\ov))^{\flat(\otau)}$
for $\ov \approx xs/(\mu\tau(x,s))$ where $\tau(x,s) = \sigma(\sqrt{\mX\mS^{-1}}\mA) + n/m$,
while \Cref{alg:short_step} uses $\nabla\Phi(\ov)/\|\nabla\Phi(\ov)\|_2$ for $\ov \approx xs$.
Previously, \Cref{alg:implement:short_step_log_barrier} (the implementation of \Cref{alg:short_step})
used \Cref{thm:gradient_maintenance_simple} to maintain $\nabla\Phi(\ov)/\|\nabla\Phi(\ov)\|_2$.
We now replace this data structure by \Cref{thm:gradient_maintenance} 
to maintain $(\nabla\Phi(\ov))^{\flat(\otau)}$ instead.

The next difference in implementation is, that we must have an accurate approximation $\otau \approx \tau(x,s)$.
Note that, while \Cref{lem:large_entry_datastructure} does provide upper bounds on the leverage scores,
it does not yield a good approximation.
So we must use a different data structure for maintaining $\otau \approx \tau(x,s)$.
Such a data structure was previously given by \cite{blss20},
and we state these results in \Cref{sec:leverage_score}.
This data structure requires that $\ox,\os$ are very accurate approximations of $x,s$,
so another difference is that we use much smaller accuracy parameters than before.
However, the accuracy is smaller by only some $\polylog(n)$-factor, 
so it will not affect the complexity besides an additional $\tilde{O}(1)$-factor.

\Cref{alg:implement:short_step_ls_barrier} gives a formal description 
of our implementation of \Cref{alg:short_step_LS}
and \Cref{lem:implement:short_step_ls_barrier} shows 
that this implementation is indeed correct.

Afterward, we show in \Cref{lem:improved_path_following} 
that the complexity of \Cref{alg:implement:path_following} improves 
when using this new implementation of the faster IPM.

\begin{algorithm2e}[t!]
\caption{Implementation of \Cref{alg:short_step_LS} \label{alg:implement:short_step_ls_barrier}}
\SetKwProg{Globals}{global variables}{}{}
\SetKwProg{Proc}{procedure}{}{}
\Globals{}{
	Same variables as in \Cref{alg:implement:short_step_log_barrier}, and additionally \\
	$D^{(x,\nabla)}$ instance of \Cref{thm:gradient_maintenance} using $\gamma/(10^5 \log^2 n)$ accuracy \\
	$D^{(s)}$ instance of \Cref{thm:vector_maintenance} using $\gamma/(10^5 \log^2 n)$ accuracy\\
	$D^{(\tau)}$ instance of \Cref{thm:leverage_score_maintenance} with accuracy $\gamma/2$ \\
	$\otau \in \R^m$ element-wise approximations of $\tau(x,s)$ \\
	$T\ge\sqrt{n}$ upper bound on how often \textsc{ShortStep} is called \\
	$\alpha=\frac{1}{4\log(\frac{4m}{n})}$ \\
}
\Proc{\textsc{ShortStep}$(\mu^{\new}>0)$}{
	\LineComment{Update $\omu$ and data structures that depend on it.}
	\If{$\omu \not\approx_{\gamma/8} \mu^\new$}{ \label{line:ls:update_mu}
		$\omu \leftarrow \mu^\new$ \\
		\lFor{$i \in [m]$}{$D^{(x,\nabla)}.\textsc{Update}(i, \gamma\ox_i, (\ox_i\os_i)/(\omu\otau_i))$}
	}
	\LineComment{Leverage Score Sampling to sparsify $(\mA \mX/\mS \mA)$ with $\gamma/2$ spectral approximation}
	$I_v \leftarrow D^{(\sample,\sigma)}.\textsc{LeverageScoreSample}(c \gamma^{-2} \log n)$ for some large enough constant $c > 0$ \\
	$v \leftarrow D^{(\sample,\sigma)}.\textsc{LeverageScoreBound}(c \gamma^{-2} \log n, I_v)$\\
	$v_i \leftarrow 1/v_i$ for $i$ with $v_i \neq 0$ \\
	\LineComment{Perform \textsc{ShortStep} (\Cref{alg:short_step_LS})}
	$h' \leftarrow D^{(x,\nabla)}.\textsc{QueryProduct()}$ 
			\Comment{$h' = \gamma \mA \omX \nabla\Phi(\ov) / \| \nabla\Phi(\ov) \|_2$} \\
	$h'' \leftarrow$ solve $(\mA^\top \mV \mA)^{-1} (h' + \Delta)$ with $\gamma/2$ accuracy via Laplacian Solver (\Cref{lem:laplacian_solver}). 
			\Comment{$h'' = \omH^{-1} (h' + (\mA^\top x - b))$, $\delta_r = \mS^{-1} \mA h''$} \\
	$\mR \leftarrow \textsc{SamplePrimal}(c_1 (\sqrt{m}/n) \cdot T \gamma \log^2 n \log(mT), D^{(\sample)}, D^{(\sample, \sigma)}, h'')$ for some large enough constant $c_1 > 0$ such that $x \approx_{\gamma/(10^5 \log^2 n)} \widehat{x}$ via \Cref{lem:close_to_stable_sequence} \\
	$x^{\tmp}, I_x \leftarrow D^{(x,\nabla)}.\textsc{QuerySum}(-\mR \omX \omS^{-1} \mA h'')$ \\
	$s^{\tmp}, I_s \leftarrow D^{(s)}.\textsc{Add}((\mA^\top \mV \mA)^{-1} h', \gamma/16)$ \\
	$\Delta \leftarrow \Delta + h' - \mA^\top \mR \omX \omS^{-1} \mA h''$ 
			\Comment{Maintain $\Delta = \mA^\top x - b$}\\
	\LineComment{Update $\ox, \os$ and data structures that depend on them.}
	\For{$i \in I_x \cup I_s$}{ \label{line:ls:update_xs}
		\If{$x^{\tmp}_i \not\approx_{\gamma/(10^5 \log n)} \ox_i$ 
					or $s^{\tmp}_i \not\approx_{\gamma/(10^5 \log^2 n)} \os_i$}{
			$\ox_i \leftarrow x^{\tmp}_i$, $\os_i \leftarrow s^{\tmp}_i$ \\
			$D^{(x,\nabla)}.\textsc{Update}(i, \gamma\ox_i, (\ox_i\os_i)/(\omu\otau_i))$ \\
			$D^{(\sample)}.\textsc{Scale}(i, 1/\os_i)$ \\
			$D^{(\sample, \sigma)}.\textsc{Scale}(i, \sqrt{\ox_i/\os_i})$ \\
			$D^{(\tau)}.\textsc{Scale}(i, (\ox_i/\os_i)^{1-2\alpha})$ \Comment{So $\otau \approx \sigma(\mX^{1/2-\alpha}\mS^{-1/2-\alpha})+n/m$}\\
		}
	}
	\LineComment{Update $\otau$ and data structures that depend on it.}
	$I_\tau, \otau \leftarrow D^{(\tau)}.\textsc{Query}()$ \\
	\For{$i \in I_\tau$}{ \label{line:ls:update_tau}
		$D^{(x,\nabla)}.\textsc{Update}(i, \gamma\ox_i, \otau_i, (\ox_i\os_i)/(\omu\otau_i))$
	}
}
\end{algorithm2e}

\begin{lemma}\label{lem:implement:short_step_ls_barrier}
\Cref{alg:implement:short_step_ls_barrier} implements \Cref{alg:short_step_LS}.
\end{lemma}

\begin{proof}
Note that the only difference between \Cref{alg:short_step} 
and \Cref{alg:short_step_LS} is the gradient.
While in \Cref{alg:short_step} we use $\nabla \Phi(\ov)/\|\nabla \Phi(\ov)\|_2$ for $\ov \approx_\gamma xs/\mu$,
the faster \Cref{alg:short_step_LS} uses $(\nabla \Phi(\ov))^{\flat(\ttau)}$ for $\ov \approx_\gamma xs/(\mu\tau(x,s))$ 
and $\ttau \approx_\gamma \tau(x,s)$.

Likewise \Cref{alg:implement:short_step_ls_barrier} 
is almost the same as \Cref{alg:implement:short_step_log_barrier}.
The main difference being that we replaced the data structure 
of \Cref{thm:gradient_maintenance_simple} 
for maintaining $\nabla \Phi(\ov)/\|\nabla \Phi(\ov)\|_2$
by the data structure of \Cref{thm:gradient_maintenance} 
for maintaining $(\nabla \Phi(\ov))^{\flat(\otau)}$.
The next difference is that we now use 
\Cref{thm:leverage_score_maintenance} 
to maintain the leverage scores.

Given that the rest of the implementation (\Cref{alg:implement:short_step_ls_barrier}) of the LS-barrier method 
is identical to the implementation (\Cref{alg:implement:short_step_log_barrier}) of the log-barrier method,
we only need to verify that these modifications maintain the desired values with sufficient accuracy.

\paragraph{Leverage Scores}

Note that \Cref{line:ls:update_xs} makes sure that
\Cref{thm:leverage_score_maintenance} always uses $g = (\ox/\os)^{1-2\alpha}$.
As $D^{(x,\nabla)}$ and $D^{(s)}$ return $\gamma/(10^5 \log^2 n)$ approximations
and we change $\ox,\os$ whenever these approximations change by some $\gamma/(10^5 \log^2 n)$ factor,
we have $(\ox/\os)^{1-2\alpha} \approx_{4\gamma/(10^5 \log^2 n)} (x/s)^{1-2\alpha}$.
This then implies
$$
\otau 
\approx_{\gamma/2} 
\sigma(\omX^{1/2-\alpha}\omS^{-1/2-\alpha}\mA) 
\approx_{4\gamma/(10^5 \log^2 n)} 
\sigma(\omX^{1/2-\alpha}\omS^{-1/2-\alpha}\mA) 
=: 
\tau(x,s),
$$
assuming $D^{(\tau)}$ (\Cref{thm:leverage_score_maintenance}) is initialized with accuracy $\gamma/4$.
Hence \Cref{thm:leverage_score_maintenance} maintains $\otau \approx_{3\gamma/4} \tau(x,s)$.

\paragraph{Gradient}

In \Cref{line:ls:update_mu}, \Cref{line:ls:update_xs} and \Cref{line:ls:update_tau} 
we make sure that \Cref{thm:gradient_maintenance} 
uses $z = \ox\os/(\omu\otau)$.
Hence the data structure maintains $(\nabla\phi(\ov))^{\flat(\ttau)}$
for $\ov \approx_{\gamma/16} z = \ox\os/(\omu\otau) \approx_{15\gamma/16} xs/\tau(x,s)$, 
so $\ov \approx_{\gamma} xs/(\tau(x,s))$.
Further, $\ttau \approx_{\gamma/(10^5 \log^2 n)} \otau \approx_{3\gamma/4} \tau(x,s)$,
so $\ttau \approx_{\gamma} \tau(x,s)$.

\end{proof}

We now show that the complexity of \Cref{alg:implement:path_following}
improves when using \Cref{alg:implement:short_step_ls_barrier} 
for the \textsc{ShortStep} function.

\begin{lemma}\label{lem:improved_path_following}
We can improve \Cref{alg:implement:path_following} to run in
\begin{align*}
\tilde{O}\left((m + n^{1.5}) \cdot |\log \mu^{\target}/\mu^{\init} |\cdot (\log W + |\log \mu^{\target} / \mu^{\init} | ) \right)
\end{align*}
time, where $\log W$ is a bound on the ratio of largest to smallest entry in $x^{\init}$ and $s^{\init}$.
\end{lemma}

\begin{proof}
The main difference is that we run 
\Cref{alg:implement:short_step_ls_barrier} (the LS-barrier) for \textsc{ShortStep} 
instead of \Cref{alg:implement:short_step_log_barrier} (the log-barrier).
For that we must first initialize \Cref{thm:leverage_score_maintenance} via
$D^{(\tau)}.\textsc{Initialize}((x^{\init}/s^{\init})^{1/2-\alpha}, \gamma/4)$.
This additional initialization takes $\tilde{O}(m)$ time.
We also run \Cref{thm:gradient_maintenance}
instead of \Cref{thm:gradient_maintenance_simple}
to maintain $(\nabla\Phi(\ov))^{\flat(\ttau)}$ instead of $\nabla\Phi(\ov)/\|\nabla\Phi(\ov)\|_2$,
but the complexities of these two data structures are the same.
Note that we must initialize $D^{(x,\nabla)}$ and $D^{(s)}$ for a higher accuracy ($\gamma/(10^5 \log^2 n)$)
than in the previous log-barrier implementation (which had accuracy $\gamma/16$).
However, this increases the complexity by only a $\tilde{O}(1)$ factor.

\paragraph{Complexity of \Cref{alg:implement:short_step_ls_barrier}}

The main difference to the previous analysis from \Cref{lem:complexity_pathfollowing} 
is that now \textsc{ShortStep} is called only 
$K = \tilde{O}(\sqrt{n} |\log(\mu^{\init}/\mu^{\target})|)$ times.
We will now list all the differences 
compared to the previous analysis of \Cref{lem:complexity_pathfollowing}.
Note that \Cref{alg:implement:short_step_ls_barrier} assumes a global variable $T \ge \sqrt{n}$,
which is an upper bound on how often \textsc{ShortStep} is called.
So we modify \Cref{alg:implement:path_following} to set $T = K = \tilde{O}(\sqrt{n} |\log(\mu^{\init}/\mu^{\target})|)$.
Further note that as before, 
we can bound the ratio $W'$ of largest to smallest entry in $x$ and $s$ throughout all calls to \textsc{ShortStep}
by $\log W' = \tilde{O}(\log W + |\log(\mu^{\init}/\mu^{\target})|)$
via \Cref{lem:bitlength}.

\paragraph{Sampling the primal}
We scale the sampling probability for \textsc{SamplePrimal} 
by $\tilde{O}(\sqrt{m} K / n)$.
The complexity per call thus increases to
$$
\tilde{O}(n (\log W + |\log(\mu^{\init}/\mu^{\target})|) + K m / n)
$$
by \Cref{lem:sample_primal}.
Further note that by \Cref{lem:close_to_stable_sequence} the sampling probability is large enough for
our sequence of primal solutions $x^{(1)},x^{(2)},...$ 
to be close to some sequence $\widehat{x}^{(1)},\widehat{x}^{(2)},...$, 
such that $x^{(t)} \approx_{\gamma/(10^5 \log^2 n} \widehat{x}^{(t)}$ 
and $\|(\widehat{\mX}^{(t-1)})^{-1}(\widehat{x}^{(t)} - \widehat{x}^{(t-1)})\|_\tau$ 
for all $t$.

\paragraph{$D^{(x,\nabla)}$ data structure}
The number of times $\omu$ changes is now bounded by $\tilde{O}(K/\sqrt{n})$,
so the total cost of updating $D^{(x,\nabla)}$ in \Cref{line:ls:update_mu} 
is now bounded by $\tilde{O}(Km/\sqrt{n})$.
The cost of $D^{(x,\nabla)}.\textsc{QueryProduct}$ and \textsc{QuerySum} is bounded by
\begin{align*}
&~\tilde{O}(
Kn + K n (\log W + |\log(\mu^{\init}/\mu^{\target})|) + K \cdot \|\omX (\nabla\Phi(\ov))^{\flat(\otau)} / x\|_2^2
) \\
=&~
\tilde{O}(
Kn + K n (\log W + |\log(\mu^{\init}/\mu^{\target})|) + K \cdot m/n
),
\end{align*}
by \Cref{thm:gradient_maintenance},
where we use that 
\begin{align*}
\|\omX (\nabla\Phi(\ov))^{\flat(\otau)} / x\|_2^2 
&= 
O(1) \cdot \|\nabla\Phi(\ov))^{\flat(\otau)}\|_2^2 
= 
O(m/n) \cdot \|\nabla\Phi(\ov))^{\flat(\otau)}\|_{\otau} \\
&=
\tilde{O}(m/n) \cdot \|\nabla\Phi(\ov))^{\flat(\otau)}\|_{\otau+\infty}
=
\tilde{O}(m/n)
\end{align*}

\paragraph{$D^{(s)}$ data structure}
The complexity of calls to $D^{(s)}.\textsc{Add}$ (\Cref{thm:vector_maintenance})
is now bounded by $\tilde{O}(K^2 \cdot m/n + Kn (\log W + |\log(\mu^{\init}/\mu^{\target})|))$,
because $\|\mS^{-1} \delta_s \|_2^2 \le O(m/n) \|\mS^{-1} \delta_s \|_{\tau(x,s)}^2 = O(m/n)$.

\paragraph{$D^{(\tau)}$ data structure}
Any change to $\otau$ causes $\tilde{O}(1)$ cost in \Cref{line:ls:update_tau}.
The number of times any entry of $\otau$ changes is bounded by
the \textsc{Query} complexity of \Cref{thm:leverage_score_maintenance}.
This complexity can be bounded by
$\tilde{O}(K^2 m/n + Kn(\log W + |\log(\mu^{\init}/\mu^{\target})|))$,
when applying \Cref{lem:bitlength}
and the fact that we scale the sampling probability for \textsc{SampleProbability}
by $\tilde{O}(\sqrt{m} K / n)$.
Note that this complexity bound only holds when for the given input sequence $g^{(t)} := (\ox^{(t)}/\os^{(t)})^{1-\alpha}$,
there exists sequence $\widehat{g}^{(t)} \approx_{\gamma / (2 \cdot 16\cdot144 \log^2 n)} g^{(t)}$ which satisfies
\eqref{eq:closeness_assumption} and \eqref{eq:stability_assumption}.
This is satisfied because of the accuracy of 
$\ox \approx_{2\gamma/(10^5 \log^2 n)} x$, 
$\os \approx_{2\gamma/(10^5 \log^2 n)} s$, 
the stability property $\|\mS^{-1} \delta_s\|_\tau \le 2\gamma$ (\Cref{lem:ls_second}), 
and the stable sequence $\widehat{x}^{(t)} \approx_{\gamma/(10^5 \log^2 n)} x^{(t)}$ that exists according to \Cref{lem:close_to_stable_sequence}.
So the sequence 
$\widehat{g}^{(t)} := (\widehat{x}^{(t)}/s^{(t)})^{1-2\alpha}$ 
satisfies
$\widehat{g}^{(t)} \approx_{\gamma / (2 \cdot 16\cdot144 \log^2 n)} g^{(t)} = (x^{(t)}/s^{(t)})^{1-2\alpha}$, \eqref{eq:closeness_assumption} and \eqref{eq:stability_assumption}.

\paragraph{Updating the data structure}
The number of times that any entry of $\os$ changes 
now bounded by $\tilde{O}(K^2 m/n)$,
because $\|\mS^{-1} \delta_s \|_2^2 = O(m/n)$.
Likewise the number of times that any entry of $\os$ changes is bounded by
$\tilde{O}(Kn (\log W + |\log(\mu^{\init}/\mu^{\target})| + K^2 m/n)$
because $x$ changes by $\omX (\nabla \Phi(\ov))^{\flat(\otau)} - \mR\omX\omS^{-1} \mA h''$,
where for the first term we have $\|\omX (\nabla\Phi(\ov))^{\flat(\otau)} / x\|_2^2 = O(m/n)$
and the second term as only at most 
$\tilde{O}(n (\log W + |\log(\mu^{\init}/\mu^{\target})|) + K m / n)$ non-zero entries.
So updating all the data structures via calls to 
\textsc{Update} and \textsc{Scale} in \Cref{line:ls:update_xs}
now takes $\tilde{O}(Kn (\log W + |\log(\mu^{\init}/\mu^{\target})|) + K^2 \cdot m/n)$ time.

\paragraph{Total Complexity}

The total complexity is thus
\begin{align*}
&~
\tilde{O}(
m 
+ 
K^2 m/n
+
Kn (\log W + |\log(\mu^{\init}/\mu^{\target})|)
)\\
=&~
\tilde{O}(
(m + n^{1.5}) |\log(\mu^{\init}/\mu^{\target})| (\log W + |\log(\mu^{\init}/\mu^{\target})|)
).
\end{align*}

\end{proof}

Note that \Cref{lem:initial_point_simple}, as used in our $\tilde{O}(n\sqrt{m})$-time algorithm,
only provides an initial point $(x,s)$ with $xs = \mu$ for some $\mu \in \R$.
For the faster IPM, however, we require $xs \approx \mu \tau(x,s)$
with $\tau(x,s) := \sigma(x,s) + n/m$ as defined in \Cref{def:LS_weight}.
The following \Cref{lem:initial_point} shows that such a point can be constructed,
and it is proven in \Cref{sec:initial_point:initial_point}.

\begin{restatable}{lemma}{lemInitialPoint}\label{lem:initial_point}
We are given a minimum weight perfect $b$-matching instance 
on graph $G=(U,V,E)$ with cost vector $c \in \R^E$.

Let $G' = (V', E')$ with demand vector $d'$ be the corresponding starred flow graph
(see \Cref{def:starred_flow_graph})
and let $n$ be the number of nodes of $G'$ and $m$ the number of edges.

For $\tau(x,s) := \sigma(x,s) + m/n$ and for any $\epsilon > 0$
we can construct in $\tilde{O}(m \poly(\epsilon^{-1}))$ time 
a cost vector $c' \in \R^{E'}$ for $G'$
and a feasible primal dual solution pair $(x,s)$ 
for the minimum cost flow problem on $G'$ with cost vector $c'$,
where the solution satisfies
\begin{align}
xs \approx_{\epsilon} \tau(x,s).
\end{align}
The cost vector $c'$ satisfies
\begin{align*}
\frac{n}{4 m (1+\|b\|_\infty)} \le c' \le 3n.
\end{align*}

\end{restatable}

We now obtain \Cref{thm:perfect_b_matching_fast}
by replacing the slow $\sqrt{m}$-iteration IPM,
used in \Cref{alg:perfect_b_matching},
with the faster $\sqrt{n}$-iteration IPM.

\thmFastMatching*

\begin{proof}
This follows directly from \Cref{alg:perfect_b_matching} 
(\Cref{thm:perfect_b_matching_slow}),
when we use the faster variant of \Cref{alg:implement:path_following}
(\Cref{lem:improved_path_following} instead of \Cref{lem:complexity_pathfollowing}).
To construct the initial point we now use \Cref{lem:initial_point} instead of \Cref{lem:initial_point_simple}.
\end{proof}

\subsection{More Applications to Matching, Flow, and Shortest Paths}

\label{sec:app}

In this section, we apply known reductions and list applications
of our main results (see \Cref{fig:reductions}). 
In \Cref{sec:app int}, we show that we can obtain
\emph{exact} algorithms for many fundamental problems when all the numbers
from the problem instance are integral. 
For a history overview for some of these problems see 
\Cref{table:results:SSSP,table:results:perfect matching}.
Then, in \Cref{sec:app frac},
we show that even when the numbers are not integral, we can still
$(1+\epsilon)$-approximately solve all the problems with $\log^{2}(1/\epsilon)$
dependency in the running time.

\begin{figure}[ht]
	\begin{centering}
		\includegraphics[trim=2cm 8.5cm 16.75cm 0.5cm,clip,width=0.5\textwidth]{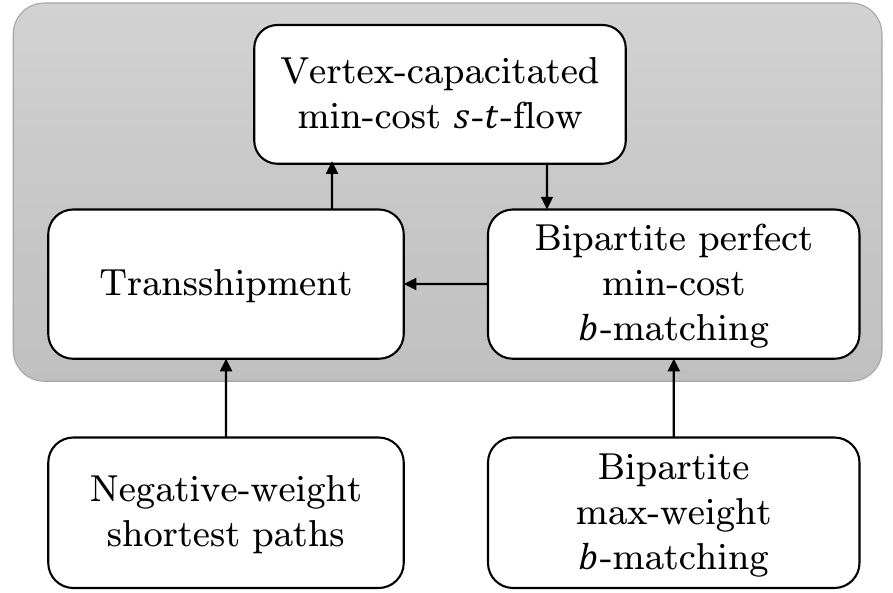}
		\par\end{centering}
	\caption{\label{fig:reductions}Reductions between problems}
	
\end{figure}

\begin{table}[ht]
    \centering
    \begin{tabular}{|c|c|c|c|c|}
    \hline
        {\bf Year} & {\bf Author} & {\bf Reference} & {\bf Complexity} & {\bf Notes}  \\ \hline
        1955 & Shimbel & \cite{shimbel55} & $O(n^{4})$ &  \\ \hline
        1956 & Ford & \cite{ford56} & $O(Wn^{2}m)$ & \\ \hline
        1958 & Bellman, Moore & \cite{bellman58} \cite{moore59} & $O(nm)$ & * \\ \hline
        1983 & Gabow & \cite{Gabow85} & $O(n^{3/4}m\log W)$ & \\ \hline
        1989 & Gabow and Tarjan & \cite{GabowT89} & $O(\sqrt{n}m\log(nW))$ & \\ \hline
        1993 & Goldberg & \cite{goldberg93} & $O(\sqrt{n}m\log(W))$ & * \\ \hline
        2005 & Sankowski; Yuster and Zwick & \cite{s05}, \cite{yz05} & $\tilde{O}(Wn^{\omega})$ &  \\ \hline
        2017 & Cohen, Madry, Sankowski, Vladu & \cite{cmsv17} & $\tilde{O}(m^{10/7}\log W)$ & \\ \hline
        2020 & Axiotis, Madry, Vladu & \cite{amv20} & $\tilde{O}(m^{4/3+o(1)}\log W)$ & * \\ \hline
        2020 & & This paper & $\tilde{O}((m+n^{1.5})\log^2 W)$ & * \\ \hline
    \end{tabular}
    \caption{The complexity results for the {\bf SSSP} problem with negative weights
		({*} indicates asymptotically the best bound for some range of parameters). The table is extended from \cite{cmsv17}}
    \label{table:results:SSSP}
\end{table}%

\begin{table}[ht]
    \centering
    \begin{tabular}{|c|c|c|c|c|}
         \hline
        {\bf Year} & {\bf Author} & {\bf Reference} & {\bf Complexity} & {\bf Notes}  \\ \hline
        1931 &  Egervary & \cite{egervary31} & $O(Wn^{2}m)$ & \\\hline
			1955 & Khun and Munkers & \cite{kuhn55},\cite{munkres57} & $O(Wn^{2}m)$ & \\ \hline 
			1960 & Iri & \cite{iri60} & $O(n^{2}m)$  & \\ \hline
			1969 & Dinic and Kronrod & \cite{dinic-kronrod-69} & $O(n^{3})$ & \\ \hline
			1970 & Edmonds and Karp & \cite{EdmondsK72} & $O(nm+n^{2}\log n)$ & * \\ \hline
			1983 & Gabow & \cite{Gabow85} & $O(n^{\frac{3}{4}}m\log W)$  & \\ \hline
			1989 & Gabow and Tarjan & \cite{GabowT89} & $O(\sqrt{n}m\log(nW))$ & * \\ \hline 
			1999 & Kao, Lam, Sung and Ting & \cite{kao-99} & $O(W\sqrt{n}m)$ &  \\ \hline 
			2006 & Sankowski & \cite{s06} & $O(Wn^{\omega})$ &   \\ \hline
			2017 & Cohen, Madry, Sankowski, Vladu & \cite{cmsv17} & $\tilde{O}(m^{10/7}\log W)$ & \\ \hline
			2020 & Axiotis, Madry, Vladu & \cite{amv20} & $\tilde{O}(m^{4/3+o(1)}\log W)$ & * \\ \hline
			2020 & & This paper & $\tilde{O}((m+n^{1.5})\log^2 W)$ & * \\ \hline
    \end{tabular}
    \caption{The complexity results for the {\bf minimum-weight bipartite perfect matching}
		problem ({*} indicates asymptotically the best bound for some range
		of parameters). The table is extended from \cite{cmsv17}.}
    \label{table:results:perfect matching}
\end{table}%

\subsubsection{Exact Algorithms for Integral Problems}

\label{sec:app int}

In this section, we define several problems such that the numbers
from the problem instance are integral, and show how to solve them
exactly. We will later allow fractions in the problem instance in
\Cref{sec:app frac}, and show how to solve those problems approximately.
\begin{definition}
[$b$-matching]For any graph $G = (V,E)$, a demand vector $b \in \Z_{\ge0}^{V}$,
and a cost vector $c \in \Z^{E}$, a \emph{$b$-matching} is an assignment
$x \in \Z_{\ge 0}^{E}$ where $\sum_{e \ni v} x_{e} \le b_{v}$ for all
$v \in V$. 
We say that $x$ is \emph{perfect} if $\sum_{e\ni v} x_{e} = b_{v}$
for all $v \in V$. The \emph{cost/weight} of $x$ is $c^{\top}x = \sum_{e\in E} c_{e} x_{e}$.
When $b$ is an all-one vector, the $b$-matching is referred to simply as a {\em matching}. 
We say that $x$ is \emph{fractional} if we allow $x \in \R_{\ge0}^{E}$.
Similarly, $b$ is fractional if $b \in \R_{\ge0}^{V}$. 
\end{definition}

In the \emph{maximum weight bipartite $b$-matching} problem, given
a bipartite graph $G$, a demand vector $b$,
and a cost vector $c$, we need to find a $b$-matching with maximum
weight. In the \emph{minimum/maximum weight bipartite perfect $b$-matching}
problem, we need to find a perfect $b$-matching with minimum/maximum
weight or report that no perfect $b$-matching exists. Below, we show
a relation between these problems. 
We let $W$ denote the largest absolute value used for specifying any value in the problem
(i.e. $b_v \le W$ for all $v\in V$ and $-W \le c_{u,v} \le W$ for all $(u,v) \in E$).

\paragraph{Equivalence of variants of perfect $b$-matching.}

Most variants of the problems for perfect $b$-matching
are equivalent. First, we can assume that the cost is non-negative by
working with the cost vector $c'$ where $c'(e)=c(e)+W$ for all $e\in E$.
This is because the cost of \emph{every} perfect $b$-matching is
increased by exactly $W \cdot \| b \|_1 / 2$. Second, the maximization
variant is equivalent with the minimization variant by working with
the cost $c'(e) = -c(e)$ because the cost of every perfect $b$-matching
is negated. 

\paragraph{Maximum weight non-perfect $b$-matching $\rightarrow$ Maximum weight perfect $b$-matching.}

We can reduce the non-perfect variant to the perfect one. Given
a problem instance $(G,b,c)$ of maximum weight bipartite $b$-matching
where $G=(U,V,E)$ is a bipartite graph, $b$ is a demand vector and
$c$ is a cost vector, we construct an instance $(G',b',c')$ of 
maximum weight bipartite perfect $b$-matching as follows. 
First, note that we can assume that all entries in $c$ are non-negative since we never add edges of negative cost to the optimal maximum weight non-perfect $b$-matching. 
Next, we add
new vertices $u_{0}$ into $U$ and $v_{0}$ into $V$. Set $b'_{ u_0 } = \sum_{v\in V} b_{v}$,
$b'_{v_{0}} = \sum_{u \in U} b_{u}$, and $b'_{x}=b_{x}$ for all $x\in U\cup V$.
We add edges connecting $u_{0}$ to all vertices in $V$ (including
$v_{0}$) and connecting $v_{0}$ to all vertices in $U$.
The cost of all these edges incident to $u_{0}$ or $v_{0}$ is $0$.
It is easy to see that, given a perfect $b'$-matching $x'$ in $G'$,
we can obtain a $b$-matching $x$ with the same cost in $G$ simply
by deleting all edges incident to $u_{0}$ and $v_{0}$. Also, given
any $b$-matching $x$ in $G$, there is a perfect $b'$-matching
$x'$ in $G'$ with the same cost that is obtained from $x$ by ``matching
the residual demand'' on vertices in $U$ and $V$ to $v_{0}$ and
$u_{0}$ respectively. More formally, we set $x'_{e} = x_{e} \forall e \in E$,
$x'_{(u_{0},v)} = b_{v} - \sum_{(u,v) \in E}x_{(u,v)} \forall v \in V$,
$x'_{(u,v_{0})} = b_{u} - \sum_{(u,v)\in E}x_{(u,v)} \forall u \in U$,
and $x'_{(u_{0},v_{0})} = \sum_{(u,v)\in E}x_{(u,v)}$. It is easy to
verify that $x'$ is a perfect $b'$-matching with the same cost as
$x$. 

From the above discussion, we conclude:
\begin{proposition}
Consider the following problems when the input graph has $n$ vertices
and $m$ edges, and let $W$ denote the largest number in the problem
instance:
\begin{itemize}
\item minimum weight bipartite perfect $b$-matching (with positive cost
vector),
\item minimum weight bipartite perfect $b$-matching, and
\item maximum weight bipartite perfect $b$-matching.
\end{itemize}
If one of the above problems can be solved in $T(m,n,W)$ time, then
we can solve the other problems in $T(m,n,2W)$ time. Additionally, we
can solve the maximum weight bipartite $b$-matching problem in $T(m+O(n),n,O(Wn))$
time.
\end{proposition}

\begin{definition}
[Flow]For any directed graph $G=(V,E)$, a demand vector $b\in\Z^{V}$ where
$\sum_{v\in V}b_{v}=0$, and a cost vector  $c\in\Z^{E}$, a \emph{flow} (w.r.t.~$b$) is an assignment $f\in\Z_{\ge0}^{E}$ such that, for all $v\in V$,
$\sum_{(u,v)\in E}f_{(u,v)}-\sum_{(v,u)\in E}f_{(v,u)}=b_{v}$. The
\emph{cost} of $f$ is $c^{\top}f=\sum_{e\in E}c_{e}f_{e}$. We say
that $f$ is \emph{fractional} if we allow $f\in\R_{\ge0}^{E}$.

An $s$-$t$ flow $f$ is a flow such that $b_{v}=0$ for all $v\in V\setminus\{s,t\}$.
So $b_{s}=k=-b_{t}$ for some $k\in\R$. We say that $k$ is the \emph{value}
of the $s$-$t$ flow $f$. Let $\CAP\in\Z_{>0}^{V\setminus\{s,t\}}$
denote vertex capacities. We say that $f$ \emph{respects} the vertex
capacity if $\sum_{(u,v)\in E}f_{(u,v)}\le\CAP_{v}$ for all $v\in V\setminus\{s,t\}$.
\end{definition}

In the \emph{uncapacitated minimum-cost flow} problem (a.k.a.~the
\emph{transshipment} problem), given a directed graph $G=(V,E)$, a demand
vector $b$, and a cost vector $c$, we need either (1) report that  no flow w.r.t.~$b$ exists, (2) report that the optimal cost is $-\infty$ (this can happen if there is a negative cycle w.r.t.~$c$)
or (3) find such flow with minimum cost.
We will need the following observation:
\begin{claim}
	\label{claim:lower bound cost transhipment}If there is a feasible solution
	for the transshipment problem and the optimal cost is not $-\infty$,
	then the optimal cost is between $[-W^{2}n^{3},W^{2}n^{3}]$.
\end{claim}

\begin{proof}
	If the optimal cost is not $-\infty$, there is no negative cycle.
	So an optimal flow $f$ can be decomposed into a collection of paths.
	The total flow value summing overall flow paths is exactly $\sum_{v\in V}\frac{||b||_{1}}{2}$.
	So the flow value on every edge is at most $\sum_{v\in V}\frac{||b||_{1}}{2}$.
	Therefore, the cost of $f$ is at least $\sum_{v\in V}\frac{||b||_{1}}{2}\times\min_{e\in E}c_{e}\times|E|\ge-W^{2}n^{3}$
	and at most $\sum_{v\in V}\frac{||b||_{1}}{2}\times\max_{e\in E}c_{e}\times|E|\le W^{2}n^{3}$
	.
\end{proof}

In the \emph{vertex-capacitated minimum-cost $s$-$t$ flow} problem, given
a directed $G$ and a pair of vertices $(s,t)$, a target value $k$,
a vertex capacity vector $\CAP$, and a cost vector $c$, we either
report that no such $s$-$t$ flow with value $k$ respecting vertex
capacities $\CAP$ exists, or find such flow with minimum cost. 

Below, we show the equivalence between these two problems and the minimum
weight perfect $b$-matching problem.

\paragraph{Minimum weight perfect $b$-matching $\rightarrow$ Transshipment. }

Observe that minimum weight bipartite perfect $b$-matching
is a special case of the transshipment problem when we only allow
edges that direct from vertices with positive demands (for the transshipment
problem) to vertices with negative demands.

\paragraph{Transshipment $\rightarrow$ Vertex-capacitated minimum-cost $s$-$t$
flow.}

Given a transshipment instance $(G=(V,E),b,c)$, we can construct
an instance $(G',k,\CAP,c')$ for vertex-capacitated $s$-$t$ flow
as follows. First, we add a special source $s$ and $t$ into the
graph. For each vertex $v$ with positive demand $b_{v}>0$, we add
a dummy vertex $s_{v}$ and two edges $(s,s_{v})$ and $(s_{v},v)$.
For each vertex $v$ with positive demand $b_{v}<0$, we add a dummy
vertex $t_{v}$ and two edges $(v,t_{v})$ and $(t_{v},t)$. Let $G'$
denote the resulting graph. The vertex capacity of each $s_{v}$ is
$\CAP_{s_{v}}=b_{v}$. The vertex capacity of each $t_{v}$ is $\CAP_{t_{v}}=-b_{v}$.
For original vertices $v\in V$, we set $\CAP_{v}=3n^3 W^2 + n||b||_{1}/2$. We set the cost $c'_{e}=c_{e}$
for all original edges $e\in E$, otherwise $c'_{e}=0$ for all new
edges $e$. Lastly, we set the target value $k=\sum_{v:b_{v}>0}b_{v}=||b||_{1}/2$. 

First, we claim there is an $s$-$t$ flow in $G'$ with value $k$
if and only if there is a flow in $G$ satisfying the demand $b$. To see this, consider any $s$-$t$ flow $f'$ with value $k$ respecting vertex
capacities in $G'$. Observe that the vertex capacities of all
$s_{v}$ must be exactly saturated, because there are $k$ units of flow going out from $s$, 
but the total capacities of all $s_v$ are at most $k$. This similarly holds for all $t_v$.
Therefore, for all
original vertices $v\in V$, $\sum_{(u,v)\in E}f'_{(u,v)}-\sum_{(v,u)\in E}f'_{(v,u)}=b_{v}$.
Observe that, by deleting all non-original vertices in $G'$, $f'$
corresponds to a flow $f$ (w.r.t. $b$) in $G$ with the same cost
as $f$. 

Therefore, if the algorithm for solving vertex-capacitated $s$-$t$ flow reports that there is no
$s$-$t$ flow in $G'$ with value $k$, then we just report that
there is no flow in $G$ satisfying the demand $b$. 
Otherwise, we obtain an $s$-$t$ flow $f'$ in $G'$ with value $k$.

We claim that the cost of $f'$ is less than $-W^{2}n^{3}$ if and only if there
is a negative cycle w.r.t.~$c$ in $G$. If the cost of $f'$ is
less than $-W^{2}n^{3}$, then there is a flow in $f$ satisfying
the demand $b$ with the same cost by the above argument. So, by \Cref{claim:lower bound cost transhipment},
the optimal cost of flow in $G$ must be $-\infty$ and, hence,
there is a negative cycle in $G$. In the other direction, suppose
there is a negative cycle $C$ in $G$, then it suffices to construct a $s$-$t$
flow $f'_{1}$ in $G'$ with value $k$ with cost less than $-W^{2}n^{3}$. Start with an arbitrary $s$-$t$ flow $f'_{2}$ in $G'$
with value $k$ such that $f'_{2}$ can be decomposed into a collection
of paths only (no cycle). We know that $f'_{2}$ exists because $f'$
exists. Observe that the flow value in $f'_{2}$ is at most $k$ on
every edge and at most $kn$ on every vertex. So the cost of $f'_{2}$
is at most $k\times\max_{e\in E}c_{e}\times|E|\le n^{3}W^{2}$. Now,
$f'_{1}$ is obtained from $f'_{2}$ by augmenting a flow along the
negative cycle $C$ with value $3n^{3}W^{2}$. Note that this augmentation
is possible because the capacity of every original vertex $v\in V$ is 
$nk+3n^{3}W^{2}$. So the cost of $f'_{1}$ is at most $n^{3}W^{2}-3n^{3}W^{2}\le-2n^{3}W^{2}$. 

By the above argument, we simply check the cost of $f'$. If it is
less than $-W^{2}n^{3}$, we report that the optimal cost for the
transshipment problem is $-\infty$. Otherwise, we return the corresponding
flow $f$ from the discussion above.

\paragraph{Vertex-capacitated minimum-cost $s$-$t$ flow $\rightarrow$ Minimum
weight perfect $b$-matching.}

We show a reduction that directly generalizes the well-known reduction
from minimum-cost vertex-disjoint $s$-$t$ paths to minimum weight
perfect matching (see e.g.~Section 3.3 of \cite{Sankowski09} and
Chapter 16.7c of \cite{schrijver2003combinatorial}). 

Suppose we are given an instance $(G=(V,E),k,\CAP,c)$ for the vertex-capacitated
minimum-cost $s$-$t$ flow problem (possibly with negative costs). First, we assume that $s$ has
no incoming edges and $t$ has no outgoing edges. To assume this, we
can add vertices $s_{0}$ and $t_{0}$, and edges $(s_{0},s)$ and
$(t,t_{0})$ both with zero cost. Set $\CAP_{s}=\CAP_{t}=\sum_{v\in V\setminus\{s,t\}}\CAP_{v}$.
By treating $s_{0}$ and $t_{0}$ as the new source and sink, our assumption
is justified.

Next, we construct a bipartite graph $G'=(L,R,E')$, a cost vector
$c'\in\Z^{E'}$, a demand vector $b'\in\Z^{L\cup R}$ as follows.
For each $v\in V\setminus\{t\}$, we create a copy $v_{l}\in L$.
For each $v\in V\setminus\{s\}$, we create a copy $v_{r}\in R$.
For $v\in V\setminus\{s,t\}$, add an edge $(v_{l},v_{r})\in E'$
with cost $c'_{(v_{l},v_{r})}=0$. For each $e=(u,v)\in E$, add $(u_{l},v_{r})\in E'$
with cost $c'_{e}=c_{e}$. (Note that here we exploit the assumption that
$s$/$t$ has no incoming/outgoing edges.) We assign the demand
vector $b'$ as follows: $b'_{s_{l}}=k$, $b'_{t_{r}}=k$, and, for
every $v\in V\setminus\{s,t\}$, $b'_{v_{l}}=b'_{v_{r}}=\CAP_{v}$.

Observe that every perfect $b'$-matching $x$ in $G'$ corresponds
to a set $\cC$ of cycles in $G\cup\{(t,s)\}$ where each vertex $v$
is part of at most $\CAP_{v}$ cycles and $(t,s)$ is part of exactly
$k$ cycles. Hence, $\cC$ in turn corresponds to an $s$-$t$ flow
$f$ in $G$ of value $k$ that respect the vertex capacities $\CAP$.
Moreover, the cost of $x$ is exactly the same as the cost of $f$. 

From the above three reductions, we conclude:
\begin{proposition}
Consider the following problems when the input graph has $n$ vertices
and $m$ edges, and $W$ denotes the largest number in the problem
instance:
\begin{itemize}
\item minimum weight bipartite perfect $b$-matching,
\item uncapacitated minimum-cost flow (a.k.a. transshipment), and
\item vertex-capacitated minimum-cost $s$-$t$ flow.
\end{itemize}
If one of the above problems can be solved in $T(m,n,W)$ time, then
we can solve the other problems in $T(m+O(n),O(n),O(n^3W^2))$ time.
\end{proposition}

\paragraph{Negative-weight single-source shortest paths $\rightarrow$ Transshipment
/ Dual of minimum weight perfect $b$-matching.}

We also consider the \emph{negative-weight single-source shortest
path} problem. In this problem, we are given a graph $G=(V,E)$,
a source $s\in V$, and a cost vector $c\in\Z^{E}$, the goal is to
compute a shortest path tree rooted at $s$ w.r.t. the cost $c$, or detect a negative cycle.

The problem can be reduced to the transshipment problem as follows. 
First, by performing depth first search, we assume that the source $s$ can reach all other vertices. (The unreachable vertices have distance $\infty$ and can be removed from the graph.)
Next, we
make sure that the shortest path is unique by randomly perturbing
the weight of each edge (exploiting the Isolation lemma as in, e.g., \cite[Section~6]{Cabello2013multiple}).
Then, we set the demand vector $b_{s}=n-1$ and $b_{v}=-1$ for all
$v\neq s$. As $s$ can reach all other vertices, we know that there exists a flow satisfying the demand $b$. 
Finally, we invoke the algorithm for solving the transshipment instance $G$ with this demand $b$. 
If the algorithm reports that the optimal cost is $-\infty$, then we report that there is a negative cycle. Otherwise, we obtain a flow $f$ that can be decomposed into $s$-$v$ paths $P_v$ for all $v \in V\setminus \{s\}$. Observe that $P_v$ must be a shortest $s$-$v$ path.
Since the shortest paths are unique, we have that the union of $P_v$ over all $v \in V \setminus \{s\}$, which is exactly the support of $f$, form a shortest path tree. 
Therefore, we simply return the support of $f$ as the answer.

\paragraph{Minimum mean cycle, deterministic MDP $\rightarrow$ Negative-weight single-source shortest path.} 
Given a cost function $c$ on edges,
the \emph{mean cost} of a directed cycle $\sigma$ is $\bar{c}(\sigma) = \frac{1}{|\sigma|}\cdot \sum_{e \in \sigma} c(e)$. The \emph{minimum mean cycle} problem is to find a cycle of minimum mean cost. The \emph{deterministic markov decision process (MDP)} problem \cite{Madani02} is the same problem except that we want to find a cycle with maximum mean cost. These two problems are clearly equivalent by negating all the edge cost.

Given a subroutine for detecting a negative cycle, Lawler \cite{Lawler1972} shows that we can also solve the minimum mean cycle problem on a graph of size $n$ by calling the negative cycle subroutine $O(\log (nW))$ times on graphs with the same order.

\paragraph{Plugging Algorithms to Reductions.}
Fianlly, by plugging the main algorithm from \Cref{thm:perfect_b_matching_fast}
into the above reductions, we immediately obtain several applications. 
\begin{corollary}[Exact Algorithms for Integral Problems]
\label{cor:app exact}
There are algorithms with running time $\Ot((m+n^{1.5})\log^{2}W)$
for solving the following problems \emph{exactly}:
\begin{itemize}[noitemsep]
\item minimum weight bipartite perfect $b$-matching, 
\item maximum weight bipartite perfect $b$-matching,
\item maximum weight bipartite $b$-matching,
\item uncapacitated minimum-cost flow (a.k.a. transshipment), 
\item vertex-capacitated minimum-cost $s$-$t$ flow, and
\item negative-weight single-source shortest path,
\item minimum mean cycle, and
\item deterministic MDP
\end{itemize}
when the input graph has $n$ vertices and $m$ edges, and $W$ denotes
the largest absolute value used for specifying any value in the problem.
Here, we assume that all numbers used for specifying the input are integral.
\end{corollary}

\subsubsection{Approximation Algorithms for Fractional Problems}

\label{sec:app frac}

Generally, for each the problem considered in the previous section, we can consider the same problem where the values that describe the input might not be integral.
We say that an algorithm solves a problem with additive approximation factor of $\epsilon$ if it always returns a solution whose cost is at most $\opt+\epsilon$ where $\opt$ is the cost of the optimal solution.

\begin{corollary}[Approximation Algorithms for Fractional Problems]
There are algorithms with running time $\Ot((m+n^{1.5})\log^{2}\frac{W}{\epsilon})$
for solving the following problems with \emph{additive} approximation factor of $\epsilon$:
\begin{itemize}[noitemsep]
\item minimum weight bipartite perfect $b$-matching (a.k.a.~optimal transport), 
\item maximum weight bipartite perfect $b$-matching,
\item maximum weight bipartite $b$-matching,
\item uncapacitated minimum-cost flow (a.k.a. transshipment), 
\item vertex-capacitated minimum-cost $s$-$t$ flow, and
\item negative-weight single-source shortest path,
\item minimum mean cycle, and
\item deterministic MDP
\end{itemize}
when the input graph has $n$ vertices and $m$ edges, and $W$ denotes
the ratio of the largest non-zero absolute value to the smallest non-zero absolute value used for specifying any value in the problem.
\end{corollary}

We note that the \emph{optimal transport} problem (a.k.a.~the \emph{transportation}
problem) is the same as the minimum weight bipartite perfect $b$-matching
problem except that we allow all vectors to be fractional, including the demand vector $b$, the cost vector $c$, and the solution vector $x$.
Several papers in the literature assume that the input graph is a complete bipartite
graph. We do not assume this. 

The idea of the proof above is to round up all the values from the description of the input to 
(integral) multiple of $1/\poly(nW/\epsilon)$ which takes $O(\log(nW/\epsilon))$ bits.
Then, we solve the rounded problem instance exactly using the algorithm \Cref{cor:app exact} for integral problems.
Consider an optimal solution $S$ with cost $\OPT$ before rounding.
Observe that after rounding, the rounded cost of this optimal solution $S$ can increase by at most $n^{2}\times\max_{e\in E}c(e)\times\frac{1}{\poly(nW/\epsilon)}\le\epsilon$.
So the optimal cost in the rounded instance has cost at most $\OPT+\epsilon$.

\section*{Acknowledgments}

We thank Yang Liu for helpful conversations, feedback on earlier drafts of the paper, and suggesting \Cref{lem:P2order} and its application.
This project has received funding from the European Research Council (ERC) under the European
Unions Horizon 2020 research and innovation programme under grant agreement No 715672. Danupon
Nanongkai is also partially supported by the Swedish Research Council (Reg. No. 2015-04659 and 2019-05622). Aaron Sidford is supported by NSF CAREER Award CCF-1844855 and PayPal research gift. Yin Tat Lee is supported by NSF awards CCF-1749609, CCF-1740551, DMS-1839116, Microsoft Research Faculty Fellowship, a Sloan Research Fellowship.
Di Wang did part of this work while at Georiga Tech,
and was partially supported by NSF grant CCF-1718533. Richard Peng was partially supported by NSF grants CCF-1718533 and CCF-1846218. Zhao Song was partially supported by Ma Huateng Foundation, Schmidt Foundation, Simons Foundation, NSF, DARPA/SRC, Google and Amazon.

\ifdefined\DEBUG
\newpage
\fi

\bibliographystyle{alpha}
\bibliography{ref}

\newcommand{\etalchar}[1]{$^{#1}$}
\begin{thebibliography}{KNPW11}

\bibitem[AKPS19]{adil2019iterative}
Deeksha Adil, Rasmus Kyng, Richard Peng, and Sushant Sachdeva.
\newblock Iterative refinement for $\ell_p$-norm regression.
\newblock In {\em Proceedings of the Thirtieth Annual ACM-SIAM Symposium on
  Discrete Algorithms (SODA)}, pages 1405--1424. SIAM,
  \url{https://arxiv.org/pdf/1901.06764.pdf}, 2019.

\bibitem[AMV20]{amv20}
Kyriakos Axiotis, Aleksander Madry, and Adrian Vladu.
\newblock Circulation control for faster minimum cost flow in unit-capacity
  graphs.
\newblock In {\em {FOCS}}, pages 93--104. {IEEE}, 2020.
\newblock \url{https//arxiv.org/pdf/2003.04863.pdf}.

\bibitem[Ans96]{Anstreicher96}
Kurt~M. Anstreicher.
\newblock Volumetric path following algorithms for linear programming.
\newblock {\em Math. Program.}, 76:245--263, 1996.

\bibitem[AS20]{AdilS20}
Deeksha Adil and Sushant Sachdeva.
\newblock Faster $p$-norm minimizing flows, via smoothed $q$-norm problems.
\newblock In {\em SODA}, pages 892--910. SIAM, 2020.

\bibitem[ASZ20]{AndoniSZ19}
Alexandr Andoni, Clifford Stein, and Peilin Zhong.
\newblock Parallel approximate undirected shortest paths via low hop emulators.
\newblock pages 322--335, 2020.
\newblock \url{https://arxiv.org/pdf/1911.01956.pdf}.

\bibitem[AWR17]{AltschulerWR17}
Jason Altschuler, Jonathan Weed, and Philippe Rigollet.
\newblock Near-linear time approximation algorithms for optimal transport via
  sinkhorn iteration.
\newblock In {\em {NIPS}}, pages 1964--1974, 2017.

\bibitem[BBG{\etalchar{+}}20]{BernsteinBNPSS20}
Aaron Bernstein, Jan van~den Brand, Maximilian~Probst Gutenberg, Danupon
  Nanongkai, Thatchaphol Saranurak, Aaron Sidford, and He~Sun.
\newblock Fully-dynamic graph sparsifiers against an adaptive adversary.
\newblock {\em CoRR}, abs/2004.08432, 2020.

\bibitem[Bel58]{bellman58}
R.~Bellman.
\newblock {On a Routing Problem}.
\newblock {\em Quarterly of Applied Mathematics}, 16(1):87--90, 1958.

\bibitem[BGS20]{BernsteinGS20scc}
Aaron Bernstein, Maximilian~Probst Gutenberg, and Thatchaphol Saranurak.
\newblock Deterministic decremental reachability, scc, and shortest paths via
  directed expanders and congestion balancing.
\newblock pages 1123--1134, 2020.

\bibitem[BJKS18]{BlanchetJKS18}
Jose~H. Blanchet, Arun Jambulapati, Carson Kent, and Aaron Sidford.
\newblock Towards optimal running times for optimal transport.
\newblock {\em CoRR}, abs/1810.07717, 2018.

\bibitem[BLSS20]{blss20}
Jan van~den Brand, Yin~Tat Lee, Aaron Sidford, and Zhao Song.
\newblock Solving tall dense linear programs in nearly linear time.
\newblock In {\em {STOC}}, pages 775--788. {ACM}, 2020.
\newblock \url{https://arxiv.org/pdf/2002.02304.pdf}.

\bibitem[BNS19]{BrandNS19}
Jan van~den Brand, Danupon Nanongkai, and Thatchaphol Saranurak.
\newblock Dynamic matrix inverse: Improved algorithms and matching conditional
  lower bounds.
\newblock In {\em {FOCS}}, pages 456--480. {IEEE} Computer Society, 2019.

\bibitem[Bra20]{b20}
Jan van~den Brand.
\newblock A deterministic linear program solver in current matrix
  multiplication time.
\newblock In {\em SODA}, pages 259--278. {SIAM}, 2020.

\bibitem[CCE13]{Cabello2013multiple}
Sergio Cabello, Erin~W Chambers, and Jeff Erickson.
\newblock Multiple-source shortest paths in embedded graphs.
\newblock {\em SIAM Journal on Computing}, 42(4):1542--1571, 2013.

\bibitem[CGL{\etalchar{+}}20]{cglnps20}
Julia Chuzhoy, Yu~Gao, Jason Li, Danupon Nanongkai, Richard Peng, and
  Thatchaphol Saranurak.
\newblock A deterministic algorithm for balanced cut with applications to
  dynamic connectivity, flows, and beyond.
\newblock In {\em {FOCS}}, 2020.
\newblock \url{https://arxiv.org/pdf/1910.08025.pdf}.

\bibitem[CGP{\etalchar{+}}18]{ChuGPSSW18}
Timothy Chu, Yu~Gao, Richard Peng, Sushant Sachdeva, Saurabh Sawlani, and
  Junxing Wang.
\newblock Graph sparsification, spectral sketches, and faster resistance
  computation, via short cycle decompositions.
\newblock In {\em {FOCS}}, pages 361--372. {IEEE} Computer Society, 2018.

\bibitem[CKM{\etalchar{+}}14]{CohenKMPPRX14}
Michael~B. Cohen, Rasmus Kyng, Gary~L. Miller, Jakub~W. Pachocki, Richard Peng,
  Anup~B. Rao, and Shen~Chen Xu.
\newblock Solving sdd linear systems in nearly $m\log^{1/2}n$ time.
\newblock In {\em Proceedings of the 46th Annual ACM Symposium on Theory of
  Computing (STOC)}, pages 343--352, 2014.

\bibitem[CLM{\etalchar{+}}15]{clmmps15}
Michael~B Cohen, Yin~Tat Lee, Cameron Musco, Christopher Musco, Richard Peng,
  and Aaron Sidford.
\newblock Uniform sampling for matrix approximation.
\newblock In {\em Proceedings of the 2015 Conference on Innovations in
  Theoretical Computer Science (ITCS)}, pages 181--190, 2015.

\bibitem[CLS19]{cls19}
Michael~B. Cohen, Yin~Tat Lee, and Zhao Song.
\newblock Solving linear programs in the current matrix multiplication time.
\newblock In {\em {STOC}}, pages 938--942. {ACM}, 2019.
\newblock \url{https://arxiv.org/pdf/1810.07896}.

\bibitem[CMSV17]{cmsv17}
Michael~B Cohen, Aleksander Madry, Piotr Sankowski, and Adrian Vladu.
\newblock Negative-weight shortest paths and unit capacity minimum cost flow in
  ${O}(m^{10/7} \log {W})$ time.
\newblock In {\em Proceedings of the Twenty-Eighth Annual ACM-SIAM Symposium on
  Discrete Algorithms (SODA)}, pages 752--771. SIAM, 2017.

\bibitem[CP15]{CohenP15}
Michael~B. Cohen and Richard Peng.
\newblock $\ell_p$ row sampling by lewis weights.
\newblock In {\em STOC}, pages 183--192. {ACM}, 2015.

\bibitem[Din70]{Dinic70}
Efim~A Dinic.
\newblock Algorithm for solution of a problem of maximum flow in networks with
  power estimation.
\newblock In {\em Soviet Math. Doklady}, volume~11, pages 1277--1280, 1970.

\bibitem[DK69]{dinic-kronrod-69}
E.~A. Dinic and M.~A. Kronrod.
\newblock {An Algorithm for the Solution of the Assignment Problem}.
\newblock {\em Soviet Math. Dokl}, 10:1324--1326, 1969.

\bibitem[DP14]{DuanP14}
Ran Duan and Seth Pettie.
\newblock Linear-time approximation for maximum weight matching.
\newblock {\em J. {ACM}}, 61(1):1:1--1:23, 2014.

\bibitem[DS08]{ds08}
Samuel~I Daitch and Daniel~A Spielman.
\newblock Faster approximate lossy generalized flow via interior point
  algorithms.
\newblock In {\em Proceedings of the fortieth annual ACM symposium on Theory of
  computing (STOC)}, pages 451--460, 2008.

\bibitem[Ege31]{egervary31}
E.~Egerv\'{a}ry.
\newblock Matrixok kombinatorius tulajdons\'{a}gair\'{o}l (hungarian) on
  combinatorial properties of matrices.
\newblock {\em Matematikai \'{e}s Fizikai Lapok}, 38:16--28, 1931.

\bibitem[EK72]{EdmondsK72}
Jack Edmonds and Richard~M. Karp.
\newblock Theoretical improvements in algorithmic efficiency for network flow
  problems.
\newblock {\em J. {ACM}}, 19(2):248--264, 1972.

\bibitem[Gab85]{Gabow85}
Harold~N. Gabow.
\newblock Scaling algorithms for network problems.
\newblock {\em J. Comput. Syst. Sci.}, 31(2):148--168, 1985.
\newblock announced at FOCS'83.

\bibitem[Gal14]{Gall14a}
Fran{\c{c}}ois~Le Gall.
\newblock Powers of tensors and fast matrix multiplication.
\newblock In {\em {ISSAC}}, pages 296--303. {ACM}, 2014.

\bibitem[GLPS10]{glps10}
Anna~C Gilbert, Yi~Li, Ely Porat, and Martin~J Strauss.
\newblock Approximate sparse recovery: optimizing time and measurements.
\newblock {\em SIAM Journal on Computing 2012 (A preliminary version of this
  paper appears in STOC 2010)}, 41(2):436--453, 2010.

\bibitem[Gol93]{goldberg93}
Andrew~V. Goldberg.
\newblock Scaling algorithms for the shortest paths problem.
\newblock In {\em SODA '93: Proceedings of the fourth annual ACM-SIAM Symposium
  on Discrete algorithms}, pages 222--231. Society for Industrial and Applied
  Mathematics, 1993.

\bibitem[GRST20]{GoranciRST20hierarchy}
Gramoz Goranci, Harald R{\"{a}}cke, Thatchaphol Saranurak, and Zihan Tan.
\newblock The expander hierarchy and its applications to dynamic graph
  algorithms.
\newblock {\em CoRR}, abs/2005.02369, 2020.

\bibitem[GT89]{GabowT89}
Harold~N. Gabow and Robert~Endre Tarjan.
\newblock Faster scaling algorithms for network problems.
\newblock {\em {SIAM} J. Comput.}, 18(5):1013--1036, 1989.

\bibitem[GT91]{GabowT91}
Harold~N. Gabow and Robert~Endre Tarjan.
\newblock Faster scaling algorithms for general graph-matching problems.
\newblock {\em J. {ACM}}, 38(4):815--853, 1991.

\bibitem[HIKP12]{hikp12a}
Haitham Hassanieh, Piotr Indyk, Dina Katabi, and Eric Price.
\newblock Nearly optimal sparse {F}ourier transform.
\newblock In {\em Proceedings of the forty-fourth annual ACM symposium on
  Theory of computing (STOC)}, pages 563--578. ACM,
  \url{https://arxiv.org/pdf/1201.2501.pdf}, 2012.

\bibitem[HK73]{HopcroftK73}
John~E. Hopcroft and Richard~M. Karp.
\newblock An $n^{5/2}$ algorithm for maximum matchings in bipartite graphs.
\newblock {\em {SIAM} J. Comput.}, 2(4):225--231, 1973.
\newblock Announced at FOCS'71.

\bibitem[IM81]{IbarraM81}
Oscar~H. Ibarra and Shlomo Moran.
\newblock Deterministic and probabilistic algorithms for maximum bipartite
  matching via fast matrix multiplication.
\newblock {\em Inf. Process. Lett.}, 13(1):12--15, 1981.

\bibitem[Iri60]{iri60}
M.~Iri.
\newblock A new method for solving transportation-network problems.
\newblock {\em Journal of the Operations Research Society of Japan}, 3:27--87,
  1960.

\bibitem[JSWZ20]{JiangSWZ20}
Shunhua Jiang, Zhao Song, Omri Weinstein, and Hengjie Zhang.
\newblock Faster dynamic matrix inverse for faster lps.
\newblock {\em CoRR}, abs/2004.07470, 2020.

\bibitem[Kap17]{k17}
Michael Kapralov.
\newblock Sample efficient estimation and recovery in sparse {FFT} via
  isolation on average.
\newblock In {\em 58th Annual IEEE Symposium on Foundations of Computer Science
  (FOCS)}. \url{https://arxiv.org/pdf/1708.04544}, 2017.

\bibitem[Kar73]{Karzanov73}
Alexander~V Karzanov.
\newblock On finding maximum flows in networks with special structure and some
  applications.
\newblock {\em Matematicheskie Voprosy Upravleniya Proizvodstvom}, 5:81--94,
  1973.

\bibitem[Kar84]{Karmarkar84}
Narendra Karmarkar.
\newblock A new polynomial-time algorithm for linear programming.
\newblock {\em Combinatorica}, 4(4):373--396, 1984.
\newblock Announced at STOC'84.

\bibitem[KLP{\etalchar{+}}16]{KyngLPSS16}
Rasmus Kyng, Yin~Tat Lee, Richard Peng, Sushant Sachdeva, and Daniel~A.
  Spielman.
\newblock Sparsified cholesky and multigrid solvers for connection laplacians.
\newblock In {\em {\em STOC'16: Proceedings of the 48th Annual ACM Symposium on
  Theory of Computing}}, 2016.

\bibitem[KLST99]{kao-99}
M.-Y. Kao, T.~W. Lam, W.-K. Sung, and H.-F. Ting.
\newblock A decomposition theorem for maximum weight bipartite matchings with
  applications to evolutionary trees.
\newblock In {\em {\em Proceedings of the 7th Annual European Symposium on
  Algorithms}}, pages 438--449, 1999.

\bibitem[KMP10]{KoutisMP10}
Ioannis Koutis, Gary~L. Miller, and Richard Peng.
\newblock Approaching optimality for solving {SDD} systems.
\newblock In {\em Proceedings of the 51st Annual IEEE Symposium on Foundations
  of Computer Science (FOCS)}, pages 235--244, 2010.

\bibitem[KMP11]{KoutisMP11}
Ioannis Koutis, Gary~L. Miller, and Richard Peng.
\newblock A nearly $m \log n$-time solver for {SDD} linear systems.
\newblock In {\em Proceedings of the 52nd Annual IEEE Symposium on Foundations
  of Computer Science (FOCS)}, pages 590--598, 2011.

\bibitem[KNPW11]{knpw11}
Daniel~M Kane, Jelani Nelson, Ely Porat, and David~P Woodruff.
\newblock Fast moment estimation in data streams in optimal space.
\newblock In {\em Proceedings of the forty-third annual ACM symposium on Theory
  of computing (STOC)}, pages 745--754, 2011.

\bibitem[KOSZ13]{KelnerOSZ13}
Jonathan~A. Kelner, Lorenzo Orecchia, Aaron Sidford, and Zeyuan~Allen Zhu.
\newblock A simple.
\newblock In {\em STOC'13: Proceedings of the 45th Annual ACM Symposium on the
  Theory of Computing}, pages 911--920. combinatorial algorithm for solving
  {SDD} systems in nearly-linear time. In, 2013.

\bibitem[KPSW19]{kpsw19}
Rasmus Kyng, Richard Peng, Sushant Sachdeva, and Di~Wang.
\newblock Flows in almost linear time via adaptive preconditioning.
\newblock In {\em Proceedings of the 51st Annual {ACM} {SIGACT} Symposium on
  Theory of Computing (STOC)}, pages 902--913.
  \url{https://arxiv.org/pdf/1906.10340.pdf}, 2019.

\bibitem[KS01]{klivans2001randomness}
Adam~R Klivans and Daniel Spielman.
\newblock Randomness efficient identity testing of multivariate polynomials.
\newblock In {\em Proceedings of the thirty-third annual ACM symposium on
  Theory of computing (STOC)}, pages 216--223, 2001.

\bibitem[KS16]{KyngS16}
Rasmus Kyng and Sushant Sachdeva.
\newblock Approximate gaussian elimination for laplacians - fast, sparse, and
  simple.
\newblock In {\em {FOCS}}, pages 573--582. {IEEE} Computer Society, 2016.

\bibitem[Kuh55]{kuhn55}
H.~W Kuhn.
\newblock The hungarian method for the assignment problem.
\newblock {\em Naval Research Logistics Quarterly}, 2:83--97, 1955.

\bibitem[Law72]{Lawler1972}
Eugene~L. Lawler.
\newblock {\em Optimal Cycles in Graphs and the Minimal Cost-To-Time Ratio
  Problem}, pages 37--60.
\newblock Springer Vienna, Vienna, 1972.

\bibitem[LHJ19]{LinHJ19}
Tianyi Lin, Nhat Ho, and Michael~I. Jordan.
\newblock On efficient optimal transport: An analysis of greedy and accelerated
  mirror descent algorithms.
\newblock In {\em {ICML}}, volume~97 of {\em Proceedings of Machine Learning
  Research}, pages 3982--3991. {PMLR}, 2019.

\bibitem[Li20]{Li19}
Jason Li.
\newblock Faster parallel algorithm for approximate shortest path.
\newblock In {\em {STOC}}, volume \url{https://arxiv.org/pdf/1911.01626.pdf},
  pages 308--321. {ACM}, 2020.

\bibitem[LNNT16]{lnnt16}
Kasper~Green Larsen, Jelani Nelson, Huy~L Nguyen, and Mikkel Thorup.
\newblock Heavy hitters via cluster-preserving clustering.
\newblock In {\em 57th Annual Symposium on Foundations of Computer Science
  (FOCS)}, pages 61--70. IEEE, \url{https://arxiv.org/pdf/1604.01357}, 2016.

\bibitem[LS13]{lee2013efficient}
Yin~Tat Lee and Aaron Sidford.
\newblock Efficient accelerated coordinate descent methods and faster
  algorithms for solving linear systems.
\newblock In {\em 2013 IEEE 54th Annual Symposium on Foundations of Computer
  Science}, pages 147--156. IEEE, 2013.

\bibitem[LS14]{ls14}
Yin~Tat Lee and Aaron Sidford.
\newblock Path finding methods for linear programming: Solving linear programs
  in ${O}( \sqrt{rank} )$ iterations and faster algorithms for maximum flow.
\newblock In {\em 55th Annual IEEE Symposium on Foundations of Computer Science
  (FOCS)}, pages 424--433. \url{https://arxiv.org/pdf/1312.6677.pdf},
  \url{https://arxiv.org/pdf/1312.6713.pdf}, 2014.

\bibitem[LS15]{LeeS15}
Yin~Tat Lee and Aaron Sidford.
\newblock Efficient inverse maintenance and faster algorithms for linear
  programming.
\newblock In {\em {FOCS}}, pages 230--249. {IEEE} Computer Society, 2015.

\bibitem[LS19]{ls19}
Yin~Tat Lee and Aaron Sidford.
\newblock Solving linear programs with $\sqrt{rank}$ linear system solves.
\newblock In {\em arXiv preprint}. \url{https://arxiv.org/pdf/1910.08033.pdf},
  2019.

\bibitem[LS20a]{ls20_focs}
Yang~P Liu and Aaron Sidford.
\newblock Faster divergence maximization for faster maximum flow.
\newblock In {\em arXiv preprint}. \url{https://arxiv.org/pdf/2003.08929.pdf},
  2020.

\bibitem[LS20b]{ls20_stoc}
Yang~P Liu and Aaron Sidford.
\newblock Faster energy maximization for faster maximum flow.
\newblock In {\em STOC}. \url{https://arxiv.org/pdf/1910.14276.pdf}, 2020.

\bibitem[LSZ19]{lsz19}
Yin~Tat Lee, Zhao Song, and Qiuyi Zhang.
\newblock Solving empirical risk minimization in the current matrix
  multiplication time.
\newblock In {\em {COLT}}, volume~99 of {\em Proceedings of Machine Learning
  Research}, pages 2140--2157. {PMLR}, 2019.
\newblock \url{https://arxiv.org/pdf/1905.04447}.

\bibitem[Mad02]{Madani02}
Omid Madani.
\newblock Polynomial value iteration algorithms for detrerminstic mdps.
\newblock In {\em UAI}, pages 311--318, 2002.

\bibitem[Mad13]{m13}
Aleksander Madry.
\newblock Navigating central path with electrical flows: From flows to
  matchings, and back.
\newblock In {\em {FOCS}}, pages 253--262. {IEEE} Computer Society, 2013.

\bibitem[Mad16]{m16}
Aleksander Madry.
\newblock Computing maximum flow with augmenting electrical flows.
\newblock In {\em 2016 IEEE 57th Annual Symposium on Foundations of Computer
  Science (FOCS)}, pages 593--602. IEEE, 2016.

\bibitem[Moo59]{moore59}
E.~F. Moore.
\newblock {The Shortest Path Through a Maze}.
\newblock In {\em {\em Proceedings of the International Symposium on the Theory
  of Switching}}, pages 285--292, 1959.

\bibitem[MS04]{MuchaS04}
Marcin Mucha and Piotr Sankowski.
\newblock Maximum matchings via gaussian elimination.
\newblock In {\em {FOCS}}, pages 248--255. {IEEE} Computer Society, 2004.

\bibitem[Mun57]{munkres57}
J.~Munkres.
\newblock {Algorithms for the Assignment and Transportation Problems}.
\newblock {\em Journal of SIAM}, 5(1):32--38, 1957.

\bibitem[{Net}56]{ford56}
R.~Ford {Network Flow Theory}.
\newblock {\em Paper P-923}.
\newblock The RAND Corperation, Santa Moncia, California, 1956.

\bibitem[NS17]{ns17}
Danupon Nanongkai and Thatchaphol Saranurak.
\newblock Dynamic spanning forest with worst-case update time: adaptive, las
  vegas, and ${O}(n^{1/2 - \epsilon})$-time.
\newblock In {\em Proceedings of the 49th Annual ACM SIGACT Symposium on Theory
  of Computing (STOC)}, pages 1122--1129, 2017.

\bibitem[NS19]{ns19}
Vasileios Nakos and Zhao Song.
\newblock Stronger l\({}_{\mbox{2}}\)/l\({}_{\mbox{2}}\) compressed sensing;
  without iterating.
\newblock In {\em {STOC}}, pages 289--297. {ACM}, 2019.

\bibitem[NSW17]{nsw17}
Danupon Nanongkai, Thatchaphol Saranurak, and Christian Wulff{-}Nilsen.
\newblock Dynamic minimum spanning forest with subpolynomial worst-case update
  time.
\newblock In {\em FOCS}, pages 950--961. {IEEE} Computer Society, 2017.

\bibitem[NSW19]{nsw19}
Vasileios Nakos, Zhao Song, and Zhengyu Wang.
\newblock (nearly) sample-optimal sparse fourier transform in any dimension;
  ripless and filterless.
\newblock In {\em {FOCS}}, pages 1568--1577. {IEEE} Computer Society, 2019.

\bibitem[NT97]{NesterovT97}
Yurii~E. Nesterov and Michael~J. Todd.
\newblock Self-scaled barriers and interior-point methods for convex
  programming.
\newblock {\em Math. Oper. Res.}, 22(1):1--42, 1997.

\bibitem[Pag13]{p13}
Rasmus Pagh.
\newblock Compressed matrix multiplication.
\newblock {\em ACM Transactions on Computation Theory (TOCT)}, 5(3):1--17,
  2013.

\bibitem[Qua19]{Quanrud19}
Kent Quanrud.
\newblock Approximating optimal transport with linear programs.
\newblock In {\em SOSA@SODA}, volume~69 of {\em {OASICS}}, pages 6:1--6:9.
  Schloss Dagstuhl - Leibniz-Zentrum f{\"{u}}r Informatik, 2019.

\bibitem[Ren88]{Renegar88}
James Renegar.
\newblock A polynomial-time algorithm, based on newton's method, for linear
  programming.
\newblock {\em Math. Program.}, 40(1-3):59--93, 1988.

\bibitem[San04]{Sankowski04}
Piotr Sankowski.
\newblock Dynamic transitive closure via dynamic matrix inverse (extended
  abstract).
\newblock In {\em {FOCS}}, pages 509--517. {IEEE} Computer Society, 2004.

\bibitem[San05]{s05}
Piotr Sankowski.
\newblock {\em Algorithms -- ESA 2005: 13th Annual European Symposium, Palma de
  Mallorca, Spain, October 3-6, 2005. Proceedings}.
\newblock chapter Shortest Paths in Matrix Multiplication Time, pages 770--778.
  Springer Berlin Heidelberg, Berlin, Heidelberg, 2005.

\bibitem[San06]{s06}
Piotr Sankowski.
\newblock {\em Automata, Languages and Programming: 33rd International
  Colloquium, ICALP 2006, Venice, Italy, July 10-14, 2006, Proceedings, Part
  I}.
\newblock chapter Weighted Bipartite Matching in Matrix Multiplication Time,
  pages 274--285. Springer Berlin Heidelberg, Berlin, Heidelberg, 2006.

\bibitem[San09]{Sankowski09}
Piotr Sankowski.
\newblock Maximum weight bipartite matching in matrix multiplication time.
\newblock {\em Theor. Comput. Sci.}, 410(44):4480--4488, 2009.

\bibitem[Sch03]{schrijver2003combinatorial}
Alexander Schrijver.
\newblock {\em Combinatorial optimization: polyhedra and efficiency},
  volume~24.
\newblock Springer Science \& Business Media, 2003.

\bibitem[She17]{Sherman17}
Jonah Sherman.
\newblock Generalized preconditioning and undirected minimum-cost flow.
\newblock In {\em SODA}, pages 772--780. SIAM, 2017.

\bibitem[Shi55]{shimbel55}
A.~Shimbel.
\newblock {Structure in Communication Nets}.
\newblock In {\em In Proceedings of the Symposium on Information Networks},
  pages 199--203, Brooklyn, 1955. Polytechnic Press of the Polytechnic
  Institute of Brooklyn.

\bibitem[ST03]{SpielmanT03}
Daniel~A. Spielman and Shang-Hua Teng.
\newblock Solving sparse.
\newblock In {\em FOCS'03: Proceedings of the 44th Annual IEEE Symposium on
  Foundations of Computer Science}, pages 416--427, diagonally-dominant linear
  systems in time ${O}(m^{1.31})$. In, 2003. symmetric.

\bibitem[ST04]{SpielmanT04}
Daniel~A. Spielman and Shang{-}Hua Teng.
\newblock Nearly-linear time algorithms for graph partitioning, graph
  sparsification, and solving linear systems.
\newblock In {\em STOC'04: Proceedings of the 36th Annual ACM Symposium on the
  Theory of Computing}, pages 81--90. {ACM}, 2004.

\bibitem[ST11]{SpielmanT11}
Daniel~A. Spielman and Shang{-}Hua Teng.
\newblock Spectral sparsification of graphs.
\newblock {\em {SIAM} J. Comput.}, 40(4):981--1025, 2011.

\bibitem[SW19]{sw19}
Thatchaphol Saranurak and Di~Wang.
\newblock Expander decomposition and pruning: Faster, stronger, and simpler.
\newblock In {\em SODA}, pages 2616--2635. {SIAM}, 2019.

\bibitem[VA93]{VaidyaA93}
Pravin~M Vaidya and David~S Atkinson.
\newblock A technique for bounding the number of iterations in path following
  algorithms.
\newblock In {\em Complexity in Numerical Optimization}, pages 462--489. World
  Scientific, 1993.

\bibitem[Vai87]{Vaidya87}
Pravin~M. Vaidya.
\newblock An algorithm for linear programming which requires ${O}(( (m+n)n^2 +
  (m+n)^{1.5} n){L})$ arithmetic operations.
\newblock In {\em {STOC}}, pages 29--38. {ACM}, 1987.

\bibitem[Vai89]{Vaidya89a}
Pravin~M. Vaidya.
\newblock Speeding-up linear programming using fast matrix multiplication
  (extended abstract).
\newblock In {\em FOCS}, pages 332--337. {IEEE} Computer Society, 1989.

\bibitem[Vai91]{Vaidya90}
Pravin~M. Vaidya.
\newblock Solving linear equations with symmetric diagonally dominant matrices
  by constructing good preconditioners.
\newblock Technical report, Unpublished manuscript, UIUC 1990. A talk based on
  the manuscript was presented at the IMA Workshop on Graph Theory and Sparse
  Matrix Computation Mineapolis, October 1991.

\bibitem[Wil12]{Williams12}
Virginia~Vassilevska Williams.
\newblock Multiplying matrices faster than coppersmith-winograd.
\newblock In {\em STOC}, pages 887--898. {ACM}, 2012.

\bibitem[Wul17]{w17}
Christian Wulff{-}Nilsen.
\newblock Fully-dynamic minimum spanning forest with improved worst-case update
  time.
\newblock In {\em Proceedings of the 49th Annual ACM SIGACT Symposium on Theory
  of Computing (STOC)}, pages 1130--1143, 2017.

\bibitem[YZ05]{yz05}
Raphael Yuster and Uri Zwick.
\newblock Answering distance queries in directed graphs using fast matrix
  multiplication.
\newblock pages 389--396, 2005.

\end{thebibliography}

\newpage
\appendix

\section*{Appendix}

\section{Initial Point and Solution Rounding}
\label{sec:initial_point}

In this section we discuss how to construct the initial point for our IPM 
and how to convert fractional to integral solutions.

In \Cref{sec:initial_point:initial_point} we show how to construct a feasible primal solution $x$
and a feasible dual slack $s$, such that $xs \approx \tau(x,s)$, as required by the IPM.
Unfortunately, for a given minimum cost $b$-matching instance we are not able to construct such an initial point.
Instead, we replace the given cost vector $c$ by some alternative cost vector $c'$. 
The idea is to then run the IPM in reverse,
so that our solution pairs $(x,s)$ becomes less tailored to $c'$ 
and further away from the constraints, 
i.e. the points move closer towards the center of the polytope. 
Once we are sufficiently close to the center
it does not matter if we use cost vector $c$ or $c'$, so we can simply switch the cost vector.
That switching the cost vector is indeed possible is proven in \Cref{sec:initial_point:switch}.
After switching the cost vector, we then run the IPM again 
as to move the point $(x,s)$ closer towards the optimal solution.

Our IPM is not able to maintain a truly feasible primal solution $x$,
instead the solution is only approximately feasible.
Thus we must discuss how to turn such an approximately feasible solution into a truly feasible solution,
which is done in \Cref{sec:initial_point:feasible}.
Further, the solution we obtain this way might be fractional,
so at last in \Cref{sec:initial_point:integral} we show how to obtain an integral solution.

\subsection{Finding an Initial Point 
(Proof of \Cref{lem:implement:short_step_log_barrier,lem:implement:short_step_ls_barrier})}
\label{sec:initial_point:initial_point}

We want to solve minimum cost $b$-matching instances.
In order to construct an initial point, 
we reduce this problem to uncapacitated minimum cost flows.
For the reduction we first construct the underlying graph for the flow instance 
from a given bipartite $b$-matching instance
as described in \Cref{def:starred_flow_graph}.
We are left with constructing a feasible primal solution $x$, 
a cost vector $c'$ 
and a feasible slack $s$ of the dual 
w.r.t. the cost vector $c'$,
such that $xs \approx \tau(x,s)$ as required by our IPM.

In \Cref{sec:log_barrier_method} and \Cref{sec:ls_barrier_method} we present two different IPMs,
where the first uses $\tau(x,s) = 1$ 
and the second uses $\tau(x,s) = \sigma(x,s) + n/m$
where $\sigma(x,s):=\sigma(\mx^{1/2-\alpha}\ms^{-1/2-\alpha}\ma)$
for $\alpha =\frac{1}{4\log(4m/n)}$.
Hence we will also present two lemmas for construct the initial point,
one for each option of $\tau(x,s)$.

\lemInitialPointSimple*

\begin{proof}
We first construct the primal solution $x$, 
then we construct the dual slack $s$,
and at last we construct a cost vector $c$ for which there is a dual solution $y$ with slack $s$.

\paragraph{Primal solution (feasible flow)}

Let $n,m$ be the number of nodes and edges in $G'$.
For our primal solution $x \in \R^{E'}$ be define the following flow
\begin{align*}
x_{u,v} &= \frac{1}{n} \text{ for $u \in U$, $v \in V$,$(u,v)\in E$,} \\
x_{u,z} &= b_u - \frac{\deg_G(u)}{n} \text{ for $u \in U$,} \\
x_{z,v} &= b_v - \frac{\deg_G(v)}{n} \text{ for $v \in V$.} 
\end{align*}
Here we have $x > 0$ because $b_v \ge 1$ 
and $\deg_G(v) < n$ for all $v \in V \cup U$.
The flow also satisfies the demand $d$, 
because for every $u \in U$ there is a total amount of 
$\sum_{(u,v) \in E} 1/n + b_u - \deg_G(u)/n = b_u = d_u$ leaving $u$, 
and likewise a total flow of $b_v$ reaching $v \in V$.
Further, the flow reaching and leaving $z$ is $0$, because 
\begin{align*}
\sum_{u\in U} (b_u - \deg_G(u)/n) 
&= 
(\sum_{u \in U} b_u) + |E|/n \\
&= 
(\sum_{v \in V} b_v) + |E|/n \\
& = \sum_{v\in V} (b_v - \deg_G(v)/n),
\end{align*}
where the first term i the flow reaching $z$ and the last term is the flow leaving $z$.
Thus in summary, the flow $x$ is a feasible flow.

\paragraph{Dual slack and cost vector}

Now, we define the slack $s$ and the cost vector $c'$.
The slack of the dual are of the form $s_{u,v} = c'_{u,v} - y _u + y_v$ for $(u,v) \in E'$.
We simply set $y = \zerovec_n$ and $s_{u,v} = c'_{u,v} = 1/x_{u,v}$ for every $(u,v) \in E'$,
so then $xs = 1$.
As $1/n \le x_{u,v} \le b_u - \deg_G(u)/n \le \|b\|_\infty$,
we get $ \|b\|_\infty^{-1} \le c' \le n$.

\end{proof}

\begin{lemma}[\cite{CohenP15}, Theorem 12 of \cite{blss20}]\label{lem:lewis_weight}
Given $p\in(0,4)$, $\eta>0$ and incidence matrix $\mA\in\R^{m\times n}$ w.h.p.~in $n$, 
we can compute $w\in\R^m_{>0}$ with $w\approx_{\epsilon}\sigma(\mW^{\frac{1}{2}-\frac{1}{p}}\mA)+\eta\onevec$ in $\tilde{O}(m\poly(1/\epsilon))$ time. 
\end{lemma}

We remark that the complexity stated in \cite{blss20} is $\tilde{O}((\nnz(\mA)+n^\omega)\poly(1/\epsilon))$,
because it is stated for general matrices. However, when using a Laplacian Solver (\Cref{lem:laplacian_solver}) 
the $n^\omega$ term becomes $\tilde{O}(m)$.

\lemInitialPoint*

\begin{proof}
We first construct the primal solution $x$, 
then we construct the dual slack $s$,
and at last we construct a cost vector $c$ for which there is a dual solution $y$ with slack $s$.

\paragraph{Primal solution (feasible flow)}

The construction of the primal solution $x$ is the same as the one constructed in the proof of~\Cref{lem:initial_point_simple}, and its feasibility follows the same argument.
\paragraph{Dual slack}

Let $\mA$ be the incidence matrix of graph $G'$ and define $\mA' := \mX \mA$.
We construct a vector $s'$ with the property 
$$
s' \approx_\epsilon \sigma(\mS'^{-1/2-\alpha} \mA') + \frac{n}{m},
$$
This construction is done via \Cref{lem:lewis_weight}.

Next, define $s := s'/x$, so the vectors $x$ and $s$ satisfy
\begin{align*}
xs
&=
s' 
\approx_\epsilon
\sigma(\mS'^{-1/2-\alpha} \mA') + \frac{n}{m} 
=
\sigma(\mS'^{-1/2-\alpha} \mX \mA) + \frac{n}{m} \\
&=
\sigma((\mX\mS)^{-1/2-\alpha} \mX \mA) + \frac{n}{m}
=
\sigma(\mX^{1/2-\alpha}\mS^{-1/2-\alpha} \mA) + \frac{n}{m} \\
&=
\tau(x,s)
\end{align*}
Thus the vector $s$ satisfies $xs \approx_\epsilon \tau(x,s)$ 
and we are only left with making sure that $s$ is indeed a slack vector of the dual.

\paragraph{Constructing cost vector}

We define the cost vector $c' := s$ and consider the dual solution $y = \zerovec_{n}$.
The dual problem of the uncapacitated min-cost-flow is
\begin{align*}
c'_{u,v} &\ge y_u - y_v \text{ for $(u,v) \in E'$}.
\end{align*}
So by setting $y = \zerovec_{n}$ the dual slack is exactly $c' = s$,
i.e. the vector $s$ is indeed a slack of the dual problem.
As $c' := s$, we can bound
\begin{align*}
c'_i &\le e^\epsilon \tau(x,s) /x_i \le 3n \text{ and }\\
c'_i &\ge e^{-\epsilon} \tau(x,s)/x_i \ge \frac{n}{3m\|d\|_\infty},
\end{align*}
where $x \le \|d\|_\infty$ comes from the fact 
that any edge incident to $u \in U$ or $v \in V$
can carry at most $|d_u|$ or $|d_v|$ units of flow,
because all edges are directed away from $u$ and towards $v$.
\end{proof}

\subsection{Switching the Cost Vector  (Proof of \Cref{lem:switch_cost})}
\label{sec:initial_point:switch}

The construction of our initial point replaces the given cost vector by some other cost vector.
We now show that it is possible to revert this replacement, 
provided that our current point $(x,s)$ 
is far enough away from the optimal solution. Below we restate \Cref{lem:switch_cost}. 

\lemSwitchCost*

\begin{proof}

Let $y$ be the dual solution with $\mA y + s = c$,
and let $s' := s+c'-c$ be the slack vector when replacing cost vector $c$ by $c'$.
Then $\mA y + s' = \mA y + s + c' - c = c'$, so if $s' \ge 0$, then it is indeed a valid slack vector.

In order to show that $s' \ge 0$ and $xs' \approx_{\epsilon/2} \mu \tau(x,s')$,
we will bound the relative difference between $s$ and $s'$, 
i.e. we want an upper bound on
$$
\frac{s' - s}{s} = \frac{c' - c}{s}.
$$
This relative difference is small, when $s$ is large. 
So we want to lower bound $s$. 
By $xs \approx_{1/2} \mu \tau(x,s)$ we have
\begin{align}
s_i \ge \frac{\mu \tau(x,s)_i}{2x_i} \ge \frac{\mu n}{2mx_i} \text{ for all }i,
\label{eq:upper_bound_dual_slack}
\end{align}
where we used that $\tau(x,s) \ge m/n$.
In order to bound this term, we must first find an upper bound on $x$.

By \Cref{lem:feasible_projection} and the assumption 
$\|\mA x - d \|_{(\mA \mX \mS^{-1} \mA)^{-1}} \le \sqrt{\mu} / 4$
we have that there is a feasible $x'$ with $\|x^{-1}(x'-x)\|_\infty \le 1/2$.
For this feasible $x'$ we know that $x' \le \|d\|_\infty$
and thus $x \le 1.5 \|d \|_\infty$.

By \eqref{eq:upper_bound_dual_slack} we have
$$
s_i \ge \frac{\mu n}{2mx_i} \ge \frac{\mu n}{1.5 \|d\|_\infty}
$$
which in turn implies
$$
\|\frac{s' - s}{s}\|_\infty 
\le 
(\|c'\|_\infty + \|c\|_\infty) \frac{1.5 \|d\|_\infty}{\mu n}
\le
(\|c'\|_\infty + 3n) \frac{1.5 \|d\|_\infty}{\mu n}
\le
\frac{4.5 \|c'\|_\infty \|d\|_\infty}{\mu}
$$
Thus for $\mu \ge 45 \|c'\|_\infty \|d\|_\infty \epsilon$ we can bound
$$
\left\|\frac{s' - s}{s}\right\|_\infty
\le
\frac{\epsilon}{10}.
$$
This then implies $xs' \approx_{\epsilon/2} \mu \tau(x,s')$.
\end{proof}

\subsection{Rounding an Approximate to a Feasible Solution (Proof of \Cref{lem:make_flow_feasible})}
\label{sec:initial_point:feasible}

Our IPM is not able to maintain a feasible primal solution $x$,
instead we can only guarantee that $\frac{1}{\mu} \| \mA^\top x - b \|_{(\mA^\top \mX \mS^{-1} \mA)^{-1}}^2$ is small.
We now show that this is enough to construct a feasible solution
that does not differ too much from the provided approximately feasible solution.
We further show how small the progress parameter $\mu$ of the IPM must be, 
in order to obtain a small additive $\epsilon$ error compared to the optimal solution. Below we restate \Cref{lem:make_flow_feasible}. 

\lemMakeFlowFeasible*

In order to prove \Cref{lem:make_flow_feasible},
we will first prove a helpful lemma that bounds $\| \mA^\top x - b \|_2$
with respect to $\frac{1}{\mu} \| \mA^\top x - b \|_{(\mA^\top \mX \mS^{-1} \mA)^{-1}}^2$.

\begin{lemma}\label{lem:demand_error_bound}
Let $G'$ be a uncapacitated flow instance for demand $d$ and cost $c$
where for every feasible flow $f$ we have $f \le \|d\|_\infty$.
Let $\mA$ be the incidence matrix of $G'$.
Given a primal dual pair $(x,s)$ with infeasible $x$ and $xs \approx_1 \tau(x,s)$,
define 
$$
\delta := \frac{1}{\mu} \cdot \| \mA^\top x - d \|_{(\mA^\top \mX \mS^{-1} \mA)^{-1}}^2.
$$
Then
$$
\| \mA^\top x - d \|_2 \le 6 m \|d\|_\infty^2 \delta
$$
\end{lemma}

\begin{proof}
By \Cref{lem:feasible_projection}
we know that there is a feasible $x'$ with $\|x^{-1}(x'-x)\|_\infty \le 3\delta$.
\begin{align*}
\| \mA^\top x - d \|_2^2
&\le
\frac{1}{\lambda_{\min}((\mA^\top \mX \mS^{-1} \mA)^{-1})}
\| \mA^\top x - d \|_{(\mA^\top \mX \mS^{-1} \mA)^{-1}}^2 \\
&=
\lambda_{\max}(\mA^\top \mX \mS^{-1} \mA)
\mu \delta \\
&\le
\max_{i} \{ x_i/s_i \} \lambda_{\max}(\mA^\top \mA) \mu \delta \\
&\le
2n \max_{i} \{ x_i/s_i \} \mu \delta \\
\end{align*}
Because of $xs \approx_1 \mu \tau(x,s)$ and $x \le \|d\|_\infty$ we have
\begin{align*}
s_i \ge \frac{\mu \tau(x,s)}{3 x_i} \ge \frac{n\mu}{3 m \|d\|_\infty}
\end{align*}
which implies
\begin{align*}
\frac{x_i}{s_i}
\le
\frac{3m\|d\|_\infty^2}{n\mu}
\end{align*}
and thus
\begin{align*}
\| \mA^\top x - d \|_2^2
\le
6 m \|d\|_\infty^2 \delta.
\end{align*}
\end{proof}

We can now prove \Cref{lem:make_flow_feasible}.

\begin{proof}[Proof of \Cref{lem:make_flow_feasible}]
The algorithm works as follows.
First we compute
$$
x' = x + \mX \mS^{-1} \mA \mH^{-1} (d - \mA^\top x)
$$
for some $\mH^{-1} \approx_{\delta} (\mA^\top \mX \mS^{-1} \mA)^{-1}$.
This can be done in $\tilde{O}(m\log \delta^{-1})$ time via a Laplacian solver.

Then for every $u \in U$ with $(\mA^\top x')_u > d_u$ 
we reduce the flow on some edges incident to $u$ 
such that the demand of $u$ is satisfied.
We also do the same for every $v \in V$ with $(\mA^\top x')_v < d_v$.
Let $x''$ be the resulting flow.
After that we route $(\mA^\top x'')_u - d_u$ flow from $u \in U$ to $z$ 
and likewise $d_v - (\mA^\top x'')_v$ flow from $z$ to $v \in V$.
Let $f$ be the resulting flow, then $f$ is a feasible flow as all demands are satisfied.
Constructing $f$ from $x'$ this way takes $\tilde{O}(m)$ time.

\paragraph{Optimality}

We are left with proving claim $c^\top f \le \OPT + \epsilon$.
By \Cref{lem:feasible_projection} the vector $x'$ satisfies
\begin{align}
\|\mA^\top x' - d\|_{(\mA \mX \mS^{-1} \mA)^{-1}}^2
&\le 
5\delta \cdot \|\mA^\top x - d\|_{(\mA \mX \mS^{-1} \mA)^{-1}}^2
\le
\frac{\delta\mu}{2}
\text{, and} \\
\|\mX^{-1}(x - x')\|_\infty^2
&\le 
\frac{5}{\mu} \cdot \| \mA x - b \|_{(\mA^\top \mX \mS^{-1} \mA)^{-1}}^2
\le
\frac{1}{2} \label{eq:x_is_close}
\end{align}
So by \Cref{lem:demand_error_bound} we have 
$$
\|\mA^\top x' - d\|_2^2
\le 
\frac{6 m \|d\|_\infty^2 \|\mA^\top x' - d\|_{(\mA \mX \mS^{-1} \mA)^{-1}}^2 }{\mu}
\le
3 m \|d\|_\infty^2 \delta.
$$
As this bounds the maximum error we have for the demand of each node,
this also bounds how much we might change the flow on some edge,
i.e. we obtain the bound
$$
\|x' - f\|_2
\le 
\|d\|_\infty \sqrt{3 m \delta}.
$$
Thus we have $$
c^\top f 
\le
c^\top x'
+
\|c\|_2 \|x-f\|_2
\le
c^\top x'
+
m \|c\|_\infty \|d\|_\infty \sqrt{3 \delta}.
$$
We are left with bounding $c^\top x'$. For this note that the duality gap of $(x',s)$ can be bounded by
\begin{align*}
\sum_{i \in [m]} x'_i s_i
\le
2
\sum_{i \in [m]} x_i s_i 
\le
6
\mu \sum_{i \in [m]} \tau(x,s)_i 
\le
12 \mu n,
\end{align*}
where we use \eqref{eq:x_is_close}, $xs \approx_1 \mu \tau(x,s)$, and $\sum_{i\in[m]}\tau(x,s)_i = 2n$.
As the duality gap is an upper bound on the difference to the optimal value, 
we have
$$
c^\top x'
\le
\OPT + 12\mu n
$$
so for $\mu \le \frac{\epsilon}{24 n}$
and $\delta \le \frac{\epsilon^2}{12(\|c\|_\infty\|d\|_\infty m)^2}$
we have
$$
c^\top f 
\le
\OPT + 12 \mu n + m \|c\|_\infty \|d\|_\infty\sqrt{3\delta}
\le
\OPT +
\epsilon.
$$
\end{proof}

\subsection{Rounding to an Integral Solution (Proof of \Cref{lem:isolation_lemma})}
\label{sec:initial_point:integral}
\label{sec:isolation}

In \Cref{sec:initial_point:feasible} we showed how to obtain a feasible nearly optimal fractional solution.
We now show that we can also obtain a feasible truly optimal integral solution 
by rounding each entry of the solution vector $x$ to the nearest integer.
This only works, if the optimal solution is unique.
In order to make sure that the optimal solution is unique, we use a variant of the Isolation-lemma. 
In high-level, this lemma shows that by adding a small random cost to each edge, 
the minimum-cost flow becomes unique. Below we restate \Cref{lem:isolation_lemma}.

\lemRoundIntegral*

A variant of \Cref{lem:isolation_lemma} was proven in \cite{ds08} via the following result from \cite{klivans2001randomness}.

\begin{lemma}[{\cite[Lemma 4]{klivans2001randomness}}]
\label{lem:cited_isolation_lemma}
Let $C$ be any collection of distinct linear forms in variables $z_1,...,z_\ell$ 
with coefficients in the range $\{0,...,K\}$.
If $z_1,...,z_\ell$ are independently chosen uniformly at random in $\{0,...,K\ell/\epsilon\}$, 
then, with probability greater than $1-\epsilon$, 
there is a unique form of minimal value at $z_1,...,z_\ell$.
\end{lemma}

Here we prove \Cref{lem:isolation_lemma} via a similar proof to the one in \cite{ds08}.

\begin{proof}[Proof of \Cref{lem:isolation_lemma}]
We start by proving that $\tilde{\Pi}$ has a unique optimal solution.
We then show that rounding a nearly optimal fractional matching results in the optimal integral matching.

\paragraph{Unique Optimal Solution}

Let $C \subset \R^E$ be the set of optimal integral flows.
Then we can interpret each $f \in C$ as a linear form 
$z \mapsto \langle f, z\rangle$ in the variables $z_1,...,z_m$.
Further we have $0 \le f_e \le W$ for all $f \in C$ and $e \in E$. 

Now consider $z \in \R^m$ where each $z_i$ 
is an independent uniformly sampled element from $\{0,...,K\ell/\epsilon\}$.
By \Cref{lem:isolation_lemma} the minimizer $\argmin_{f \in C} \langle f, z \rangle$ is unique with probability at least $1/2$.
As $\langle f, c \rangle = \langle f', c \rangle$ for all $f,f' \in C$,
we have that for $c' := c + z/(4m^2W^2)$ there also exists a unique $f \in C$ that minimizes $\langle f, c' \rangle$.
Further note that $\langle f, c' \rangle < \langle f, c \rangle + 1$
so this $f$ is also the unique optimal solution to $\tPi$.

\paragraph{Rounding to the optimal integral matching}

Since we added random multiples of $\frac{1}{4m^2W^2}$ to each edge cost,
the second best \emph{integral} matching $f'$ has cost at least $\opt(\tPi) + \frac{1}{4m^2W^2}$.

Now assume we have some feasible fractional matching $\tf$.
As every feasible fractional matching is a convex combination of feasible integral matchings,
we can write $tf = \lambda f + (1-\lambda) g$,
where $\lambda \in [0,1]$, $f$ is the optimal (integral) matching, 
and $g$ is a feasible fractional matching with cost at least $\opt(\tPi) + \frac{1}{4m^2W^2}$
(i.e. $g$ is a convex combination of non-optimal feasible integral matchings).

If $\tf$ has cost at most $\opt(\tPi)+\frac{1}{12 m^2 W^3}$, then $\lambda \ge 1-\frac{1}{3W}$
and thus $\|\tf-f\|_\infty \le 1/3$.
So by rounding the entries of $\tf$ to the nearest integer we obtain the optimal flow $f$.
\end{proof}

\section{Leverage Score Maintenance}
\label{sec:leverage_score}

In this section we explain how to obtain the following data structure 
for maintaining approximate leverage scores
via results from \cite{blss20}.

\begin{theorem}\label{thm:leverage_score_maintenance}
There exists a Monte-Carlo data-structure%
, which works against an adaptive adversary, that supports the following procedures: 
\begin{itemize}
\item \textsc{Initialize}$(\mA \in \R^{m \times n}, g \in \R^m, \epsilon \in (0,1) )$: 
	Given incidence matrix $\mA \in \R^{m \times n}$, scaling $g \in \R^m$ and accuracy $\epsilon>0$, 
	the data-structure initializes in $\tilde{O} ( m \epsilon^{-2})$ time.
\item \textsc{Scale}$(I \subset [m], u \in \R^I)$: 
	Sets $g_{I} = u$ in $\tilde{O}( |I| \epsilon^{-2} )$ time.
\item \textsc{Query}$()$: 
	Let $g\t \in \R^m$ be the vector $g \in \R^m$ during $t$-th call to \textsc{Query}
	and assume $g\t \approx_{1/32} g^{(t-1)}$. 
	W.h.p. in $n$ the data-structure outputs a vector $\ttau \in \R^m$
	such that $\ttau \approx_{\epsilon} \sigma(\sqrt{\mG\t} \mA) + \frac{n}{m} \vones$. 
	Further, with probability at least $\frac{1}{2}$ the total cost over $K$ steps is 
	\begin{align*}
	\tilde{O}\Big( K^2 \cdot \epsilon^{-4} \cdot m/n + K ( n \epsilon^{-2} + n \epsilon^{-1} \log W )\Big),
	\end{align*}
	where $W$ is a bound on the ratio of largest to smallest entry of any $g\t$.
\end{itemize}

Let $g^{(k)}$
be the state of $g$
during the $k$-th call to \textsc{Query}
and define $\tau(g) := \sigma(\sqrt{\mG} \mA) + \frac{n}{m} \vones$.
Then the complexity bound on \textsc{Query} holds, 
if there exists a sequence $\tg^{(1)},...,\tg^{(K)} \in \R^m$  where for all $k\in[K]$
\begin{gather}
g_i^{(k)} \in \left( 1 \pm \frac{1}{16\cdot 144 \epsilon \log^2 n} \right) \tg_i^{(k)} \text{ for all $i \in [m]$,  
and} \label{eq:closeness_assumption}\\
\frac{1}{\epsilon^2}\|\mG^{(k)-1}(g^{(k+1)} - g^{(k)})\|_{\tau(g^{(k)})}^2 
+ \|(\mT(g^{(k)})^{-1}(\tau(g^{(k+1)}) - \tau(g^{(k)}))\|_{\tau(g^{(k)})}^2
= O(1). \label{eq:stability_assumption}
\end{gather}

\end{theorem}

We use the following \Cref{lem:leverage_score_maintenance} from \cite{blss20}
which describes a data structure for maintaining approximate leverage scores.
Note that the complexity scales with the Frobenius norm of some matrix.
We obtain \Cref{thm:leverage_score_maintenance} by combining
\Cref{lem:leverage_score_maintenance} with another data structure
that guarantees that this norm bound is small.

\begin{lemma}[{\cite[Theorem 8, Algorithm 7, when specialized to incidence matrices]{blss20}}]\label{lem:leverage_score_maintenance}
There exists a Monte-Carlo data-structure%
, which works against an adaptive adversary, 
that supports the following procedures: 
\begin{itemize}
\item \textsc{Initialize}$(\mA \in \R^{m \times n}, g \in \R^m, \epsilon \in (0 , 1) )$: 
	Given incidence matrix $\mA\in\R^{m\times n}$, scaling $g\in\R^{m}$ and accuracy $\epsilon>0$, 
	the data-structure initializes in $\tilde{O}(m \epsilon^{-2})$ time.
\item \textsc{Scale}$(i \in [m], u \in \R)$: 
	Given $i\in[m]$ and $u\in\R$ sets $g_{i}=u$ in $\tilde{O}( \epsilon^{-2} )$ time.
\item \textsc{Query}$(\Psi\t \in \R^{n \times n}, \Psi\t\safe \in \R^{n \times n})$: 
	Let $g\t \in \R^m$ be the vector $g \in \R^m$ during $t$-th call to \textsc{Query},
	assume $g\t \approx_{1/32} g^{(t-1)}$ 
	and define $\mH^{(t)} = \mA^{\top} \mG\t \mA \in \R^{n \times n}$. 
	Given random input-matrices $\Psi\t\in\R^{n\times n}$ and $\Psi\t\safe \in \R^{n \times n}$ 
	such that $$\Psi\t\approx_{ \epsilon/(24 \log n) }(\mH\t)^{-1},
	\Psi\t\safe\approx_{\epsilon/(24\log n)}(\mH\t)^{-1}$$
	and any randomness used to generate $\Psi\t\safe$ is independent of the randomness used to generate $\Psi\t$, 
	w.h.p. in $n$ the data-structure outputs a vector $\ttau\in\R^m$ independent of $\Psi^{(1)},...,\Psi\t$ 
	such that $\ttau\approx_{\epsilon}\sigma(\sqrt{\mG\t}\mA)+\frac{n}{m} \vones$. 
	Further, the total cost over $K$ steps is 
	\begin{align*}
	&\tilde{O}\Big(\Big(\sum_{t\in[K]}\|\sqrt{\mG^{(t)}}\mA\Psi^{(t)}\mA^{\top}\sqrt{\mG^{(t)}}-\sqrt{\mG^{(t-1)}}\mA\Psi^{(t-1)}\mA^{\top}\sqrt{\mG^{(t-1)}}\|_{F}\Big)^{2}\cdot\text{\ensuremath{\epsilon^{-4}}}m/n \\
	&~
	+K\left(T_{\Psi\safe}+T_{\Psi}+n\epsilon^{-2} + n\epsilon^{-1}\log W\right)\Big)
	\end{align*}
	where $T_{\Psi},T_{\Psi\safe}$ is the time required to multiply a vector with $\Psi^{(t)} \in \R^{n \times n}$ 
	and $\Psi^{(t)}\safe \in \R^{n \times n}$ respectively  
	(i.e. in case the matrices are given implicitly via data structures).
\end{itemize}
\end{lemma}

The initialization, scale and query complexity of \cite[Theorem 8, Algorithm 7]{blss20}
are larger by a factor of $n$ compared to the complexities stated in \Cref{lem:leverage_score_maintenance}.
This difference in complexity is because the data structure of \cite{blss20} was proven for general matrices $\mA$ which may have up to $n$ entries per row, 
while in \Cref{lem:leverage_score_maintenance} we consider only incidence matrices with only $2$ entries per row.
The proof of \cite[Theorem 8, Algorithm 7]{blss20} works by reducing the problem of maintaining leverage scores 
to the problem of detecting large entries of the product $\mdiag(g)\mA h$, 
for some vectors $g \in \R^m_{\ge 0}$, $h \in \R^n$.
This is exactly what our data structure of \Cref{lem:large_entry_datastructure} does.
By plugging our data structure of \Cref{lem:large_entry_datastructure} 
into the algorithm of \cite{blss20} and observing the sparsity of $\mA$, 
we obtain the faster complexities as stated in \Cref{lem:leverage_score_maintenance} 
for incidence matrices $\mA$. %

We combine \Cref{lem:leverage_score_maintenance} with the following \Cref{alg:inverse_maintenance}. 
The data structure of \Cref{alg:inverse_maintenance} guarantees 
that the Frobenius norm 
(as used in the complexity statement of \Cref{lem:leverage_score_maintenance}) 
stays small.
The properties of \Cref{alg:inverse_maintenance} are stated in \Cref{lem:inverse_maintenance}
and \Cref{lem:frobenius_property_clean}, where the latter bounds the change of the Frobenius norm.

\begin{lemma}[{\cite[Theorem 10, Algorithm 8]{blss20}}]
\label{lem:inverse_maintenance}
There exists a randomized data-structure that supports the following operations (\Cref{alg:inverse_maintenance})
\begin{itemize}
\item \textsc{Initialize}$(\og \in \R^m, \otau \in \R^m, \epsilon)$:
	Preprocess vector $\og \in \R^m, \otau \in \R^m$ and accuracy parameter $\epsilon > 0$
	in $O(m)$ time.
	Returns two vector $\tg, v \in \R^{m}$.
\item \textsc{Update}$(I \subset [m], s \in \R^{I}, t \in \R^{I})$ 
	Sets $\og_I \leftarrow s$, $\otau_I \leftarrow t$ in $\tilde{O}(|I|)$ time.
	Returns set $J \subset [m]$ and two vectors $\tg, v \in \R^{m}$.
	Here $\tg$ is returned via pointer and $J$ lists the indices $i$ where $\tg_i$ or $v_i$ changed compared to the previous output.
\end{itemize}
If the update does not depend on any of the previous outputs,
and if there is some matrix $\mA \in \R^{m \times n}$ with $\otau \approx_{0.5} \tau(\omG^{1/2} \mA)$, 
then $\tg, v$ satisfy following two properties:
(i) with high probability $\| v \|_0 = O(n \epsilon^{2} \log n)$.
(ii) $\tg \approx_\epsilon \og$,
and with high probability $\mA^\top \mV \mA \approx_{\epsilon} \mA^\top \omG \mA$.
\end{lemma}

\begin{lemma}[{\cite[Lemma 11]{blss20}}]
\label{lem:frobenius_property_clean}
Let $\og^{(k)},\otau^{(k)} \in \R^m$, $\tg^{(k)}, v^{(k)} \in \R^m$ be the vectors $\og$, $\otau$, $\tg$, $v$
right before the $k$-th call to \textsc{Update} of \Cref{lem:inverse_maintenance}.
Assume there exists a sequence $g^{(1)},...,g^{(K)} \in \R^m$ and matrix $\mA \in \R^{m \times n}$
with 
\begin{align*}
\og_i^{(k)} \in (1 \pm 1 / ( 16 \epsilon \log n) ) g_i^{(k)},
\otau_i^{(k)} \in (1 \pm 1 /( 16 \log n ) ) \tau(\mG^{(k)} \mA)_i, \forall i \in [m]
\end{align*} 
and 
\begin{align*}
\epsilon^{-2} \| (\mG^{(k)})^{-1}(g^{(k+1)} - g^{(k)})\|_{\tau(g^{(k)})}^2 + 
\|(\mT^{(k)})^{-1} (\tau^{(k+1)} - \tau^{(k)}) \|_{\tau(g^{(k)})} = O(1)
\end{align*}
for all $k \in [K-1]$, then
\begin{align*}
\E \left[
	\sum_{k\in[K-1]}
	\left\|
	\sqrt{\tmG^{(k+1)}} \mA (\mA^\top \mV^{(k+1)} \mA)^{-1} \mA^\top \sqrt{\tmG^{(k+1)}} - \sqrt{\tmG^{(k)}} \mA (\mA^\top \mV^{(k)} \mA)^{-1} \mA^\top \sqrt{\tmG^{(k)}}
	\right\|_F
	\right] 
\end{align*}
is at most $O(K\log^{5/2} n)$.
\end{lemma}

\begin{algorithm2e}[t!]
\caption{\label{alg:inverse_maintenance} Pseudocode for \Cref{lem:inverse_maintenance}, based on \cite[Algorithm 8]{blss20}.}
\SetKwProg{Members}{members}{}{}
\SetKwProg{Proc}{procedure}{}{}
\Members{}{
$\tg$, $\og$, $\ttau$, $\otau$, $\epsilon$, $v$, $y$, $\gamma := c \log n$ for some large constant $c$.
}
\Proc{\textsc{Initialize}$(\og, \otau, \epsilon)$}{
	$\og \leftarrow \og$, $\tg \leftarrow \og$, 
	$\otau \leftarrow \otau$, $\ttau \leftarrow \otau$, 
	$\epsilon \leftarrow \epsilon$,
	$y \leftarrow \zerovec_m$ \\
	$v_i \leftarrow \tg_i/ \min\{1, \gamma \epsilon^{-2} \ttau_i\}$ independently for each $i \in [m]$ with probability $\min\{1, \gamma \epsilon^{-2} \ttau_i\}$ and $v_i \leftarrow 0$ otherwise.
	
}
\Proc{\textsc{Update}$(I \subset [m],s \in \R^{I},t \in \R^{I})$}{
	$\og_I \leftarrow s$, $\otau_I \leftarrow t$ \\
	$y_i \leftarrow \frac{8}{\epsilon}(\og_{i}/\tg_{i}-1)$
		and $y_{i+n}\leftarrow 2(\otau_{i}/\ttau_{i}-1)$ for $i\in I$ \\
	Let $\pi:[2n]\rightarrow[2n]$ be a sorting permutation such that $|y_{\pi(i)}|\ge |y_{\pi(i+1)}|$ \\
	For each integer $\ell$, we define $i_{\ell}$ be the smallest
		integer such that $\sum_{j\in[i_{\ell}]}\otau_{\pi(j)}\geq2^{\ell}$. \\
	Let $k$ be the smallest integer such that $|y_{\pi(i_{k})}|\leq1-\frac{k}{2\left\lceil \log n\right\rceil }$
	$J \leftarrow \emptyset$ \\	
	\For{each coordinate $j\in[i_{k}]$}{
		Set $i=\pi(j)$ if $\pi(j)\leq n$ and set $i=\pi(j)-n$ otherwise \\
		$\tg_{i}\leftarrow \og_{i}$, $\ttau_{i}\leftarrow\otau_{i}$,
		$y_j \leftarrow 0$, $J \leftarrow J \cup \{ i \}$ \\
		$v_{i}\leftarrow\begin{cases}
			\tg_{i}/\min\{1,\gamma\epsilon^{-2}\cdot\ttau_{i}\} 
				& \text{with probability }\min\{1,\gamma\epsilon^{-2}\cdot\ttau_{i}\}\\
			0 & \text{otherwise}
		\end{cases}$
	}
	\Return $J$, $\tg$, $v$
}
\end{algorithm2e}

The complexity of \textsc{Update} in \cite[Lemma 11]{blss20} is a bit higher since they consider dense matrices.
Here we show that for sparse matrices the amortized update time is nearly linear in the input size.

\begin{lemma}
The function \textsc{Update} can be implemented to run in amortized $\tilde{O}(|I|)$ time.
\end{lemma}

\begin{proof}
We maintain a balanced binary search tree to sort the $y$. This takes $\tilde{O}(|I|)$ amortized time.
Then use that binary search to tree as a prefix data-structure,
i.e. each node is marked by the sum of it's children to the left.
Then the prefix of first $i$ elements is obtained by going from the root down to the $i$th element 
and adding the value of a node whenever we take the right edge.
When removing/inserting a node, updating these partial sums takes only $O(depth) = \tilde{O}(1)$ time.
Hence the amortized time for maintaining this structure is $\tilde{O}(|I|)$.

The rest of \textsc{Update} simply uses that data-structure, 
and can find the right value for $i_\ell$ and $k$ via binary-search.
Note that every index $i$ for which the loop of \textsc{Update} is executed, 
has been part of some input set $I$ before and will not be included in that loop again, 
unless it is part of another future input set $I$.
Hence we can charge the cost of the loop to previous calls to \textsc{Update}.
So we obtain $O(|I|)$ amortized cost for the loop per call to \textsc{Update}.
\end{proof}

\begin{algorithm2e}
\caption{\label{alg:leverage_score}Pseudocode for the data structure of \Cref{thm:leverage_score_maintenance}}
\SetKwProg{Members}{members}{}{}
\SetKwProg{Proc}{procedure}{}{}
\Members{}{
$D^{(\mathrm{stable})}$, $D^{(\tau)}$, $\tg \in \R^m$, $v \in \R^m$
}
\Proc{\textsc{Initialize}$(\mA\in\R^{m\times n},g\in\R^{m},\epsilon\in(0,1))$}{
	$\otau \leftarrow$ compute $\sigma(\sqrt{\mG} \mA) + n/m$. \\
	$\tg, v \leftarrow D^{(\mathrm{stable})}.\textsc{Initialize}(g, \otau, \epsilon / (144 \log n))$ (\Cref{lem:inverse_maintenance})
		\label{line:ls:initA}\\
	$D^{(\tau)}.\textsc{Initialize}(\mA, \tg, \epsilon / 10)$ (\Cref{lem:leverage_score_maintenance})
		\label{line:ls:initB}\\
}
\Proc{\textsc{Scale}$(I\subset[m],u\in\R^{|I|})$}{
	$J, \tg, v \leftarrow D^{(\mathrm{stable})}.\textsc{Update}(i, u)$ \\
	$D^{(\tau)}.\textsc{Scale}(i, \tg_i)$ for all $i\in J$. \label{line:ls:scale}\\
}
\Proc{\textsc{Query}$()$}{
	$\Psi\leftarrow$ Laplacian Solver for $(\mA^\top \mV \mA)^{-1}$ \label{line:ls:ls_sample_1}\\
	$\Psi\safe\leftarrow$ Laplacian Solver for $(\mA^\top \mW \mA)^{-1}$ where $w \in \R^m$ is a leverage score sample from $\otau$. \label{line:ls:ls_sample_2}\\
	$I, \otau \leftarrow D^{(\tau)}.\textsc{Query}(\Psi, \Psi\safe)$ \\
	\Return $I$, $\otau$
}
\end{algorithm2e}

By combining \Cref{lem:leverage_score_maintenance} with \Cref{lem:inverse_maintenance}
to get a good bound on the Frobenius norm via \Cref{lem:frobenius_property_clean}
we then obtain \Cref{thm:leverage_score_maintenance}.

\begin{proof}[Proof of \Cref{thm:leverage_score_maintenance}]

The algorithm is given by \Cref{alg:leverage_score}.
We will first prove correctness, and then bound the complexity.

\paragraph{Correctness}

Note that by \Cref{line:ls:initB} and \Cref{line:ls:scale}, 
the data structure $D^{(\tau)}$ always maintains the leverage score of $\sqrt{\tg} \mA$, 
where $\tg$ is the vector maintained by $D^{(sample)}$ (\Cref{lem:inverse_maintenance}).
We will argue that $D^{(\tau)}$ maintains $\otau \approx_{\epsilon/2} \tau(\sqrt{\tmG} \mA)$.
By $\tg \approx_{\epsilon/4} g$ this then implies $\otau \approx_\epsilon \tau(\mG \mA)$.

For \Cref{lem:inverse_maintenance} to return a good approximation,
we require that (i) $\Psi \approx_{1/(48 \log n)} \mA \tmG \mA$, 
(ii) $\Psi\safe \approx_{1/(48 \log n)} \mA \tmG \mA$,
and (iii) the randomness of $\Psi\safe$ does not depend on the randomness in $\Psi$.

By \Cref{lem:inverse_maintenance} the vector $v$ satisfies $\mA^\top \mV \mA \approx_{\epsilon/(144 \log n)} \mA \mG \mA$,
and $\tg \approx_{\epsilon/(144 \log n)} g$. 
So for a $\epsilon/(144 \log n)$ accurate Laplacian Solver (\Cref{lem:laplacian_solver}) 
we have $\Psi \approx_{\epsilon/(48 \log n)} \mA \tmG \mA$.
Thus condition (i) is satisfied.

We assume that the result of the last call to $D^{(\tau)}.\textsc{Query}$ satisfied $\otau \approx \tau(\sqrt{\tg} \mA)$.
Then by the assumption $g^{(t)} \approx_{1/32} g^{(t-1)}$ 
(i.e. the vector $g$ does not change too much between two calls to \textsc{Query}) 
we have that any entry of $\tau(\sqrt{\mG} \mA)$ can change by at most a constant factor 
between any two calls to \textsc{Query}.
As $\tg \approx g$, 
this means that this old $\otau$ from the previous iteration 
is still a constant factor approximation of the new $\tau(\sqrt{\tg} \mA)$.
So we can use the old $\otau$ of the previous iteration for the leverage score sampling in \Cref{line:ls:ls_sample_2}
to obtain a good approximation of $\mA \tmG \mA$.
Note that the approximation can be arbitrarily (i.e. $\delta > 0$) close by sampling $O(n\delta^{-2}\log n)$ entries.
Hence we can have a Laplacian solver with $\Psi\safe \approx_{\epsilon/(48 \log n)} \mA \tmG \mA$.
Also note that the randomness of $\Psi\safe$ depends only on $\otau$, 
which by \Cref{lem:leverage_score_maintenance} does not depend on $\Psi$.
Thus (iii) is also satisfied.

\paragraph{Complexity}

The complexities for \textsc{Initialization} and \textsc{Scale} 
come directly from \Cref{lem:leverage_score_maintenance} and \Cref{lem:inverse_maintenance}.
So let us consider \textsc{Query} instead.

By assumption \eqref{eq:stability_assumption} and \eqref{eq:closeness_assumption} the requirements for \Cref{lem:frobenius_property_clean} are satisfied and we have
\begin{align*}
\E\left[
	\sum_{k\in[K-1]}
	\left\|
	\sqrt{\mG^{(k+1)}} \mA (\mA^\top \mV^{(k+1)} \mA)^{-1} \mA^\top \sqrt{\mG^{(k+1)}} - \sqrt{\mG^{(k)}} \mA (\mA^\top \mV^{(k)} \mA)^{-1} \mA^\top \sqrt{\mG^{(k)}}
	\right\|_F
	\right] 
\end{align*}
is at most $O(K\log^{5/2} n)$.

Note that the Laplacian Solver for $(\mA^\top \mV^{(k)} \mA)^{-1}$ can be made $\delta >0$ accurate at cost $O(\log \delta^{-1})$ (\Cref{lem:laplacian_solver}),
so by choosing some $\delta = O(\epsilon/\poly(n))$, we can bound
$$
\E\left[
	\sum_{k\in[K-1]}
	\left\|
	\sqrt{\mG^{(k+1)}} \mA \Psi^{(k+1)} \mA^\top \sqrt{\mG^{(k+1)}}
	- \sqrt{\mG^{(k)}} \mA \Psi^{(k)} \mA^\top \sqrt{\mG^{(k)}}
	\right\|_F
	\right] 
\le
O(K\log^{5/2} n).
$$
By Markov inequality we thus have with probability at least $1/2$
that the queries to \Cref{lem:leverage_score_maintenance} have total complexity
$
\tilde{O}(
K^2 \cdot \epsilon^{-4} \cdot m/n + K(n\epsilon^{-2} + n \epsilon^{-1} \log W)
),$
where we used that both $\Psi$ and $\Psi\safe$ can be applied in $\tilde{O}(n \epsilon^{-2})$ time,
because of the sparsity of $v$ and $w$.

\end{proof}

\section{Degeneracy of $\mA$}
\label{sec:degenerate}

For all our data structures and the final algorithm from \Cref{sec:matching}, for simplicity we always assumed that the constraint matrix $\mA \in \R^{m \times n}$ of the linear program 
is the incidence matrix of some directed graph. 
However, this comes with some other problem:
Note that the rank of such a matrix is at most $n-1$, so $\mA$ is a so called \emph{degenerate} matrix,
and the Laplacian $\mL = (\mA^\top \mA)$ is singular, i.e. $\mL^{-1}$ does not exist.
However, the IPM presented in \Cref{sec:ipm} assumes 
that the constraint matrix is non-degenerate and that $(\mA^\top \mA)^{-1}$ exists.
So technically it is not clear if the IPM can indeed be used 
to solve the perfect $b$-matching problem as we did in \Cref{sec:matching}.
We outline here how we can assume that $\mA$ is indeed non-degenerate
and what small modifications must be performed to our data structures
to handle the new matrix $\mA$.

Assume $\mA$ is the incidence matrix of some minimum weight perfect bipartite $b$-matching instance,
and the corresponding formulation as a linear program is
$\min_{\mA^\top x = b, x \ge 0} c^\top x$.
Instead of using the incidence matrix $\mA$ and cost vector $c$, 
we use
$$
\mA' = \left[ \begin{array}{c}
\mA \\
\mI_{n \times n}
\end{array}\right], \text{ and }
c' = \left[ \begin{array}{c}
c \\
\|b\|_1\|c\|_\infty \cdot \onevec_n.
\end{array}\right]
$$
That is, we add an $n \times n$ identity block to the bottom of $\mA$
and add a large cost at the bottom of $c$ to make sure the optimal solution of
$\min_{\mA'^\top x'=b, x'\ge0} c'^\top x'$
has $x'_k =0$ for all $m < k \le m+n$,
so that the first $m$ coordinates of $x'$ are an optimal solution for the original linear program.

After this modification, the matrix $\mA'$ is full rank and the IPM from \Cref{sec:ipm} can be applied.
However, the matrix $\mA'$ is no longer an incidence matrix, 
so we have to modify our data structures that assume $\mA$ is an incidence matrix.

\paragraph{\Cref{lem:large_entry_datastructure}:}
We can detect the large entries of $\mA' h$ by using \Cref{lem:large_entry_datastructure}
for the first $m$ entries and simply computing the bottom $n$ entries of $\mA' h$ explicitly.
This additional cost of $O(n)$ is subsumed by \Cref{lem:large_entry_datastructure}.
We can handle the sampling in \textsc{Sample} the same way 
by computing the impact of the bottom $n$ entries of $\mA' h$ to the norm
by computing the bottom $n$ entries explicitly.
For the upper bound on the leverage scores and the leverage score sample,
note that adding more rows to $\mA$ only decreases the leverage score of the original rows.
Hence we can use the previous upper bound for the top $m$ rows
and simply use $1$ as an upper bound on the leverage scores of the bottom $n$ rows.

\paragraph{Laplacian Solvers:}
We use Laplacian solvers in several algorithms.
While $\mA'^\top \mG \mA'$ is no longer a Laplacian matrix,
it is still an SDDM (symmetric diagonally dominant matrix)
and solving systems in such matrices also has fast nearly linear time solvers.

\paragraph{Other data structures}

All other data structures do not require any special structure to the matrix $\mA'$
besides of very sparse rows.
While results such as \Cref{thm:vector_maintenance} state that $\mA$ must be an incidence matrix,
this is only because it uses \Cref{lem:large_entry_datastructure} internally,
which has already been covered.
Likewise, \Cref{lem:leverage_score_maintenance} is obtained via
a reduction from \cite{blss20} to the data structure \Cref{lem:large_entry_datastructure}.

\paragraph{Initial Points}

Finding initial points is even easier now, as we no longer have to add a star to the bipartite graph.
To construct a feasible $x'$, we again route $1/n$ flow on each edge of the bipartite graph,
and then set the bottom $n$ coordinates of $x'$ such that we satisfy $\mA'^\top x' = b$.
The slack of the dual is constructed in the same way as before in \Cref{lem:initial_point}
and \Cref{lem:initial_point_simple}.

\end{document}